\pdfoutput=1
\documentclass[11pt]{article}

\usepackage[a4paper,margin=1in]{geometry}
\usepackage{amsmath,amssymb,amsthm,bm,mathtools}
\usepackage{enumitem}
\usepackage{booktabs}
\usepackage{graphicx}
\usepackage{placeins}
\usepackage{float}
\usepackage{natbib}
\usepackage[colorlinks=true,linkcolor=blue,citecolor=blue,urlcolor=blue]{hyperref}
\hypersetup{
  pdftitle={Elliptical Regularized Hotelling Tests for High-Dimensional Change-Point Detection},
  pdfauthor={Fengyi Song, Mengtao Wen, and Long Feng},
  pdfsubject={High-dimensional change-point detection under elliptical distributions},
  pdfkeywords={change point, elliptical distribution, regularized Hotelling test, spatial median, spatial-sign covariance matrix}
}

\allowdisplaybreaks
\setlist[itemize]{leftmargin=2em}
\setlist[enumerate]{leftmargin=2em}

\newcommand{\R}{\mathbb{R}}
\newcommand{\Pbb}{\mathbb{P}}
\newcommand{\E}{\mathbb{E}}
\newcommand{\Var}{\operatorname{Var}}
\newcommand{\Cov}{\operatorname{Cov}}
\newcommand{\tr}{\operatorname{tr}}
\newcommand{\diag}{\operatorname{diag}}
\newcommand{\argmax}{\operatorname*{arg\,max}}
\newcommand{\argmin}{\operatorname*{arg\,min}}
\newcommand{\ind}{\mathbf{1}}
\newcommand{\bfI}{\mathbf{I}}
\newcommand{\bfR}{\mathbf{R}}
\newcommand{\bfS}{\mathbf{S}}
\newcommand{\bfT}{\mathbf{T}}
\newcommand{\bfa}{\mathbf{a}}
\newcommand{\bfA}{\mathbf{A}}
\newcommand{\bfB}{\mathbf{B}}
\newcommand{\bfb}{\mathbf{b}}
\newcommand{\bfd}{\mathbf{d}}
\newcommand{\bfc}{\mathbf{c}}
\newcommand{\bfe}{\mathbf{e}}
\newcommand{\bfdelta}{\bm{\delta}}
\newcommand{\bfp}{\mathbf{p}}
\newcommand{\bfr}{\mathbf{r}}
\newcommand{\bfC}{\mathbf{C}}
\newcommand{\bfD}{\mathbf{D}}
\newcommand{\bfE}{\mathbf{E}}
\newcommand{\bfF}{\mathbf{F}}

\newcommand{\bfK}{\mathbf{K}}
\newcommand{\bfL}{\mathbf{L}}
\newcommand{\bfM}{\mathbf{M}}
\newcommand{\bfJ}{\mathbf{J}}
\newcommand{\bfG}{\mathbf{G}}
\newcommand{\bfH}{\mathbf{H}}
\newcommand{\bfQ}{\mathbf{Q}}
\newcommand{\bfU}{\mathbf{U}}
\newcommand{\bfV}{\mathbf{V}}
\newcommand{\bfX}{\mathbf{X}}
\newcommand{\bfx}{\mathbf{x}}
\newcommand{\bfu}{\mathbf{u}}
\newcommand{\bfv}{\mathbf{v}}
\newcommand{\bfY}{\mathbf{Y}}
\newcommand{\bfy}{\mathbf{y}}
\newcommand{\bfz}{\mathbf{z}}
\newcommand{\bfZ}{\mathbf{Z}}
\newcommand{\bfW}{\mathbf{W}}
\newcommand{\bfw}{\mathbf{w}}
\newcommand{\bfone}{\mathbf{1}}
\newcommand{\bftheta}{\bm{\theta}}
\newcommand{\bfDelta}{\bm{\Delta}}
\newcommand{\bfvarepsilon}{\bm{\varepsilon}}
\newcommand{\bfOmega}{\bm{\Omega}}
\newcommand{\bfO}{\mathbf{O}}
\newcommand{\bfP}{\mathbf{P}}
\newcommand{\bfLambda}{\bm{\Lambda}}
\newcommand{\bfbeta}{\bm{\beta}}
\newcommand{\bfGamma}{\bm{\Gamma}}
\newcommand{\bfvarsigma}{\bm{\varsigma}}
\newcommand{\bfupsilon}{\bm{\upsilon}}
\newcommand{\calS}{\mathcal{S}}
\newcommand{\calR}{\mathcal{R}}
\newcommand{\calF}{\mathcal{F}}
\newcommand{\calG}{\mathcal{G}}
\newcommand{\dto}{\overset{d}{\longrightarrow}}
\newcommand{\pto}{\overset{P}{\longrightarrow}}
\newcommand{\opnorm}[1]{\left\|#1\right\|_{\rm op}}
\newcommand{\norm}[1]{\left\|#1\right\|}
\newcommand{\abs}[1]{\left|#1\right|}
\newcommand{\floor}[1]{\lfloor #1\rfloor}

\newtheorem{theorem}{Theorem}[section]
\newtheorem{assumption}{Assumption}[section]
\newtheorem{proposition}{Proposition}[section]
\newtheorem{lemma}{Lemma}[section]
\newtheorem{corollary}{Corollary}[section]
\newtheorem{remark}{Remark}[section]
\newtheorem{convention}{Convention}[section]
\theoremstyle{definition}
\newtheorem{algorithm}{Algorithm}[section]
\theoremstyle{plain}

\title{\bfseries Elliptical Regularized Hotelling Tests for High-Dimensional Change-Point Detection}
\author{Fengyi Song, Mengtao Wen and Long Feng\\ School of Statistics and Data Science, Nankai University}
\date{}

\begin{document}
\maketitle

\begin{abstract}
We propose an elliptical regularized Hotelling (ERHT) procedure for detecting location changes in high-dimensional sequences with heavy-tailed, cross-sectionally dependent observations. ERHT contrasts spatial medians on adjacent segments using a ridge-regularized inverse of the pooled centered spatial-sign covariance matrix, thereby combining robustness to radial variation with dependence-aware weighting. We establish Gaussian-process limits for the single- and multiple-change scans and joint convergence over a finite set of regularization parameters. These results provide asymptotically exact calibration of a Cauchy-aggregated adaptive test through the joint Gaussian limit, together with guarantees for local power and single-change localization. We further embed the ERHT score in wild binary segmentation and prove consistency for estimating the number and locations of multiple changes. Simulations show that ERHT is generally well calibrated and delivers competitive power under heavy-tailed distributions, particularly when cross-sectional dependence is substantial. An analysis of the Fama--French 49 industry portfolios reveals persistent evidence of location instability and identifies four structural breaks.
\end{abstract}


\section{Introduction}

High-dimensional change-point analysis studies whether the location structure of a multivariate sequence remains stable over time.  It is now a standard component of modern data analysis in genomics, neuroimaging, finance, environmental monitoring and network monitoring, where the dimension may be comparable with, or far exceed, the available sample size.  In this regime, classical likelihood-ratio and CUSUM procedures face two related obstacles.  First, covariance estimation is unstable without additional regularization.  Secondly, the dependence among coordinates can strongly affect both calibration and power.  These issues have led to several lines of high-dimensional mean change-point methodology.

Early and foundational work developed simultaneous or uniform testing procedures for many coordinates.  
\citet{zhang2010simultaneous} studied simultaneous changepoints in multiple sequences. \citet{jirak2015} developed uniform high-dimensional change-point tests based on coordinatewise CUSUM processes, and \citet{wen2024activation} proposed an FDR-controlling procedure for discovering coordinates that contain change-points.  
A second line focuses on localization and multiple-change segmentation.  Contributions in this direction include binary segmentation \citep{cho2015sbs}, projection-based estimation \citep{wangsamworth2018}, detection-boundary theory \citep{enikeeva2019}, multiple-change detection for high-dimensional sequences \citep{wang2019multiple}, breakpoint inference for dependent high-dimensional time series \citep{chen2022breakpoints}, dating the break in high-dimensional data \citep{wangshao2023dating}, and two-way moving-sum inference \citep{li2024twoway}.  A third line addresses data-adaptive or distributionally robust calibration through resampling, self-normalization or aggregation; representative examples include \citet{zhang2018self}, \citet{liu2020unified}, \citet{yu2021finite}, \citet{wang2022self} and \citet{zhang2022adaptive}.  Recent developments in this broad area are reviewed by \citet{liu2022review}.

A particularly relevant contribution is \citet{wangfeng2023}.  They showed that max-type and sum-type high-dimensional change-point statistics are asymptotically independent under light-tailed observations and weak cross-sectional dependence, and they used this property to construct a computationally efficient data-adaptive test.  This result clarifies why combining coordinatewise and quadratic summaries can be effective when the dependence is sufficiently weak.  However, the method is still mean-based and its theory does not target strongly correlated elliptical observations.  Heavy-tailed radial components can destabilize sample means and covariance-type quantities, while strong cross-sectional dependence can make unregularized or weak-dependence calibrations inefficient.

Robust high-dimensional change-point inference replaces moment-based summaries by spatial ranks, spatial signs, or self-normalized statistics.  \citet{shu2022spatialrank} proposed a spatial-rank method based on random integration, and \citet{jiang2023robust} developed robust high-dimensional inference using spatial signs and self-normalization.  More closely related to this paper, \citet{liu2025sscpd} introduced SSCPD, a spatial-sign-based high-dimensional change-point procedure that uses spatial medians and spatial signs to reduce sensitivity to heavy tails.  SSCPD is attractive for elliptically distributed data, but its theoretical guarantees are formulated for weak coordinate dependence and its normalization does not exploit a regularized inverse of a shape matrix.  Consequently, it may lose efficiency when the change is aligned with major dependence directions.

Regularized Hotelling methodology provides a different way to use dependence information in high dimension.  Instead of inverting an unstable sample covariance matrix, ridge-regularized Hotelling procedures use a ridge-regularized inverse to obtain a stable covariance-adjusted quadratic contrast.  This idea has been developed for high-dimensional two-sample testing by \citet{chen2011rht} and \citet{li2020adaptable}.  For change-point testing, \citet{li2026change} proposed covariance-based regularized Hotelling's \(T^2\) (RHT) scan statistics under light-tailed assumptions, and \citet{zhao2026rhtcp} studied Cauchy aggregation over finite grids of ridge-regularization parameters.  These procedures are powerful when the covariance is informative, but their covariance-based construction is vulnerable to heavy-tailed radial variation.  This observation motivates a robust RHT-type change-point procedure for elliptically symmetric data.

\subsection{Our contribution}

We propose ERHT, an elliptical regularized Hotelling \(T^2\) method for high-dimensional location change-point testing.  For each candidate scan interval, ERHT compares the spatial medians of two adjacent segments and standardizes the contrast by a ridge-regularized inverse of the pooled centered spatial-sign covariance matrix.  The proposed method combines two complementary ingredients: spatial medians and spatial signs remove the effect of radial magnitudes, while ridge regularization retains information about the dependence structure.
This construction is designed for settings where heavy-tailed elliptical distributions and non-negligible cross-sectional dependence occur simultaneously.
Concretely, we summarize our contributions in the following three points:

\begin{itemize}
  \item[1.] \textbf{Asymptotic theory for the ERHT statistic.}
  The analysis for the proposed statistic is nontrivial because the spatial median is nonlinear, the spatial-sign covariance matrix is computed after pooled recentering, and the ridge-regularized inverse couples all coordinates through a random high-dimensional matrix. 
  We develop a unified asymptotic theory for the ERHT statistic that accommodates both a fixed ridge-regularization parameter and a finite grid of ridge-regularization parameters, thereby providing a common foundation for parameter-specific and adaptively aggregated inference.
  \item[2.] \textbf{Adaptive aggregation over regularization parameters.} We propose a regularization-adaptive test that combines scan evidence across a finite grid of ridge-regularization parameters using the Cauchy rule. The joint limiting theory accounts for dependence across parameter values and supports fixed-level calibration of the aggregated statistic.
  \item[3.] \textbf{Global testing and segmentation under multiple changes.} 
  We extend ERHT to the multiple-change setting. 
  Additionally, motivated by wild binary segmentation \citep[WBS,][]{fryzlewicz2014}, we propose an ERHT-based WBS method for estimating the number and locations of changes. 
  The accompanying theory establishes the null validity of the global test and the consistency of the estimated number and locations of the changes.
\end{itemize}

\subsection{Organization}

The rest of the paper is organized as follows.  Section~\ref{sec:methods}
introduces the model, the ERHT statistic, and the calibration schemes.
Section~\ref{sec:theory-global} develops the deterministic equivalents and
the single- and multiple-change testing theory.
Section~\ref{sec:estimation} gives the WBS-based
multiple-change estimator and its consistency theorem.
Sections~\ref{sec:simulation} and~\ref{sec:realdata} report the results of the simulation studies
and real-data analysis.  
Technical lemmas and proofs are given in the
Supplementary Material.

\section{Methodology}
\label{sec:methods}

\subsection{Model, scan intervals, and raw statistic}\label{sec:model}

Let $\bfX_1,\ldots,\bfX_n\in\R^p$ be independent \(p\)-dimensional observations following the elliptical model \citep{fang1990}
\begin{equation}
  \label{eq:elliptical}
  \bfX_i=\bftheta_i+ \bfvarepsilon_i,\qquad \bfvarepsilon_i = \sqrt p\,R_i\bfOmega_p^{1/2}\frac{\bfG_i}{\norm{\bfG_i}},
\end{equation}
Here, \(\bftheta_i\in\R^p\) is the location vector of interest; \(\bfG_i\sim N_p(\mathbf 0,\bfI_p)\), so that \(\bfG_i/\|\bfG_i\|\) is the noise direction; \(\bfOmega_p\) is a positive-definite shape matrix normalized by \(p^{-1}\tr(\bfOmega_p)=1\); and \(R_i>0\) is a radial variable independent of \(\bfG_i\).
For each \((n,p)\), the pairs \((R_i,\bfG_i)\) are assumed to be independent and identically distributed across \(i\), although their common distribution may depend on \(p\).
We consider testing the equality of the location vectors, that is,
\begin{equation*}
  H_0:\quad \bftheta_1=\cdots=\bftheta_n=\bftheta.
\end{equation*}
We consider two alternatives corresponding to the single- and multiple-change problems.  
For the single-change scenario, we consider
\begin{equation*}
  H_{1,{\rm sc}}:\quad
\bftheta_i=
  \begin{cases}
    \bftheta^{(1)}, & i\le \floor{n\tau_\star},\\
    \bftheta^{(2)}, & i>\floor{n\tau_\star},
  \end{cases}
\end{equation*}
where \(\bfDelta_p=\bftheta^{(2)}-\bftheta^{(1)}\ne\mathbf{0}\) is the change vector, and $\tau_\star\in[\varepsilon,1-\varepsilon]$ for a fixed $\varepsilon\in(0,1/2)$ is the change-point fraction.  
When testing the existence of multiple changes, we consider
\begin{equation*}
  H_{1,{\rm mc}}:\quad
  \bftheta_i=\bftheta^{(\ell)}, \qquad 
  \floor{n\tau_{\ell-1}}<i\le \floor{n\tau_\ell},\qquad 
  \ell=1,\ldots,q+1,
\end{equation*}
where $q\ge1$ is the number of change-points and $0=\tau_0<\tau_1<\cdots<\tau_q<\tau_{q+1}=1$ are the change-point fractions. 
Adjacent location levels are assumed to differ, i.e., \(\bftheta^{(\ell)}\neq \bftheta^{(\ell+1)}\) for \(\ell = 1, \ldots, q\).  

To consider these two types of alternatives in a unified framework, following \citet{li2026change} and \citet{zhao2026rhtcp},
we introduce an ordered triple
\[
s = (t_1, t_2, t_3), \qquad 0\le t_1<t_2<t_3\le 1,
\]
representing the starts and ends of two adjacent intervals \(I_1(s)\) and \(I_2(s)\), where $\iota_n(t)=\floor{nt}+1$ for $t\in[0,1]$, and
\begin{align*}
  I_1(s)=\{\iota_n(t_1),\ldots,\iota_n(t_2)-1\},\qquad 
  I_2(s)=\{\iota_n(t_2),\ldots,\iota_n(t_3)-1\},
\end{align*}
are the adjacent intervals with sizes \(n_1(s) = |I_1(s)|\) and \(n_2(s) = |I_2(s)|\).
This representation permits different scan domains for the two alternatives, with the global statistic obtained by taking the supremum of the local two-sample statistic over the relevant candidate triples.  
For the single-change problem, define
\begin{equation*}
  \calS_{\rm sc}(\varepsilon)=\{(0,t,1):\varepsilon\le t\le1-\varepsilon\},
\end{equation*}
whereas the multiple-change problem uses the full adjacent-triple domain
\begin{equation*}
  \calS_{\rm mc}(\varepsilon)=\{(t_1,t_2,t_3)\in[0,1]^3:
  t_1<t_2<t_3,
  \ t_2-t_1\ge\varepsilon,
  \ t_3-t_2\ge\varepsilon\}.
\end{equation*}
The construction of the two-sample statistic on the adjacent intervals indexed by \(s\) parallels that of an RHT statistic. 
Concretely, we define the difference of sample spatial medians on \(I_1(s)\) and \(I_2(s)\) as 
\[
\widehat\bfDelta_s=\widehat\bftheta_{2,s}-\widehat\bftheta_{1,s}, \quad \text{with} \quad \widehat\bftheta_{a,s}=\widehat\bftheta(I_a(s)) \quad \text{for}\quad a = 1, 2,
\]
where \(\widehat\bftheta(I)\in\argmin_{\bfu\in\R^p}\sum_{i\in I}\norm{\bfX_i-\bfu}\) is the sample spatial median for any finite set \(I\subseteq \{1, \ldots, n\}\); if the minimizer is not unique, a fixed deterministic tie-breaking rule is used.
Let \(J(s)\subseteq \{1, \ldots, n\}\) denote the indices used to estimate the scatter normalizer, with \(m_s = |J(s)|\).  We use the ridge-regularized inverse of the pooled centered spatial-sign covariance matrix \citep{li2020adaptable},
\[
\widehat\bfQ_{\rho, s}=(\widehat\bfR_s+\rho\bfI_p)^{-1}, \quad \text{with} \quad \widehat\bfR_s=\frac1{m_s}\sum_{i\in J(s)} \widehat\bfY_{i, s}\widehat\bfY_{i, s}^{\top} \quad \text{and}\quad \widehat\bfY_{i, s}=\sqrt p\,U(\bfX_i-\widehat\bftheta_{0, s}),
\]
where \(\rho\) is the ridge-regularization parameter, \(U(\bfx)=\bfx/\norm{\bfx}\) with the convention \(U(\bfx)=\mathbf0\) for \(\bfx=\mathbf0\), and \(\widehat\bftheta_{0, s}=\widehat\bftheta(J(s))\) is the spatial median on \(J(s)\).
For both global testing problems, we use the common full-sample pool
\begin{equation}\label{eq:Js}
J(s) = \{1, \ldots, n\},
\end{equation}
for every candidate in \(\calS_{\rm sc}\) or \(\calS_{\rm mc}\), and retain the subscript \(s\) for notational convenience.
The raw local ERHT statistic is
\begin{equation}
\label{eq:Vraw}
  V^{\rm raw}_{\rho}(s)=N_s\widehat\bfDelta_s^{\top}\widehat\bfQ_{\rho, s}\widehat\bfDelta_s,
\end{equation}
where \(N_s={n_1(s)n_2(s)}/\{n_1(s)+n_2(s)\}\) is the effective sample size.
Because the high-dimensional quadratic form has a nonzero null mean,
\(V_{\rho}^{\mathrm{raw}}(s)\) cannot be compared directly across candidate
triples.  We therefore center and studentize the raw statistic before
constructing the scan statistic.

\begin{remark}
  The multiple-change scan set \(\calS_{\rm mc}(\varepsilon)\) is a continuous three-dimensional domain. For computational efficiency, one may use the grid
  \begin{equation}
  \label{eq:Smc-star}
    \calG_\varepsilon=\{j\varepsilon:j\in\mathbb Z\}\cap[0,1],\qquad
    \calS_{\rm mc}^*(\varepsilon)=\calS_{\rm mc}(\varepsilon)\cap\calG_\varepsilon^3.
  \end{equation}
  This discretization reduces computation by restricting the scan to finitely many candidate triples while leaving each local statistic unchanged.
\end{remark}

\subsection{Local studentization and scan statistic}\label{sec:calibration}

The null distribution of the raw ERHT statistic \eqref{eq:Vraw} is difficult to derive directly because of the nonlinearity of the sample spatial medians, the dependence between the spatial-median contrast and the estimated scatter matrix, and the high-dimensional regime.
We therefore first introduce the score representation of the spatial-median contrast.
Specifically, define the inverse-distance averages
\begin{equation*}
  \widehat e_{a,s}=\frac1{n_a(s)}\sum_{i\in I_a(s)}
  \widehat w_{i,a,s},
  \qquad a=1,2,
\end{equation*}
where the inverse-distance weight \(\widehat w_{i,a,s} = \sqrt{p}/\|\bfX_i-\widehat\bftheta_{a,s}\|\) for \(\bfX_i\ne\widehat\bftheta_{a,s}\) and \(\widehat w_{i,a,s} = 0\) for \(\bfX_i = \widehat\bftheta_{a,s}\). 
The inversion factors \(\widehat e_{a,s}^{-1}\) account for the Jacobian appearing in the first-order expansion of the two segment spatial medians.
Then, the score-CUSUM weights are 
\begin{equation}
\label{eq:beta-hat}
  \widehat\beta_i(s)=\sqrt{N_s}
  \left\{
  \frac{\ind(i\in I_2(s))}{n_2(s)\widehat e_{2,s}}
  -\frac{\ind(i\in I_1(s))}{n_1(s)\widehat e_{1,s}}
  \right\}.
\end{equation}
These weights are the spatial-median analogues of the ordinary adjacent-segment CUSUM weights. 
Let \(\widehat\bfbeta_s=(\widehat\beta_i(s):i\in J(s))^\top\) be the weight vector, and $\widehat\bfY_s$ be the $p\times m_s$ matrix whose columns are $\widehat\bfY_{i,s}$ in their time order. 
The spatial-median linearization motivates the following score approximation,
\[
\sqrt{N_s}\,\widehat\bfDelta_s\approx\sum_{i\in J(s)}\widehat\beta_i(s)\widehat \bfY_{i,s} = \widehat \bfY_{s} \widehat \bfbeta_s.
\]
Thus, the nonlinear contrast between the two spatial medians of the adjacent intervals is approximated by a
weighted linear combination of pooled spatial signs. 

Next, we introduce the companion matrix to reformulate the approximated quadratic form of the raw local statistic \(V^{\rm raw}_{\rho}(s)\).
Define the companion matrix by
\begin{equation*}
  \widehat\bfA_{\rho,s}=\frac1{m_s}\widehat\bfY_s^{\top}\widehat\bfQ_{\rho,s}\widehat\bfY_s,
\end{equation*}
whose $(i,j)$th entry is
\(\widehat A_{ij,\rho,s}
=
m_s^{-1}
\widehat \bfY_{i,s}^{\top}
\widehat \bfQ_{\rho,s}
\widehat \bfY_{j,s}\).
This gives
\[
V^{\rm raw}_{\rho}(s) \approx \widehat \bfbeta_s^\top \widehat \bfY_s^\top \widehat \bfQ_{\rho, s} \widehat \bfY_s \widehat \bfbeta_s = m_s \widehat\bfbeta_s^{\top}
\widehat \bfA_{\rho,s}
\widehat\bfbeta_s.
\]
This representation separates the two principal components of the raw local statistic:
the weight vector $\widehat\bfbeta_s$ captures the
temporal contrast, whereas the companion matrix
$\widehat \bfA_{\rho,s}$ captures the cross-sectional geometry of the
pooled spatial signs after weighting by the ridge-regularized inverse
\(\widehat\bfQ_{\rho,s}\).

Furthermore, motivated by the central symmetry of the elliptical error distribution, we use the Rademacher representation to derive the conditional expectation and conditional variance of the companion quadratic form. 
In particular, let \(\bfupsilon_s = (\upsilon_i:i\in J(s))^\top\) have independent Rademacher components satisfying
\(
\Pbb(\upsilon_i=1)
=
\Pbb(\upsilon_i=-1)
= 1/2\), independent of the
data. 
Define the companion quadratic form as
\begin{equation}
\label{eq:companion-quadratic-form}
\widehat V^{\rm comp}_{\rho}(s)
  =m_s\{\widehat\bfbeta_s\circ\bfupsilon_s\}^\top
    \widehat\bfA_{\rho,s}
    \{\widehat\bfbeta_s\circ\bfupsilon_s\},
\end{equation}
where $\circ$ denotes componentwise multiplication.
Expanding
\eqref{eq:companion-quadratic-form} gives
\begin{align*}
\widehat V_{\rho}^{\mathrm{comp}}(s)
&=
m_s
\sum_{i\in J(s)}
\widehat\beta_i(s)^2
\widehat A_{ii,\rho,s}
+
m_s
\sum_{\substack{i,j\in J(s)\\i\ne j}}
\widehat\beta_i(s)
\widehat\beta_j(s)
\upsilon_i\upsilon_j
\widehat A_{ij,\rho,s}.
\end{align*}
The diagonal component of \(\widehat V_{\rho}^{\mathrm{comp}}(s)\) is invariant
to the Rademacher signs and determines the conditional center, while the
off-diagonal component has conditional mean zero and determines the
stochastic variation.
Accordingly, define the centering quantity and the variance quantity as
\begin{equation*}
\widehat\kappa_{\rho}(s)
=
\sum_{i\in J(s)}
\widehat\beta_i(s)^2
\widehat A_{ii,\rho,s},\qquad
\widehat\sigma_{\rho}^2(s)
=
2m_s
\sum_{\substack{i,j\in J(s)\\i\ne j}}
\widehat\beta_i(s)^2
\widehat\beta_j(s)^2
\widehat A_{ij,\rho,s}^2,
\end{equation*}
respectively.
Consequently, we have 
\[
\E\{\widehat V_\rho^{\mathrm{comp}}(s) \mid \widehat \bfA_{\rho, s}, \widehat\bfbeta_s\} = m_s \widehat \kappa_\rho(s) \quad\text{and}\quad \Var\{\widehat V_\rho^{\mathrm{comp}}(s) \mid \widehat \bfA_{\rho, s}, \widehat\bfbeta_s\} = m_s \widehat \sigma_\rho^2(s),
\] 
which motivates the studentized local statistic
\begin{equation}
\label{eq:Z-local}
  Z_{\rho}(s)=
  \frac{V^{\rm raw}_{\rho}(s)-m_s\widehat\kappa_{\rho}(s)}
  {\{m_s\widehat\sigma_{\rho}^2(s)\}^{1/2}}.
\end{equation}

Finally, for the studentized local statistic, we take the supremum over a
chosen scan set
\(\calS\) to obtain the global statistic
\[
  T_{\rho}(\calS)=\sup_{s\in\calS}Z_{\rho}(s).
\]
For the single-change alternative, the scan set is \(\calS_{\rm sc}(\varepsilon)\), whereas for the multiple-change alternative it is \(\calS_{\rm mc}(\varepsilon)\) (or its discretized version \(\calS_{\rm mc}^*(\varepsilon)\)). 
Thus, the two global statistics are
\[
\begin{aligned}
T_{\rho}^{\rm sc}
  &:=T_{\rho}\{\calS_{\rm sc}(\varepsilon)\}
    =\sup_{s\in\calS_{\rm sc}(\varepsilon)}Z_{\rho}(s),\\
T_{\rho}^{\rm mc}
  &:=T_{\rho}\{\calS_{\rm mc}(\varepsilon)\}
    =\sup_{s\in\calS_{\rm mc}(\varepsilon)}Z_{\rho}(s),
\end{aligned}
\]
for testing the single-change alternative and the multiple-change alternative, respectively.

\subsection{Adaptive aggregation over ridge-regularization parameters}\label{sec:ridge-aggregation}

The ridge-regularization parameter controls the strength of shape adjustment: smaller
values exploit the estimated dependence structure more aggressively but may
amplify noise in estimated directions with small eigenvalues, whereas larger values improve
numerical stability at the cost of shrinking the procedure toward isotropic
weighting.  
Thus, the choice of regularization parameter balances exploitation of the estimated dependence structure against numerical stability.

For a fixed integer \(K\ge1\), let
\[
  \calR_K=\{\rho_n^{(1)},\ldots,\rho_n^{(K)}\}
\]
be a finite grid of ridge-regularization parameters, with each
\(\rho_n^{(k)}\) deterministic.  We combine the scan evidence computed at
these parameter values.
Let \(F_{\calS}\) denote the limiting Gaussian-supremum distribution of \(T_\rho(\calS)\) for the scan set \(\calS\). 
We define
\begin{equation}
\label{eq:ridge-marginal-pvalues}
  P_k=1-F_{\calS}\{T_{\rho_n^{(k)}}(\calS)\},
  \qquad k=1,\ldots,K.
\end{equation}
Under the common-pool convention \eqref{eq:Js}, the supremum law \(F_\calS\) depends on the scan domain but not on the regularization-parameter value because all marginal Gaussian processes have the same covariance kernel.
Finally, the Cauchy aggregate and its analytic transformation are
\begin{equation}
\label{eq:CCT}
  T_{\rm CC}=\sum_{k=1}^K\varpi_k
  \tan\left[\pi\left\{\frac12-P_k\right\}\right],
  \qquad
  P_{\rm CC}=\frac12-\frac1\pi\arctan(T_{\rm CC}),
\end{equation}
where $\varpi_k>0$ are the combination weights satisfying $\sum_{k=1}^K\varpi_k=1$.
Two aggregate calibrations should be distinguished.  Exact asymptotic calibration rejects for large \(T_{\rm CC}\) using a quantile of its joint Gaussian limit, which accounts for dependence across regularization-parameter values.  The simpler analytic rule rejects when \(P_{\rm CC}\le\alpha\); Theorems~\ref{thm:sc-cauchy} and~\ref{thm:mc-cauchy} characterize its limiting rejection probability without assuming that \(P_{\rm CC}\) is exactly uniform.  In the empirical sections, ``Gaussian-supremum calibration'' refers to using \(F_{\calS}\) for each \(P_k\), whereas ``time-permutation calibration'' refers to replacing the \(P_k\)'s by parameter-specific permutation p-values before applying the same analytic Cauchy transformation.

\section{Asymptotic Theory}
\label{sec:theory-global}

\subsection{Primitive assumptions}\label{sec:assumptions}

This section develops the asymptotic theory for global testing.  
We begin with the primitive assumptions used throughout.  
Define the inverse radial variables by \(\xi_i=R_i^{-1}\), \(i=1,\ldots,n\).
Let \(\lambda_{1,p},\ldots,\lambda_{p,p}\) and
\(\bfp_{1,p},\ldots,\bfp_{p,p}\) be the eigenvalues and eigenvectors of
\(\bfOmega_p\), respectively.
We impose the following assumptions.

\begin{assumption}[Elliptical error distribution]\label{ass:elliptical}
Assume that the elliptical model \eqref{eq:elliptical} holds for \(i=1,\ldots,n\).
For some \(\eta>0\) and \(c_0>0\),
\(\E\xi_i\to\zeta_{-1}\in(0,\infty)\), \(\limsup_{n\to\infty}\E\xi_i^{4+\eta}<\infty\), and
\(\Pbb\left\{\max_{1\le i\le n}\xi_i>(\log n)^{c_0}\right\}\to0\).
\end{assumption}

\begin{assumption}[Dimension and spectrum]\label{ass:spectrum}
(a) As \(n,p\to\infty\), assume \(\gamma_n={p}/{n}\to\gamma\in(0,\infty)\).
(b) There exists a constant \(0<\omega_+<\infty\) such that \(0<\lambda_{j,p}\le \omega_+\)
for all \(j=1,\ldots,p\) and all \(p\).   
(c) The empirical spectral distribution
\(H_p(x)=p^{-1}\sum_{j=1}^p\ind(\lambda_{j,p}\le x)\)
converges weakly to a probability distribution \(H\) supported on \([0,\omega_+]\).  
\end{assumption}

\begin{assumption}[Grid of ridge-regularization parameters]\label{ass:grid}
There exist constants \(0<\rho_0<\rho_1<\infty\) such that, for fixed \(K\ge1\),
\[
  \calR_K=\{\rho_n^{(1)},\ldots,\rho_n^{(K)}\}
  \subseteq[\rho_0,\rho_1].
\]
For some \(\rho^{(1)},\ldots,\rho^{(K)}\in[\rho_0,\rho_1]\),
\(\max_{1\le k\le K}|\rho_n^{(k)}-\rho^{(k)}|\to0\).
\end{assumption}

Assumption~\ref{ass:elliptical} specifies the radial regularity needed beyond the elliptical representation. 
The convergence of \(\E\xi_i\) stabilizes the Jacobian factor in the spatial-median expansion.  The \(4+\eta\) moment bound controls its higher-order terms and the quadratic-form remainders, while the
polylogarithmic maximum bound provides the uniform control needed over all candidate scan intervals.  
Assumption~\ref{ass:spectrum} places the problem in the proportional-growth regime.  
The uniform upper spectral bound rules out directions whose scale diverges with the dimension, and the weak convergence of \(H_p\) stabilizes the normalized resolvent traces that determine the null variance and local power. 
Assumption~\ref{ass:grid} keeps the deterministic grid of ridge-regularization parameters finite and inside a compact interval bounded away from zero.  

We also summarize the design of different scan sets for different inferential targets.

\begin{convention}[Scan designs and inferential targets]\label{conv:scan}
The scan set is chosen according to the change-point problem under study.
\begin{enumerate}[label=(\roman*)]
\item \emph{Single-change scan}: 
\begin{equation*}
  \calS=\calS_{\rm sc}(\varepsilon),\qquad
  \varepsilon\in(0,1/2),\qquad
  J(s)=\{1,\ldots,n\}\quad(s\in\calS).
\end{equation*}
\item \emph{Multiple-change scan}: 
\begin{equation*}
  \calS=\calS_{\rm mc}(\varepsilon),\qquad
  \varepsilon\in(0,1/2),\qquad
  J(s)=\{1,\ldots,n\}\quad(s\in\calS).
\end{equation*}
\end{enumerate}
The two cases correspond to two different inferential targets. 
The single-change problem scans the one-dimensional path \((0,t,1)\), whereas the multiple-change testing problem scans the full adjacent-triple domain \((t_1,t_2,t_3)\).  
\end{convention}

\subsection{Core approximation and pointwise null law}\label{sec:core-null}

In this section, we present fundamental results used in both the single-change and
multiple-change tests.  
We first introduce notation used in their statements.
Under \(H_0\), define the oracle-scaled spatial
sign and inverse-distance weight by
\(\bfY_i=\sqrt p\,U(\bfvarepsilon_i)\)
and 
\(w_i={\sqrt p}/{\norm{\bfvarepsilon_i}}\), respectively.
Accordingly, we define \(\bfY_s\) as the
\(p\times m_s\) matrix whose columns are \(\bfY_i\) for \(i\in J(s)\) in time
order, and put
\[
  \bfR_s^0=\frac{1}{m_s}\bfY_s\bfY_s^\top,
  \qquad
  \bfQ_{\rho,s}^0=(\bfR_s^0+\rho\bfI_p)^{-1},
  \qquad
  \bfA_{\rho,s}^0=\frac{1}{m_s}\bfY_s^\top
  \bfQ_{\rho,s}^0\bfY_s.
\]
For the score weights, let \(e_{a,s}=n_a(s)^{-1}\sum_{i\in I_a(s)}w_i\) for \(a = 1, 2\), and set \(\bfbeta_s=(\beta_i(s):i\in J(s))^\top\) with elements
\[
  \beta_i(s)=\sqrt{N_s}
  \left\{
  \frac{\ind(i\in I_2(s))}{n_2(s)e_{2,s}}
  -\frac{\ind(i\in I_1(s))}{n_1(s)e_{1,s}}
  \right\}.
\]
Finally, define the oracle score
\(\bfB_s=\sum_{i\in J(s)}\beta_i(s)\bfY_i\) and the oracle quadratic form
\[
  \widetilde V_\rho^0(s)
  =\bfB_s^\top\bfQ_{\rho,s}^0\bfB_s
  =m_s\bfbeta_s^\top\bfA_{\rho,s}^0\bfbeta_s.
\]
Writing \(A_{ij,\rho,s}^0\) for the entries of
\(\bfA_{\rho,s}^0\), define the oracle center and variance by
\[
  \kappa_\rho^0(s)
  =\sum_{i\in J(s)}\beta_i(s)^2A_{ii,\rho,s}^0,
  \qquad
  \sigma_\rho^{0,2}(s)
  =2m_s\sum_{\substack{i,j\in J(s)\\i\neq j}}
  \beta_i(s)^2\beta_j(s)^2(A_{ij,\rho,s}^0)^2.
\]
Let \(m_*=\inf_{s\in\calS}m_s\), and let
\(\ell_n=(\log n)^{c_\ell}\) for some constant \(c_\ell > 0\) denote the fixed, sufficiently large polylogarithmic envelope used in Section~\ref{app:auxiliary} of the Supplementary Material.  

\begin{proposition}[Raw-to-score reduction]\label{prop:raw-score}
Under \(H_0\), Assumptions~\ref{ass:elliptical}--\ref{ass:grid}, and Convention~\ref{conv:scan},
\begin{equation*}
  \sup_{s\in\calS}\sup_{\rho\in[\rho_0,\rho_1]}
  \abs{V^{\rm raw}_{\rho}(s)-\widetilde V_{\rho}^0(s)}
  =O_P(\ell_n).
\end{equation*}
Since \(m_*=n\) under Convention~\ref{conv:scan}, it follows after null standardization that
\[
  \sup_{s\in\calS}\sup_{\rho\in[\rho_0,\rho_1]}
  \frac{\abs{V^{\rm raw}_{\rho}(s)-\widetilde V_{\rho}^0(s)}}{\sqrt{m_s}}
  =O_P(\ell_nm_*^{-1/2})=o_P(1).
\]
\end{proposition}

Proposition~\ref{prop:raw-score} shows that the raw local statistic is asymptotically equivalent to the oracle quadratic form \(\widetilde V_{\rho}^0(s)\), uniformly over the candidate scan intervals and over \(\rho\in[\rho_0,\rho_1]\). 
Building on this proposition, we derive the weak convergence of the studentized local statistic \(Z_\rho(s)\) in the following theorem.

\begin{theorem}[Pointwise null law at a fixed ridge-regularization parameter]
  \label{thm:pointwise}
Suppose \(H_0\) and Assumptions~\ref{ass:elliptical}--\ref{ass:grid} hold.  Let
\(s=s_n\) be any deterministic candidate sequence satisfying
\(I_1(s)\cup I_2(s)\subseteq J(s)\), \(m_s\asymp n\),
and, for some fixed \(\varepsilon_s>0\),
\(\min\{n_1(s),n_2(s)\}\ge \varepsilon_s m_s\).
Then, for every fixed \(\rho\in[\rho_0,\rho_1]\),
\begin{equation}
\label{eq:pointwise-null}
  Z_{\rho}(s)
  =\frac{V^{\rm raw}_{\rho}(s)-m_s\widehat\kappa_{\rho}(s)}
  {\{m_s\widehat\sigma_{\rho}^2(s)\}^{1/2}}
  \dto N(0,1).
\end{equation}
\end{theorem}

\subsection{Single-change inference}\label{sec:single-theory}

\subsubsection{Null limits and aggregation over regularization parameters}\label{sec:theory}

For \(s=(t_1,t_2,t_3)\), define the temporal contrast
\(\varphi_s(x)=(t_3 - t_2)^{-1}{\ind(x\in[t_2, t_3))} - (t_2 - t_1)^{-1}{\ind(x\in[t_1, t_2))}\).
For two triples \(s\) and \(r\), we put
\(\psi(s,r)=\int_0^1\varphi_s(x)\varphi_r(x)\,dx\),
\(\psi(s)=\psi(s,s)\), and 
\(K_0(s,r)={\psi(s,r)^2}/\{\psi(s)\psi(r)\}\).
The kernel \(K_0\) records the overlap between two temporal contrasts.  
The following theorem characterizes the weak convergence of the supremum of the studentized local statistics over the single-change scan set \(\calS_{\rm sc}(\varepsilon)\).

\begin{theorem}[Single-change scan limit at a fixed ridge-regularization parameter]\label{thm:sc-process}
Suppose \(H_0\), Assumptions~\ref{ass:elliptical}--\ref{ass:grid}, and Convention~\ref{conv:scan}(i) hold.  For every fixed \(\rho\in[\rho_0,\rho_1]\),
\begin{equation}
\label{eq:sc-process-limit}
  \{Z_\rho(s):s\in\calS_{\rm sc}(\varepsilon)\}
  \dto \{G_\rho^{\rm sc}(s):s\in\calS_{\rm sc}(\varepsilon)\}
  \quad\text{in }\ell^\infty\{\calS_{\rm sc}(\varepsilon)\},
\end{equation}
where \(G_\rho^{\rm sc}\) is a centered Gaussian process with almost surely continuous sample paths and covariance
\(\Cov\{G_\rho^{\rm sc}(s),G_\rho^{\rm sc}(r)\}=K_0(s,r)\).
Consequently,
\begin{equation}
\label{eq:sc-T-limit}
  T_\rho^{\rm sc}=\sup_{s\in\calS_{\rm sc}(\varepsilon)}Z_\rho(s)
  \dto \sup_{s\in\calS_{\rm sc}(\varepsilon)}G_\rho^{\rm sc}(s) := T_0^{\rm sc}.
\end{equation}
\end{theorem}

Let \(F_{\rm sc}\) denote the distribution of \(T_0^{\rm sc}\).  For each regularization-parameter sequence \(\rho_n^{(k)}\), define the Gaussian-limit-calibrated p-value and its limiting counterpart by
\begin{align*}
  P_k^{\rm sc}
  =1-F_{\rm sc}(T^{\rm sc}_{\rho_n^{(k)}}),\qquad P_{\infty,k}^{\rm sc}
  =1-F_{\rm sc}\left\{
    \sup_{s\in\calS_{\rm sc}(\varepsilon)}
    G_{\rho^{(k)}}^{\rm sc}(s)\right\}.
\end{align*}
The following theorem shows the joint limit of the derived p-values \((P_1^{\rm sc},\ldots,P_K^{\rm sc})^\top\).

\begin{theorem}[Joint single-change scan limit over the finite parameter grid \(\calR_K\)]\label{thm:sc-joint}
Under the conditions of Theorem~\ref{thm:sc-process}, for the finite grid of ridge-regularization parameters \(\calR_K=\{\rho_n^{(1)},\ldots,\rho_n^{(K)}\}\), we have
\begin{equation}
\label{eq:sc-joint-process}
  \{Z_{\rho_n^{(k)}}(s):s\in\calS_{\rm sc}(\varepsilon),\ k=1,\ldots,K\}
  \dto
  \{G_{\rho^{(k)}}^{\rm sc}(s):s\in\calS_{\rm sc}(\varepsilon),\ k=1,\ldots,K\},
\end{equation}
where the right-hand side is a jointly centered Gaussian process with covariance
\begin{equation}
\label{eq:sc-joint-cov}
\Cov\{G_{\rho}^{\rm sc}(s),G_{\rho'}^{\rm sc}(r)\}
  =r_E(\rho, \rho')K_0(s,r).
\end{equation}
Here, \(r_E(\rho,\rho')\) is the full-pool correlation across regularization parameters, defined in Lemma~\ref{lem:cross-factorization} of the Supplementary Material.
Furthermore,
\begin{equation}
\label{eq:sc-pvector-limit}
  (P_1^{\rm sc},\ldots,P_K^{\rm sc})^\top
  \dto
  (P_{\infty,1}^{\rm sc},\ldots,P_{\infty,K}^{\rm sc})^\top.
\end{equation}
\end{theorem}

The joint convergence of \((P_1^{\rm sc},\ldots,P_K^{\rm sc})^\top\) yields the corresponding Cauchy-aggregated statistic \(T_{\rm CC}^{\rm sc}\) and analytic transformation \(P_{\rm CC}^{\rm sc}\) by substitution into \eqref{eq:CCT}. 
The next theorem establishes asymptotic size for calibration based on the exact joint-limit quantile and characterizes the rejection probability of the analytic Cauchy transformation.

\begin{theorem}[Single-change Cauchy aggregation over the finite parameter grid \(\calR_K\)]\label{thm:sc-cauchy}
Suppose that the conditions of Theorem~\ref{thm:sc-joint} hold.
Let 
\(T_{{\rm CC},\infty}^{\rm sc}=\sum_{k=1}^K\varpi_k
  \tan\big[\pi \{\frac12-P_{\infty,k}^{\rm sc}\}\big]\).  
If \(c_{\alpha,K}^{\rm sc}\) is a continuity point of the cumulative distribution function of \(T_{{\rm CC},\infty}^{\rm sc}\) satisfying \(\Pbb(T_{{\rm CC},\infty}^{\rm sc}>c_{\alpha,K}^{\rm sc})=\alpha\), then under \(H_0\), 
\begin{equation*}
  \Pbb(T_{\rm CC}^{\rm sc}>c_{\alpha,K}^{\rm sc})\to\alpha.
\end{equation*}
Moreover, for the analytic Cauchy transformation \(P_{\rm CC}^{\rm sc}\), at every \(\alpha\in(0,1)\) for which \(\cot(\pi\alpha)\) is a continuity point of the distribution of \(T_{{\rm CC},\infty}^{\rm sc}\), we have
\begin{equation*}
  \Pbb(P_{\rm CC}^{\rm sc}\le\alpha)
  \to
  \Pbb\{T_{{\rm CC},\infty}^{\rm sc}\ge\cot(\pi\alpha)\}.
\end{equation*}
\end{theorem}

\subsubsection{Local power and localization}\label{sec:theory:single:power}

To study power, we impose an additional assumption on the signal and consider separate local- and strong-signal regimes.

\begin{assumption}[Local and strong alternatives]\label{ass:alternatives}
(a) For local single-change power, the true change occurs at $\tau_\star\in[\varepsilon,1-\varepsilon]$ and \(\bfDelta_p=n^{-1/4}\bfd_p\) with \(\norm{\bfd_p}\le C\).
  The spectral signal measure
  \(G_p=\sum_{j=1}^p(\bfp_{j,p}^{\top}\bfd_p)^2\delta_{\lambda_{j,p}}\)
  converges weakly to a finite nonzero measure $G_\Delta$ on $[0,\omega_+]$.  
(b) For consistency under a strong single-change alternative, the true change occurs at \(\tau_\star\in[\varepsilon,1-\varepsilon]\), and the shift satisfies \(\ell_n\norm{\bfDelta_p}\to0\) and \(\sqrt n\norm{\bfDelta_p}^2\to\infty\).
\end{assumption}

These two regimes are used separately. Concretely, Assumption~\ref{ass:alternatives}(a) places the change at the \(n^{-1/4}\) local scale of the quadratic statistic and requires the distribution of the rescaled signal energy across the eigenspaces of \(\bfOmega_p\) to converge to a stable, nondegenerate spectral limit, while Assumption~\ref{ass:alternatives}(b) considers shrinking but detectable shifts: the condition \(\ell_n\norm{\bfDelta_p}\to0\) keeps the spatial-median and spatial-sign expansions within the null local-perturbation regime uniformly over the scan, whereas \(\sqrt n\norm{\bfDelta_p}^2\to\infty\) makes the standardized quadratic signal diverge.

Let $a_{\rho,n}$ be the positive solution of the negative-real-axis Mar\v{c}enko--Pastur equation \(a_{\rho,n}=1-\gamma_n+\gamma_n\rho m_{\rho,n}\), where \(m_{\rho,n}=\int\{a_{\rho,n}x+\rho\}^{-1}\,dH_p(x)\).
Define the deterministic equivalent of the ridge-regularized inverse
\begin{equation*}
  \bfD_{\rho,n}=(a_{\rho,n}\bfOmega_p+\rho\bfI_p)^{-1}.
\end{equation*}
Under \(G_p\Rightarrow G_\Delta\) in Assumption~\ref{ass:alternatives}(a),
Lemma~\ref{lem:ridge-stability} in the Supplementary Material gives
\begin{equation*}
  \mathfrak q_{\rho,n}(\bfd_p):=\bfd_p^{\top}\bfD_{\rho,n}\bfd_p
  \longrightarrow
  \mathfrak q_\rho(G_\Delta):=\int\frac{1}{a_\rho x+\rho}\,dG_\Delta(x),
\end{equation*}
where $a_\rho$ is the limit of $a_{\rho,n}$.
Let \(\sigma_\rho\) denote the positive square root of the full-pool variance limit, defined in Lemma~\ref{lem:oracle-variance-limit} of the Supplementary Material. 
For a single change and $s_t=(0,t,1)$, set
\begin{equation*}
  \vartheta(t;\tau_\star)=
  \begin{cases}
    (1-\tau_\star)/(1-t), & t<\tau_\star,\\[2mm]
    \tau_\star/t, & t\ge\tau_\star,
  \end{cases}
  \qquad
  \mathfrak h(t;\tau_\star)=t(1-t)\vartheta(t;\tau_\star)^2.
\end{equation*}
The following theorem gives the limiting power under the local alternative and consistency under the strong alternative.

\begin{theorem}[Single-change local power and consistency]\label{thm:power}
(a) Suppose Assumptions~\ref{ass:elliptical}--\ref{ass:grid} and~\ref{ass:alternatives}(a), together with Convention~\ref{conv:scan}(i), hold.  For a fixed \(\rho\in[\rho_0,\rho_1]\) and after piecewise-constant interpolation from the natural scan grid,
\begin{equation*}
  \{Z_\rho(s_t):t\in[\varepsilon,1-\varepsilon]\}
  \Rightarrow
  \{G_\rho^{\rm sc}(s_t)+\Lambda_\rho(G_\Delta)\mathfrak h(t;\tau_\star):
    t\in[\varepsilon,1-\varepsilon]\}
  \quad\text{in }\ell^\infty[\varepsilon,1-\varepsilon],
\end{equation*}
where \(\Lambda_\rho(G_\Delta)=\mathfrak q_\rho(G_\Delta)/\sigma_\rho\).
For the level-\(\alpha\) scan test at this regularization parameter, let \(c_{\alpha,{\rm sc}}\) be the \((1-\alpha)\)-quantile of \(\sup_{s\in\calS_{\rm sc}(\varepsilon)}G_\rho^{\rm sc}(s)\).
\begin{equation*}
  \lim_{n\to\infty}\Pbb_{H_1}(T_\rho^{\rm sc}>c_{\alpha,{\rm sc}})
  =\Pbb\left[\sup_{t\in[\varepsilon,1-\varepsilon]}\{G_\rho^{\rm sc}(s_t)+\Lambda_\rho(G_\Delta)\mathfrak h(t;\tau_\star)\}>c_{\alpha,{\rm sc}}\right].
\end{equation*}
(b) Suppose Assumptions~\ref{ass:elliptical}--\ref{ass:grid} and~\ref{ass:alternatives}(b), together with Convention~\ref{conv:scan}(i), hold.  For every fixed \(\rho\in[\rho_0,\rho_1]\),
\begin{equation*}
  T_\rho^{\rm sc}\pto\infty.
\end{equation*}
Consequently, for the finite grid in Assumption~\ref{ass:grid} with positive Cauchy weights,
\begin{equation*}
  P_{\rm CC}^{\rm sc}\pto0,
  \qquad
  \Pbb_{H_1}(P_{\rm CC}^{\rm sc}\le\alpha)\to1.
\end{equation*}
\end{theorem}

\begin{remark}[Comparison with covariance-based RHT]
Theorem~\ref{thm:power} characterizes the local power of ERHT, but it
does not by itself yield a high-dimensional asymptotic relative efficiency (ARE) comparison with a
covariance-based RHT\@. Under a nondegenerate radial distribution, the
sample covariance has an asymptotically separable form whose resolvent
depends on the full limiting law of \(R_i^2\), rather than only on
\(\E(R_i^2)\); see \citet{elkaroui2009}. Consequently, existing RHT
limits cannot be transferred to fixed-degree multivariate \(t\) or
nondegenerate scale-mixture observations by simply replacing
\(\bfOmega_p\) with \(\E(R_i^2)\bfOmega_p\). 
In Section~\ref{sec:simulation}, we compare these procedures numerically.
\end{remark}

For a single change, define the natural scan grid
\(\calG_n^{\rm sc}=\{k/n:\lceil n\varepsilon\rceil\le k\le\lfloor n(1-\varepsilon)\rfloor\}\),
and, for a fixed ridge-regularization parameter \(\rho\), define
\(\widehat\tau_\rho\in\argmax_{t\in\calG_n^{\rm sc}}Z_\rho(0,t,1)\), with ties broken deterministically.
We next establish its localization rate. 

\begin{theorem}[Single-change localization at a fixed ridge-regularization parameter]\label{thm:localization}
Under Assumptions~\ref{ass:elliptical}--\ref{ass:grid} and Convention~\ref{conv:scan}(i), suppose there is one change at \(\tau_\star\in[\varepsilon,1-\varepsilon]\) and the shift satisfies the uniform small-perturbation condition \(\ell_n\norm{\bfDelta_p}\to0\).  
Let \(\mathfrak s_{n,\rho}=(\sqrt n\,\bfDelta_p^{\top}\bfD_{\rho,n}\bfDelta_p)/\sigma_\rho\) and \(\mathfrak e_{\sigma,n}=\ell_n\norm{\bfDelta_p}+\norm{\bfDelta_p}^2+\ell_nn^{-1/2}\).
If $\mathfrak s_{n,\rho}\to\infty$, then
\begin{equation}
\label{eq:localization-consistency}
  |\widehat\tau_\rho-\tau_\star|
  =O_P\{\mathfrak s_{n,\rho}^{-1}+\mathfrak e_{\sigma,n}\}.
\end{equation}
Additionally, if $\mathfrak s_{n,\rho}\mathfrak e_{\sigma,n}=O(1)$, then \(|\widehat\tau_\rho-\tau_\star|=O_P(\mathfrak s_{n,\rho}^{-1})\).
\end{theorem}

The appearance of the term \(\mathfrak e_{\sigma,n}\) in Theorem~\ref{thm:localization} is because the variance estimator in the studentized statistic differs from its centered-error counterpart by
\(O_P(\mathfrak e_{\sigma,n})\) under a small but non-local shift. 
This perturbation is negligible for consistency and for power, but it can affect the sharper localization rate unless \(\mathfrak s_{n,\rho}\mathfrak e_{\sigma,n}=O(1)\).  For example, if \(\norm{\bfDelta_p}=n^{-b}\), then the conditions \(\mathfrak s_{n,\rho}\to\infty\) and \(\mathfrak s_{n,\rho}\mathfrak e_{\sigma,n}=O(1)\) are simultaneously satisfied for exponents in the nonempty range \(1/6<b<1/4\), up to logarithmic factors.

\subsection{Multiple-change global testing}\label{sec:mc}

For multiple-change testing, recall that the statistic is \(T_\rho^{\rm mc}=\sup_{s\in\calS_{\rm mc}(\varepsilon)}Z_\rho(s)\).
The common full-sample scatter pool \eqref{eq:Js} is retained, so for each fixed regularization parameter the covariance kernel remains \(K_0\); the difference from the single-change problem lies in the scan domain and hence in the limiting supremum distribution. 
As in Section~\ref{sec:theory}, let \(F_{\rm mc}\) denote the distribution of the limit of \(T_\rho^{\rm mc}\). 
Define the corresponding Gaussian-limit-calibrated p-values and their limiting counterparts by
\(P_k^{\rm mc}=1-F_{\rm mc}(T_{\rho_n^{(k)}}^{\rm mc})\) and
\(P_{\infty,k}^{\rm mc}=1-F_{\rm mc}\{\sup_{s\in\calS_{\rm mc}(\varepsilon)}G_{\rho^{(k)}}^{\rm mc}(s)\}\), respectively. 
Accordingly, we define \(T_{\rm CC}^{\rm mc}\) and its analytic transformation \(P_{\rm CC}^{\rm mc}\) by \eqref{eq:CCT} with \((P_1^{\rm mc}, \ldots, P_K^{\rm mc})^\top\).  
The corresponding null limits are as follows.  

\begin{theorem}[Multiple-change scan limit at a fixed ridge-regularization parameter]\label{thm:mc-process}
Suppose \(H_0\), Assumptions~\ref{ass:elliptical}--\ref{ass:grid}, and Convention~\ref{conv:scan}(ii) hold.  For every fixed \(\rho\in[\rho_0,\rho_1]\),
\begin{equation}
\label{eq:mc-process-limit}
  \{Z_\rho(s):s\in\calS_{\rm mc}(\varepsilon)\}
  \dto \{G_\rho^{\rm mc}(s):s\in\calS_{\rm mc}(\varepsilon)\}
  \quad\text{in }\ell^\infty\{\calS_{\rm mc}(\varepsilon)\},
\end{equation}
where \(G_\rho^{\rm mc}\) is a centered Gaussian process with almost surely continuous sample paths and covariance \(\Cov\{G_\rho^{\rm mc}(s),G_\rho^{\rm mc}(r)\}=K_0(s,r)\) for \(s,r\in\calS_{\rm mc}(\varepsilon)\).
Consequently,
\begin{equation}
\label{eq:mc-T-limit}
  T_\rho^{\rm mc}=\sup_{s\in\calS_{\rm mc}(\varepsilon)}Z_\rho(s)
  \dto \sup_{s\in\calS_{\rm mc}(\varepsilon)}G_\rho^{\rm mc}(s) := T_0^{\rm mc}.
\end{equation}
\end{theorem}

\begin{theorem}[Joint multiple-change scan limit over the finite parameter grid \(\calR_K\)]\label{thm:mc-joint}
Under the conditions of Theorem~\ref{thm:mc-process}, for the finite grid of ridge-regularization parameters \(\calR_K\),
\begin{equation}
\label{eq:mc-joint-process}
  \{Z_{\rho_n^{(k)}}(s):s\in\calS_{\rm mc}(\varepsilon),\ k=1,\ldots,K\}
  \dto
  \{G_{\rho^{(k)}}^{\rm mc}(s):s\in\calS_{\rm mc}(\varepsilon),\ k=1,\ldots,K\},
\end{equation}
where the right-hand side is a jointly centered Gaussian process with covariance
\begin{equation}
\label{eq:mc-joint-cov}
  \Cov\{G_{\rho}^{\rm mc}(s),G_{\rho'}^{\rm mc}(r)\}
  =r_E(\rho,\rho')K_0(s,r).
\end{equation}
Furthermore,
\begin{equation}
\label{eq:mc-pvector-limit}
  (P_1^{\rm mc},\ldots,P_K^{\rm mc})^\top
  \dto
  (P_{\infty,1}^{\rm mc},\ldots,P_{\infty,K}^{\rm mc})^\top.
\end{equation}
\end{theorem}

\begin{theorem}[Multiple-change Cauchy aggregation over the finite parameter grid \(\calR_K\)]\label{thm:mc-cauchy}
Suppose that the conditions of Theorem~\ref{thm:mc-joint} hold. Define
\(T_{{\rm CC},\infty}^{\rm mc}=\sum_{k=1}^K\varpi_k
  \tan[\pi\{\frac12-P_{\infty,k}^{\rm mc}\}]\).
If \(c_{\alpha,K}^{\rm mc}\) is a continuity point of the distribution of \(T_{{\rm CC},\infty}^{\rm mc}\) satisfying \(\Pbb(T_{{\rm CC},\infty}^{\rm mc}>c_{\alpha,K}^{\rm mc})=\alpha\), then under \(H_0\),
\begin{equation*}
  \Pbb(T_{\rm CC}^{\rm mc}>c_{\alpha,K}^{\rm mc})\to\alpha.
\end{equation*}
Moreover, for the analytic Cauchy transformation \(P_{\rm CC}^{\rm mc}\), at every \(\alpha\in(0,1)\) for which \(\cot(\pi\alpha)\) is a continuity point of the distribution of \(T_{{\rm CC},\infty}^{\rm mc}\), we have
\begin{equation*}
  \Pbb(P_{\rm CC}^{\rm mc}\le\alpha)
  \to
  \Pbb\{T_{{\rm CC},\infty}^{\rm mc}\ge\cot(\pi\alpha)\}.
\end{equation*}
\end{theorem}

These theorems are the multiple-change analogues of Theorems~\ref{thm:sc-process}--\ref{thm:sc-cauchy}; the only change is the scan domain and the resulting supremum distribution.

\section{Multiple Change-Point Estimation}\label{sec:estimation}

\subsection{WBS-ERHT segmentation algorithm}\label{sec:wbs}
To estimate all change locations, we combine the ERHT score with wild binary segmentation (WBS) \citep{fryzlewicz2014}. 
The WBS intervals provide multiscale localization, and the ERHT statistic quantifies evidence for a change within each sampled interval.

In particular, write \(|I|=r-\ell+1\) for an integer interval \(I=[\ell,r]=\{\ell,\ell+1,\ldots,r\}\).  For a fixed trimming fraction \(\varepsilon\in(0,1/2)\), define the integer adjacent-triple collection inside \(I\) by
\begin{equation*}
  \mathcal A(I,\varepsilon)=\{(a,b,c):\ell\le a\le b<c\le r,
  \ b-a+1\ge \varepsilon |I|,
  \ c-b\ge \varepsilon |I|\}.
\end{equation*}
The triple \((a,b,c)\) compares \(\{a,\ldots,b\}\) with \(\{b+1,\ldots,c\}\), and the common scatter pool is the whole WBS interval, i.e., \(J(I)=I\).  Let \(Z_{\rho,I}(a,b,c)\)
be the standardized statistic in Section~\ref{sec:calibration}, computed after replacing \(I_1(s),I_2(s),J(s)\) by these two adjacent windows and the pool \(J(I)\).  
Define the interval score and its estimated split by
\begin{equation*}
  \mathcal T_\rho(I)=\max_{(a,b,c)\in \mathcal A(I,\varepsilon)}Z_{\rho,I}(a,b,c),
\end{equation*}
and choose
\[
(\widehat a_\rho(I),\widehat k_\rho(I),\widehat c_\rho(I))
  \in\argmax_{(a,b,c)\in \mathcal A(I,\varepsilon)}Z_{\rho,I}(a,b,c).
\]
Ties are broken by a fixed deterministic rule.
For a finite grid of ridge-regularization parameters, define
\begin{equation}
\label{eq:wbs-grid-score}
  \mathcal T(I)=\max_{1\le k\le K}\mathcal T_{\rho_n^{(k)}}(I),
  \qquad
  \widehat k_{\calR}(I)=\widehat k_{\rho_n^{(\widehat\jmath(I))}}(I),
\end{equation}
where \(\widehat\jmath(I)\) is the smallest index of a regularization parameter attaining the maximum.  
Having defined the aggregated interval score \(\mathcal T(I)\), we generate WBS intervals \(\mathcal I_{M_n}=\{[\ell_v^{\rm WBS},u_v^{\rm WBS}]:v=1,\ldots,M_n\}\) independently of the observations.  
The algorithm below uses the narrowest-over-threshold selection rule of \citet{baranowski2019not} within the WBS framework \citep{fryzlewicz2014}.

\begin{algorithm}[WBS-ERHT segmentation]\label{alg:wbs-erht}
Fix a finite grid \(\calR_K\) of ridge-regularization parameters, a trimming parameter \(\varepsilon\), a threshold \(\mathfrak t_n^{\rm WBS}\), a minimum interval length \(m_{\min}\), a refinement radius \(g_n\), a deletion radius \(h_n\), and WBS intervals \(\mathcal I_{M_n}\).  Initialize \(\widehat{\mathcal K}=\varnothing\), and define the recursive routine \(\operatorname{WBS}(\ell,r)\) as follows.
\begin{enumerate}[label=\arabic*.]
\item If \(r-\ell+1<m_{\min}\), stop.
\item Form
\(\mathcal I(\ell,r)=\{I\in\mathcal I_{M_n}:I\subset[\ell,r],\ |I|\ge m_{\min}\}\cup\{[\ell,r]\}\).
For every \(I\in\mathcal I(\ell,r)\), compute \(\mathcal T(I)\) and \(\widehat k_{\calR}(I)\) from \eqref{eq:wbs-grid-score}.
\item Let \(\mathcal I^+(\ell,r)=\{I\in\mathcal I(\ell,r):\mathcal T(I)>\mathfrak t_n^{\rm WBS}\}\).  If \(\mathcal I^+(\ell,r)=\varnothing\), stop.
\item Choose
\(I^*\in\argmin_{I\in\mathcal I^+(\ell,r)}|I|\) and 
\(\widetilde k=\widehat k_{\calR}(I^*)\) 
breaking ties by the larger score and any remaining ties deterministically.  Form the local refinement interval
\[
  B(\widetilde k;\ell,r)=
  [\max\{\ell,\widetilde k-g_n\},\min\{r,\widetilde k+g_n\}].
\]
If \(|B(\widetilde k;\ell,r)|\ge m_{\min}\), recompute the score on this interval and set
\(\widehat k=\widehat k_{\calR}(B(\widetilde k;\ell,r))\);
otherwise set \(\widehat k=\widetilde k\).  Add \(\widehat k\) to the estimated change set \(\widehat{\mathcal K}\) and run
\[
  \operatorname{WBS}(\ell,\widehat k-h_n),
  \qquad
  \operatorname{WBS}(\widehat k+h_n+1,r).
\]
\end{enumerate}

The minimum-length restriction excludes unstable short-window statistics, the refinement step provides a locally balanced window, and the deletion radius prevents repeated selection of the same change-point.
Call \(\operatorname{WBS}(1,n)\).  The final estimator \(\widehat{\mathcal K}\) is the sorted set of all selected boundaries.    
\end{algorithm}

\subsection{Theory for WBS-ERHT segmentation}\label{sec:wbs-theory}

We next establish the consistency of WBS-ERHT segmentation.
Let the true integer change-points be \(k_j=\lfloor n\tau_j\rfloor\), \(j=1,\ldots,q\), and put \(k_0=0\), \(k_{q+1}=n\).  
Assume the minimum-spacing condition
\begin{equation}
\label{eq:wbs-spacing}
  0=\tau_0<\tau_1<\cdots<\tau_q<\tau_{q+1}=1,
  \qquad
  \min_{0\le j\le q}(\tau_{j+1}-\tau_j)\ge \delta_0>0.
\end{equation}
Let 
\(\bfDelta_j=\bftheta^{(j+1)}-\bftheta^{(j)}\),
\(r_n=\max_{1\le j\le q}\|\bfDelta_j\|\), and 
\(d_n=\min_{1\le j\le q}\|\bfDelta_j\|\).
We define
\begin{equation}
\label{eq:wbs-envelope-scales}
\begin{aligned}
  \overline{\mathfrak s}_n^{\rm WBS}&=\sqrt n\,r_n^2,\qquad
  \underline{\mathfrak s}_n^{\rm WBS}=\sqrt n\,d_n^2,\\
  \mathfrak e_n^{\rm WBS}&=\ell_nr_n+r_n^2+\ell_nn^{-1/2}.
\end{aligned}
\end{equation}
We impose two additional assumptions on the random-interval design, jump sizes, and tuning hierarchy.

\begin{assumption}[WBS design and geometry]\label{ass:wbs-design}
The number of change-points \(q\) is fixed and \eqref{eq:wbs-spacing} holds.  The
intervals \([\ell_v^{\rm WBS},u_v^{\rm WBS}]\), \(v=1,\ldots,M_n\), are
independent of the observations and independent across \(v\); each is obtained
by drawing two distinct endpoints uniformly from \(\{1,\ldots,n\}\) and
ordering them.  Moreover, \(M_n\to\infty\).  There exist fixed constants
\(\delta_{\min},\delta_R>0\) such that
\(0<\delta_{\min}<(1-2\varepsilon)\delta_R\),
\(\delta_R+\delta_{\min}< {\delta_0}/{2}\),
and the algorithmic window parameters satisfy
\({m_{\min}}/{n}\to\delta_{\min}\) and 
\(g_n=\lfloor\delta_Rn\rfloor\).
\end{assumption}

\begin{assumption}[Jump size and tuning hierarchy]\label{ass:wbs-signal-tuning}
The jumps satisfy \(\ell_nr_n\to0\).  The deterministic threshold
\(\mathfrak t_n^{\rm WBS}>0\) and integer deletion radius \(h_n\ge0\)
satisfy
\[
\mathfrak t_n^{\rm WBS}=o(\underline{\mathfrak s}_n^{\rm WBS}),
\qquad
\sqrt{\log n}
  +\overline{\mathfrak s}_n^{\rm WBS}\mathfrak e_n^{\rm WBS}
  =o(\mathfrak t_n^{\rm WBS}),
\qquad
h_n=o(n),
\qquad
\frac{n\mathfrak t_n^{\rm WBS}}
     {\underline{\mathfrak s}_n^{\rm WBS}}
  =o(h_n).
\]
\end{assumption}

Assumption~\ref{ass:wbs-design} fixes the random-interval design and the deterministic
geometry needed for isolation and refinement.  Assumption~\ref{ass:wbs-signal-tuning} requires the threshold to dominate the stochastic and feasible-expansion errors while remaining smaller than the weakest jump signal; it also makes the deletion radius larger than the resulting localization scale but asymptotically smaller than \(n\).

\begin{theorem}[Consistency of WBS-ERHT segmentation]\label{thm:wbs-consistency}
Suppose Assumptions~\ref{ass:elliptical}--\ref{ass:grid} and
Assumptions~\ref{ass:wbs-design}--\ref{ass:wbs-signal-tuning} hold.  Let
\(\widehat{\mathcal K}=\{\widehat k_1<\cdots<\widehat k_{\widehat q}\}\)
be the output of Algorithm~\ref{alg:wbs-erht}.  Then
\begin{equation*}
  \Pbb(\widehat q=q)\to1.
\end{equation*}
On the event \(\{\widehat q=q\}\), order the estimates and true change-points increasingly.  With this convention,
\begin{equation*}
  \max_{1\le j\le q}\frac{|\widehat k_j-k_j|}{n}
  =O_P\left(
    \frac{\mathfrak t_n^{\rm WBS}}
         {\underline{\mathfrak s}_n^{\rm WBS}}
  \right)
  =O_P\left(
    \frac{\mathfrak t_n^{\rm WBS}}
         {\sqrt n\min_{1\le j\le q}\|\bfDelta_j\|^2}
  \right).
\end{equation*}
\end{theorem}

Theorem~\ref{thm:wbs-consistency} establishes simultaneous model-selection and localization
consistency: WBS-ERHT recovers the number of change-points with probability
tending to one, and its largest normalized localization error vanishes at
rate \(\mathfrak t_n^{\rm WBS}/\underline{\mathfrak s}_n^{\rm WBS}\).
This rate is governed by the weakest jump, while the lower bound on the
threshold prevents stochastic and feasible-expansion errors from producing
spurious detections.

\section{Simulation study}\label{sec:simulation}

\subsection{Common design, methods, and experimental protocol}
\label{sec:sim-design}

Throughout, observations are generated from
\(\bfX_i=\bftheta_i+\bfvarepsilon_i\) for \(i=1,\ldots,n\),
where every shape matrix is standardized by \(\tr(\bfOmega_p)=p\), and the null
model sets \(\bftheta_i=\mathbf 0\) for all \(i\).
Let \(\bfO_p\) be an orthogonal rotation matrix and write
\(
  \mathcal D(d_1,\ldots,d_p)
  ={p\,\bfO_p\big\{\diag(d_1,\ldots,d_p)/\sum_{j=1}^p d_j\big\}\bfO_p^\top}\).
We consider three different shape designs as follows.  
\begin{enumerate}[label=(\roman*)]
\item \textit{Identity}: \(\bfOmega_p=\bfI_p\).
\item \textit{Polynomial decay}: \(\bfOmega_p=\mathcal D(d_1^{\rm poly},\ldots,d_p^{\rm poly})\) with \(d_j^{\rm poly}=0.01+(p-j+0.1)^2\).
\item \textit{Exponential decay}: \(\bfOmega_p=\mathcal D(d_1^{\rm exp},\ldots,d_p^{\rm exp})\) with \(d_j^{\rm exp}=\exp(-3j/p)\).
\end{enumerate}
These designs are labelled Identity, Poly, and Exp, respectively, in the subsequent tables and figures.  
The Identity and Exp designs have eigenvalues bounded away from zero and infinity, whereas the Poly design has its smallest eigenvalues approaching zero.  This is allowed by Assumption~\ref{ass:spectrum}, which requires only a uniform upper spectral bound.
For each shape design, we consider three elliptically symmetric error distributions.  Let \(\bfZ_i\sim N_p(\mathbf 0,\bfOmega_p)\), \(S_i\sim\chi_3^2\), and \(B_i\sim\operatorname{Bernoulli}(0.2)\) be mutually independent.  We generate
\begin{enumerate}[label=(\roman*)]
\item \textit{Gaussian errors}: \(\bfvarepsilon_i=\bfZ_i\), equivalently \(\bfvarepsilon_i\sim N_p(\mathbf 0,\bfOmega_p)\).
\item \textit{Multivariate \(t_3\) errors}: 
\(\bfvarepsilon_i={\bfZ_i}/{(S_i/3)^{1/2}}\).
\item \textit{Gaussian-mixture errors}:
\(\bfvarepsilon_i\sim 0.8\,N_p(\mathbf 0,\bfOmega_p) +0.2\,N_p(\mathbf 0,100\bfOmega_p)\),
which is equivalently generated as \(\bfvarepsilon_i=(1+9B_i)\bfZ_i\).
\end{enumerate}
In Section~\ref{app:assumption1-examples} of the Supplementary Material, we give the radial representations explicitly and verify the radial moment and maximal-truncation conditions in Assumption~\ref{ass:elliptical} for these three error distributions.

Two shift profiles, Uniform and Constant, are considered across the power and localization
experiments.  
Writing \(c_{\rm sig}>0\) as the signal strength,  
we set
\begin{enumerate}[label=(\roman*)]
\item \textit{Uniform-direction shift (Uniform)}:
\(\bfDelta_p^{\rm unif}(c_{\rm sig})
  =(\delta_1,\ldots,\delta_p)^\top\)
with
\(\delta_j{\sim}
  N\big(0, {c_{\rm sig}^2}/{(np)}\big)\) independently for \(j = 1, \ldots, p\),
so that
\(\E\{\|\bfDelta_p^{\rm unif}(c_{\rm sig})\|^2\}=c_{\rm sig}^2/n\).
\item \textit{Constant shift}:
\(\bfDelta_p^{\rm const}(c_{\rm sig})=c_{\rm sig}\,\bfone_p\),
which gives an equal-coordinate location change.
\end{enumerate}

For the single-change testing and localization experiments, we compare four methods: the proposed ERHT Cauchy-combination test (ERHT-CC), the covariance-based RHT Cauchy-combination test (RHT-CC) \citep{li2026change,zhao2026rhtcp}, the sum-type data-adaptive mean-shift test (DMS0) \citep{wangfeng2023}, and the spatial-sign change-point test (SSCPD0) \citep{liu2025sscpd}. 
The suffix ``0'' indicates implementations without regularized inverse-shape normalization.
The nominal level is \(\alpha=0.05\).  For ERHT-CC, the regularization ratios \(\rho/\gamma\) take values in \(\{0.05,0.10,\ldots,0.50\}\), where \(\gamma=p/n\); RHT-CC uses the analogous grid for its covariance-regularization parameter.  Both Cauchy combinations use equal weights.
For each regularization-parameter value, the single-change p-value is computed using the Gaussian-supremum calibration in Theorem~\ref{thm:sc-process}; ERHT-CC reports the analytic Cauchy aggregation of these p-values.  

For the multiple-change global testing experiment, we compare the adjacent-triple RHT scan (RHT-MC) \citep{li2026change} with the proposed adjacent-triple ERHT Cauchy-combination scan (ERHT-MC).  
The multiple-change localization experiment below additionally compares the WBS implementation of the proposed statistic with INSPECT \citep{wangsamworth2018}, HDBINSEG \citep{cho2015sbs}, and ECP \citep{matteson2014nonparametric,james2014ecp}.

\subsection{Single-change global testing}
\label{sec:sim-single-testing}

We first examine empirical size under \(H_0\).  Tables~\ref{tab:size-n200} and~\ref{tab:size-n400} report rejection percentages for \(n=200\) and \(n=400\), respectively, with \(p\in\{100,200,400\}\).  Each entry is based on 1000 Monte Carlo replications, so the Monte Carlo standard error near a five-percent rejection probability is about 0.69 percentage points.  
The proposed ERHT-CC test is close to the nominal level across most settings.  The covariance-based RHT-CC test is reasonably calibrated under Gaussian errors but becomes markedly oversized under the heavy-tailed \(t_3\) and Gaussian-mixture distributions.  DMS0 tends to be conservative under heavy tails, whereas SSCPD0 is often liberal when \(p\) is small or the shape matrix is nontrivial.

\begin{table}[tb]
  \centering
  \small
  \caption{Empirical sizes of the single-change global tests for \(n=200\).
  Gaussian-supremum p-values computed separately at each regularization-parameter value are followed by analytic Cauchy
  aggregation.  Entries are rejection percentages at the nominal
  five-percent level based on 1000 replications.}
  \label{tab:size-n200}
  \begin{tabular}{llrrrrr}
    \toprule
    Shape & Error & \(p\) & ERHT-CC & RHT-CC & DMS0 & SSCPD0 \\
    \midrule
    Identity & Normal & 100 & 4.5 & 5.0 & 6.7 & 7.0 \\
    Identity & Normal & 200 & 4.1 & 4.3 & 8.0 & 8.4 \\
    Identity & Normal & 400 & 8.1 & 6.4 & 6.4 & 6.3 \\
    \addlinespace
    Identity & $t_3$ & 100 & 5.1 & 9.4 & 3.1 & 9.9 \\
    Identity & $t_3$ & 200 & 5.0 & 9.9 & 3.1 & 7.9 \\
    Identity & $t_3$ & 400 & 7.3 & 13.2 & 2.0 & 6.3 \\
    \addlinespace
    Identity & Mixture & 100 & 4.8 & 17.9 & 1.3 & 9.8 \\
    Identity & Mixture & 200 & 5.8 & 19.8 & 1.2 & 9.1 \\
    Identity & Mixture & 400 & 5.7 & 19.4 & 1.7 & 7.8 \\
    \addlinespace
    Poly & Normal & 100 & 4.8 & 6.7 & 8.6 & 11.4 \\
    Poly & Normal & 200 & 5.0 & 6.0 & 7.5 & 9.0 \\
    Poly & Normal & 400 & 6.2 & 5.6 & 8.8 & 7.5 \\
    \addlinespace
    Poly & $t_3$ & 100 & 3.9 & 9.6 & 3.6 & 10.1 \\
    Poly & $t_3$ & 200 & 5.1 & 9.3 & 2.5 & 8.2 \\
    Poly & $t_3$ & 400 & 5.5 & 12.0 & 2.3 & 6.3 \\
    \addlinespace
    Poly & Mixture & 100 & 6.1 & 17.5 & 1.7 & 11.4 \\
    Poly & Mixture & 200 & 4.4 & 19.8 & 1.9 & 11.3 \\
    Poly & Mixture & 400 & 5.3 & 23.1 & 1.2 & 8.4 \\
    \addlinespace
    Exp & Normal & 100 & 4.9 & 6.5 & 8.7 & 11.4 \\
    Exp & Normal & 200 & 5.7 & 5.2 & 7.7 & 8.8 \\
    Exp & Normal & 400 & 6.2 & 5.6 & 8.7 & 8.4 \\
    \addlinespace
    Exp & $t_3$ & 100 & 5.9 & 9.5 & 2.6 & 11.3 \\
    Exp & $t_3$ & 200 & 5.9 & 9.5 & 3.3 & 8.9 \\
    Exp & $t_3$ & 400 & 5.2 & 10.7 & 2.7 & 7.0 \\
    \addlinespace
    Exp & Mixture & 100 & 3.9 & 18.6 & 1.8 & 12.8 \\
    Exp & Mixture & 200 & 5.2 & 21.8 & 1.3 & 10.7 \\
    Exp & Mixture & 400 & 5.8 & 20.9 & 1.5 & 7.5 \\
    \bottomrule
  \end{tabular}
\end{table}

\begin{table}[tb]
  \centering
  \small
  \caption{Empirical sizes of the single-change global tests for \(n=400\).
  Gaussian-supremum p-values computed separately at each regularization-parameter value are followed by analytic Cauchy
  aggregation.  Entries are rejection percentages at the nominal
  five-percent level based on 1000 replications.}
  \label{tab:size-n400}
  \begin{tabular}{llrrrrr}
    \toprule
    Shape & Error & \(p\) & ERHT-CC & RHT-CC & DMS0 & SSCPD0 \\
    \midrule
    Identity & Normal & 100 & 4.5 & 7.6 & 7.3 & 7.5 \\
    Identity & Normal & 200 & 5.8 & 6.3 & 6.8 & 5.7 \\
    Identity & Normal & 400 & 5.9 & 4.1 & 6.2 & 5.3 \\
    \addlinespace
    Identity & $t_3$ & 100 & 5.1 & 9.3 & 4.1 & 8.1 \\
    Identity & $t_3$ & 200 & 6.1 & 9.5 & 2.1 & 5.1 \\
    Identity & $t_3$ & 400 & 4.5 & 8.1 & 2.4 & 4.7 \\
    \addlinespace
    Identity & Mixture & 100 & 5.0 & 18.7 & 2.2 & 6.4 \\
    Identity & Mixture & 200 & 4.1 & 19.9 & 1.5 & 6.8 \\
    Identity & Mixture & 400 & 4.3 & 20.7 & 1.5 & 5.2 \\
    \addlinespace
    Poly & Normal & 100 & 5.2 & 9.2 & 9.6 & 10.9 \\
    Poly & Normal & 200 & 5.2 & 6.8 & 7.9 & 8.6 \\
    Poly & Normal & 400 & 3.9 & 4.2 & 6.8 & 5.0 \\
    \addlinespace
    Poly & $t_3$ & 100 & 3.7 & 9.6 & 3.5 & 8.3 \\
    Poly & $t_3$ & 200 & 5.4 & 8.6 & 2.4 & 7.7 \\
    Poly & $t_3$ & 400 & 4.2 & 9.1 & 1.8 & 6.5 \\
    \addlinespace
    Poly & Mixture & 100 & 4.8 & 15.7 & 2.8 & 9.1 \\
    Poly & Mixture & 200 & 3.3 & 16.7 & 2.1 & 6.7 \\
    Poly & Mixture & 400 & 4.6 & 20.1 & 1.7 & 6.3 \\
    \addlinespace
    Exp & Normal & 100 & 6.3 & 9.4 & 8.6 & 9.7 \\
    Exp & Normal & 200 & 4.7 & 5.1 & 6.9 & 7.3 \\
    Exp & Normal & 400 & 5.6 & 5.4 & 7.4 & 7.2 \\
    \addlinespace
    Exp & $t_3$ & 100 & 5.5 & 9.8 & 3.8 & 9.4 \\
    Exp & $t_3$ & 200 & 5.1 & 10.1 & 2.6 & 7.7 \\
    Exp & $t_3$ & 400 & 4.8 & 7.9 & 2.6 & 6.0 \\
    \addlinespace
    Exp & Mixture & 100 & 4.5 & 18.8 & 2.8 & 7.8 \\
    Exp & Mixture & 200 & 5.0 & 18.7 & 2.2 & 6.9 \\
    Exp & Mixture & 400 & 4.6 & 20.5 & 1.6 & 6.5 \\
    \bottomrule
  \end{tabular}
\end{table}

For the power experiment, the change occurs at the middle of the sample, so
that \(\tau_\star=1/2\), and the data are generated under \(H_1\) as \(\bfX_i=\bfvarepsilon_i\) for \(i\le \lfloor n\tau_\star\rfloor\) and \(\bfX_i=\bfDelta_p^{\nu_{\rm sh}}(c_{\rm sig})+\bfvarepsilon_i\) for \(i>\lfloor n\tau_\star\rfloor\).
Here \(\nu_{\rm sh} \in \{{\rm unif}, {\rm const}\}\), 
\(c_{\rm sig}\) is the signal strength displayed on the horizontal axis,
and the Uniform and Constant alternatives are those defined in
Section~\ref{sec:sim-design}.
For the reported power curves, the critical value for each method and
data-generating model is the empirical 95th percentile from 1000 independent null
simulations; empirical powers are computed from 1000 independent Monte Carlo
replications.  
We fix \(n=200\) and \(p=100\), use the three shape designs and the three error distributions described above, and consider the two location-shift profiles defined above.  The Uniform profile represents a random dense direction whose total signal norm is controlled by \(c_{\rm sig}/\sqrt n\), whereas the Constant profile represents an equal-coordinate location shift.  
The displayed values of \(c_{\rm sig}\) are chosen in two stages.  First, preliminary ERHT runs over candidate signal strengths are used to select values spanning a broad range of rejection probabilities.  The same selected values are then used for all four methods.  This construction makes the curves comparable within each panel and avoids choosing method-specific alternatives.

\begin{figure}[tbp]
  \centering
  \includegraphics[width=.95\textwidth]{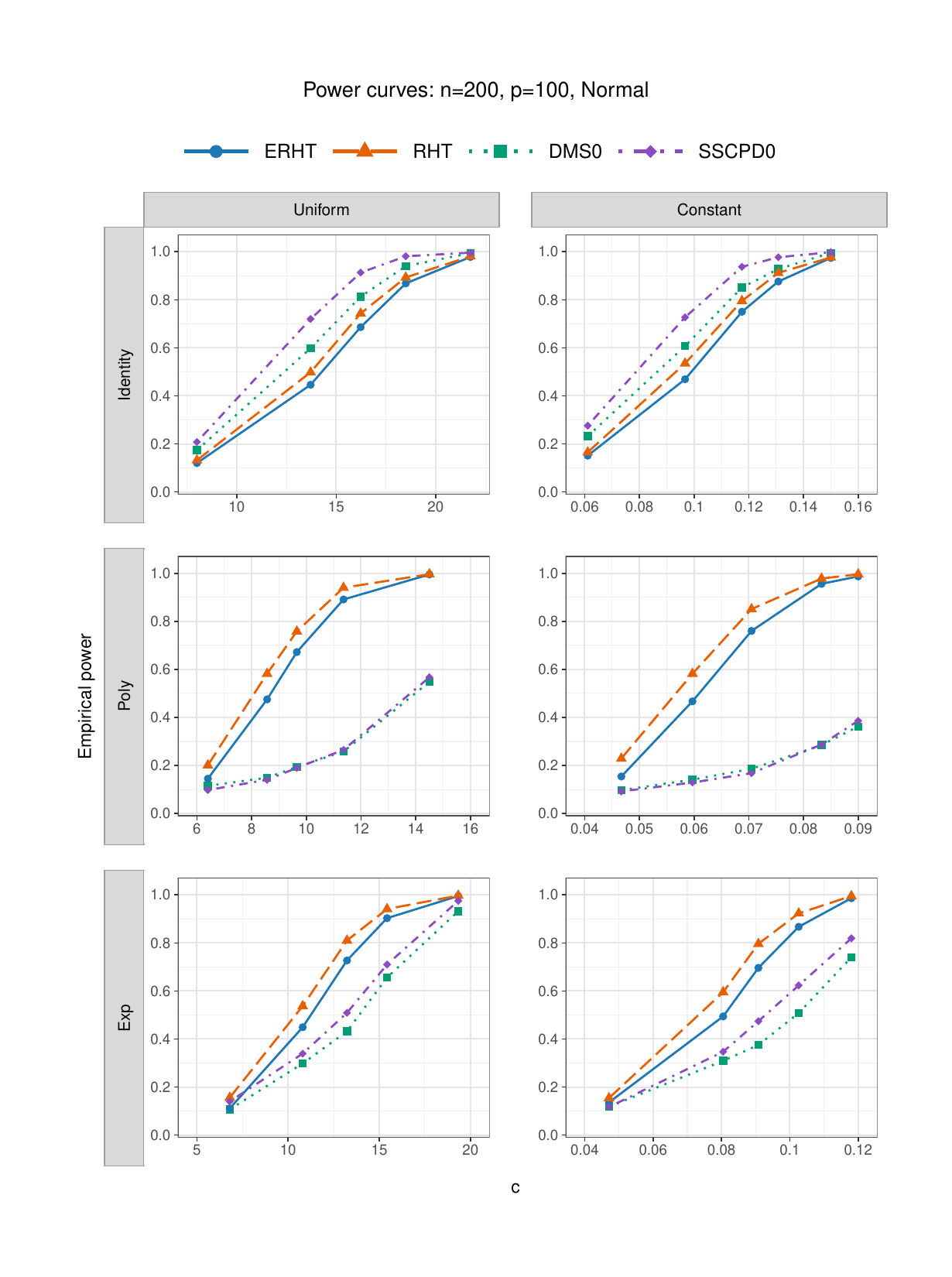}
  \caption{Size-corrected empirical power under Gaussian errors with
  \(n=200\) and \(p=100\).  The two columns correspond to the Uniform and
  Constant shift profiles, and the three rows correspond to the identity,
  polynomial-decay and exponential-decay shape designs.}
  \label{fig:power-normal}
\end{figure}

\begin{figure}[tbp]
  \centering
  \includegraphics[width=.95\textwidth]{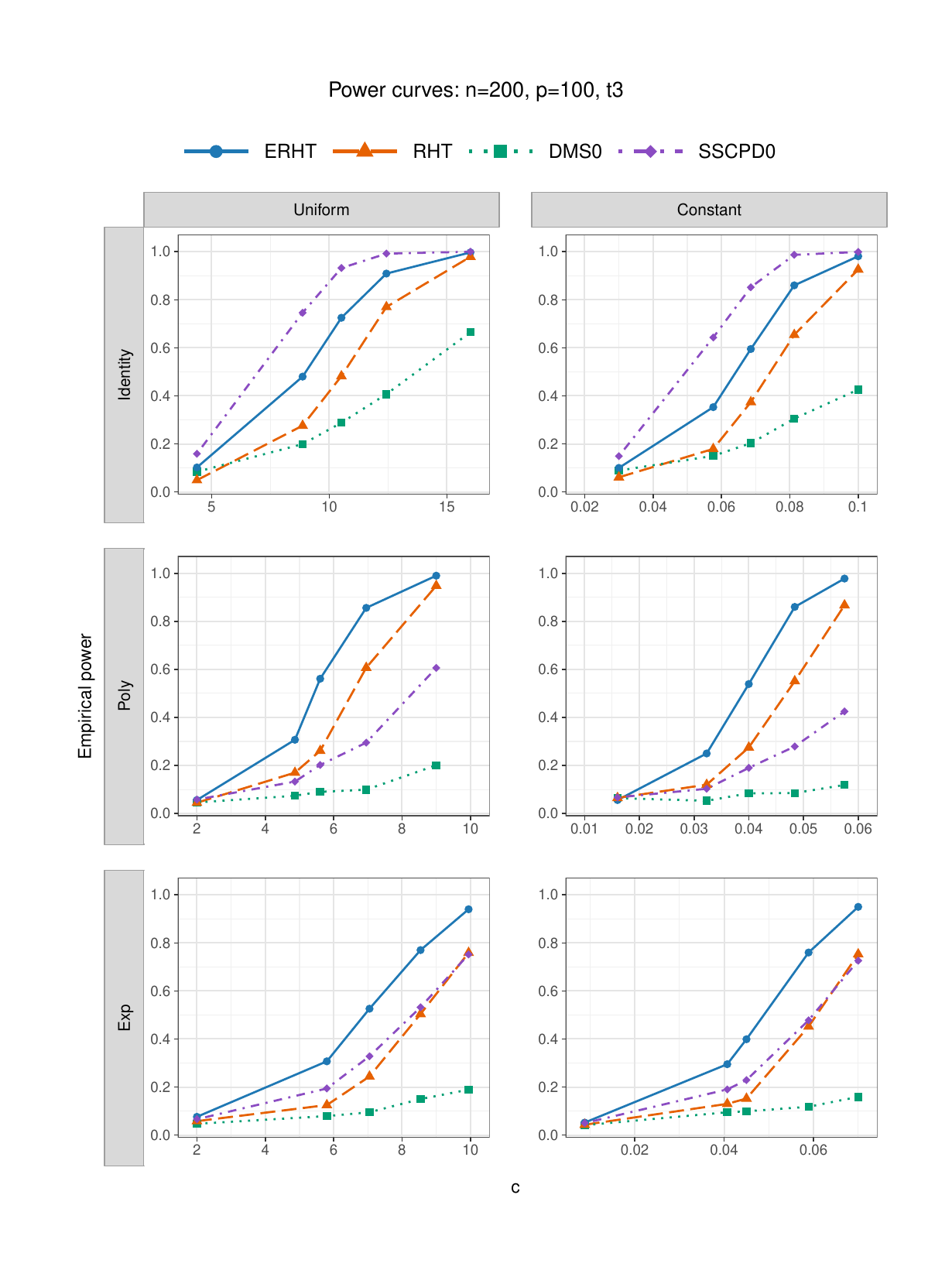}
  \caption{Size-corrected empirical power under multivariate \(t_3\)
  errors with \(n=200\) and \(p=100\).  Panel labels are the same as in
  Figure~\ref{fig:power-normal}.}
  \label{fig:power-t3}
\end{figure}

Figure~\ref{fig:power-normal} reports the Gaussian case.  In these settings, covariance-based normalization entails no material loss of power.  In the identity design, SSCPD0 is highly competitive because no nontrivial cross-sectional adjustment is needed.  In the polynomial- and exponential-decay designs, ERHT and RHT dominate the two competitors without regularized shape normalization over most of the signal range, with RHT slightly ahead.  This pattern is consistent with the fact that, under Gaussian errors, the sample covariance retains magnitude information that a sign-based method deliberately removes.

Figure~\ref{fig:power-t3} shows a different pattern under the heavy-tailed \(t_3\) distribution.  The proposed ERHT test is substantially more powerful than RHT in the polynomial- and exponential-decay designs for both shift profiles.  The low power of DMS0 in these settings is consistent with the sensitivity of mean-based aggregation to heavy tails.  SSCPD0 performs well in the identity design but loses power under the two nonidentity shapes.  These results show the benefit of combining spatial-median robustness with normalization by a ridge-regularized inverse of the centered spatial-sign covariance matrix.

\begin{figure}[tbp]
  \centering
  \includegraphics[width=.95\textwidth]{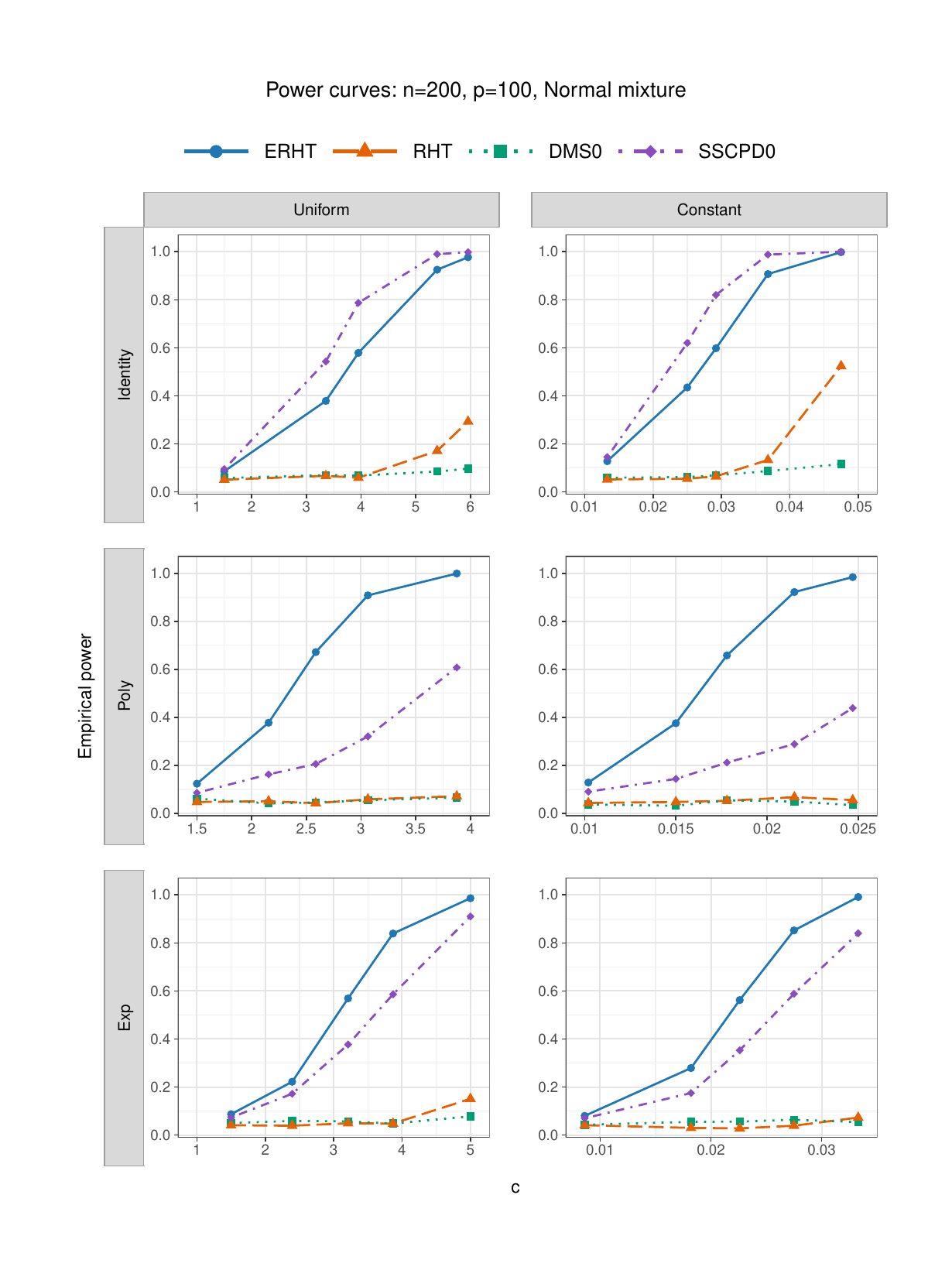}
  \caption{Size-corrected empirical power under Gaussian-mixture errors
  with \(n=200\) and \(p=100\).  Panel labels are the same as in
  Figure~\ref{fig:power-normal}.}
  \label{fig:power-mixture}
\end{figure}

The Gaussian-mixture experiment in Figure~\ref{fig:power-mixture} is the most challenging for covariance-based procedures.  ERHT gives the largest power in nearly all panels and reaches high rejection probabilities as the signal increases.  By contrast, RHT and DMS0 often remain close to the nominal level in the polynomial- and exponential-decay designs, even at the largest displayed signal strengths.  SSCPD0 is more robust than the mean-based competitors, but it does not use the spectral information in the centered spatial-sign covariance matrix and is therefore less effective than ERHT when the shape matrix is nontrivial.  Overall, the three power experiments provide an empirical comparison: covariance-based RHT can be competitive under light tails, whereas ERHT has a clear advantage in these heavy-tailed designs when cross-sectional dependence is informative.

\subsection{Multiple-change global testing}
\label{sec:sim-multiple-testing}

We next study the multiple-change global testing problem.  
In all multiple-change testing experiments, we set \(n=300\), \(p=100\), \(\varepsilon=0.1\), and use the same Identity, Poly, and Exp shape designs.  To make the signal scales comparable across radial distributions, the non-Gaussian errors are variance-standardized before the shape matrix is applied.  Equivalently, with \(\bfG_i\sim N_p(\mathbf 0,\bfI_p)\), \(S_i\sim\chi_3^2\), and \(B_i\sim\operatorname{Bernoulli}(0.2)\), the three noise cases are generated as
\[
  \bfOmega_p^{1/2}\bfG_i,
  \qquad
  \bfOmega_p^{1/2}\frac{\bfG_i}{\sqrt{S_i/3}\sqrt{3}},
  \qquad
  \bfOmega_p^{1/2}\frac{(1+9B_i)\bfG_i}{\sqrt{0.8+100\cdot0.2}},
\]
respectively.  
For RHT-MC we set its covariance-regularization parameter \(\lambda\) such that \(\lambda/\gamma=0.1\), where \(\gamma=p/n\). 
For ERHT-MC we use \(\rho/\gamma\in\{0.05,0.10,\ldots,0.50\}\) and combine the parameter-specific p-values by the Cauchy rule.

We first report empirical sizes under the null model.  
For this diagnostic experiment, each simulated data matrix is recalibrated by randomly
permuting the time order \(B=99\) times.  
Table~\ref{tab:mc-size-permutation}
reports only these permutation-calibrated rejection percentages. 
For ERHT-MC, the listed value is the final Cauchy-combined test, not any individual parameter-specific test.
Table~\ref{tab:mc-size-permutation} shows that permutation calibration gives stable null rejection rates for the adjacent-triple multiple-change scan.
Across the nine distribution--shape settings, ERHT-MC stays between 3.4\% and
5.5\%, while RHT-MC ranges from 2.7\% to 4.6\%.

\begin{table}[tbp]
  \centering
  \small
  \caption{Empirical sizes of the multiple-change global tests using
  time-permutation calibration.  ERHT-MC uses p-values computed separately at each ridge-regularization parameter, followed by analytic Cauchy aggregation; RHT-MC uses a permutation p-value at \(\lambda/\gamma=0.1\).
  Entries are rejection percentages at the nominal five-percent level,
  based on 1000 null replications; each replication uses \(B=99\)
  permutations.}
  \label{tab:mc-size-permutation}
  \begin{tabular}{llrr}
    \toprule
    Error & Shape & RHT-MC & ERHT-MC \\
    \midrule
    Normal & Identity & 4.6 & 4.6 \\
    Normal & Poly & 2.7 & 4.0 \\
    Normal & Exp & 3.9 & 4.1 \\
    \addlinespace
    $t_3$ & Identity & 3.5 & 4.8 \\
    $t_3$ & Poly & 4.2 & 3.7 \\
    $t_3$ & Exp & 3.5 & 3.8 \\
    \addlinespace
    Mixture & Identity & 4.6 & 3.4 \\
    Mixture & Poly & 4.3 & 4.6 \\
    Mixture & Exp & 2.9 & 5.5 \\
    \bottomrule
  \end{tabular}
\end{table}

We then consider the epidemic alternative with two change-points
\(k_{1,\mathrm{test}}=\lfloor 0.35n\rfloor\) and
\(k_{2,\mathrm{test}}=\lfloor 0.65n\rfloor\).
For each \(\nu_{\rm sh}\in\{\mathrm{unif},\mathrm{const}\}\), the observations are
\begin{equation}\label{eq:epidemic}
  \bfX_i=
  \begin{cases}
  \bfvarepsilon_i, & 1\le i\le k_{1,\mathrm{test}},\\
  \bfDelta_p^{\nu_{\rm sh}}(c_{\rm sig})+\bfvarepsilon_i,
    & k_{1,\mathrm{test}}<i\le k_{2,\mathrm{test}},\\
  \bfvarepsilon_i, & k_{2,\mathrm{test}}<i\le n.
  \end{cases}
\end{equation}
The Uniform and Constant profiles are
\(\bfDelta_p^{\rm unif}(c_{\rm sig})\) and
\(\bfDelta_p^{\rm const}(c_{\rm sig})\), respectively, as defined in
Section~\ref{sec:sim-design}.  For
Figures~\ref{fig:mc-power-normal}--\ref{fig:mc-power-mixture}, the critical
value for each method, shape, and radial distribution is the empirical 95th
percentile from 1000 independent null simulations; empirical powers are computed from
1000 independent Monte Carlo replications.

\begin{figure}[tbp]
  \centering
  \includegraphics[width=.95\textwidth]{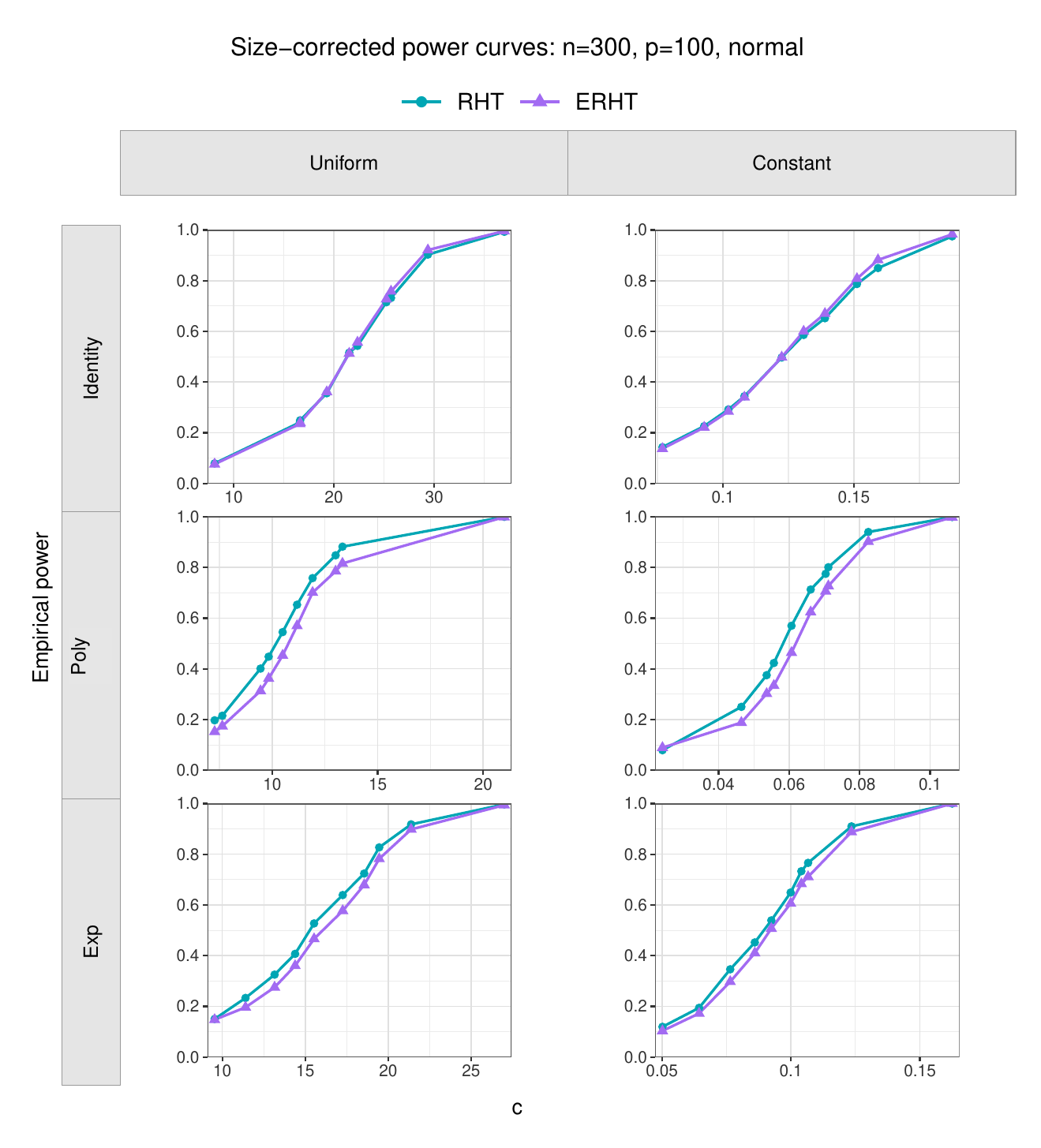}
  \caption{Size-corrected empirical power of the multiple-change
  global tests under Gaussian errors with \(n=300\) and \(p=100\).  The two
  columns correspond to the Uniform and Constant epidemic shifts, and the
  three rows correspond to the Identity, Poly and Exp shape designs.}
  \label{fig:mc-power-normal}
\end{figure}

\begin{figure}[tbp]
  \centering
  \includegraphics[width=.95\textwidth]{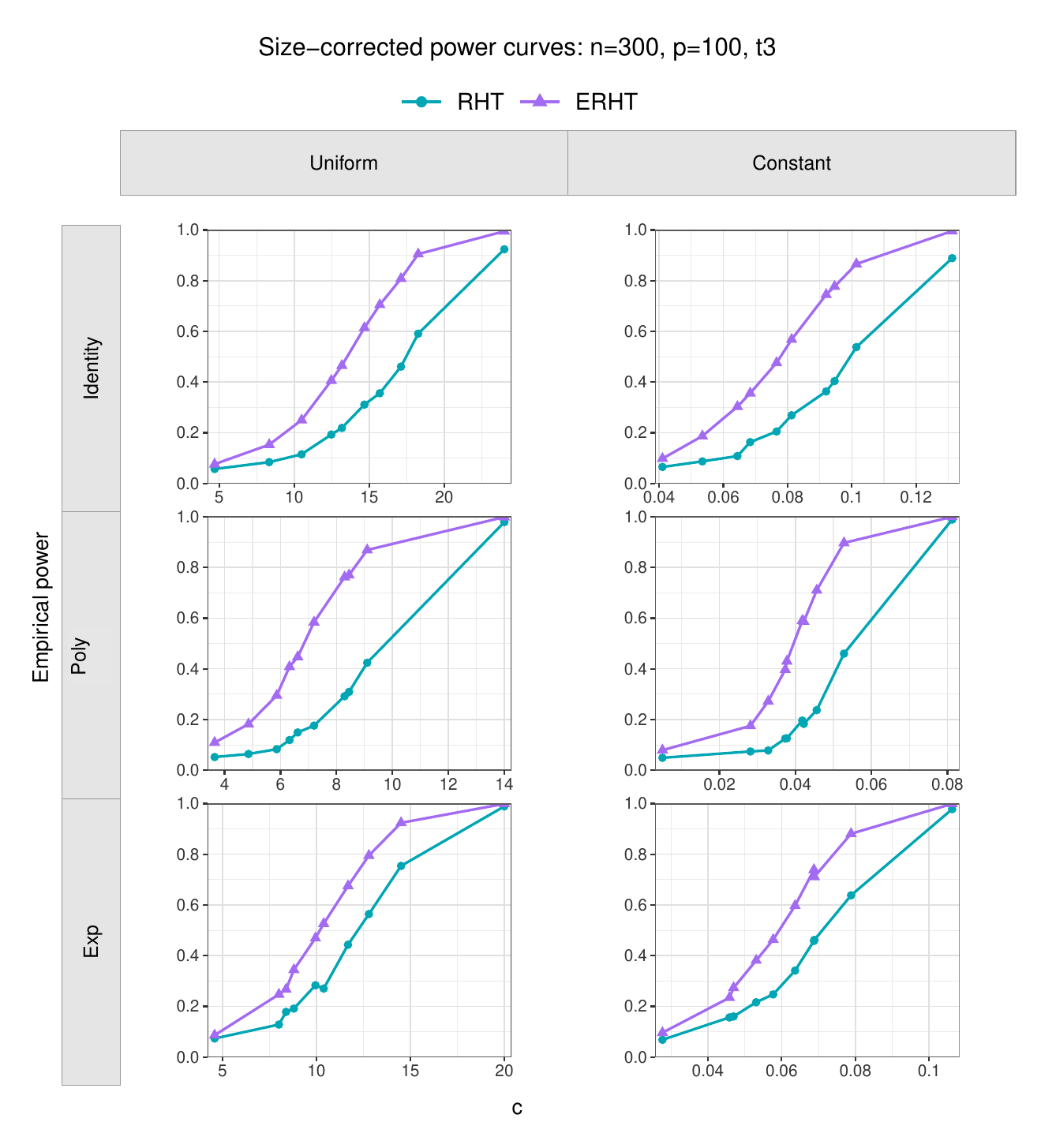}
  \caption{Size-corrected empirical power of the multiple-change
  global tests under variance-standardized multivariate \(t_3\) errors with
  \(n=300\) and \(p=100\).  Panel labels are the same as in
  Figure~\ref{fig:mc-power-normal}.}
  \label{fig:mc-power-t3}
\end{figure}

\begin{figure}[tbp]
  \centering
  \includegraphics[width=.95\textwidth]{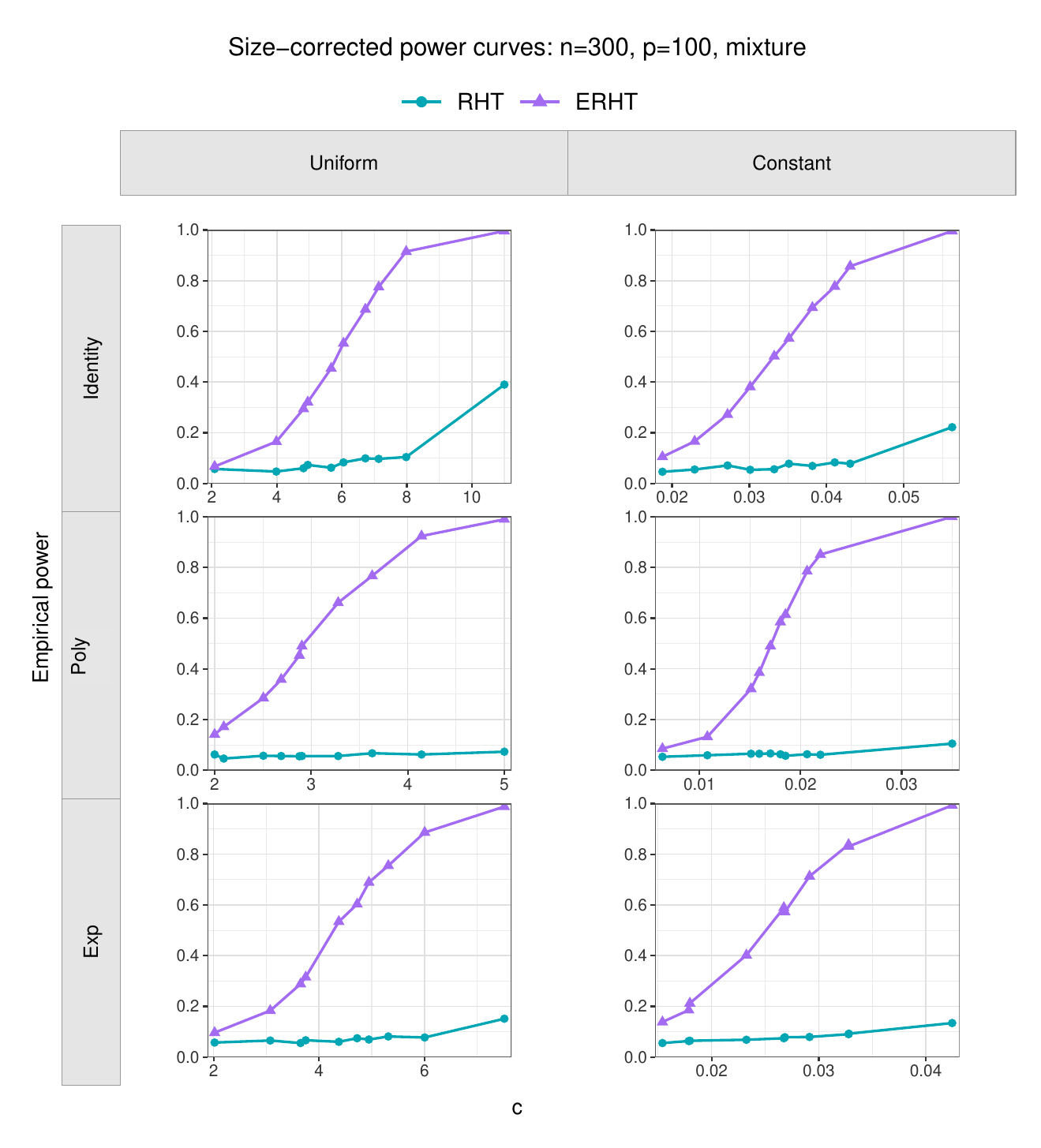}
  \caption{Size-corrected empirical power of the multiple-change
  global tests under variance-standardized Gaussian-mixture errors with
  \(n=300\) and \(p=100\).  Panel labels are the same as in
  Figure~\ref{fig:mc-power-normal}.}
  \label{fig:mc-power-mixture}
\end{figure}

Figures~\ref{fig:mc-power-normal}--\ref{fig:mc-power-mixture} show that the relative behavior of the two global scans is consistent with the single-change evidence.  Under Gaussian errors, RHT-MC and ERHT-MC are very close in the identity design, while RHT-MC can be slightly more powerful in the nonidentity shape designs because sample-covariance normalization is efficient under light tails.  Under variance-standardized \(t_3\) errors, ERHT-MC dominates RHT-MC in all six panels, with especially large gains in the Poly design and under the Constant shift.  Under Gaussian-mixture errors, the contrast is sharper: RHT-MC remains close to the size level over much of the displayed signal-strength range, whereas ERHT-MC increases rapidly to high power.  These results indicate that the spatial-median contrast and normalization by a ridge-regularized inverse of the SSCM remain beneficial for the full adjacent-triple multiple-change scan, not only for the single-change path.

\subsection{Change-point localization}
\label{sec:sim-localization}

We first examine single-change localization accuracy.  The true change location is again \(\tau_\star=1/2\).  
For each method and replication \(r=1,\ldots,1000\), let
\(\widehat k^{(r)}\) denote the estimated integer change-point location and
let \(k_\star=\lfloor n\tau_\star\rfloor\) denote the true integer
change-point location.  We report
\[
  100\times\frac{1}{1000}\sum_{r=1}^{1000}
  \frac{|\widehat k^{(r)}-k_\star|}{n},
\]
so smaller values indicate more accurate localization.
As in the global-testing experiments, a preliminary ERHT experiment is used to select \(c_{\rm sig}\) so that the settings span a comparable range of localization difficulty; the selected value is then used for all competing methods.

\begin{table}[H]
  \centering
  \small
  \caption{Single-change localization accuracy for \(n=200\) and \(p=100\).  Entries are \(100\times \mathrm{MAE}/n\), based on 1000 replications; smaller values are better.}
  \label{tab:location-single}
  \begin{tabular}{lllrrrr}
    \toprule
    Error & Shape & Shift & ERHT & RHT & DMS0 & SSCPD0 \\
    \midrule
    Normal & Identity & Uniform & 5.23 & 4.13 & 3.10 & 2.03 \\
    Normal & Identity & Constant & 4.46 & 3.57 & 2.88 & 1.85 \\
    Normal & Poly & Uniform & 1.67 & 1.42 & 17.08 & 8.40 \\
    Normal & Poly & Constant & 1.40 & 1.08 & 16.78 & 7.80 \\
    Normal & Exp & Uniform & 3.90 & 3.14 & 9.67 & 4.24 \\
    Normal & Exp & Constant & 5.00 & 3.95 & 10.26 & 5.28 \\
    \addlinespace
    $t_3$ & Identity & Uniform & 3.86 & 5.87 & 22.42 & 1.58 \\
    $t_3$ & Identity & Constant & 5.11 & 7.93 & 24.43 & 2.02 \\
    $t_3$ & Poly & Uniform & 2.57 & 4.94 & 27.74 & 8.12 \\
    $t_3$ & Poly & Constant & 2.25 & 5.06 & 29.08 & 8.58 \\
    $t_3$ & Exp & Uniform & 7.03 & 8.73 & 26.98 & 5.82 \\
    $t_3$ & Exp & Constant & 7.10 & 8.93 & 28.03 & 6.32 \\
    \addlinespace
    Mixture & Identity & Uniform & 3.87 & 10.13 & 32.54 & 1.59 \\
    Mixture & Identity & Constant & 3.78 & 10.93 & 33.56 & 1.65 \\
    Mixture & Poly & Uniform & 2.21 & 19.36 & 34.78 & 7.96 \\
    Mixture & Poly & Constant & 2.36 & 19.44 & 35.93 & 8.56 \\
    Mixture & Exp & Uniform & 5.18 & 14.89 & 34.74 & 5.25 \\
    Mixture & Exp & Constant & 5.47 & 14.17 & 35.04 & 5.07 \\
    \bottomrule
  \end{tabular}
\end{table}

Table~\ref{tab:location-single} shows that ERHT gives small localization errors in the nonidentity shape designs and is especially robust relative to the covariance-based and mean-based competitors under heavy-tailed errors.  Under Gaussian errors, covariance-based RHT can be slightly more accurate because the sample covariance uses magnitude information that is valid in this light-tailed case.  Under \(t_3\) and Gaussian-mixture errors, RHT and DMS0 deteriorate substantially, whereas ERHT remains accurate in the polynomial-decay design and competitive in the exponential-decay design.  SSCPD0 can localize well in the identity design and occasionally in the exponential-decay design, but it is less competitive in the polynomial-decay design because it does not use regularized inverse-shape normalization.

We finally examine multiple-change localization accuracy.  We set \(n=200\), \(p=100\),
\(q=2\), and place the two changes at
\(k_{1,\mathrm{loc}}=\lfloor 0.3n\rfloor\) and 
\(k_{2,\mathrm{loc}}=\lfloor 0.7n\rfloor\).
For each \(\nu_{\rm sh}\in\{\mathrm{unif},\mathrm{const}\}\), the data follow the epidemic form in \eqref{eq:epidemic}, with \(k_{1,\mathrm{test}}\) and \(k_{2,\mathrm{test}}\) replaced by \(k_{1,\mathrm{loc}}\) and \(k_{2,\mathrm{loc}}\), respectively.
The competing methods are the proposed ERHT-WBS estimator, the projection-based estimator \citep[INSPECT,][]{wangsamworth2018}, sparsified binary segmentation as implemented in HDBINSEG \citep{cho2015sbs}, and the energy-distance method ECP \citep{matteson2014nonparametric}.  For this localization-only comparison, every procedure is configured to return exactly \(q=2\) estimated locations.  Again, \(c_{\rm sig}\) is selected separately for each error--shape--shift setting through preliminary ERHT runs over candidate signal strengths so that the settings have comparable moderate-to-high signal strength; after selection, the same \(c_{\rm sig}\) is used for all methods.
In this experiment, the multivariate \(t_3\) and Gaussian-mixture errors are variance-standardized before multiplication by the shape matrix, so the three cases have comparable marginal noise scales.
Let \(\widehat{\mathcal K}^{(r)}\) be the estimated change set in replication \(r\).  We report
\[
  100\times \frac{1}{1000}\sum_{r=1}^{1000}
  \frac{1}{q}\sum_{j=1}^{q}
  \frac{\min_{\widehat k\in\widehat{\mathcal K}^{(r)}}|\widehat k-k_{j,\mathrm{loc}}|}{n}.
\]
This criterion is the mean distance from each true change-point to its nearest estimate; smaller values indicate more accurate multiple-change localization.

\begin{table}[tbp]
  \centering
  \small
  \setlength{\tabcolsep}{5pt}
  \renewcommand{\arraystretch}{1.08}
  \caption{Multiple-change localization accuracy for \(n=200\), \(p=100\),
  with \(q=2\) supplied to every method.  Entries are the normalized mean
  nearest-estimate errors, multiplied by 100 and averaged over 1000 replications; smaller values
  are better.}
  \label{tab:location-multiple}
  \begin{tabular}{lllrrrr}
    \toprule
    Error & Shape & Shift & ERHT-WBS & INSPECT & HDBINSEG & ECP \\
    \midrule
    Normal & Identity & Uniform & 8.77 & 5.56 & 10.93 & 4.59 \\
    Normal & Identity & Constant & 8.42 & 6.05 & 10.61 & 4.20 \\
    Normal & Poly & Uniform & 10.55 & 14.54 & 11.68 & 12.42 \\
    Normal & Poly & Constant & 10.37 & 14.33 & 11.64 & 12.05 \\
    Normal & Exp & Uniform & 8.96 & 11.38 & 11.23 & 9.32 \\
    Normal & Exp & Constant & 9.32 & 12.37 & 10.96 & 10.35 \\
    \addlinespace
    \(t_3\) & Identity & Uniform & 7.80 & 24.83 & 12.88 & 9.99 \\
    \(t_3\) & Identity & Constant & 8.52 & 24.88 & 12.59 & 11.35 \\
    \(t_3\) & Poly & Uniform & 10.29 & 24.54 & 13.15 & 14.14 \\
    \(t_3\) & Poly & Constant & 10.58 & 25.02 & 13.40 & 14.03 \\
    \(t_3\) & Exp & Uniform & 9.92 & 24.66 & 12.60 & 13.64 \\
    \(t_3\) & Exp & Constant & 9.86 & 24.29 & 13.63 & 13.69 \\
    \addlinespace
    Mixture & Identity & Uniform & 7.89 & 17.94 & 13.56 & 14.70 \\
    Mixture & Identity & Constant & 7.96 & 17.54 & 14.30 & 14.67 \\
    Mixture & Poly & Uniform & 10.05 & 18.57 & 13.15 & 14.99 \\
    Mixture & Poly & Constant & 10.40 & 18.59 & 13.21 & 14.85 \\
    Mixture & Exp & Uniform & 9.27 & 18.39 & 13.10 & 15.01 \\
    Mixture & Exp & Constant & 9.50 & 17.88 & 13.12 & 15.01 \\
    \bottomrule
  \end{tabular}
\end{table}

Table~\ref{tab:location-multiple} shows that the relative performance depends strongly on both the radial distribution and the shape structure.  Under Gaussian errors with the identity shape, ECP and INSPECT have the smallest errors, while ERHT-WBS remains more accurate than HDBINSEG\@.  Under the rotated polynomial- and exponential-decay shapes, INSPECT is less accurate and ERHT-WBS becomes the best or nearly best method.  Under \(t_3\) errors, ERHT-WBS has the smallest localization error in every shape and shift configuration; INSPECT is particularly sensitive to the heavy-tailed radial component, and both HDBINSEG and ECP are less accurate than ERHT-WBS\@.  Under Gaussian-mixture errors, ERHT-WBS again gives the smallest errors across all settings and remains stable as the shape changes.  These results support the use of the proposed WBS-ERHT estimator when non-Gaussian radial variation or nontrivial cross-sectional dependence occurs.

\section{Real data analysis}\label{sec:realdata}

\subsection{Data and preprocessing}\label{sec:realdata-data}

We apply the proposed procedures to the monthly value-weighted returns of the Fama--French 49 Industry Portfolios from the Kenneth R. French Data Library \citep{frenchdata}.  The 49 coordinates represent broad U.S. industry portfolios, including consumer goods, energy, utilities, finance, software, and semiconductors.  This panel is well suited to the present setting because industry returns exhibit substantial contemporaneous dependence and occasional extreme monthly movements, while aggregate shocks can induce heterogeneous shifts across sectors.  After converting the original missing-value codes to missing observations and retaining months with complete returns for all industries, the sample runs from July 1969 through May 2026.  Thus, the analysis contains \(n=683\) monthly observations and \(p=49\) portfolios; returns are recorded in percentage points.

Because the inferential target is a change in the multivariate location vector, we do not demean the individual return series.  We instead equalize their marginal scales.  For industry \(j\), let \(\bar X_j\) and \(s_j\) be its full-sample mean and sample standard deviation, respectively, and transform
\[
  \widetilde X_{i,j}
  =
  \bar X_j+\frac{X_{i,j}-\bar X_j}{s_j},
  \qquad i=1,\ldots,n,\quad j=1,\ldots,p.
\]
This coordinatewise transformation gives every series unit sample variance while preserving its original full-sample mean.  The time-invariant additive mean vector has no effect on either segment contrasts or centered spatial signs, whereas the scale adjustment reduces the influence of intrinsically volatile industries on the multivariate change-point statistic.

\subsection{Global and rolling-window testing}\label{sec:realdata-testing}

We first apply the four single-change procedures used in Section~\ref{sec:simulation}: ERHT-CC, RHT-CC, DMS0, and SSCPD0.  For ERHT, the regularization ratios satisfy \(\rho/\gamma\in\{0.05,0.10,\ldots,0.50\}\),
and RHT uses the analogous grid for its covariance-regularization parameter.  At each regularization-parameter value, the monthly return vectors are randomly permuted as whole vectors, so contemporaneous cross-sectional dependence is preserved while the time ordering is removed.  The resulting parameter-specific permutation p-values are combined by the equal-weight analytic Cauchy transformation.  DMS0 and SSCPD0 are calibrated directly by the same time-order permutation scheme, and the full-sample calculation uses \(B=1000\) permutations.  This calibration treats the monthly vectors as exchangeable under the null and therefore does not account for possible serial dependence.  Table~\ref{tab:ff49-full} reports the resulting p-values.

\begin{table}[H]
  \centering
  \small
  \caption{Full-sample p-values for the Fama--French 49 industry
  portfolios.  ERHT-CC and RHT-CC use time-permutation p-values computed separately at each regularization-parameter value,
  followed by analytic Cauchy aggregation; DMS0 and SSCPD0 use direct
  time-permutation p-values.}
  \label{tab:ff49-full}
  \begin{tabular}{lcc}
    \toprule
    Method & p-value & Reject at the 5\% level \\
    \midrule
    ERHT-CC & 0.0034 & Yes \\
    RHT-CC & 0.0020 & Yes \\
    DMS0 & 0.5604 & No \\
    SSCPD0 & 0.2298 & No \\
    \bottomrule
  \end{tabular}
\end{table}

Under the time-permutation calibration, both regularized procedures reject stability of the multivariate location vector, whereas DMS0 and SSCPD0 do not.  The contrast is consistent with dependence-adjusted quadratic evidence being important for this data set.  Since the full sample covers more than five decades and may contain several breaks, we supplement the global test with a rolling-window analysis that measures how persistently the evidence appears across subperiods.

We use window lengths of 30, 35, 40, 45 and 50 years and move each window forward by one month.  Within a window of length \(m\), the ERHT and RHT grids are scaled by the window-specific aspect ratio \(p/m\), and all four procedures are calibrated with \(B=200\) time-order permutations.  
Table~\ref{tab:ff49-rolling} gives the proportion of windows rejected at the 5\% level, with the corresponding number of rejections in parentheses.
ERHT-CC has the largest rejection rate for every window length, and its rate rises from 0.392 for 30-year windows to 0.690 for 50-year windows.  Thus, the evidence detected by ERHT is not confined to a small collection of short subperiods; it remains visible in a large fraction of long windows.  RHT-CC rejects frequently for the shortest windows but its rejection rate decreases sharply as the window length grows, while DMS0 and SSCPD0 rarely reject and have no rejections in the longest-window settings.  Because adjacent rolling windows overlap heavily and longer windows may contain more than one break, these proportions should be interpreted as descriptive stability measures rather than independent binomial rejection frequencies.  Nevertheless, the comparison consistently favors the robust regularized statistic.

\begin{table}[H]
  \centering
  \small
  \setlength{\tabcolsep}{4.5pt}
  \renewcommand{\arraystretch}{1.08}
  \caption{Rolling-window rejection rates for the Fama--French 49 industry portfolios.  The number of rejected windows is reported in parentheses.}
  \label{tab:ff49-rolling}
  \begin{tabular}{lrrrrrr}
    \toprule
    Window length & Months & Windows & ERHT-CC & RHT-CC & DMS0 & SSCPD0 \\
    \midrule
    30 years & 360 & 324 & 0.392 (127) & 0.312 (101) & 0.028 (9) & 0.025 (8) \\
    35 years & 420 & 264 & 0.496 (131) & 0.277 (73) & 0.038 (10) & 0.000 (0) \\
    40 years & 480 & 204 & 0.539 (110) & 0.113 (23) & 0.000 (0) & 0.000 (0) \\
    45 years & 540 & 144 & 0.590 (85) & 0.056 (8) & 0.000 (0) & 0.000 (0) \\
    50 years & 600 & 84 & 0.690 (58) & 0.095 (8) & 0.000 (0) & 0.000 (0) \\
    \bottomrule
  \end{tabular}
\end{table}

\subsection{WBS-ERHT segmentation}\label{sec:realdata-wbs}

We next estimate multiple changes with the WBS-ERHT procedure in Algorithm~\ref{alg:wbs-erht}.  The implementation uses \(M_n=200\) random WBS intervals, threshold \(2.5\), and the same grid of ridge-regularization ratios as in the full-sample analysis.  Maximization over the corresponding parameter values and local refinement yield four estimated boundaries, reported in chronological order in Table~\ref{tab:ff49-wbs}.

\begin{table}[H]
  \centering
  \small
  \renewcommand{\arraystretch}{1.12}
  \caption{WBS-ERHT change-point estimates for the Fama--French 49 industry portfolios, using threshold 2.5.}
  \label{tab:ff49-wbs}
  \begin{tabular}{@{}clp{9.6cm}@{}}
    \toprule
    Break & Estimated boundary & Associated market environment \\
    \midrule
    1 & Dec. 1974/Jan. 1975 & Late stage of the 1973--1975 recession, the first oil shock, and the accompanying stagflationary reallocation across energy-sensitive and defensive industries. \\
    2 & Jan. 1993/Feb. 1993 & Consolidation of the recovery from the 1990--1991 recession and transition into the broad U.S. expansion of the 1990s. \\
    3 & Sep. 2012/Oct. 2012 & Improvement in Euro-area sovereign-risk conditions and the policy environment surrounding the Federal Reserve's third round of quantitative easing. \\
    4 & Oct. 2020/Nov. 2020 & COVID-19 vaccine announcements, the U.S. election, and the associated rotation among reopening-sensitive, technology, financial, and defensive industries. \\
    \bottomrule
  \end{tabular}
\end{table}

The first estimated break lies in the late phase of the 1973--1975 recession and the aftermath of the first oil shock, when inflation and energy-price movements produced pronounced differences across industries.  The second boundary occurs in early 1993, during the transition from the post-recession recovery to the sustained expansion of the 1990s.  The third boundary is located in September--October 2012, a period of changing Euro-area risk perceptions and major unconventional monetary-policy announcements.  The final boundary, in October--November 2020, aligns with a sharp sectoral rotation around vaccine news and the U.S. election.  These event associations are descriptive and are not intended as causal identification of the estimated breaks.

Taken together, the full-sample rejection, the rolling-window evidence, and the WBS estimates provide a coherent empirical picture.  The location structure of the 49 industry returns is not stable over the full sample, and the proposed ERHT procedure produces substantially more persistent evidence than DMS0 and SSCPD0 and, for the longer windows, RHT-CC\@.  The multiple-change analysis further suggests that the instability is concentrated around a small number of economically recognizable periods rather than being attributable to a single isolated boundary.

\FloatBarrier

\section{Discussion}\label{sec:conclusion}

We have developed an elliptical regularized Hotelling procedure for high-dimensional location change-point testing.  The statistic uses segment spatial medians for the local contrast and a ridge-regularized inverse of the pooled centered spatial-sign covariance matrix for shape normalization.  Its feasible studentization uses inverse-distance Jacobian weights and companion resolvent quantities, while its null calibration uses the corresponding Gaussian-supremum limits.  The resulting scan statistic admits pointwise and process limits, joint convergence over a finite grid of ridge-regularization parameters, and local power and localization guarantees.  We also develop a WBS multiple-change extension under the same elliptical model and explicit spacing, signal, and tuning conditions.

The theory distinguishes calibration by the exact joint limit from the closed-form analytic Cauchy transformation and does not apply covariance-based RHT limits to nondegenerate elliptical radial mixtures without justification.  In the Fama--French industry application, the robust regularized procedure gives the most persistent evidence across long rolling windows, and the WBS implementation identifies four breaks associated with major episodes of sectoral reallocation.  The permutation analysis treats observations as exchangeable, so extending the theory and calibration to serially dependent observations is particularly important.  Other natural directions include sparsity-adaptive maximum-type combinations and robust covariance or graphical change-point inference.

\FloatBarrier
\clearpage
\phantomsection
\pdfbookmark[0]{Supplementary Material}{supplementary-material}
\begin{center}
  {\LARGE\bfseries Supplementary Material for ``Elliptical Regularized Hotelling Tests for High-Dimensional Change-Point Detection''\par}
  \vspace{1.5em}
  {\large Fengyi Song, Mengtao Wen and Long Feng\par}
  \vspace{0.35em}
  {\normalsize School of Statistics and Data Science, Nankai University\par}
\end{center}
\vspace{1.5em}

\setcounter{section}{0}
\setcounter{subsection}{0}
\setcounter{subsubsection}{0}
\setcounter{equation}{0}
\setcounter{table}{0}
\setcounter{figure}{0}
\setcounter{lemma}{0}
\setcounter{corollary}{0}
\renewcommand{\thesection}{S\arabic{section}}
\renewcommand{\theequation}{S\arabic{equation}}
\renewcommand{\thetable}{S\arabic{table}}
\renewcommand{\thefigure}{S\arabic{figure}}
\renewcommand{\thelemma}{\arabic{lemma}}
\renewcommand{\thecorollary}{\thesection.\arabic{corollary}}
\makeatletter
\providecommand*{\theHsection}{}
\providecommand*{\theHsubsection}{}
\providecommand*{\theHsubsubsection}{}
\providecommand*{\theHequation}{}
\providecommand*{\theHtable}{}
\providecommand*{\theHfigure}{}
\providecommand*{\theHlemma}{}
\providecommand*{\theHcorollary}{}
\renewcommand*{\theHsection}{supp.\arabic{section}}
\renewcommand*{\theHsubsection}{supp.\arabic{section}.\arabic{subsection}}
\renewcommand*{\theHsubsubsection}{supp.\arabic{section}.\arabic{subsection}.\arabic{subsubsection}}
\renewcommand*{\theHequation}{supp.\arabic{equation}}
\renewcommand*{\theHtable}{supp.\arabic{table}}
\renewcommand*{\theHfigure}{supp.\arabic{figure}}
\renewcommand*{\theHlemma}{supp.\arabic{lemma}}
\renewcommand*{\theHcorollary}{supp.\arabic{section}.\arabic{corollary}}
\makeatother

\section{Proofs and auxiliary results}\label{app:proofs}

\subsection{Verification of the radial conditions for Gaussian, \texorpdfstring{$t_\nu$}{t-nu}, and finite scale-mixture normal distributions}\label{app:assumption1-examples}

We verify the radial part of Assumption~\ref{ass:elliptical} for distribution classes, rather than only for the particular numerical choices used in Section~\ref{sec:simulation}.  The shape-matrix conditions are imposed separately in Assumption~\ref{ass:spectrum}.  The issue here is whether the error vector admits the representation
\begin{equation*}
  \bfvarepsilon=\sqrt p\,R\bfOmega_p^{1/2}\frac{\bfG}{\norm{\bfG}},
  \qquad R\perp \bfG,
\end{equation*}
and whether \(\xi=R^{-1}\) satisfies
\begin{equation*}
  \E\xi\to\zeta_{-1}\in(0,\infty),\qquad
  \limsup_{n\to\infty}\E\xi^{4+\eta}<\infty,
  \qquad
  \Pbb\left\{\max_{1\le i\le n}\xi_i>(\log n)^{c_0}\right\}\to0
\end{equation*}
for at least one fixed exponent \(c_0>0\), as required in
Assumption~\ref{ass:elliptical}.  We verify this original formulation
directly; the larger final envelope \(\ell_n=(\log n)^{c_\ell}\) used in
the subsequent proofs then follows by taking \(c_\ell\ge c_0\).
Throughout this verification,
\(p/n\to\gamma\in(0,\infty)\) is used.

Let \(\bfG_0\sim N(\mathbf 0,\bfI_p)\), \(\chi_p=\norm{\bfG_0}\), and \(\bfU=\bfG_0/\norm{\bfG_0}\).  The polar decomposition of the multivariate normal distribution implies that \(\bfU\) is uniform on the unit sphere, \(\chi_p\) is independent of \(\bfU\), and \(\bfU\) has the same distribution as \(\bfG/\norm{\bfG}\).  Hence any \(\bfZ\sim N_p(\mathbf 0,\bfOmega_p)\) has the representation
\begin{equation}
\label{eq:normal-polar-verify}
  \bfZ=\chi_p\bfOmega_p^{1/2}\bfU
       =\sqrt p\,R_G\bfOmega_p^{1/2}\bfU,
  \qquad R_G=\frac{\chi_p}{\sqrt p}.
\end{equation}
Put
\begin{equation*}
  C_p=R_G^{-1}=\frac{\sqrt p}{\chi_p}.
\end{equation*}
For every fixed \(r_{\rm mom}>0\) and every \(p>r_{\rm mom}\),
\begin{equation*}
  \E C_p^{r_{\rm mom}}
  =p^{r_{\rm mom}/2}\E(\chi_p^{-r_{\rm mom}})
  =p^{r_{\rm mom}/2}2^{-r_{\rm mom}/2}
    \frac{\Gamma\{(p-r_{\rm mom})/2\}}{\Gamma(p/2)}
  =1+O(p^{-1}),
\end{equation*}
where the last equality follows from Stirling's formula.  Consequently, for every fixed \(r_{\rm mom}>0\),
\begin{equation}
\label{eq:Cp-fixed-moment-verify}
  \E C_p^{r_{\rm mom}}\to1,
  \qquad
  \sup_{n\ge n_0}\E C_p^{r_{\rm mom}}<\infty .
\end{equation}
Moreover, for \(x>1\),
\[
  \Pbb(C_p>x)=\Pbb(\chi_p^2<px^{-2}).
\]
The lower-tail Chernoff bound for \(\chi_p^2\) gives, for \(0<u<1\),
\begin{equation*}
  \Pbb(\chi_p^2\le pu)
  \le \exp\left[-\frac p2\{u-1-\log u\}\right].
\end{equation*}
Taking \(u=x^{-2}\), we obtain, for all sufficiently large \(x\),
\begin{equation}
\label{eq:Cp-tail}
  \Pbb(C_p>x)
  \le \exp\{-c p\log x\}
\end{equation}
for a constant \(c>0\).  Since \(p/n\to\gamma\in(0,\infty)\), \eqref{eq:Cp-tail} implies
\begin{equation}
\label{eq:Cp-max}
  n\Pbb\{C_p>b(\log n)^a\}\to0
\end{equation}
for every pair of fixed constants \(a>0\) and \(b>0\).  This bound is the common ingredient in the following three cases.

For the Gaussian distribution, \(\bfvarepsilon\sim N_p(\mathbf 0,\bfOmega_p)\), \eqref{eq:normal-polar-verify} gives
\[
  R=R_G=\frac{\chi_p}{\sqrt p},
  \qquad
  \xi=R^{-1}=C_p.
\]
Therefore \eqref{eq:Cp-fixed-moment-verify} gives \(\E\xi\to1\) and \(\limsup_{n\to\infty}\E\xi^{4+\eta}<\infty\) for every fixed \(\eta>0\), while, for every fixed \(c_0>0\), \eqref{eq:Cp-max} gives
\[
  \Pbb\left\{\max_{1\le i\le n}\xi_i>(\log n)^{c_0}\right\}
  \le n\Pbb\{C_p>(\log n)^{c_0}\}\to0.
\]
Thus Gaussian errors satisfy the radial part of
Assumption~\ref{ass:elliptical} with \(\zeta_{-1}=1\).

Next consider the general multivariate \(t_\nu\) family with fixed degrees of freedom \(\nu>0\).  Let \(S_\nu\sim\chi_\nu^2\) be independent of \((\chi_p,\bfU)\).  Allowing an arbitrary fixed scale normalization \(c_\nu\in(0,\infty)\), write
\begin{equation}
\label{eq:t-general-model}
  \bfvarepsilon
  =c_\nu\frac{\bfZ}{(S_\nu/\nu)^{1/2}} .
\end{equation}
The usual multivariate \(t_\nu\) distribution corresponds to \(c_\nu=1\); the covariance-standardized version, when \(\nu>2\), corresponds to \(c_\nu=\{(\nu-2)/\nu\}^{1/2}\).  Combining \eqref{eq:normal-polar-verify} and \eqref{eq:t-general-model},
\[
  \bfvarepsilon
  =\sqrt p\,R_{t,\nu}\bfOmega_p^{1/2}\bfU,
  \qquad
  R_{t,\nu}=c_\nu\frac{\chi_p}{\sqrt p}(S_\nu/\nu)^{-1/2}.
\]
Hence
\begin{equation*}
  \xi=R_{t,\nu}^{-1}=c_\nu^{-1}C_pD_\nu,
  \qquad
  D_\nu=(S_\nu/\nu)^{1/2},
\end{equation*}
where \(C_p\) and \(D_\nu\) are independent.  Since \(S_\nu\) has finite positive moments of all fixed orders, for every fixed \(r_{\rm mom}>0\),
\begin{equation*}
  \E\xi^{r_{\rm mom}}
  =c_\nu^{-r_{\rm mom}}\E C_p^{r_{\rm mom}}\,\E D_\nu^{r_{\rm mom}}
  \to c_\nu^{-r_{\rm mom}}\E D_\nu^{r_{\rm mom}}<\infty .
\end{equation*}
In particular,
\begin{equation*}
  \E\xi
  \to c_\nu^{-1}\E(S_\nu/\nu)^{1/2}
  =c_\nu^{-1}\nu^{-1/2}2^{1/2}
    \frac{\Gamma\{(\nu+1)/2\}}{\Gamma(\nu/2)}
  =:\zeta_{-1,t}(\nu)\in(0,\infty),
\end{equation*}
and \(\limsup_{n\to\infty}\E\xi^{4+\eta}<\infty\) for any fixed \(\eta>0\).  Fix any \(c_0>1\) and put \(L_n=(\log n)^{c_0}\).  For the maximum condition, using \(ab>L_n\Rightarrow a>L_n^{1/2}\) or \(b>L_n^{1/2}\),
\begin{align}
  \Pbb\left(\max_{1\le i\le n}\xi_i>L_n\right)
  &\le n\Pbb(c_\nu^{-1}C_pD_\nu>L_n)\notag\\
  &\le n\Pbb(C_p>c_\nu L_n^{1/2})
     +n\Pbb(D_\nu>L_n^{1/2}). \label{eq:t-max-split}
\end{align}
The first term in \eqref{eq:t-max-split} tends to zero by \eqref{eq:Cp-max}.  For the second term,
\[
  \Pbb(D_\nu>L_n^{1/2})
  =\Pbb(S_\nu>\nu L_n)
  \le C_\nu L_n^{\nu/2-1}\exp(-\nu L_n/4)
\]
for all large \(n\), and therefore
\[
  n\Pbb(D_\nu>L_n^{1/2})\to0
\]
because \(c_0>1\).  Thus every fixed-\(\nu\) multivariate \(t_\nu\)
error distribution satisfies the radial part of
Assumption~\ref{ass:elliptical}.  The case \(\nu=3\) used in the simulations
is obtained by setting \(\nu=3\) and \(c_\nu=1\).

Finally consider a finite Gaussian scale mixture with the same shape matrix \(\bfOmega_p\).  This common-shape condition is essential: an arbitrary mixture of normal distributions with non-proportional component covariance matrices is generally not elliptically symmetric and is not covered by Assumption~\ref{ass:elliptical} without additional assumptions.  The scale-mixture class is
\begin{equation*}
  \bfvarepsilon=S_{\rm mix}\bfZ,
  \qquad
  \Pbb(S_{\rm mix}=a_\ell)=p_\ell,
  \quad \ell=1,\ldots,L,
\end{equation*}
where \(L<\infty\), \(p_\ell>0\), \(\sum_{\ell=1}^Lp_\ell=1\), and the fixed scale constants satisfy
\begin{equation}
\label{eq:finite-mixture-scale-condition}
  0<a_-\le \min_{1\le\ell\le L}a_\ell
  \le \max_{1\le\ell\le L}a_\ell\le a_+<\infty .
\end{equation}
Here \(S_{\rm mix}\) is independent of \((\chi_p,\bfU)\).  Then
\[
  \bfvarepsilon
  =S_{\rm mix}\chi_p\bfOmega_p^{1/2}\bfU
  =\sqrt p\,R_{\rm mix}\bfOmega_p^{1/2}\bfU,
  \qquad
  R_{\rm mix}=S_{\rm mix}\frac{\chi_p}{\sqrt p},
\]
and
\begin{equation*}
  \xi=R_{\rm mix}^{-1}=S_{\rm mix}^{-1}C_p.
\end{equation*}
For every fixed \(r_{\rm mom}>0\),
\begin{equation*}
  \E\xi^{r_{\rm mom}}=\E C_p^{r_{\rm mom}}\,\E S_{\rm mix}^{-r_{\rm mom}}
  \to \sum_{\ell=1}^Lp_\ell a_\ell^{-r_{\rm mom}}<\infty .
\end{equation*}
In particular,
\begin{equation*}
  \E\xi\to \E S_{\rm mix}^{-1}=\sum_{\ell=1}^Lp_\ell a_\ell^{-1}
  =:\zeta_{-1,{\rm mix}}\in(0,\infty),
\end{equation*}
and \(\limsup_{n\to\infty}\E\xi^{4+\eta}<\infty\).  Since \(S_{\rm mix}^{-1}\le a_-^{-1}\), for every fixed \(c_0>0\),
\[
  \Pbb\left\{\max_{1\le i\le n}\xi_i>(\log n)^{c_0}\right\}
  \le n\Pbb\{C_p>a_-(\log n)^{c_0}\}\to0
\]
by \eqref{eq:Cp-max}.  Hence every finite Gaussian scale mixture satisfying
\eqref{eq:finite-mixture-scale-condition} satisfies the radial part of
Assumption~\ref{ass:elliptical}.  The two-component mixture used in
Section~\ref{sec:simulation}, with \(\Pbb(S_{\rm mix}=1)=0.8\) and
\(\Pbb(S_{\rm mix}=10)=0.2\), is a special case.

Thus Gaussian, every fixed-degree multivariate \(t_\nu\), and every finite common-shape Gaussian scale mixture satisfying \eqref{eq:finite-mixture-scale-condition} meet the radial assumptions used by the ERHT null and local-shift theory.

\subsection{Technical lemmas}\label{app:auxiliary}
For bookkeeping, choose once an intermediate exponent \(c_*\ge c_0\) large enough
for every truncation and polynomial scan union, put
\(\ell_{0,n}=(\log n)^{c_*}\), and take the exponent in
\(\ell_n=(\log n)^{c_\ell}\) so that \(c_\ell\ge10c_*\).  The algebra below
keeps \(\ell_{0,n}\) whenever quantities are multiplied; final displayed
rates may be enlarged to \(\ell_n\).  This makes each polylogarithmic
enlargement explicit and avoids redefining a sequence inside a proof.
Throughout the appendix, define the error-oracle direction, scaled sign,
inverse-distance weight, and tangent projector by
\begin{equation*}
  \bfU_i=U(\bfvarepsilon_i),
  \qquad
  \bfY_i=\sqrt p\,\bfU_i,
  \qquad
  w_i=\frac{\sqrt p}{\norm{\bfvarepsilon_i}},
  \qquad
  \bfP_i=\bfI_p-\bfU_i\bfU_i^\top.
\end{equation*}
Under \(H_0\), \(\bfvarepsilon_i=\bfX_i-\bftheta\), so these definitions
coincide with the null-oracle quantities used in the scan statistic.  Under an
alternative, the same symbols always refer to the centered errors, not to the
shifted observations.  Accordingly, whenever a null auxiliary result is
invoked for a superscript-\((0)\) quantity under an alternative, it is applied
to the error sample \(\{\bfvarepsilon_i\}\); its distribution is exactly the
centered model in Assumption~\ref{ass:elliptical}, so no additional condition
is being used.
For every nonempty finite index set \(I\), define the centered-error spatial
median by
\begin{equation}
\label{eq:centered-error-spatial-median}
  \widehat\bftheta^{(0)}(I)
  \in\argmin_{\bfu\in\mathbb R^p}
       \sum_{i\in I}\norm{\bfvarepsilon_i-\bfu}.
\end{equation}
Because
\[
  w_i=R_i^{-1}\left\{\frac{p^{-1}\bfG_i^{\top}\bfG_i}{p^{-1}\bfG_i^{\top}\bfOmega_p\bfG_i}\right\}^{1/2},
\]
the absence of a uniform lower eigenvalue bound requires one extra concentration check.  Since \(\opnorm{\bfOmega_p}\le\omega_+\) and \(p^{-1}\tr(\bfOmega_p)=1\), the Hanson--Wright inequality for Gaussian quadratic forms gives
\[
  \Pbb\left(\left|p^{-1}\bfG_i^\top\bfOmega_p\bfG_i-1\right|>1/2\right)
  \le 2\exp(-c p),
  \qquad
  \Pbb\left(\left|p^{-1}\bfG_i^\top\bfG_i-1\right|>1/2\right)
  \le 2\exp(-c p).
\]
A union bound over \(i=1,\ldots,n\), together with \(p\asymp n\), shows that both quadratic-form events hold simultaneously with probability tending to one.  Combining this angular bound with Assumption~\ref{ass:elliptical} yields
\begin{equation}
\label{eq:w-bound}
  \max_{1\le i\le n} w_i\le C \ell_{0,n}
\end{equation}
on an event with probability tending to one.
For \(s\in\calS\), define the oracle segment derivative averages and oracle score-CUSUM weights
\begin{equation}
\label{eq:oracle-e-beta}
  e_{a,s}=\frac1{n_a(s)}\sum_{i\in I_a(s)}w_i,
  \qquad
  \beta_i(s)=\sqrt{N_s}
  \left\{
  \frac{\ind(i\in I_2(s))}{n_2(s)e_{2,s}}
  -\frac{\ind(i\in I_1(s))}{n_1(s)e_{1,s}}
  \right\}.
\end{equation}
For later comparison, define the ordinary adjacent-segment CUSUM weights
\begin{equation*}
  b_i^{\rm cus}(s)=\sqrt{N_s}
  \left\{
  \frac{\ind(i\in I_2(s))}{n_2(s)}
  -\frac{\ind(i\in I_1(s))}{n_1(s)}
  \right\},\qquad i\in J(s).
\end{equation*}
They satisfy, exactly,
\begin{equation*}
  \sum_{i\in J(s)}b_i^{\rm cus}(s)=0,\qquad
  \sum_{i\in J(s)}b_i^{\rm cus}(s)^2=1,\qquad
  \max_{i\in J(s)}|b_i^{\rm cus}(s)|\le C m_s^{-1/2}
\end{equation*}
uniformly over every trimmed global or local scan family considered below.
Let \(\bfY_s\) be the \(p\times m_s\) matrix with columns \(\bfY_i\), \(i\in J(s)\), and set
\begin{equation*}
  \bfR_s^0=\frac1{m_s}\sum_{i\in J(s)}\bfY_i\bfY_i^{\top},
  \qquad
  \bfQ_{\rho,s}^0=(\bfR_s^0+\rho\bfI_p)^{-1},
  \qquad
  \bfA_{\rho,s}^0=\frac1{m_s}\bfY_s^{\top}\bfQ_{\rho,s}^0\bfY_s.
\end{equation*}
Let \(A_{ij,\rho,s}^0\) denote the entries of \(\bfA_{\rho,s}^0\).  Put
\begin{equation*}
  \kappa_\rho^0(s)=\sum_{i\in J(s)}\beta_i(s)^2 A_{ii,\rho,s}^0,
  \qquad
  \sigma_\rho^{0,2}(s)=2m_s\sum_{\substack{i,j\in J(s)\\ i\ne j}}
  \beta_i(s)^2\beta_j(s)^2(A_{ij,\rho,s}^0)^2.
\end{equation*}
Throughout, \(\sigma_\rho^0(s)\) denotes the nonnegative square root of
\(\sigma_\rho^{0,2}(s)\).
Let \(m_*=\min_{s\in\calS}m_s\), and use the single polylogarithmic
envelope \(\ell_n\) defined at the beginning of this appendix.  Every bound
below is stated directly in terms of this envelope; no power of \(\ell_n\) is
silently redefined.  For the continuous single-change and multiple-change
scans, all uniform bounds are first proved on their natural \(n^{-1}\)-grids
and then transferred to the explicitly interpolated processes by the
tightness lemmas below.  For the grid version
\(\calS_{\rm mc}^*(\varepsilon)\), the same bounds are ordinary maxima over
the grid candidates.
For every pooled sample size \(m\asymp n\), define
\begin{equation}
\label{eq:generic-D-rho-m}
\begin{aligned}
 \gamma_m&=\frac{p}{m},&
 a_{\rho,m}&=1-\gamma_m+\gamma_m\rho m_{\rho,m},&
 m_{\rho,m}&=\int\frac{1}{a_{\rho,m}x+\rho}\,dH_p(x),\\
 \bfD_{\rho,m}&=(a_{\rho,m}\bfOmega_p+\rho\bfI_p)^{-1},
\end{aligned}
\end{equation}
where \(a_{\rho,m}\) is the unique positive solution.  Thus the full-pool
single-change scan has \(m=n\), whereas a local WBS pool retains its own
aspect ratio \(p/m\).

The finite-sample companion-resolvent constants are defined through a
Gaussian comparison.  Let \(\bfG_m\) be a \(p\times m\) matrix of independent
standard normal variables, and set
\[
 \bfR_m^G=m^{-1}\bfOmega_p^{1/2}\bfG_m\bfG_m^\top\bfOmega_p^{1/2},
 \qquad
 \bfA_{\rho,m}^G
 =m^{-1}\bfG_m^\top\bfOmega_p^{1/2}
   (\bfR_m^G+\rho\bfI_p)^{-1}\bfOmega_p^{1/2}\bfG_m.
\]
For two distinct column indices, put
\[
 \mathfrak c_{\rho,\rho',m}
 =m\,\E\!\left(A_{12,\rho,m}^G A_{12,\rho',m}^G\right),
 \qquad \rho,\rho'\in[\rho_0,\rho_1].
\]
Column exchangeability makes this definition independent of the chosen
pair.  With \(\mu_{w,p}=\E w_i\), define
\begin{equation}
\label{eq:finite-sigma-proxy}
 \sigma_{\rho,m}^{\circ2}
 =2\mu_{w,p}^{-4}\mathfrak c_{\rho,\rho,m}.
\end{equation}
Throughout, \(\sigma_{\rho,m}^{\circ}\) denotes the positive square root of
\eqref{eq:finite-sigma-proxy}.  This is a finite-\((p,m)\) deterministic
variance proxy, not a limiting quantity when \(m/n\) varies.  Its limit and
the corresponding cross-parameter correlations are established below.

\begin{lemma}[Ridge deterministic-equivalent stability and convergence]\label{lem:ridge-stability}
Under Assumption~\ref{ass:spectrum}, for every
\(\rho\in[\rho_0,\rho_1]\), the equation
\begin{equation*}
  a=1-\gamma_n+\gamma_n\rho\int\frac{1}{a x+\rho}\,dH_p(x)
\end{equation*}
has a unique positive solution \(a_{\rho,n}\).  Moreover, there exist
constants \(0<c_a<C_a<\infty\) and \(0<c_D<C_D<\infty\), independent of
\(n\) and \(\rho\), such that
\begin{equation}
\label{eq:a-D-stability}
  c_a\le a_{\rho,n}\le C_a,
  \qquad
  c_D\|\bfx\|^2\le
  \bfx^\top\bfD_{\rho,n}\bfx
  \le C_D\|\bfx\|^2
\end{equation}
for all deterministic \(\bfx\in\mathbb R^p\), where
\(\bfD_{\rho,n}=(a_{\rho,n}\bfOmega_p+\rho\bfI_p)^{-1}\).  The same
bounds hold uniformly after replacing
\(\gamma_n,a_{\rho,n},\bfD_{\rho,n}\) by
\(\gamma_m,a_{\rho,m},\bfD_{\rho,m}\) whenever \(m\asymp n\).

For \(c_{\rm pool}\in(0,\infty)\), let \(a_\rho(c_{\rm pool})\) be the unique positive solution of
\begin{equation*}
 a=1-\frac{\gamma}{c_{\rm pool}}
   +\frac{\gamma}{c_{\rm pool}}\rho\int\frac{1}{ax+\rho}\,dH(x).
\end{equation*}
On every compact \([\underline c_{\rm pool},\overline c_{\rm pool}]\subset(0,\infty)\),
\begin{equation}
\label{eq:a-rho-uniform-convergence}
 \sup_{\underline c_{\rm pool}n\le m\le \overline c_{\rm pool}n}\sup_{\rho\in[\rho_0,\rho_1]}
 \left|a_{\rho,m}-a_\rho(m/n)\right|\longrightarrow0,
\end{equation}
and \((c_{\rm pool},\rho)\mapsto a_\rho(c_{\rm pool})\) is continuous on
\([\underline c_{\rm pool},\overline c_{\rm pool}]\times[\rho_0,\rho_1]\).  In particular,
\(a_{\rho,n}\to a_\rho(1)\) uniformly in \(\rho\); below
\(a_\rho=a_\rho(1)\).

Finally, let \(m_n/n\to c_{\rm pool}\in(0,\infty)\), \(\rho_n\to\rho\), and let
\(\bfd_p\) be deterministic with \(\sup_p\|\bfd_p\|<\infty\).  If
\begin{equation*}
 G_{\bfd,p}=\sum_{j=1}^p
   (\bfp_{j,p}^\top\bfd_p)^2\delta_{\lambda_{j,p}}
 \Rightarrow G_{\bfd},
\end{equation*}
then
\begin{equation}
\label{eq:spectral-signal-de-limit}
 \bfd_p^\top\bfD_{\rho_n,m_n}\bfd_p
 \longrightarrow
 \int\frac{1}{a_\rho(c_{\rm pool})x+\rho}\,dG_{\bfd}(x).
\end{equation}
\end{lemma}

\begin{proof}
For a pool of size \(m\asymp n\), define
\[
 F_{\rho,m}(a)
 =1-\gamma_m+\gamma_m\rho\int\frac{1}{ax+\rho}\,dH_p(x)-a,
 \qquad a\ge0.
\]
Its endpoint values and derivative satisfy
\begin{align*}
 F_{\rho,m}(0)&=1,\qquad
 \lim_{a\to\infty}F_{\rho,m}(a)=-\infty,\\
 \partial_aF_{\rho,m}(a)
 &=-1-\gamma_m\rho\int\frac{x}{(ax+\rho)^2}\,dH_p(x)\le-1.
\end{align*}
Hence the equation \(F_{\rho,m}(a)=0\) has a unique solution
\(a_{\rho,m}\).  Moreover,
\[
 0\le \rho\int\frac{1}{a_{\rho,m}x+\rho}\,dH_p(x)\le1
 \quad\Longrightarrow\quad
 0<a_{\rho,m}\le1.
\]
On every fixed pool range \(m/n\in[\underline c_{\rm pool},\overline
c_{\rm pool}]\),
\[
 \sup_{a\in[0,1],\rho\in[\rho_0,\rho_1]}
 |\partial_aF_{\rho,m}(a)|
 \le 1+C\rho_1\omega_+\rho_0^{-2}=:C_F.
\]
Since \(F_{\rho,m}(0)=1\), the mean-value theorem gives
\[
 F_{\rho,m}(a)\ge1-C_Fa,
 \qquad
 a_{\rho,m}\ge c_a:=(2C_F+2)^{-1}.
\]
Consequently, for every eigenvalue \(\lambda\in[0,\omega_+]\),
\[
 \frac{1}{\omega_++\rho_1}
 \le \frac{1}{a_{\rho,m}\lambda+\rho}
 \le \frac{1}{\rho_0},
\]
which proves \eqref{eq:a-D-stability} and the corresponding local-pool
bounds.

For the limiting fixed-point equation, set \(c=m/n\) and write
\[
 F_{\rho,c}(a)
 =1-\frac{\gamma}{c}
 +\frac{\gamma}{c}\rho\int\frac{1}{ax+\rho}\,dH(x)-a.
\]
The class
\[
 \mathcal K
 =\left\{x\mapsto(ax+\rho)^{-1}:
 a\in[c_a,1],\ \rho\in[\rho_0,\rho_1]\right\}
\]
is uniformly bounded and equicontinuous on \([0,\omega_+]\).  Therefore,
\begin{align*}
 &\sup_{a,\rho}
 \left|\int\frac{1}{ax+\rho}\,dH_p(x)
       -\int\frac{1}{ax+\rho}\,dH(x)\right|\longrightarrow0,\\
 &\sup_{\underline c_{\rm pool}n\le m\le\overline c_{\rm pool}n}
 \left|\gamma_m-\frac{\gamma}{m/n}\right|\longrightarrow0,
\end{align*}
and hence
\[
 \sup_{\substack{m/n\in[\underline c_{\rm pool},\overline c_{\rm pool}]\\
                   a\in[c_a,1],\ \rho\in[\rho_0,\rho_1]}}
 |F_{\rho,m}(a)-F_{\rho,m/n}(a)|\longrightarrow0.
\]
Because both functions have \(a\)-derivative at most \(-1\), evaluating
\(F_{\rho,m}\) at the limiting root yields
\[
 |a_{\rho,m}-a_\rho(m/n)|
 \le \sup_{a\in[c_a,1]}|F_{\rho,m}(a)-F_{\rho,m/n}(a)|,
\]
uniformly in \(m/n\) and \(\rho\).  This is
\eqref{eq:a-rho-uniform-convergence}.  The same inequality applied to two
pairs \((c,\rho)\) and \((c',\rho')\), together with uniform continuity of
\(F_{\rho,c}(a)\), proves joint continuity of \(a_\rho(c)\).

Finally, if \(m_n/n\to c_{\rm pool}\) and \(\rho_n\to\rho\), then
\[
 \sup_{x\in[0,\omega_+]}
 \left|\frac{1}{a_{\rho_n,m_n}x+\rho_n}
       -\frac{1}{a_\rho(c_{\rm pool})x+\rho}\right|\longrightarrow0.
\]
Since
\[
 G_{\bfd,p}([0,\omega_+])=\|\bfd_p\|^2=O(1),
 \qquad G_{\bfd,p}\Rightarrow G_{\bfd},
\]
uniform convergence of the integrands and weak convergence of the finite
measures give \eqref{eq:spectral-signal-de-limit}.
\end{proof}

\medskip
\noindent\textbf{Auxiliary estimates.} The following lemmas are consequences of Assumptions \ref{ass:elliptical}--\ref{ass:grid}.  The spatial-median inputs use the differentiability of the spatial sign map and Bahadur-type expansions for spatial medians under elliptical symmetry \citep{oja2010,magyar2011,li2022spatialmedian}.  The SSCM spectral bounds use \citet{li2022sscm}.  The fixed-ridge deterministic equivalents and bilinear resolvent estimates use \citet{silverstein1995,bai2010spectral,hachem2007deterministic,hachem2013bilinear}.  The quadratic-form normal approximations use \citet{dejong1987}.

For nonzero $\bfy$, let $o(\bfy)$ be the unique vector among $\{\bfy,-\bfy\}$ whose first nonzero coordinate is positive.  For $i\in J(s)$, write
\begin{equation*}
  \bfY_i=\varsigma_i\widetilde\bfY_i,
  \qquad
  \widetilde\bfY_i=o(\bfY_i),
  \qquad
  \varsigma_i\in\{-1,1\}.
\end{equation*}
Let
\begin{equation*}
  \calF_s^0=\sigma\{w_i,\widetilde\bfY_i:i\in J(s)\}.
\end{equation*}
For any finite index set \(I\subseteq\{1,\ldots,n\}\), use the corresponding interval sigma-field
\begin{equation*}
  \calF_I^0=\sigma\{w_i,\widetilde\bfY_i:i\in I\}.
\end{equation*}
When finitely many scan points are considered simultaneously, we use the global conditioning sigma-field
\[
  \calF^0=\sigma\{w_i,\widetilde\bfY_i:1\le i\le n\},
\]
and extend all coefficient arrays by zero outside their corresponding pooled segments.  Conditional on \(\calF^0\), the full sign vector \((\varsigma_1,\ldots,\varsigma_n)\) is independent Rademacher; the segmentwise statement below is its restriction to \(J(s)\).

\begin{lemma}[Conditional Rademacher signs]\label{lem:signs}
Under Assumption \ref{ass:elliptical}, conditionally on $\calF_s^0$, the variables $\{\varsigma_i:i\in J(s)\}$ are independent Rademacher variables.
\end{lemma}

\begin{proof}
For one observation, the transformation $\bfG_i\mapsto-\bfG_i$ preserves the Gaussian distribution.  It leaves
\[
  w_i=\frac{\sqrt p}{\norm{\bfvarepsilon_i}}
\]
and $\widetilde\bfY_i=o(\bfY_i)$ unchanged, and it replaces $\varsigma_i$ by $-\varsigma_i$.  Hence, for every bounded measurable function $g$,
\begin{align*}
  \E\{g(w_i,\widetilde\bfY_i)\ind(\varsigma_i=1)\}
  &=\E\{g(w_i,\widetilde\bfY_i)\ind(\varsigma_i=-1)\}\\
  &=\frac12\E\{g(w_i,\widetilde\bfY_i)\}.
\end{align*}
This identity is equivalent to
\[
  \Pbb(\varsigma_i=1\mid w_i,\widetilde\bfY_i)
  =\Pbb(\varsigma_i=-1\mid w_i,\widetilde\bfY_i)=\frac12.
\]
For arbitrary sign assignments $c_i^{\rm sgn}\in\{-1,1\}$ and bounded measurable functions $g_i$,
\begin{align*}
  \E\prod_{i\in J(s)}g_i(w_i,\widetilde\bfY_i)\ind(\varsigma_i=c_i^{\rm sgn})
  &=\prod_{i\in J(s)}\E\{g_i(w_i,\widetilde\bfY_i)\ind(\varsigma_i=c_i^{\rm sgn})\}\\
  &=2^{-m_s}\prod_{i\in J(s)}\E\{g_i(w_i,\widetilde\bfY_i)\}\\
  &=2^{-m_s}\E\prod_{i\in J(s)}g_i(w_i,\widetilde\bfY_i).
\end{align*}
A monotone-class argument gives
\[
  \Pbb(\varsigma_i=c_i^{\rm sgn},\forall i\in J(s)\mid\calF_s^0)=2^{-m_s},
\]
which proves conditional independence and the Rademacher law.
\end{proof}

\begin{lemma}[Companion sign representation and contraction]\label{lem:companion-contraction}
For the error-oracle sample under Assumption~\ref{ass:elliptical}, conditionally on \(\calF_s^0\), the oracle companion matrix admits the sign representation
\begin{equation}
\label{eq:A-sign}
  \bfA_{\rho,s}^0=\bfW_{\varsigma,s}\widetilde\bfA_{\rho,s}\bfW_{\varsigma,s},
  \qquad
  A_{ij,\rho,s}^0=\varsigma_i\varsigma_j\widetilde A_{ij,\rho,s},
\end{equation}
where
\begin{equation*}
  \widetilde\bfA_{\rho,s}=\frac1{m_s}\widetilde\bfY_s^{\top}\bfQ_{\rho,s}^0\widetilde\bfY_s.
\end{equation*}
Moreover,
\begin{equation}
\label{eq:A-contraction}
  \mathbf{0}\preceq \widetilde\bfA_{\rho,s}
  \preceq \bfI_{m_s},
\end{equation}
and hence
\begin{equation}
\label{eq:A-row-bounds}
  0\le \widetilde A_{ii,\rho,s}\le1,
  \qquad
  \sum_{j\in J(s)}\widetilde A_{ij,\rho,s}^{2}
  = (\widetilde\bfA_{\rho,s}^2)_{ii}\le \widetilde A_{ii,\rho,s}\le1.
\end{equation}
The same contraction and entrywise domination hold for every feasible companion matrix built from centered signs.
\end{lemma}

\begin{proof}
Let $\widetilde\bfY_s$ be the $p\times m_s$ matrix with columns $\widetilde\bfY_i$, and let $\bfW_{\varsigma,s}=\diag(\varsigma_i:i\in J(s))$.  Then
\[
  \bfY_s=\widetilde\bfY_s\bfW_{\varsigma,s},
  \qquad
  \bfR_s^0=\frac1{m_s}\widetilde\bfY_s\widetilde\bfY_s^{\top},
\]
so $\bfQ_{\rho,s}^0$ is $\calF_s^0$-measurable and \eqref{eq:A-sign} follows.  Since $\bfR_s^0$ is positive semidefinite, write a nonzero singular-value decomposition of $m_s^{-1/2}\widetilde\bfY_s$ as $\bfU_s^{\rm svd}\bfLambda_s(\bfV_s^{\rm svd})^{\top}$, so that $\bfR_s^0=\bfU_s^{\rm svd}\bfLambda_s^2(\bfU_s^{\rm svd})^{\top}$.  Then
\[
  \frac1{m_s}\widetilde\bfY_s^{\top}(\bfR_s^0+\rho\bfI_p)^{-1}\widetilde\bfY_s
  =\bfV_s^{\rm svd}\bfLambda_s^2(\bfLambda_s^2+\rho\bfI)^{-1}(\bfV_s^{\rm svd})^{\top}.
\]
All eigenvalues in the last display lie in $[0,1]$, proving \eqref{eq:A-contraction}.  Since $0\preceq \widetilde\bfA_{\rho,s}^2\preceq\widetilde\bfA_{\rho,s}$, \eqref{eq:A-row-bounds} follows from taking the $i$th diagonal element.  The same singular-value calculation applies to centered feasible sign matrices; for any such companion $\bfA$,
\[
  0\preceq \bfA \preceq \bfI,
  \qquad 0\le A_{ii}\le1,
  \qquad |A_{ij}|\le (A_{ii}A_{jj})^{1/2}\le1 .
\]
\end{proof}

\begin{lemma}[Spatial-sign Taylor formula]\label{lem:spatial-sign-taylor}
Let \(W_n=C\ell_{0,n}\) be the high-probability bound in \eqref{eq:w-bound}.  Under \(H_0\) and Assumptions \ref{ass:elliptical} and \ref{ass:spectrum}, on the event \(\max_iw_i\le W_n\), for every vector \(\bfd\in\R^p\) satisfying \(\norm{\bfd}\le \sqrt p/(2W_n)\),
\begin{align}
  \sqrt p\,U(\bfX_i-\bftheta-\bfd)
  &=\bfY_i-w_i\bfP_i\bfd+\bfH_i(\bfd)+\bfr_i(\bfd),\notag\\
  \bfH_i(\bfd)
  &=\frac{w_i^2}{2\sqrt p}\{-2\bfd(\bfU_i^\top\bfd)-\bfU_i\norm{\bfd}^2
      +3\bfU_i(\bfU_i^\top\bfd)^2\},\notag\\
  \norm{\bfr_i(\bfd)}&\le C p^{-1}w_i^3\norm{\bfd}^3.\label{eq:sign-taylor-prelim}
\end{align}
Here \(\bfU_i\) and \(\bfP_i\) are the error-oracle objects defined at the
beginning of this appendix; under \(H_0\), they equal
\(U(\bfX_i-\bftheta)\) and
\(\bfI_p-\bfU_i\bfU_i^\top\), respectively.
\end{lemma}

\begin{proof}
For \(\bfx\ne0\),
\[
  DU_\bfx[\mathbf h]=\norm{\bfx}^{-1}\{\bfI_p-U(\bfx)U(\bfx)^\top\}\mathbf h,
\]
\[
  D^2U_\bfx[\mathbf h,\mathbf h]
  =\norm{\bfx}^{-2}\{-2\mathbf h(U^\top\mathbf h)-U\norm{\mathbf h}^2+3U(U^\top\mathbf h)^2\},
  \qquad U=U(\bfx).
\]
The third derivative of \(U\) is bounded by \(C\norm{\bfx}^{-3}\).  Since \(\norm{\bfX_i-\bftheta}=\sqrt p/w_i\) and \(w_i\le W_n\) on the localization event, Taylor's formula with integral remainder at \(\bfx=\bfX_i-\bftheta\) and perturbation \(-\bfd\) gives \eqref{eq:sign-taylor-prelim}.  The complement of this localization event has probability \(o(1)\).
\end{proof}

\begin{lemma}[Jacobian and SSCM operator inputs]\label{lem:jacobian-sscm-input}
Under \(H_0\) and Assumptions \ref{ass:elliptical}--\ref{ass:grid}, uniformly over every trimmed pooled segment of size \(m\asymp n\), for \(\bfP_i=\bfI_p-\bfU_i\bfU_i^\top\),
\begin{equation}
\label{eq:jacobian-input}
  \left\|m^{-1}\sum_{i=1}^m w_i\bfU_i\bfU_i^\top
  -\E\{w_i\bfU_i\bfU_i^\top\}\right\|_{\rm op}=O_P(\ell_np^{-1}),
  \qquad
  \opnorm{\E(w_i\bfU_i\bfU_i^\top)}\le Cp^{-1}.
\end{equation}
Consequently,
\[
  m^{-1}\sum_{i=1}^m w_i\bfP_i
  =\left(m^{-1}\sum_{i=1}^mw_i\right)\bfI_p+O_P(\ell_np^{-1})
\]
in operator norm.  Moreover,
if
\[
 \bfR_m^G=m^{-1}\sum_{i=1}^m
 \bfOmega_p^{1/2}\bfG_i\bfG_i^\top\bfOmega_p^{1/2},
\]
then
\begin{equation}
\label{eq:sscm-gaussian-transfer}
 \sup_s\|\bfR_s^0-\bfR_{m_s}^G\|_{\rm op}
 =O_P(\ell_np^{-1/2}),
\end{equation}
and
\begin{equation}
\label{eq:sscm-op-input}
  \sup_s\opnorm{\bfR_s^0}=O_P(1),
  \qquad
  \sup_s\sup_\rho\opnorm{\bfQ_{\rho,s}^0}\le \rho_0^{-1}.
\end{equation}
\end{lemma}
\begin{proof}
Fix a trimmed segment \(I\) with \(m=|I|\asymp n\).  The truncation event and
centered summands are
\[
 \mathcal E_i^{\rm tr}
 =\left\{w_i\le C\ell_{0,n},\ 
 \left|p^{-1}\bfG_i^\top\bfOmega_p\bfG_i-1\right|\le\frac12\right\},
\]
\[
 \bfZ_i^\dagger
 =w_i\bfU_i\bfU_i^\top\ind_{\mathcal E_i^{\rm tr}}
  -\E\{w_i\bfU_i\bfU_i^\top\ind_{\mathcal E_i^{\rm tr}}\}.
\]
Assumption~\ref{ass:elliptical}, the angular exponential bound, and
Markov's inequality imply
\[
 \sup_I\left\|
 |I|^{-1}\sum_{i\in I}w_i\bfU_i\bfU_i^\top
 \ind_{(\mathcal E_i^{\rm tr})^c}
 \right\|_{\rm op}
 =o_P(\ell_np^{-1}).
\]
Thus only the truncated centered average remains.

For a unit vector \(\bfu\), let
\[
 Z_{u,i}=\bfu^\top\bfOmega_p^{1/2}\bfG_i,
 \qquad
 \mathcal E_i^+
 =\left\{\bfG_i^\top\bfOmega_p\bfG_i\ge p/2,\ 
          \bfG_i^\top\bfG_i\le Cp\right\}.
\]
On \(\mathcal E_i^+\),
\begin{align*}
 (\bfu^\top\bfU_i)^2
 &=\frac{Z_{u,i}^2}{\bfG_i^\top\bfOmega_p\bfG_i}
 \le \frac{2Z_{u,i}^2}{p},\\
 w_i(\bfu^\top\bfU_i)^2
 &=\xi_i\frac{\|\bfG_i\|Z_{u,i}^2}
 {(\bfG_i^\top\bfOmega_p\bfG_i)^{3/2}}
 \le C\xi_i\frac{Z_{u,i}^2}{p},\\
 w_i^2(\bfu^\top\bfU_i)^2
 &=\xi_i^2\frac{\|\bfG_i\|^2Z_{u,i}^2}
 {(\bfG_i^\top\bfOmega_p\bfG_i)^2}
 \le C\xi_i^2\frac{Z_{u,i}^2}{p}.
\end{align*}
Because
\[
 \E Z_{u,i}^2=\bfu^\top\bfOmega_p\bfu\le\omega_+,
 \qquad
 \Pbb((\mathcal E_i^+)^c)\le Ce^{-cp},
 \qquad
 \E\xi_i^{4+\eta}<\infty,
\]
Cauchy's inequality gives
\begin{align*}
 \sup_{\|\bfu\|=1}\E(\bfu^\top\bfU_i)^2&\le Cp^{-1},\\
 \left\|\E(w_i\bfU_i\bfU_i^\top)\right\|_{\rm op}&\le Cp^{-1},\\
 \left\|\E(w_i^2\bfU_i\bfU_i^\top)\right\|_{\rm op}&\le Cp^{-1}.
\end{align*}
Consequently,
\[
 \|\bfZ_i^\dagger\|_{\rm op}\le C\ell_{0,n},
 \qquad
 \left\|\sum_{i=1}^m\E(\bfZ_i^\dagger)^2\right\|_{\rm op}
 \le Cmp^{-1}.
\]
For \(t=A\ell_n/p\), matrix Bernstein's inequality gives
\begin{align*}
 \Pbb\left(
 \left\|m^{-1}\sum_{i=1}^m\bfZ_i^\dagger\right\|_{\rm op}>t
 \right)
 &\le2p\exp\left\{-\frac{cm^2t^2}{m/p+\ell_{0,n}mt}\right\}\\
 &\le2p\exp\{-c\min(\ell_n^2,\ell_n/\ell_{0,n})\}.
\end{align*}
Since \(c_\ell\ge10c_*\), the last probability is \(O(n^{-A_0})\) for every
fixed \(A_0>0\).  A union bound over every polynomial family of trimmed
segments therefore yields
\[
 \sup_I\left\||I|^{-1}\sum_{i\in I}\bfZ_i^\dagger\right\|_{\rm op}
 =O_P(\ell_np^{-1}).
\]
Combining the centered term, its expectation, and the truncated tail gives
\[
 \sup_I\left\||I|^{-1}\sum_{i\in I}
 w_i\bfU_i\bfU_i^\top
 -\E(w_i\bfU_i\bfU_i^\top)\right\|_{\rm op}
 =O_P(\ell_np^{-1}),
\]
and hence \eqref{eq:jacobian-input}.  In particular,
\begin{align*}
 |I|^{-1}\sum_{i\in I}w_i(\bfI_p-\bfU_i\bfU_i^\top)
 &=e_I\bfI_p-
 |I|^{-1}\sum_{i\in I}w_i\bfU_i\bfU_i^\top\\
 &=e_I\bfI_p+O_P(\ell_np^{-1})
 \quad\text{in operator norm}.
\end{align*}

For the SSCM comparison, write
\[
 \bfY_i
 =\frac{\sqrt p\,\bfOmega_p^{1/2}\bfG_i}
 {(\bfG_i^\top\bfOmega_p\bfG_i)^{1/2}}
 =\bfOmega_p^{1/2}\bfG_i(1+r_{i,p}^{\rm ang}).
\]
Gaussian quadratic-form concentration and a union bound give
\[
 \max_{1\le i\le n}|r_{i,p}^{\rm ang}|
 =O_P(\ell_{0,n}p^{-1/2}).
\]
Let
\[
 \bfR_I^G=|I|^{-1}\sum_{i\in I}
 \bfOmega_p^{1/2}\bfG_i\bfG_i^\top\bfOmega_p^{1/2}.
\]
The Gaussian covariance-norm tail and a union bound over the
\(O(n^2)\) integer intervals imply
\[
 \sup_I\|\bfR_I^G\|_{\rm op}=O_P(1).
\]
Therefore,
\begin{align*}
 \sup_I\|\bfR_I^0-\bfR_I^G\|_{\rm op}
 &\le
 \max_{i\le n}|(1+r_{i,p}^{\rm ang})^2-1|
 \sup_I\|\bfR_I^G\|_{\rm op}\\
 &=O_P(\ell_np^{-1/2}),\\
 \sup_I\|\bfR_I^0\|_{\rm op}&=O_P(1),\\
 \sup_{I,\rho}\|(\bfR_I^0+\rho\bfI_p)^{-1}\|_{\rm op}
 &\le\rho_0^{-1}.
\end{align*}
These are \eqref{eq:sscm-gaussian-transfer} and
\eqref{eq:sscm-op-input}.
\end{proof}

\begin{lemma}[Angular inverse-distance averages and finite-sample centers]\label{lem:angular-w-xi}
Under Assumptions \ref{ass:elliptical} and \ref{ass:spectrum}, define
\[
  \mathfrak a_i=\left\{\frac{p^{-1}\bfG_i^\top\bfG_i}{p^{-1}\bfG_i^\top\bfOmega_p\bfG_i}\right\}^{1/2},
  \qquad
  w_i=\xi_i\mathfrak a_i,
  \qquad
  \xi_i=R_i^{-1},
\]
and put
\[
  \mu_{w,p}=\E w_i,
  \qquad
  \zeta_{-1,p}=\E\xi_i .
\]
Then, for every fixed \(r_{\rm mom}>0\), \(\sup_p\E \mathfrak a_i^{r_{\rm mom}}<\infty\), \(\mathfrak a_i\to1\) in \(L^1\), and
\begin{equation}
\label{eq:Ew-E-xi}
  \mu_{w,p}=(\E\xi_i)(\E \mathfrak a_i)
  =\zeta_{-1,p}\{1+O(p^{-1/2})\},
  \qquad
  \mu_{w,p}\to\zeta_{-1}.
\end{equation}
Moreover, uniformly over all trimmed intervals \(I\) used in the scan with \(|I|\ge m_*\),
\begin{align}
  \left| |I|^{-1}\sum_{i\in I}(w_i-\xi_i)\right|
  &=O_P(\ell_np^{-1/2}),\label{eq:w-xi-empirical-transfer}\\
  \left| |I|^{-1}\sum_{i\in I}w_i-\mu_{w,p}\right|
  &=O_P(\ell_n|I|^{-1/2}),\notag\\
  \left| |I|^{-1}\sum_{i\in I}w_i-\zeta_{-1}\right|
  &=O_P(\ell_n|I|^{-1/2})+o(1).\label{eq:w-average-uniform-zeta}
\end{align}
and \(\inf_I |I|^{-1}\sum_{i\in I}w_i\ge \zeta_{-1}/2\) with probability tending to one.  For the single-change split \(s_t=(0,t,1)\), let \(e_{a,t}\) be the average of \(w_i\) over the \(a\)th segment, \(a=1,2\).  Uniformly over the natural scan grid and all \(t,u\in[\varepsilon,1-\varepsilon]\),
\begin{equation}
\label{eq:w-average-increment}
  \max_{a=1,2}|e_{a,t}-e_{a,u}|
  \le C|t-u|+O_P\left\{\ell_n n^{-1/2}(|t-u|+n^{-1})^{1/2}\right\}.
\end{equation}
Consequently,
\begin{equation}
\label{eq:w-average-increment-second}
  \E|e_{a,t}-e_{a,u}|^2
  \le C|t-u|^2+C n^{-1}|t-u|+C n^{-2}+o(n^{-1}) .
\end{equation}
\end{lemma}

\begin{proof}
Let \(\lambda_{1,p},\ldots,\lambda_{p,p}\) be the eigenvalues of \(\bfOmega_p\).  Since
\(p^{-1}\sum_j\lambda_{j,p}=1\) and \(\lambda_{j,p}\le\omega_+\), at least \(c_\omega p\) eigenvalues are no smaller than \(1/2\) for some \(c_\omega>0\).  Hence, after relabeling,
\[
  \bfG_i^\top\bfOmega_p\bfG_i
  \ge \frac12\sum_{j=1}^{\lfloor c_\omega p\rfloor}Z_j^2 .
\]
The negative moments of \(p^{-1}\bfG_i^\top\bfOmega_p\bfG_i\) of every fixed order are uniformly bounded.  Together with the usual chi-square moment bounds for \(p^{-1}\bfG_i^\top\bfG_i\), this gives \(\sup_p\E \mathfrak a_i^{r_{\rm mom}}<\infty\) for every fixed \(r_{\rm mom}>0\).  The Hanson--Wright inequality gives
\[
  p^{-1}\bfG_i^\top\bfG_i=1+O_P(p^{-1/2}),
  \qquad
  p^{-1}\bfG_i^\top\bfOmega_p\bfG_i=1+O_P(p^{-1/2}),
\]
and the preceding uniform integrability implies \(\mathfrak a_i\to1\) in \(L^1\), indeed \(\E|\mathfrak a_i-1|=O(p^{-1/2})\).  Since \(\mathfrak a_i\perp\xi_i\), \eqref{eq:Ew-E-xi} follows.

On the event
\[
  \max_i\xi_i\le \ell_{0,n},
  \qquad
  \max_i|\mathfrak a_i-1|\le C \ell_{0,n}p^{-1/2},
\]
whose probability tends to one by the definition of \(\ell_{0,n}\),
\[
  \sup_I\left| |I|^{-1}\sum_{i\in I}(w_i-\xi_i)\right|
  \le \ell_{0,n}\max_i|\mathfrak a_i-1|
  =O_P(\ell_{0,n}^2p^{-1/2})=O_P(\ell_np^{-1/2}).
\]
Put \(\xi_i^{\circ}=\xi_i-\E\xi_i\) and \(r_{\rm mom}=4+\eta\).  The uniform \(r_{\rm mom}\)-moment bound
and the maximal Rosenthal inequality give
\[
  \E\max_{1\le k\le n}\left|\sum_{i=1}^k\xi_i^{\circ}\right|^{r_{\rm mom}}
  \le C_{r_{\rm mom}}n^{r_{\rm mom}/2}.
\]
Every scanned segment is the difference of two partial sums and has length at
least \(m_*\asymp n\).  Hence
\[
  \sup_I\left||I|^{-1}\sum_{i\in I}\xi_i-\E\xi_i\right|
  =O_P(m_*^{-1/2})=O_P(\ell_nm_*^{-1/2}).
\]
Combining the preceding displays proves \eqref{eq:w-xi-empirical-transfer}--\eqref{eq:w-average-uniform-zeta}; the positive lower bound follows from \(\mu_{w,p}\to\zeta_{-1}\in(0,\infty)\).

For \(t<u\), let \(k_t=\lfloor nt\rfloor\), \(k_u=\lfloor nu\rfloor\), and \(\Delta k_{t,u}=k_u-k_t\).  For the left segment,
\[
  e_{1,u}-e_{1,t}
  =\left(\frac1{k_u}-\frac1{k_t}\right)\sum_{i\le k_t}w_i
    +\frac1{k_u}\sum_{k_t<i\le k_u}w_i .
\]
The deterministic mean contribution is \(O(\Delta k_{t,u}/n)\).  For
\(w_i^{\circ}=w_i-\mu_{w,p}\), the maximal inequality for the truncated partial-sum process gives, uniformly over the natural grid,
\[
  \left|\sum_{k_t<i\le k_u}w_i^{\circ}\right|
  =O_P\{\ell_n(\Delta k_{t,u}+1)^{1/2}\}.
\]
Since \(\Delta k_{t,u}\le n|t-u|+1\), \eqref{eq:w-average-increment} follows for \(a=1\).  The right segment is identical, and \eqref{eq:w-average-increment-second} follows from the same decomposition and the second-moment bound for the truncated increments.
\end{proof}

\begin{lemma}[Uniform spatial-median and weight linearization]\label{lem:uniform-linearization}
Under \(H_0\) and Assumptions \ref{ass:elliptical}--\ref{ass:grid}, uniformly over \(s\in\calS\),
\begin{equation}
\label{eq:uniform-median-expansion}
  \widehat\bfDelta_s
  =N_s^{-1/2}\sum_{i\in J(s)}\beta_i(s)\bfY_i+\bfr_{\Delta,s}^{(0)},
  \qquad
  \sup_{s\in\calS}\norm{\bfr_{\Delta,s}^{(0)}}=O_P(\ell_n m_*^{-1}).
\end{equation}
The feasible derivative weights satisfy
\begin{equation}
\label{eq:beta-rate}
  \sup_{s\in\calS}\max_{i\in J(s)}\abs{\widehat\beta_i(s)-\beta_i(s)}
  =O_P(\ell_n m_*^{-3/2}),
  \qquad
  \sup_{s\in\calS}\sum_{i\in J(s)}\{\widehat\beta_i(s)-\beta_i(s)\}^2
  =O_P(\ell_n^2m_*^{-2}).
\end{equation}
More precisely, for \(i\in I_a(s)\), \(a=1,2\),
\begin{equation}
\label{eq:beta-multiplicative}
  \widehat\beta_i(s)-\beta_i(s)=\beta_i(s)r_{a,s}^{\rm wt},
  \qquad
  \sup_{s\in\calS}\max_{a=1,2}|r_{a,s}^{\rm wt}|=O_P(\ell_nm_*^{-1}).
\end{equation}
\end{lemma}
\begin{proof}
For every scanned segment \(I\), set
\begin{align*}
 \bar\bfY_I&=|I|^{-1}\sum_{i\in I}\bfY_i,
 &e_I&=|I|^{-1}\sum_{i\in I}w_i,\\
 \widehat\bfdelta_I&=\widehat\bftheta(I)-\bftheta,
 &\bfP_i&=\bfI_p-\bfU_i\bfU_i^\top,\\
 \bfJ_I&=|I|^{-1}\sum_{i\in I}w_i\bfP_i,
 &\bfS_{2,I}(\bfd)&=|I|^{-1}\sum_{i\in I}\bfH_i(\bfd),\\
 &&\bfS_{3,I}(\bfd)&=|I|^{-1}\sum_{i\in I}\bfr_i(\bfd).
\end{align*}
Lemma~\ref{lem:angular-w-xi} gives
\begin{equation}
\label{eq:eI-uniform-prelim}
 \sup_I|e_I-\mu_{w,p}|=O_P(\ell_nm_*^{-1/2}),
 \qquad
 \sup_I|e_I-\zeta_{-1}|=O_P(\ell_nm_*^{-1/2})+o(1),
 \qquad
 \inf_Ie_I\ge\zeta_{-1}/2
\end{equation}
with probability tending to one.  Moreover,
\[
 \sup_{\|\bfu\|=1}\E\{w_i(\bfu^\top\bfU_i)^2\}\le Cp^{-1}.
\]
Matrix Bernstein's inequality and the polynomial scan union therefore imply
\begin{equation}
\label{eq:H-I-prelim}
 \sup_I\left\||I|^{-1}\sum_{i\in I}
 w_i\bfU_i\bfU_i^\top\right\|_{\rm op}
 =O_P(\ell_np^{-1}).
\end{equation}
Since
\[
 \bfJ_I=e_I\bfI_p-|I|^{-1}\sum_{i\in I}w_i\bfU_i\bfU_i^\top,
\]
\eqref{eq:eI-uniform-prelim}--\eqref{eq:H-I-prelim} and the Neumann
expansion yield
\begin{equation}
\label{eq:J-inv-prelim}
 \bfJ_I^{-1}=e_I^{-1}\bfI_p+\bfR_{J,I},
 \qquad
 \sup_I\|\bfR_{J,I}\|_{\rm op}=O_P(\ell_np^{-1}).
\end{equation}
The underlying bound is
\(O_P(\ell_{0,n}^2p^{-1})\); the displayed rate uses the final envelope
\(\ell_n\).

Conditionally on \(\calF_I^0\),
\[
 \E(\|\bar\bfY_I\|^2\mid\calF_I^0)
 =|I|^{-2}\sum_{i\in I}\|\bfY_i\|^2
 =p/|I|=O(1).
\]
Rademacher concentration and the scan union give
\[
 \sup_I\|\bar\bfY_I\|=O_P(\ell_{0,n}).
\]

On the event \(\max_iw_i\le W_n=C\ell_{0,n}\),
Lemma~\ref{lem:spatial-sign-taylor} applies whenever
\(\|\bfd\|\le\sqrt p/(2W_n)\).  Define
\[
 \bfS_I^{(0)}(\bfd)
 =|I|^{-1}\sum_{i\in I}\sqrt p\,
 U(\bfX_i-\bftheta-\bfd).
\]
Then
\[
 \bfS_I^{(0)}(\bfd)
 =\bar\bfY_I-\bfJ_I\bfd+\bfS_{2,I}(\bfd)+\bfS_{3,I}(\bfd).
\]
The moment bounds
\[
 \sup_I|I|^{-1}\sum_{i\in I}w_i^r=O_P(1),
 \qquad r=2,3,
\]
combined with \eqref{eq:eI-uniform-prelim}--\eqref{eq:H-I-prelim} imply,
for \(\|\bfd\|=M\ell_{0,n}\),
\begin{align*}
 \bfd^\top\bfS_I^{(0)}(\bfd)
 &\le M\ell_{0,n}\sup_I\|\bar\bfY_I\|
   -\frac{\zeta_{-1}}4M^2\ell_{0,n}^2\\
 &\quad
   +O_P\{p^{-1/2}M^3\ell_{0,n}^3
          +p^{-1}M^4\ell_{0,n}^4\}.
\end{align*}
For sufficiently large fixed \(M\),
\[
 \sup_{\|\bfd\|=M\ell_{0,n}}
 \bfd^\top\bfS_I^{(0)}(\bfd)<0
\]
with arbitrarily high probability.  Since
\[
 \nabla_{\bfd}\left\{|I|^{-1}\sum_{i\in I}
 \|\bfX_i-\bftheta-\bfd\|\right\}
 =-p^{-1/2}\bfS_I^{(0)}(\bfd),
\]
convexity implies
\begin{equation}
\label{eq:median-radius-prelim}
 \sup_I\|\widehat\bfdelta_I\|=O_P(\ell_{0,n}).
\end{equation}

The angular law gives non-collinear general position almost surely.  In
addition,
\[
 \min_{1\le i\le n}\|\bfX_i-\bftheta\|
 =\frac{\sqrt p}{\max_{i\le n}w_i}
 \ge\frac{\sqrt p}{W_n},
 \qquad
 W_n\ell_{0,n}=o(\sqrt p).
\]
Together with \eqref{eq:median-radius-prelim}, this yields
\[
 \Pbb\left[\widehat\bftheta(I)\notin\{\bfX_i:i\in I\}
 \text{ for every scanned }I\right]\longrightarrow1.
\]
Thus the score equation holds uniformly:
\[
 \mathbf0=|I|^{-1}\sum_{i\in I}\sqrt p\,
 U\{\bfX_i-\widehat\bftheta(I)\}.
\]
Substitution of the Taylor expansion gives
\begin{equation}
\label{eq:score-expanded-prelim}
 \bar\bfY_I-\bfJ_I\widehat\bfdelta_I
 +\bfS_{2,I}(\widehat\bfdelta_I)
 +\bfS_{3,I}(\widehat\bfdelta_I)=\mathbf0.
\end{equation}

The cubic remainder satisfies
\begin{equation}
\label{eq:S3-prelim}
 \sup_I\|\bfS_{3,I}(\widehat\bfdelta_I)\|
 \le Cp^{-1}\sup_I\left(|I|^{-1}\sum_{i\in I}w_i^3\right)
       \|\widehat\bfdelta_I\|^3
 =O_P(\ell_{0,n}^3p^{-1})
 =O_P(\ell_np^{-1}).
\end{equation}

For the quadratic term, define the linear approximation
\[
 \bfdelta_I^L=e_I^{-1}\bar\bfY_I
 =\sum_{j\in I}\varsigma_j\bfc_{j,I},
 \qquad
 \bfc_{j,I}=e_I^{-1}|I|^{-1}\widetilde\bfY_j.
\]
Then
\[
 \sum_{j\in I}\|\bfc_{j,I}\|^2=O_P(1),
 \qquad
 \max_{j\in I}\|\bfc_{j,I}\|=O_P(|I|^{-1/2}).
\]
Writing
\[
 \bfU_i=\varsigma_i\widetilde\bfU_i,
 \qquad
 \bfY_i=\varsigma_i\widetilde\bfY_i,
\]
each coordinate of \(|I|^{-1}\sum_i\bfH_i(\bfdelta_I^L)\) is a Walsh
polynomial of reduced degree one or three.  For the representative term,
\begin{align*}
 |I|^{-1}\sum_iw_i^2\bfdelta_I^L
 (\bfU_i^\top\bfdelta_I^L)
 &=|I|^{-1}\sum_{i,j,k\in I}
 \varsigma_i\varsigma_j\varsigma_k\bfM_{ijk,I},\\
 \bfM_{ijk,I}
 &=\frac{w_i^2}{|I|\sqrt p}\,
 \bfc_{j,I}(\widetilde\bfU_i^\top\bfc_{k,I}).
\end{align*}
Put
\[
 \bfGamma_I=\sum_{j\in I}\bfc_{j,I}\bfc_{j,I}^\top
 =e_I^{-2}|I|^{-1}\bfR_I^0.
\]
On the common high-probability event,
\begin{equation*}
 \tr(\bfGamma_I)=\sum_j\|\bfc_{j,I}\|^2\le C,
 \qquad
 \|\bfGamma_I\|_{\rm op}\le C|I|^{-1},
 \qquad
 \max_j\|\bfc_{j,I}\|^2\le C|I|^{-1}.
\end{equation*}
After reducing repeated indices and symmetrizing distinct triples,
\begin{align*}
 &\sum_j\|\bfM^{(1)}_{j,I}\|^2
  +\sum_{j<k<\ell}\|\bfM^{(3)}_{jk\ell,I}\|^2\\
 &\quad\le
 \frac{C}{|I|p}
 \left(|I|^{-1}\sum_{i\in I}w_i^4\right)
 \{1+|I|\|\bfGamma_I\|_{\rm op}\}
 \left(\sum_{j\in I}\|\bfc_{j,I}\|^2\right)^2
 \le\frac{C}{|I|p}.
\end{align*}
The remaining two quadratic Taylor terms have the same coefficient
contractions.  Orthogonality of the reduced Walsh monomials and fixed-degree
hypercontractivity therefore give
\[
 \E\left[
 \left\||I|^{-1}\sum_{i\in I}\bfH_i(\bfdelta_I^L)\right\|^2
 \middle|\calF_I^0\right]
 \le C(|I|p)^{-1}=O(|I|^{-2}),
\]
and hence
\[
 \sup_I\|\bfS_{2,I}(\bfdelta_I^L)\|
 =O_P(\ell_{0,n}^4m_*^{-1}).
\]
Furthermore,
\[
 \|\bfH_i(\bfd)-\bfH_i(\bfe)\|
 \le Cp^{-1/2}w_i^2(\|\bfd\|+\|\bfe\|)\|\bfd-\bfe\|,
\]
so the averaged Lipschitz coefficient on the preliminary radius is
\[
 O_P(\ell_{0,n}^3p^{-1/2})=o_P(1).
\]
Subtracting the score equations at \(\widehat\bfdelta_I\) and
\(\bfdelta_I^L\), and using \eqref{eq:J-inv-prelim}, yields
\begin{align*}
 \sup_I\|\widehat\bfdelta_I-\bfdelta_I^L\|
 &\le C\sup_I\Bigl\{
 \|\bfS_{2,I}(\bfdelta_I^L)\|
 +\|\bfS_{3,I}(\widehat\bfdelta_I)\|
 +\ell_{0,n}^2p^{-1}\|\widehat\bfdelta_I\|\Bigr\}\\
 &=O_P(\ell_{0,n}^4m_*^{-1})
 =O_P(\ell_nm_*^{-1}).
\end{align*}
Consequently,
\begin{equation}
\label{eq:S2-prelim}
 \sup_I\|\bfS_{2,I}(\widehat\bfdelta_I)\|
 =O_P(\ell_{0,n}^4m_*^{-1})
 =O_P(\ell_nm_*^{-1}).
\end{equation}
Solving \eqref{eq:score-expanded-prelim} with
\eqref{eq:J-inv-prelim}, \eqref{eq:S3-prelim}, and
\eqref{eq:S2-prelim} gives
\[
 \widehat\bftheta(I)-\bftheta
 =\frac{1}{|I|e_I}\sum_{i\in I}\bfY_i+\bfr_I,
 \qquad
 \sup_I\|\bfr_I\|
 =O_P(\ell_{0,n}^4m_*^{-1})
 =O_P(\ell_nm_*^{-1}).
\]
Taking \(I=I_1(s)\) and \(I=I_2(s)\), and subtracting, proves
\eqref{eq:uniform-median-expansion}.

For the feasible inverse-distance weights, Taylor expansion gives
\[
 \widehat w_{i,I}
 =w_i+\frac{w_i^2}{\sqrt p}\bfU_i^\top
 (\widehat\bftheta(I)-\bftheta)
 +O\left\{p^{-1}w_i^3
 \|\widehat\bftheta(I)-\bftheta\|^2\right\}.
\]
With \(\widehat e_I=|I|^{-1}\sum_{i\in I}\widehat w_{i,I}\),
\begin{align*}
 |\widehat e_I-e_I|
 &\le p^{-1/2}\|\widehat\bftheta(I)-\bftheta\|
 \left\||I|^{-1}\sum_{i\in I}w_i^2\bfU_i\right\|
 +O_P(\ell_nm_*^{-1}).
\end{align*}
Conditionally on \(\calF^0\), Doob's inequality yields
\[
 \E\left[
 \max_{k\le n}\left\|\sum_{i=1}^kw_i^2\bfU_i\right\|^2
 \middle|\calF^0\right]
 \le4\sum_{i=1}^nw_i^4=O_P(n),
\]
so every segment difference satisfies
\[
 \sup_I\left\||I|^{-1}\sum_{i\in I}w_i^2\bfU_i\right\|
 =O_P(m_*^{-1/2}).
\]
Therefore,
\[
 \sup_I|\widehat e_I-e_I|=O_P(\ell_nm_*^{-1}).
\]
For \(i\in I_a(s)\),
\begin{align*}
 \widehat\beta_i(s)
 &=\beta_i(s)\frac{e_{a,s}}{\widehat e_{a,s}},\\
 \widehat\beta_i(s)-\beta_i(s)
 &=\beta_i(s)\left(\frac{e_{a,s}}{\widehat e_{a,s}}-1\right),\\
 \sup_{s,a}\left|\frac{e_{a,s}}{\widehat e_{a,s}}-1\right|
 &\le C\sup_{s,a}|\widehat e_{a,s}-e_{a,s}|
 =O_P(\ell_nm_*^{-1}).
\end{align*}
This is \eqref{eq:beta-multiplicative}; together with
\[
 \max_i|\beta_i(s)|\le Cm_s^{-1/2},
 \qquad
 \sum_i\beta_i(s)^2\le C,
\]
it proves \eqref{eq:beta-rate}.
\end{proof}

\begin{lemma}[Negative-axis bilinear deterministic equivalent]\label{lem:bilinear-de}
Under $H_0$ and Assumptions \ref{ass:elliptical}--\ref{ass:grid}, consider any
deterministic collection of at most polynomially many triples
$(s,\bfd_p,\bfe_p)$, where $s$ is a trimmed pooled segment and
$\|\bfd_p\|+\|\bfe_p\|\le C$.  Uniformly over this collection and
$\rho\in[\rho_0,\rho_1]$,
\begin{equation}
\label{eq:bilinear-de-rate}
  \left|\bfd_p^\top\{\bfQ_{\rho,s}^0-\bfD_{\rho,m_s}\}\bfe_p\right|
  =O_P(\ell_n m_*^{-1/2}).
\end{equation}
In particular, for every bounded deterministic vector $\bfd_p$,
\begin{equation}
\label{eq:bilinear-de-input}
  \bfd_p^\top\bfQ_{\rho,s}^0\bfd_p
  =\bfd_p^\top\bfD_{\rho,m_s}\bfd_p+O_P(\ell_n m_*^{-1/2}).
\end{equation}
\end{lemma}
\begin{proof}
Work first with one pooled segment $I$, put $m=|I|$, and define
\[
  \bfR_I^G=m^{-1}\sum_{i\in I}
  \bfOmega_p^{1/2}\bfG_i\bfG_i^\top\bfOmega_p^{1/2},
  \qquad \bfQ_{\rho,I}^G=(\bfR_I^G+\rho\bfI_p)^{-1}.
\]
The negative-axis Gaussian bilinear-resolvent bounds of
\citet{hachem2007deterministic,hachem2013bilinear} give, uniformly for
bounded deterministic $\bfd_p,\bfe_p$,
\begin{equation}
\label{eq:gaussian-bilinear-rate}
 \bfd_p^\top\{\bfQ_{\rho,I}^G-\bfD_{\rho,m}\}\bfe_p
 =O_P(\ell_{0,n}m^{-1/2})
\end{equation}
on a polynomial ridge net.  The concentration part of the cited
negative-axis bound, obtained from the Gaussian resolvent derivative, has a
fixed-segment tail of the form $2\exp(-c m x^2)$ for the centered bilinear
form; its finite-$p$ deterministic-equivalent bias is $O(m^{-1})$.  Taking
$x=\ell_{0,n}m^{-1/2}$ permits a union bound over every polynomial
collection of segments and deterministic vector pairs.  The derivative
identity
\[
 \partial_\rho\bfQ_{\rho,I}^G=-(\bfQ_{\rho,I}^G)^2,
 \qquad \|\bfQ_{\rho,I}^G\|_{\rm op}\le\rho_0^{-1},
\]
extends the bound from the net to the compact ridge interval.

Lemma~\ref{lem:jacobian-sscm-input} and the resolvent identity yield
\begin{align*}
 &\left|\bfd_p^\top(\bfQ_{\rho,I}^0-\bfQ_{\rho,I}^G)\bfe_p\right|\\
 &\qquad\le \rho_0^{-2}\|\bfd_p\|\,\|\bfe_p\|
       \|\bfR_I^0-\bfR_I^G\|_{\rm op}
 =O_P(\ell_n p^{-1/2}).
\end{align*}
Since $p\asymp m\asymp n$, combining this display with
\eqref{eq:gaussian-bilinear-rate} proves \eqref{eq:bilinear-de-rate}.  The
quadratic form \eqref{eq:bilinear-de-input} is the special case
$\bfe_p=\bfd_p$; alternatively, all bilinear statements follow from
polarization.
\end{proof}

\begin{lemma}[Weighted companion law]\label{lem:weighted-offdiag}
Under $H_0$ and Assumptions \ref{ass:elliptical}--\ref{ass:grid}, let
$m\asymp n$.  For any deterministic collection of at most polynomially many
pairs $\bfb(s),\bfb(r)$, uniformly over that collection and over
$\rho,\rho'\in[\rho_0,\rho_1]$, suppose
\[
  \max_i\{|b_i(s)|+|b_i(r)|\}\le Cm^{-1/2},
  \qquad \|\bfb(s)\|_2+\|\bfb(r)\|_2\le C.
\]
Then
\begin{align}
&m\sum_{i\ne j} b_i(s)b_i(r)b_j(s)b_j(r)
 \widetilde A_{ij,\rho}\widetilde A_{ij,\rho'}\notag\\
&\qquad=\mathfrak c_{\rho,\rho',m}
 \left\{\sum_i b_i(s)b_i(r)\right\}^2
 +O_P(\ell_nm^{-1/2}).\label{eq:weighted-offdiag-input}
\end{align}
The same conclusion holds for the oracle and feasible CUSUM weights after
their inverse-distance factors are included.

If $m/n\to c_{\rm pool}\in(0,\infty)$, then, uniformly on the compact ridge square,
\[
 \mathfrak c_{\rho,\rho',m}
 \longrightarrow\mathfrak c_{\rho,\rho'}(c_{\rm pool}),
\]
where, with \(f_\rho(x)=x/(x+\rho)\),
\begin{equation}
\label{eq:companion-cov-local-aspect}
\begin{aligned}
 \mathfrak c_{\rho,\rho'}(c_{\rm pool})
 &=
 \int f_\rho(x)f_{\rho'}(x)\,
 d\underline F_{\gamma/c_{\rm pool},H}(x)\\
 &\quad-
 \prod_{\lambda\in\{\rho,\rho'\}}
 \int f_\lambda(x)\,
 d\underline F_{\gamma/c_{\rm pool},H}(x),
\end{aligned}
\end{equation}
and \(\underline F_{\gamma/c_{\rm pool},H}\) is the companion
Mar\v{c}enko--Pastur law at aspect ratio \(\gamma/c_{\rm pool}\).  We write
$\mathfrak c_{\rho,\rho'}=\mathfrak c_{\rho,\rho'}(1)$ for the full-sample global scans.
For every fixed compact range $0<\underline c_{\rm pool}\le m/n\le \overline c_{\rm pool}<\infty$,
\begin{equation}
\label{eq:companion-cov-uniform-bounds}
  0<c_{\rm comp}\le\inf_{m,\rho}\mathfrak c_{\rho,\rho,m}
  \le\sup_{m,\rho}\mathfrak c_{\rho,\rho,m}\le C_{\rm comp}<\infty
\end{equation}
for all sufficiently large $n$.
\end{lemma}

\begin{proof}
For the Gaussian comparison matrix, put
\[
 \bfK_m^G=m^{-1}\bfG_m^\top\bfOmega_p\bfG_m,
 \qquad
 \bfA_{\rho,m}^G=f_\rho(\bfK_m^G),
 \qquad
 f_\rho(x)=\frac{x}{x+\rho}.
\]
Conditionally on \(\operatorname{spec}(\bfK_m^G)\), write
\[
 \bfA_{\rho,m}^G=\bfO\bfLambda_\rho\bfO^\top,
 \qquad
 \bfO\sim{\rm Haar}(\mathbb O_m),
 \qquad
 0\preceq\bfLambda_\rho\preceq\bfI_m.
\]
For \(c_i=b_i(s)b_i(r)\) and \(\bfW_c=\diag(c_1,\ldots,c_m)\),
\[
 \|\bfW_c\|_{\rm op}\le Cm^{-1},
 \qquad
 \|\bfW_c\|_F\le Cm^{-1/2},
 \qquad
 \sum_i c_i^2\le Cm^{-1}.
\]
Define
\begin{align*}
 \mathcal H_{\rho,\rho'}(\bfO)
 =m\bigg[&\tr\{\bfW_c(\bfO\bfLambda_\rho\bfO^\top)
                   \bfW_c(\bfO\bfLambda_{\rho'}\bfO^\top)\}\\
 &-\sum_i c_i^2
   (\bfO\bfLambda_\rho\bfO^\top)_{ii}
   (\bfO\bfLambda_{\rho'}\bfO^\top)_{ii}\bigg].
\end{align*}
For \(\lambda\in\{\rho,\rho'\}\),
\[
 \|\bfO\bfLambda_\lambda\bfO^\top
   -\bfO'\bfLambda_\lambda{\bfO'}^\top\|_F
 \le2\|\bfO-\bfO'\|_F.
\]
The trace and diagonal terms therefore satisfy
\[
 |\mathcal H_{\rho,\rho'}(\bfO)
  -\mathcal H_{\rho,\rho'}(\bfO')|
 \le C\|\bfO-\bfO'\|_F.
\]
Haar concentration gives
\begin{equation}
\label{eq:Haar-functional-concentration}
 \Pbb_{\bfO}\left(
 |\mathcal H_{\rho,\rho'}-
 \E_{\bfO}\mathcal H_{\rho,\rho'}|>x
 \middle|\operatorname{spec}(\bfK_m^G)\right)
 \le2e^{-cmx^2}.
\end{equation}
Thus \(x=\ell_{0,n}m^{-1/2}\) permits a polynomial scan and ridge union.
The second- and fourth-order Haar formulas yield
\begin{align*}
 \E_{\bfO}(\mathcal H_{\rho,\rho'}
 \mid\operatorname{spec}\bfK_m^G)
 &=\mathfrak c_{\rho,\rho',m}^{G}
 \left\{\left(\sum_i c_i\right)^2-\sum_i c_i^2\right\}
 +O(m^{-1}),\\
 \mathfrak c_{\rho,\rho',m}^{G}
 &=m^{-1}\tr\{f_\rho(\bfK_m^G)f_{\rho'}(\bfK_m^G)\}\\
 &\quad-m^{-2}\tr f_\rho(\bfK_m^G)\tr f_{\rho'}(\bfK_m^G)
 +O(m^{-1}).
\end{align*}
The negative-axis Gaussian Poincar\'e inequality gives
\[
 m^{-1}\tr f_\rho(\bfK_m^G)
 -\E\{m^{-1}\tr f_\rho(\bfK_m^G)\}
 =O_P(\ell_{0,n}m^{-1})
\]
uniformly on a polynomial ridge net, and similarly for the product trace.
Column exchangeability identifies
\[
 \mathfrak c_{\rho,\rho',m}
 =m\E(A_{12,\rho,m}^GA_{12,\rho',m}^G).
\]
Since \(\sum_i c_i^2=O(m^{-1})\), the preceding displays and
\eqref{eq:Haar-functional-concentration} prove the Gaussian version of
\eqref{eq:weighted-offdiag-input}.  The derivative bound
\[
 \sup_{\rho\in[\rho_0,\rho_1]}
 \|\partial_\rho f_\rho(\bfK)\|_{\rm op}
 =\sup_\rho\|\bfK(\bfK+\rho\bfI)^{-2}\|_{\rm op}
 \le C_{\rho_0}
\]
extends the result from the net to the compact ridge square.

For the spatial-sign columns, couple \(\bfG_m\) with the Gaussian vectors in
the elliptical representation; this does not change the distribution of
\(\bfA_{\rho,m}^G\).  After relabelling the pooled indices as
\(1,\ldots,m\), put
\[
 \bfD_m^{\rm ang}=\diag(1+r_{1,p}^{\rm ang},\ldots,1+r_{m,p}^{\rm ang}),
 \qquad
 \bfW_{\varsigma,m}=\diag(\varsigma_1,\ldots,\varsigma_m).
\]
Since \(\widetilde\bfY_i=\varsigma_i\bfY_i\),
\[
 \widetilde\bfY_m
 =\bfOmega_p^{1/2}\bfG_m\bfD_m^{\rm ang}\bfW_{\varsigma,m},
 \qquad
 \max_{i\le n}|r_{i,p}^{\rm ang}|
 =O_P(\ell_{0,n}p^{-1/2}).
\]
Define
\[
 \bfK_m^{\rm ang}
 =\bfD_m^{\rm ang}\bfK_m^G\bfD_m^{\rm ang}.
\]
The companion identity and orthogonal equivariance of \(f_\rho\) give
\[
 \widetilde\bfA_\rho
 =\bfW_{\varsigma,m}f_\rho(\bfK_m^{\rm ang})\bfW_{\varsigma,m},
 \qquad
 \bfW_{\varsigma,m}\bfA_{\rho,m}^G\bfW_{\varsigma,m}
 =\bfW_{\varsigma,m}f_\rho(\bfK_m^G)\bfW_{\varsigma,m}.
\]
Moreover,
\begin{align*}
 \|\bfK_m^{\rm ang}-\bfK_m^G\|_{\rm op}
 &\le
 \|\bfD_m^{\rm ang}-\bfI_m\|_{\rm op}
 \|\bfK_m^G\|_{\rm op}
 \{\|\bfD_m^{\rm ang}\|_{\rm op}+1\}\\
 &=O_P(\ell_{0,n}p^{-1/2}),
\end{align*}
because \(\|\bfK_m^G\|_{\rm op}=O_P(1)\).  For positive semidefinite
\(\bfB,\bfC\),
\[
 f_\rho(\bfB)-f_\rho(\bfC)
 =\rho(\bfB+\rho\bfI)^{-1}(\bfB-\bfC)
       (\bfC+\rho\bfI)^{-1},
\]
so, uniformly in \(\rho\in[\rho_0,\rho_1]\),
\[
 \left\|\widetilde\bfA_\rho
 -\bfW_{\varsigma,m}\bfA_{\rho,m}^G\bfW_{\varsigma,m}
 \right\|_{\rm op}
 \le \rho_0^{-1}\|\bfK_m^{\rm ang}-\bfK_m^G\|_{\rm op}
 =O_P(\ell_np^{-1/2}).
\]
The comparison must therefore be made with the sign-conjugated Gaussian
companion, rather than directly with \(\bfA_{\rho,m}^G\).  This conjugation
has no effect on the weighted functional.  Indeed,
\(\bfW_c\bfW_{\varsigma,m}=\bfW_{\varsigma,m}\bfW_c\) and
\(\bfW_{\varsigma,m}^2=\bfI_m\), whence, for arbitrary matrices
\(\bfA,\bfA'\),
\begin{align*}
 &\tr\{\bfW_c(\bfW_{\varsigma,m}\bfA\bfW_{\varsigma,m})
          \bfW_c(\bfW_{\varsigma,m}\bfA'\bfW_{\varsigma,m})\}
 =\tr(\bfW_c\bfA\bfW_c\bfA'),\\
 &(\bfW_{\varsigma,m}\bfA\bfW_{\varsigma,m})_{ii}
  (\bfW_{\varsigma,m}\bfA'\bfW_{\varsigma,m})_{ii}
 =A_{ii}A'_{ii}.
\end{align*}
Equivalently, for every \(i\ne j\),
\[
 (\bfW_{\varsigma,m}\bfA\bfW_{\varsigma,m})_{ij}
 (\bfW_{\varsigma,m}\bfA'\bfW_{\varsigma,m})_{ij}
 =A_{ij}A'_{ij}.
\]
Finally, let
\[
 \bfE_\rho
 =\widetilde\bfA_\rho
  -\bfW_{\varsigma,m}\bfA_{\rho,m}^G\bfW_{\varsigma,m}.
\]
Since
\[
 \|\bfW_c\|_F^2=O(m^{-1}),
 \qquad
 \|\bfW_c\|_{\rm op}=O(m^{-1}),
 \qquad
 \sum_i c_i^2=O(m^{-1}),
\]
and every companion matrix in this comparison has operator norm at most one,
for either such matrix \(\bfB\),
\begin{align*}
 m\left|\tr(\bfW_c\bfE_\rho\bfW_c\bfB)\right|
 &\le m\|\bfW_c\|_F^2\|\bfE_\rho\|_{\rm op}\|\bfB\|_{\rm op}
 =O_P(\ell_nm^{-1/2}),\\
 m\sum_i c_i^2|E_{ii,\rho}B_{ii}|
 &\le m\sum_i c_i^2\|\bfE_\rho\|_{\rm op}\|\bfB\|_{\rm op}
 =O_P(\ell_nm^{-1/2}).
\end{align*}
Applying these bounds successively at \(\rho\) and \(\rho'\) gives a total
replacement error of \(O_P(\ell_nm^{-1/2})\).  Exact conjugation invariance
then proves the spatial-sign version of \eqref{eq:weighted-offdiag-input}.

The inverse-distance expansion is
\begin{equation}
\label{eq:beta-b0-intermediate}
 \beta_i(s)=\mu_{w,p}^{-1}b_i^{\rm cus}(s)+r_{\beta,i}(s),
 \qquad
 \|\bfr_{\beta,s}\|_2=O_P(\ell_{0,n}^2m^{-1/2}),
 \qquad
 \|\bfr_{\beta,s}\|_\infty=O_P(\ell_{0,n}^2m^{-1}).
\end{equation}
Expanding the four weight factors gives
\begin{align*}
 \text{one remainder factor}
 &\;=O_P(\ell_{0,n}^2m^{-1/2}),\\
 \text{at least two remainder factors}
 &\;=O_P(\ell_{0,n}^4m^{-1})=O_P(\ell_nm^{-1}).
\end{align*}
The feasible multiplicative correction in
\eqref{eq:beta-multiplicative} is smaller.  This proves both oracle and
feasible versions of \eqref{eq:weighted-offdiag-input}.

Let \(y=\gamma/c_{\rm pool}\),
\(X_y\sim\underline F_{y,H}\), and suppose
\(c_{\rm pool}\in[\underline c_{\rm pool},\overline c_{\rm pool}]\).
The companion support lies in a common interval \([0,M]\), and
\begin{align*}
 \E X_y&=y,\\
 \E X_y^2&=y^2+y\int\lambda^2\,dH(\lambda),\\
 \Var(X_y)&=y\int\lambda^2\,dH(\lambda)\ge y,
\end{align*}
where \(\int\lambda\,dH(\lambda)=1\) and therefore
\(\int\lambda^2\,dH(\lambda)\ge1\).  In particular,
\begin{equation*}
 \Var(X_y)\ge
 \inf_{c_{\rm pool}\in[\underline c_{\rm pool},\overline c_{\rm pool}]}
 \frac{\gamma}{c_{\rm pool}}>0.
\end{equation*}
Furthermore,
\[
 d_*:=\inf_{\substack{x\in[0,M]\\\rho\in[\rho_0,\rho_1]}}
 f_\rho'(x)
 =\inf_{x,\rho}\frac{\rho}{(x+\rho)^2}>0.
\]
For an independent copy \(X_y'\),
\begin{align*}
 \mathfrak c_{\rho,\rho}(c_{\rm pool})
 &=\Var\{f_\rho(X_y)\}\\
 &=\frac12\E\{f_\rho(X_y)-f_\rho(X_y')\}^2\\
 &\ge d_*^2\Var(X_y)
 \ge d_*^2
 \inf_{c_{\rm pool}}\frac{\gamma}{c_{\rm pool}}>0.
\end{align*}
The upper bound follows from \(0\le f_\rho\le1\).  The companion
Stieltjes equation and bounded convergence give joint continuity in
\((c_{\rm pool},\rho)\).  Hence
\[
 \mathfrak c_{\rho,\rho',m}
 \longrightarrow\mathfrak c_{\rho,\rho'}(c_{\rm pool})
 \qquad(m/n\to c_{\rm pool})
\]
uniformly on compact pool-ratio and ridge sets, which proves
\eqref{eq:companion-cov-local-aspect}.  Compactness and the strictly positive
limit imply the finite-sample bounds in
\eqref{eq:companion-cov-uniform-bounds}.
\end{proof}

\begin{lemma}[Feasible companion cancellation]\label{lem:companion-cancellation}
Under \(H_0\) and Assumptions \ref{ass:elliptical}--\ref{ass:grid}, the
weighted law in Lemma~\ref{lem:weighted-offdiag} remains valid when the oracle
companion matrix is replaced by the feasible centered-sign companion matrix.
More precisely, for a pool \(J(s)\), write \(m=m_s=|J(s)|\) and suppress
its index \(s\) on companion entries only within this lemma.  Let
\(\bfb=(b_i:i\in J(s))^\top\) be deterministic or data-dependent and suppose,
on an event with probability tending to one, that
\[
  \max_i|b_i|\le Cm^{-1/2},\qquad \sum_i b_i^2\le C .
\]
Then, uniformly over the trimmed pools and the compact ridge interval,
\begin{align}
&m\sum_{i\ne j} b_i^2b_j^2
 \{\widehat A_{ij,\rho}^2-(A_{ij,\rho}^0)^2\}
   =O_P(\ell_nm^{-1/2}),\notag\\
&\sum_i b_i^2|\widehat A_{ii,\rho}-A_{ii,\rho}^0|
   =O_P(\ell_nm^{-1}).\label{eq:cancellation-input}
\end{align}
The first line also holds with a cross-ridge product in place of the two
squares.
\end{lemma}

\begin{proof}
Fix a trimmed pool \(J(s)\), reindex it as \(\{1,\ldots,m\}\), and define
\begin{align*}
 \bfY&=(\bfY_1,\ldots,\bfY_m),
 &\widehat\bfY&=(\widehat\bfY_{1,s},\ldots,\widehat\bfY_{m,s}),\\
 \bfw&=(w_1,\ldots,w_m)^\top,
 &\bar\bfY&=m^{-1}\sum_{i=1}^m\bfY_i,\\
 &&e_0&=m^{-1}\sum_{i=1}^mw_i.
\end{align*}
Lemma~\ref{lem:uniform-linearization} gives
\[
 \widehat\bfdelta_0
 =\widehat\bftheta_{0,s}-\bftheta
 =e_0^{-1}\bar\bfY+\bfr_0,
 \qquad
 \|\widehat\bfdelta_0\|=O_P(\ell_{0,n}),
 \qquad
 \|\bfr_0\|=O_P(\ell_nm^{-1}).
\]
The conditional sign bound further gives
\[
 \max_{i\le m}|\bfU_i^\top\widehat\bfdelta_0|
 =O_P(\ell_{0,n}^2p^{-1/2}).
\]
Applying the spatial-sign Taylor formula columnwise,
\begin{align*}
 \widehat\bfY_i
 &=\bfY_i-w_i\widehat\bfdelta_0+\bfE_i^\delta,\\
 \max_{i\le m}\|\bfE_i^\delta\|
 &=O_P(\ell_{0,n}^4p^{-1/2}),\\
 \|\bfE^\delta\|_{\rm op}
 &=O_P(\ell_{0,n}^4m^{-1/2}),
 &\|\bfE^\delta\|_F&=O_P(\ell_{0,n}^4).
\end{align*}
where \(\bfE^\delta=(\bfE_1^\delta,\ldots,\bfE_m^\delta)\).  The operator
bound follows by substituting
\(\widehat\bfdelta_0=e_0^{-1}\bar\bfY+\bfr_0\), reducing the resulting
Walsh matrices to degrees one and three, and applying matrix Bernstein:
\[
 \max_i\|\text{summand}_i\|_{\rm op}=O(\ell_{0,n}m^{-1}),
 \qquad
 \left\|\sum_i\E(\text{summand}_i)^2\right\|_{\rm op}
 =O(\ell_{0,n}^2m^{-1}).
\]
The terms containing \(\bfr_0\) are of smaller order.

Put
\begin{align*}
 \bfY_L&=\bfY-\widehat\bfdelta_0\bfw^\top,
 &\bfK_0&=m^{-1}\bfY^\top\bfY,\\
 \bfK_L&=m^{-1}\bfY_L^\top\bfY_L,
 &\widehat\bfK&=m^{-1}\widehat\bfY^\top\widehat\bfY,\\
 \bfR_L&=m^{-1}\bfY_L\bfY_L^\top.
\end{align*}
Then
\begin{align}
 \widehat\bfR&=\bfR_L+\bfE_R,\notag\\
 \bfE_R&=m^{-1}\{\bfY_L(\bfE^\delta)^\top
 +\bfE^\delta\bfY_L^\top
 +\bfE^\delta(\bfE^\delta)^\top\},\notag\\
 \|\bfE_R\|_{\rm op}
 &=O_P(\ell_{0,n}^5m^{-1})
 =O_P(\ell_nm^{-1}).
 \label{eq:centered-scatter-small-remainder}
\end{align}
Indeed,
\[
 \|\bfY_L\|_{\rm op}=O_P(\ell_{0,n}\sqrt m),
 \qquad
 \|\bfE^\delta\|_{\rm op}=O_P(\ell_{0,n}^4m^{-1/2}).
\]

For \(f_\rho(x)=x/(x+\rho)\),
\[
 \bfA_\rho^0=f_\rho(\bfK_0),
 \qquad
 \bfA_{L,\rho}=f_\rho(\bfK_L),
 \qquad
 \widehat\bfA_\rho=f_\rho(\widehat\bfK).
\]
Moreover,
\[
 \bfK_L-\bfK_0
 =-m^{-1}\bfY^\top\widehat\bfdelta_0\bfw^\top
  -m^{-1}\bfw\widehat\bfdelta_0^\top\bfY
  +m^{-1}\|\widehat\bfdelta_0\|^2\bfw\bfw^\top,
\]
so
\[
 \operatorname{rank}(\bfK_L-\bfK_0)\le2.
\]
The resolvent representation
\[
 f_\rho(\bfK)-f_\rho(\bfL)
 =\rho(\bfK+\rho\bfI)^{-1}
       (\bfK-\bfL)(\bfL+\rho\bfI)^{-1}
\]
therefore implies
\begin{equation}
\label{eq:linearized-companion-nuclear}
 \sup_\rho\|\bfA_{L,\rho}-\bfA_\rho^0\|_*
 +\sup_\rho\|\bfA_{L,\rho}-\bfA_\rho^0\|_F^2
 =O_P(1).
\end{equation}

Next,
\[
 \widehat\bfK-\bfK_L
 =m^{-1}\{\bfY_L^\top\bfE^\delta
 +(\bfE^\delta)^\top\bfY_L
 +(\bfE^\delta)^\top\bfE^\delta\}.
\]
Since \(\|\bfY_L\|_F=O_P(m)\), \(p\asymp m\), and
\(\|\bfE^\delta\|_F=O_P(\ell_{0,n}^4)\),
\[
 \|\widehat\bfK-\bfK_L\|_*
 \le2m^{-1}\|\bfY_L\|_F\|\bfE^\delta\|_F
   +m^{-1}\|\bfE^\delta\|_F^2
 =O_P(\ell_{0,n}^5).
\]
Hence
\begin{equation}
\label{eq:companion-nuclear-remainder}
 \sup_{\rho\in[\rho_0,\rho_1]}
 \|\widehat\bfA_\rho-\bfA_{L,\rho}\|_*
 \le\rho_0^{-1}\|\widehat\bfK-\bfK_L\|_*
 =O_P(\ell_{0,n}^5)=O_P(\ell_n).
\end{equation}
With \(\bfE_\rho=\widehat\bfA_\rho-\bfA_\rho^0\),
\eqref{eq:linearized-companion-nuclear}--
\eqref{eq:companion-nuclear-remainder} give
\[
 \sup_\rho\|\bfE_\rho\|_*=O_P(\ell_{0,n}^5),
 \qquad
 \sup_\rho\|\bfE_\rho\|_F=O_P(\ell_{0,n}^5).
\]
For every admissible weight vector \(b\),
\begin{align*}
 \sum_i b_i^2|E_{ii,\rho}|
 &\le\|b\|_\infty^2\|\bfE_\rho\|_*
 =O_P(\ell_nm^{-1}),\\
 m\sum_{i\ne j}b_i^2b_j^2E_{ij,\rho}^2
 &\le m\|b\|_\infty^4\|\bfE_\rho\|_F^2
 =O_P(\ell_{0,n}^{10}m^{-1}).
\end{align*}
The contraction \(0\preceq\bfA_\rho^0\preceq\bfI\) gives
\[
 m\sum_{i\ne j}b_i^2b_j^2(A_{ij,\rho}^0)^2=O_P(1).
\]
Using
\[
 \widehat A_{ij,\rho}^2-(A_{ij,\rho}^0)^2
 =2A_{ij,\rho}^0E_{ij,\rho}+E_{ij,\rho}^2
\]
and Cauchy's inequality,
\begin{align*}
 m\sum_{i\ne j}b_i^2b_j^2
 |A_{ij,\rho}^0E_{ij,\rho}|
 &=O_P(\ell_{0,n}^5m^{-1/2}),\\
 m\sum_{i\ne j}b_i^2b_j^2E_{ij,\rho}^2
 &=O_P(\ell_{0,n}^{10}m^{-1}).
\end{align*}
Because \(\ell_{0,n}^{10}=O(\ell_n)\) and \(\ell_n=o(m^{1/2})\), both
terms are \(O_P(\ell_nm^{-1/2})\) or smaller.  The same expansion with
\((\rho,\rho')\) proves the cross-ridge statement.  These bounds are exactly
\eqref{eq:cancellation-input}.
\end{proof}

\begin{lemma}[Segment leverage contrasts]\label{lem:leverage-contrast}
Under the null assumptions, consider any polynomial-size family of trimmed
pools \(J(s)\) with \(m_s=|J(s)|\asymp n\), and let
\(\widetilde\bfA_{\rho,s}=(\widetilde A_{ij,\rho,s})\) be the corresponding
unsigned oracle companion matrix.  For \(v_i\in\{1,w_i\}\), put
\[
 L_{a,v,\rho}(s)=\frac1{n_a(s)}
   \sum_{i\in I_a(s)}v_iA_{ii,\rho,s}^0,\qquad a=1,2.
\]
There are deterministic numbers \(\mu_{v,\rho,m_s}\), depending on a candidate
only through its pool size and independent of \(a\), such that
\begin{equation}
\label{eq:segment-leverage-average}
 \sup_{s,\rho}\max_{a=1,2}\max_{v\in\{1,w\}}
 |L_{a,v,\rho}(s)-\mu_{v,\rho,m_s}|
 =O_P(\ell_{0,n}^2m_*^{-1/2}).
\end{equation}
Consequently,
\begin{align}
 \sup_{s,\rho}\left|\sum_i v_i\beta_i(s)A_{ii,\rho,s}^0\right|
 &=O_P(\ell_{0,n}^2),\label{eq:diagonal-leverage-contrast}\\
 \sup_{s,\rho}\left|\sum_{i\ne j}
   v_i\beta_j(s)A_{ij,\rho,s}^0\right|
 &=O_P(\ell_{0,n}^3),\qquad v\in\{1,w\}.
 \label{eq:offdiagonal-leverage-contrast}
\end{align}
\end{lemma}

\begin{proof}
Work on the truncation event \(\max_iw_i\le C\ell_{0,n}\), replacing
\(w_i\) by its clipped version outside this event.  If one unsigned column
of \(\widetilde\bfY_s\) is replaced, then
\[
 \operatorname{rank}(\Delta\bfK_s)\le2,
 \qquad
 \bfK_s=m_s^{-1}\widetilde\bfY_s^\top\widetilde\bfY_s.
\]
Because
\[
 \widetilde\bfA_{\rho,s}=f_\rho(\bfK_s),
 \qquad
 f_\rho(x)=\frac{x}{x+\rho},
\]
the resolvent identity gives
\[
 \operatorname{rank}(\Delta\widetilde\bfA_{\rho,s})\le2,
 \qquad
 \|\Delta\widetilde\bfA_{\rho,s}\|_{\rm op}\le2,
 \qquad
 \|\Delta\widetilde\bfA_{\rho,s}\|_ *\le4.
\]
Hence the bounded differences of the segment leverage averages satisfy
\[
 |\Delta L_{a,1,\rho}|\le Cn_a^{-1},
 \qquad
 |\Delta L_{a,w,\rho}|\le C\ell_{0,n}n_a^{-1}.
\]
McDiarmid's inequality therefore yields, for \(v_i\in\{1,w_i\}\),
\[
 \Pbb\{|L_{a,v,\rho}-\E L_{a,v,\rho}|>x\}
 \le2\exp\{-cm_sx^2/\ell_{0,n}^2\}.
\]
Column exchangeability gives
\[
 \E L_{a,v,\rho}
 =\E(v_1A_{11,\rho,s}^0)
 =:\mu_{v,\rho,m_s},
\]
independently of \(a\).  A polynomial scan union and the derivative bound
\[
 \sup_\rho\|\partial_\rho\widetilde\bfA_{\rho,s}\|_{\rm op}
 \le C_{\rho_0}
\]
prove \eqref{eq:segment-leverage-average}; the clipped and original arrays
coincide with probability tending to one.

The same argument applied to
\(e_{a,s}=n_a^{-1}\sum_{i\in I_a(s)}w_i\) gives
\[
 \sup_{s,a}|e_{a,s}-\mu_{w,p}|
 =O_P(\ell_{0,n}^2m_*^{-1/2}),
 \qquad
 \inf_{s,a}e_{a,s}\ge\zeta_{-1}/2.
\]
For the diagonal contrast,
\[
 \sum_i v_i\beta_i(s)A_{ii,\rho,s}^0
 =\sqrt{N_s}\left\{
 \frac{L_{2,v,\rho}(s)}{e_{2,s}}
 -\frac{L_{1,v,\rho}(s)}{e_{1,s}}
 \right\}.
\]
Since the two ratios have the same deterministic center,
\[
 \left|
 \frac{L_{2,v,\rho}(s)}{e_{2,s}}
 -\frac{L_{1,v,\rho}(s)}{e_{1,s}}
 \right|
 =O_P(\ell_{0,n}^2m_*^{-1/2}),
 \qquad N_s\le m_s/4,
\]
which proves \eqref{eq:diagonal-leverage-contrast}.

For the off-diagonal contrast, condition on the unsigned columns and write
\[
 A_{ij,\rho,s}^0
 =\varsigma_i\varsigma_j\widetilde A_{ij,\rho,s}.
\]
After grouping ordered pairs, the coefficient of
\(\varsigma_i\varsigma_j\), \(i<j\), is
\[
 (v_i\beta_j+v_j\beta_i)\widetilde A_{ij,\rho,s}.
\]
The contraction \(0\preceq\widetilde\bfA_{\rho,s}\preceq\bfI\) and
\(\max_i|\beta_i|\le Cm_s^{-1/2}\) give
\begin{align*}
 \sum_{i<j}(v_i\beta_j+v_j\beta_i)^2
 \widetilde A_{ij,\rho,s}^2
 &\le C\ell_{0,n}^2m_s^{-1}
 \tr(\widetilde\bfA_{\rho,s}^2)\\
 &\le C\ell_{0,n}^2.
\end{align*}
Fixed-degree Rademacher hypercontractivity, the scan union, and the ridge-net
argument prove \eqref{eq:offdiagonal-leverage-contrast}.
\end{proof}

\begin{lemma}[Score norm and centered-SSCM resolvent replacement]\label{lem:score-resolvent}
Under \(H_0\) and Assumptions \ref{ass:elliptical}--\ref{ass:grid}, uniformly
over \(s\in\calS\) and \(\rho\in[\rho_0,\rho_1]\), with
\[
  \bfB_s=\sum_{i\in J(s)}\beta_i(s)\bfY_i,
\]
one has
\begin{equation}
\label{eq:Bnorm-rate}
  \sup_{s\in\calS}\|\bfB_s\|=O_P(\ell_n m_*^{1/2})
\end{equation}
and
\begin{equation}
\label{eq:resolvent-quadratic-rate}
  \sup_{s,\rho}
  |\bfB_s^\top(\widehat\bfQ_{\rho,s}-\bfQ_{\rho,s}^0)\bfB_s|
  =O_P(\ell_n).
\end{equation}
\end{lemma}

\begin{proof}
Conditionally on \(\calF_s^0\),
\[
 \E(\|\bfB_s\|^2\mid\calF_s^0)
 =\sum_i\beta_i(s)^2\|\bfY_i\|^2\le Cp.
\]
Pinelis' Hilbert-space martingale inequality and the polynomial scan union
therefore give
\[
 \sup_s\|\bfB_s\|
 =O_P(\ell_{0,n}m_*^{1/2})
 =O_P(\ell_nm_*^{1/2}),
\]
which is \eqref{eq:Bnorm-rate}.

For a candidate \(s\), put
\begin{align*}
 \bar\bfY_s&=m_s^{-1}\sum_{i\in J(s)}\bfY_i,
 &e_{0,s}&=m_s^{-1}\sum_{i\in J(s)}w_i,\\
 \widehat\bfdelta_{0,s}
 &=e_{0,s}^{-1}\bar\bfY_s+\bfr_{0,s},
 &\sup_s\|\bfr_{0,s}\|&=O_P(\ell_nm_*^{-1}),\\
 \bar\bfY_{w,s}&=m_s^{-1}\sum_{i\in J(s)}w_i\bfY_i,
 &\bar w_{2,s}&=m_s^{-1}\sum_{i\in J(s)}w_i^2.
\end{align*}
Define
\[
 \bfL_s
 =-\bar\bfY_{w,s}\widehat\bfdelta_{0,s}^\top
  -\widehat\bfdelta_{0,s}\bar\bfY_{w,s}^\top
  +\bar w_{2,s}\widehat\bfdelta_{0,s}
                 \widehat\bfdelta_{0,s}^\top,
 \qquad
 \bfR_{L,s}=\bfR_s^0+\bfL_s,
\]
\[
 \bfQ_{L,\rho,s}=(\bfR_{L,s}+\rho\bfI_p)^{-1}.
\]
Equation~\eqref{eq:centered-scatter-small-remainder} gives
\[
 \widehat\bfR_s=\bfR_{L,s}+\bfE_{R,s},
 \qquad
 \sup_s\|\bfE_{R,s}\|_{\rm op}
 =O_P(\ell_{0,n}^5m_*^{-1}).
\]
The exact resolvent identity is
\[
 \widehat\bfQ_{\rho,s}-\bfQ_{\rho,s}^0
 =(\bfQ_{L,\rho,s}-\bfQ_{\rho,s}^0)
 -\bfQ_{L,\rho,s}\bfE_{R,s}\widehat\bfQ_{\rho,s}.
\]
Since every ridge inverse has norm at most \(\rho_0^{-1}\),
\begin{align*}
 &\sup_{s,\rho}
 |\bfB_s^\top\bfQ_{L,\rho,s}\bfE_{R,s}
 \widehat\bfQ_{\rho,s}\bfB_s|\\
 &\qquad\le
 \rho_0^{-2}
 \sup_s\|\bfB_s\|^2
 \sup_s\|\bfE_{R,s}\|_{\rm op}
 =O_P(\ell_{0,n}^7)=O_P(\ell_n).
\end{align*}

For the finite-rank term, set
\[
 \bfU_s^{\rm lr}=(\bar\bfY_{w,s},\widehat\bfdelta_{0,s}),
 \qquad
 \bfF_s=\begin{pmatrix}0&-1\\-1&\bar w_{2,s}\end{pmatrix},
 \qquad
 \bfL_s=\bfU_s^{\rm lr}\bfF_s(\bfU_s^{\rm lr})^\top.
\]
The two-sided identity, which does not require \(\bfF_s^{-1}\), is
\[
 \bfQ_{L,\rho,s}-\bfQ_{\rho,s}^0
 =-\bfQ_{\rho,s}^0\bfU_s^{\rm lr}\bfF_{s,\rho}
   (\bfU_s^{\rm lr})^\top\bfQ_{\rho,s}^0,
\]
\[
 \bfF_{s,\rho}
 =\bfF_s-\bfF_s(\bfU_s^{\rm lr})^\top
 \bfQ_{L,\rho,s}\bfU_s^{\rm lr}\bfF_s.
\]
The coefficient bounds are
\[
 \sup_s\|\bfF_s\|_{\rm op}=O_P(1),
 \qquad
 \sup_s\|\bfU_s^{\rm lr}\|_{\rm op}=O_P(\ell_{0,n}),
 \qquad
 \sup_{s,\rho}\|\bfF_{s,\rho}\|_{\rm op}
 =O_P(\ell_{0,n}^2).
\]
It remains to bound the two projections in
\((\bfU_s^{\rm lr})^\top\bfQ_{\rho,s}^0\bfB_s\).  First,
\[
 \bar\bfY_{w,s}^\top\bfQ_{\rho,s}^0\bfB_s
 =\sum_{i,j}w_i\beta_j(s)A_{ij,\rho,s}^0,
\]
whose diagonal and off-diagonal parts are, by
\eqref{eq:diagonal-leverage-contrast} and
\eqref{eq:offdiagonal-leverage-contrast},
\[
 O_P(\ell_{0,n}^2)
 \quad\text{and}\quad
 O_P(\ell_{0,n}^3),
\]
respectively.  Next,
\[
 \bar\bfY_s^\top\bfQ_{\rho,s}^0\bfB_s
 =\sum_{i,j}\beta_j(s)A_{ij,\rho,s}^0
 =O_P(\ell_{0,n}^3).
\]
Finally,
\[
 |\bfr_{0,s}^\top\bfQ_{\rho,s}^0\bfB_s|
 \le\rho_0^{-1}\|\bfr_{0,s}\|\,\|\bfB_s\|
 =O_P(\ell_n\ell_{0,n}m_s^{-1/2})=o_P(1).
\]
Since \(e_{0,s}\) is bounded away from zero,
\begin{equation}
\label{eq:HQB-corrected-rate}
 \sup_{s,\rho}
 \|(\bfU_s^{\rm lr})^\top\bfQ_{\rho,s}^0\bfB_s\|
 =O_P(\ell_{0,n}^3).
\end{equation}
Therefore,
\begin{align*}
 &\sup_{s,\rho}
 |\bfB_s^\top(\bfQ_{L,\rho,s}-\bfQ_{\rho,s}^0)\bfB_s|\\
 &\qquad\le
 \sup_{s,\rho}
 \|(\bfU_s^{\rm lr})^\top\bfQ_{\rho,s}^0\bfB_s\|^2
 \sup_{s,\rho}\|\bfF_{s,\rho}\|_{\rm op}\\
 &\qquad=O_P(\ell_{0,n}^8)=O_P(\ell_n).
\end{align*}
Combining the finite-rank term with the exact remainder term proves
\eqref{eq:resolvent-quadratic-rate}.
\end{proof}

\begin{lemma}[Deterministic and mixed centered-resolvent transfer]\label{lem:deterministic-centered-transfer}
Under Assumptions \ref{ass:elliptical}--\ref{ass:grid}, consider the
centered-error sample \(\{\bfvarepsilon_i\}\), and let
\[
 \widehat\bftheta_{0,s}^{(0)}=\widehat\bftheta^{(0)}(J(s)),\qquad
 \widehat\bfY_{i,s}^{(0)}
 =\sqrt p\,U(\bfvarepsilon_i-\widehat\bftheta_{0,s}^{(0)}),
\]
\[
 \widehat\bfR_s^{(0)}
 =m_s^{-1}\sum_{i\in J(s)}
   \widehat\bfY_{i,s}^{(0)}\{\widehat\bfY_{i,s}^{(0)}\}^\top,
 \qquad
 \widehat\bfQ_{\rho,s}^{(0)}
 =(\widehat\bfR_s^{(0)}+\rho\bfI_p)^{-1}.
\]
Consider any deterministic collection of polynomially many triples
\((s,\bfd_p,\bfe_p)\), where \(s\) is a trimmed pooled segment and
\(\|\bfd_p\|+\|\bfe_p\|\le C\).  Uniformly over this collection and
\(\rho\in[\rho_0,\rho_1]\),
\begin{equation}
\label{eq:centered-det-transfer}
 \left|\bfd_p^\top
 \{\widehat\bfQ_{\rho,s}^{(0)}-\bfQ_{\rho,s}^0\}\bfe_p\right|
 =O_P\{\ell_nm_*^{-1}\|\bfd_p\|\,\|\bfe_p\|\}.
\end{equation}
If
\(\bfB_s=\sum_{i\in J(s)}\beta_i(s)\bfY_i\) is the corresponding
centered-error CUSUM score, then, for every polynomial deterministic
collection \((s,\bfd_p)\),
\begin{equation}
\label{eq:centered-mixed-transfer}
 \left|\bfd_p^\top
 \{\widehat\bfQ_{\rho,s}^{(0)}-\bfQ_{\rho,s}^0\}\bfB_s\right|
 =O_P\{\ell_nm_*^{-1/2}\|\bfd_p\|\}.
\end{equation}
Consequently,
\begin{equation}
\label{eq:centered-det-de}
 \bfd_p^\top\widehat\bfQ_{\rho,s}^{(0)}\bfe_p
 =\bfd_p^\top\bfD_{\rho,m_s}\bfe_p
  +O_P\{\ell_nm_*^{-1/2}\|\bfd_p\|\,\|\bfe_p\|\}.
\end{equation}
\end{lemma}

\begin{proof}
Uniformly over the stated polynomial collection,
\[
 \widehat\bfR_s^{(0)}
 =\bfR_s^0+\bfL_s+\bfE_{R,s},
 \qquad
 \bfL_s=\bfU_s^{\rm lr}\bfF_s(\bfU_s^{\rm lr})^\top,
 \qquad
 \operatorname{rank}(\bfL_s)\le2,
\]
\[
 \bfR_s^0+\bfL_s\succeq0,
 \qquad
 \|\bfE_{R,s}\|_{\rm op}=O_P(\ell_nm_*^{-1}).
\]
Let
\[
 \bfQ_{L,\rho,s}
 =(\bfR_s^0+\bfL_s+\rho\bfI_p)^{-1}.
\]
Then
\[
 \widehat\bfQ_{\rho,s}^{(0)}-\bfQ_{\rho,s}^0
 =(\bfQ_{L,\rho,s}-\bfQ_{\rho,s}^0)
 -\bfQ_{L,\rho,s}\bfE_{R,s}
  \widehat\bfQ_{\rho,s}^{(0)},
\]
and every inverse has operator norm at most \(\rho_0^{-1}\).  Consequently,
for deterministic \(\bfd_p,\bfe_p\),
\begin{align*}
 &|\bfd_p^\top\bfQ_{L,\rho,s}\bfE_{R,s}
   \widehat\bfQ_{\rho,s}^{(0)}\bfe_p|\\
 &\qquad\le\rho_0^{-2}\|\bfd_p\|\,\|\bfe_p\|
 \|\bfE_{R,s}\|_{\rm op}
 =O_P(\ell_nm_*^{-1}\|\bfd_p\|\,\|\bfe_p\|).
\end{align*}
For the mixed form, use the sharper intermediate rates
\[
 \|\bfE_{R,s}\|_{\rm op}=O_P(\ell_{0,n}^5m_*^{-1}),
 \qquad
 \|\bfB_s\|=O_P(\ell_{0,n}m_*^{1/2}),
\]
which give
\[
 |\bfd_p^\top\bfQ_{L,\rho,s}\bfE_{R,s}
   \widehat\bfQ_{\rho,s}^{(0)}\bfB_s|
 =O_P(\ell_{0,n}^6m_*^{-1/2}\|\bfd_p\|)
 =O_P(\ell_nm_*^{-1/2}\|\bfd_p\|).
\]

For the rank-two term,
\[
 \bfu^\top(\bfQ_{L,\rho,s}-\bfQ_{\rho,s}^0)\bfv
 =-\{(\bfU_s^{\rm lr})^\top\bfQ_{\rho,s}^0\bfu\}^\top
 \bfF_{s,\rho}
 \{(\bfU_s^{\rm lr})^\top\bfQ_{\rho,s}^0\bfv\}.
\]
For deterministic \(\bfu\), conditional Rademacher concentration and the
companion contraction yield
\begin{equation}
\label{eq:HQd-rate}
 \|(\bfU_s^{\rm lr})^\top\bfQ_{\rho,s}^0\bfu\|
 =O_P\left\{
 (\ell_{0,n}^2m_*^{-1/2}+\ell_nm_*^{-1})\|\bfu\|
 \right\}.
\end{equation}
Moreover,
\[
 \|\bfF_{s,\rho}\|_{\rm op}=O_P(\ell_{0,n}^2).
\]
Applying \eqref{eq:HQd-rate} twice gives
\begin{align*}
 &|\bfd_p^\top(\bfQ_{L,\rho,s}-\bfQ_{\rho,s}^0)\bfe_p|\\
 &\qquad\le
 \|\bfF_{s,\rho}\|_{\rm op}
 \|(\bfU_s^{\rm lr})^\top\bfQ_{\rho,s}^0\bfd_p\|
 \|(\bfU_s^{\rm lr})^\top\bfQ_{\rho,s}^0\bfe_p\|\\
 &\qquad=O_P(\ell_nm_*^{-1}\|\bfd_p\|\,\|\bfe_p\|),
\end{align*}
which proves the finite-rank contribution in
\eqref{eq:centered-det-transfer}.  Combining \eqref{eq:HQd-rate} with
\eqref{eq:HQB-corrected-rate} gives
\begin{align*}
 &|\bfd_p^\top(\bfQ_{L,\rho,s}-\bfQ_{\rho,s}^0)\bfB_s|\\
 &\qquad\le
 \|\bfF_{s,\rho}\|_{\rm op}
 \|(\bfU_s^{\rm lr})^\top\bfQ_{\rho,s}^0\bfd_p\|
 \|(\bfU_s^{\rm lr})^\top\bfQ_{\rho,s}^0\bfB_s\|\\
 &\qquad=O_P(\ell_{0,n}^7m_*^{-1/2}\|\bfd_p\|)
 =O_P(\ell_nm_*^{-1/2}\|\bfd_p\|),
\end{align*}
which is \eqref{eq:centered-mixed-transfer}.  The scan and ridge unions are
valid because the projected bounds have conditional exponential tails on
the truncation event.

Finally,
\begin{align*}
 \bfd_p^\top\{\widehat\bfQ_{\rho,s}^{(0)}-\bfD_{\rho,m_s}\}\bfe_p
 &=\bfd_p^\top\{\widehat\bfQ_{\rho,s}^{(0)}-\bfQ_{\rho,s}^0\}\bfe_p\\
 &\quad+\bfd_p^\top\{\bfQ_{\rho,s}^0-\bfD_{\rho,m_s}\}\bfe_p.
\end{align*}
The first term is \(O_P(\ell_nm_*^{-1})\), while
\eqref{eq:bilinear-de-rate} makes the second
\(O_P(\ell_nm_*^{-1/2})\).  This proves
\eqref{eq:centered-det-de}.
\end{proof}

\begin{lemma}[Feasible companion centering and variance]\label{lem:feasible-companion}
Under \(H_0\) and Assumptions \ref{ass:elliptical}--\ref{ass:grid}, uniformly over \(s\in\calS\) and \(\rho\in[\rho_0,\rho_1]\),
\begin{align}
  \widehat\kappa_\rho(s)-\kappa_\rho^0(s)&=O_P(\ell_nm_*^{-1}),\label{eq:kappa-feasible-rate}\\
  \widehat\sigma_\rho^2(s)-\sigma_\rho^{0,2}(s)&=O_P(\ell_nm_*^{-1/2}).\label{eq:sigma-feasible-rate}
\end{align}
In addition, with the deterministic proxy in
\eqref{eq:finite-sigma-proxy},
\begin{equation}
\label{eq:oracle-sigma-proxy-rate}
 \sup_{s,\rho}
 |\sigma_\rho^{0,2}(s)-\sigma_{\rho,m_s}^{\circ 2}|
 =O_P(\ell_nm_*^{-1/2}).
\end{equation}
Moreover, there are constants \(0<c_\sigma<C_\sigma<\infty\) such that
\begin{equation}
\label{eq:variance-nondegenerate}
  \Pbb\left\{c_\sigma\le \inf_{s,\rho}\sigma_\rho^{0,2}(s)
  \le \sup_{s,\rho}\sigma_\rho^{0,2}(s)\le C_\sigma\right\}\to1.
\end{equation}
\end{lemma}

\begin{proof}[Proof of Lemma \ref{lem:feasible-companion}]
Throughout this proof all suprema over \((s,\rho)\) are suppressed.  We first
control the feasible centering.  For the weight perturbation, by \eqref{eq:beta-multiplicative}, for \(i\in I_a(s)\),
\[
  \widehat\beta_i(s)=\beta_i(s)(1+r_{a,s}^{\rm wt}),
  \qquad
  \max_{a=1,2}|r_{a,s}^{\rm wt}|=O_P(\ell_nm_*^{-1}).
\]
Since \(0\le A_{ii,\rho,s}^0\le1\), \(\sum_i\beta_i(s)^2\le C\), and \(\max_i|\beta_i(s)|\le Cm_s^{-1/2}\),
\begin{align*}
  \sum_i |\widehat\beta_i(s)^2-\beta_i(s)^2|A_{ii,\rho,s}^0
  &\le C\max_{a}|r_{a,s}^{\rm wt}|\sum_i\beta_i(s)^2
       +C\max_a\{r_{a,s}^{\rm wt}\}^2\sum_i\beta_i(s)^2  \\
  &=O_P(\ell_nm_*^{-1}).
\end{align*}
For the diagonal companion perturbation, the cancellation bound in \eqref{eq:cancellation-input}, applied with \(b_i=\beta_i(s)\), gives
\[
  \sum_i \beta_i(s)^2|\widehat A_{ii,\rho,s}-A_{ii,\rho,s}^0|
  =O_P(\ell_nm_*^{-1}).
\]
The additional product of the weight error and the diagonal companion
perturbation is smaller:
\[
\begin{aligned}
 &\sum_i |\widehat\beta_i(s)^2-\beta_i(s)^2|\,
       |\widehat A_{ii,\rho,s}-A_{ii,\rho,s}^0| \\
 &\qquad\le O_P(\ell_nm_*^{-1})
       \sum_i\beta_i(s)^2
       |\widehat A_{ii,\rho,s}-A_{ii,\rho,s}^0| \\
 &\qquad=O_P(\ell_n^2m_*^{-2})
       =o_P(\ell_nm_*^{-1}),
\end{aligned}
\]
because the fixed polylogarithmic envelope satisfies $\ell_n=o(m_*)$.
Therefore
\begin{align*}
  \widehat\kappa_\rho(s)-\kappa_\rho^0(s)
  &=\sum_i\{\widehat\beta_i(s)^2-\beta_i(s)^2\}A_{ii,\rho,s}^0
    +\sum_i\widehat\beta_i(s)^2\{\widehat A_{ii,\rho,s}-A_{ii,\rho,s}^0\} \\
  &=O_P(\ell_nm_*^{-1}),
\end{align*}
which proves \eqref{eq:kappa-feasible-rate}.  Notice that no entrywise diagonal bound is used here; the needed rate is a weighted diagonal cancellation.

For the variance, decompose
\begin{align*}
  \widehat\sigma_\rho^2(s)-\sigma_\rho^{0,2}(s)
  &=2m_s\sum_{i\ne j}
    \{\widehat\beta_i(s)^2\widehat\beta_j(s)^2-\beta_i(s)^2\beta_j(s)^2\}(A_{ij,\rho,s}^0)^2 \\
  &\quad +2m_s\sum_{i\ne j}\widehat\beta_i(s)^2\widehat\beta_j(s)^2
    \{\widehat A_{ij,\rho,s}^2-(A_{ij,\rho,s}^0)^2\}
    \equiv R_{1,s}+R_{2,s}.
\end{align*}
The first term is controlled by the multiplicative form of the weight error.  Indeed,
\[
  |\widehat\beta_i(s)^2\widehat\beta_j(s)^2-\beta_i(s)^2\beta_j(s)^2|
  \le C\max_a|r_{a,s}^{\rm wt}|\beta_i(s)^2\beta_j(s)^2
       +C\max_a\{r_{a,s}^{\rm wt}\}^2\beta_i(s)^2\beta_j(s)^2.
\]
Using \eqref{eq:A-row-bounds},
\begin{align*}
  m_s\sum_{i\ne j}\beta_i(s)^2\beta_j(s)^2(A_{ij,\rho,s}^0)^2
  &\le m_s\sum_i\beta_i(s)^2 \max_i\sum_j\beta_j(s)^2(A_{ij,\rho,s}^0)^2 \\
  &\le C m_s\cdot m_s^{-1}\sum_i\beta_i(s)^2
   \le C.
\end{align*}
Thus
\[
  R_{1,s}=O_P(\ell_nm_*^{-1}),
\]
which is stronger than the required \(O_P(\ell_nm_*^{-1/2})\).

For \(R_{2,s}\), apply the off-diagonal cancellation bound in \eqref{eq:cancellation-input} with \(b_i=\widehat\beta_i(s)\).  The admissibility conditions in Lemma \ref{lem:companion-cancellation} hold because \eqref{eq:beta-rate} and \eqref{eq:beta-multiplicative} imply
\[
  \max_i|\widehat\beta_i(s)|\le Cm_s^{-1/2},
  \qquad
  \sum_i\widehat\beta_i(s)^2\le C
\]
with probability tending to one.  Hence
\[
  R_{2,s}=O_P(\ell_nm_*^{-1/2}).
\]
Combining the bounds for \(R_{1,s}\) and \(R_{2,s}\) proves \eqref{eq:sigma-feasible-rate}.

The remaining point is the uniform nondegeneracy of the oracle variance.  With \(b_i=\beta_i(s)\), \eqref{eq:weighted-offdiag-input} gives
\[
  \sigma_\rho^{0,2}(s)
  =2\mathfrak c_{\rho,\rho,m_s}\left\{\sum_i\beta_i(s)^2\right\}^2+O_P(\ell_nm_*^{-1/2}).
\]
The ordinary CUSUM weights have squared norm one, and the uniform
inverse-distance law gives
\[
 \sup_s\left|\sum_i\beta_i(s)^2-\mu_{w,p}^{-2}\right|
 =O_P(\ell_nm_*^{-1/2}).
\]
Together with \eqref{eq:finite-sigma-proxy}, this proves
\eqref{eq:oracle-sigma-proxy-rate}.
The CUSUM weights satisfy, uniformly on the trimmed scan,
\[
  c_1\le \sum_i\beta_i(s)^2\le c_2,
\]
because the segment lengths are proportional to \(n\) and \(e_{a,s}\) is bounded above and below by \eqref{eq:eI-uniform-prelim}.  The companion-resolvent functional \(\mathfrak c_{\rho,\rho,m_s}\) is continuous and satisfies
\[
  0<c_3\le \inf_{\rho\in[\rho_0,\rho_1]}\mathfrak c_{\rho,\rho,m_s}
  \le \sup_{\rho\in[\rho_0,\rho_1]}\mathfrak c_{\rho,\rho,m_s}\le c_4<\infty,
\]
by Lemma~\ref{lem:weighted-offdiag}.  Hence
\eqref{eq:variance-nondegenerate} follows.
\end{proof}

\begin{lemma}[Conditional quadratic-form CLT]\label{lem:clt}
Under \(H_0\) and Assumptions \ref{ass:elliptical}--\ref{ass:grid}, for every deterministic candidate sequence $s=s_n$ satisfying $m_s\asymp n$ and fixed trimming, and every fixed $\rho$,
\[
  \frac{\widetilde V_\rho^0(s)-m_s\kappa_\rho^0(s)}{\{m_s\sigma_\rho^{0,2}(s)\}^{1/2}}
  \dto N(0,1)
\]
after conditioning on $\calF_s^0$, in probability.
\end{lemma}

\begin{proof}
By Lemmas~\ref{lem:signs} and \ref{lem:companion-contraction}, conditionally on \(\calF_s^0\),
\[
  A_{ij,\rho,s}^0=\varsigma_i\varsigma_j\widetilde A_{ij,\rho,s},
  \qquad
  (\varsigma_i:i\in J(s))\mid\calF_s^0
  \stackrel{\rm ind}{\sim}\operatorname{Unif}\{-1,1\}.
\]
Therefore
\begin{align}
  \widetilde V_\rho^0(s)
  &=m_s\sum_{i,j\in J(s)}\beta_i(s)\beta_j(s)
    \varsigma_i\varsigma_j\widetilde A_{ij,\rho,s}\notag\\
  &=m_s\sum_i\beta_i(s)^2\widetilde A_{ii,\rho,s}
    +2m_s\sum_{i<j}\beta_i(s)\beta_j(s)
    \widetilde A_{ij,\rho,s}\varsigma_i\varsigma_j .
  \label{eq:V-sign-expansion}
\end{align}
In particular,
\[
  \E\{\widetilde V_\rho^0(s)\mid\calF_s^0\}=m_s\kappa_\rho^0(s).
\]
Define a symmetric matrix $\bfC_{\rho,s}$ with zero diagonal and off-diagonal entries
\begin{equation*}
  C_{ij,\rho,s}=m_s\beta_i(s)\beta_j(s)\widetilde A_{ij,\rho,s},
  \qquad i\ne j.
\end{equation*}
Let $\bfvarsigma=(\varsigma_i:i\in J(s))^{\top}$.  From \eqref{eq:V-sign-expansion},
\begin{equation}
\label{eq:quad-form-C}
  \widetilde V_\rho^0(s)-m_s\kappa_\rho^0(s)
  =\bfvarsigma^{\top}\bfC_{\rho,s}\bfvarsigma.
\end{equation}
Since
\[
  \E(\varsigma_i\varsigma_j\varsigma_k\varsigma_\ell\mid\calF_s^0)
  =\ind\{\{i,j\}=\{k,\ell\}\},
  \qquad i<j,\ k<\ell,
\]
the conditional variance is
\begin{equation}
\label{eq:var-C}
  \Var(\bfvarsigma^{\top}\bfC_{\rho,s}\bfvarsigma\mid\calF_s^0)
  =2\tr(\bfC_{\rho,s}^2)=m_s\sigma_\rho^{0,2}(s).
\end{equation}
We use the Rademacher quadratic-form corollary of the central limit theorem
of \citet{dejong1987}.  In the present notation, that corollary says that a
zero-diagonal symmetric array \(\bfC_n\), normalized by
\(2\tr(\bfC_n^2)\), is asymptotically Gaussian if
\[
 \frac{\tr(\bfC_n^4)}{\{\tr(\bfC_n^2)\}^2}\longrightarrow0,
 \qquad
 \frac{\max_i\sum_j C_{n,ij}^2}{\tr(\bfC_n^2)}\longrightarrow0.
\]
The first condition follows from
\(\tr(\bfC_n^4)\le\|\bfC_n\|_{\rm op}^2\tr(\bfC_n^2)\), so it is enough below to
verify the normalized operator-norm and row-influence bounds.

Let $\bfW_{\beta,s}=\diag(\beta_i(s):i\in J(s))$.  The matrix with diagonal included is
\[
  m_s\bfW_{\beta,s}\widetilde\bfA_{\rho,s}\bfW_{\beta,s}.
\]
By \eqref{eq:A-contraction}, $\opnorm{\widetilde\bfA_{\rho,s}}\le1$.  From \eqref{eq:oracle-e-beta} and trimming, there is a constant $C$ such that, with probability tending to one,
\begin{equation}
\label{eq:beta-max-sum-proof}
  \max_i\beta_i(s)^2\le C m_s^{-1},
  \qquad
  \sum_i\beta_i(s)^2\le C.
\end{equation}
Thus
\begin{equation}
\label{eq:C-op-bound}
  \opnorm{\bfC_{\rho,s}}
  \le m_s\opnorm{\bfW_{\beta,s}}\opnorm{\widetilde\bfA_{\rho,s}}\opnorm{\bfW_{\beta,s}}
  \le C.
\end{equation}
The diagonal removal changes the operator norm by at most
\[
  \max_i m_s\beta_i(s)^2\widetilde A_{ii,\rho,s}\le C,
\]
so \eqref{eq:C-op-bound} remains valid for the zero-diagonal $\bfC_{\rho,s}$.

The variance non-degeneracy \eqref{eq:variance-nondegenerate} gives
\begin{equation}
\label{eq:trC2-lower}
  2\tr(\bfC_{\rho,s}^2)=m_s\sigma_\rho^{0,2}(s)
  \ge c_\sigma m_s
\end{equation}
with probability tending to one.  Therefore
\begin{equation}
\label{eq:eigen-negligible}
  \frac{\opnorm{\bfC_{\rho,s}}}{\{\tr(\bfC_{\rho,s}^2)\}^{1/2}}
  \le \frac{C}{(c_\sigma m_s/2)^{1/2}}\to0.
\end{equation}
For the row influence, by \eqref{eq:A-row-bounds} and \eqref{eq:beta-max-sum-proof},
\begin{align}
  \max_i\sum_j C_{ij,\rho,s}^2
  &\le m_s^2\max_i\beta_i(s)^2
  \max_i\sum_j\beta_j(s)^2\widetilde A_{ij,\rho,s}^2\notag\\
  &\le m_s^2\cdot C m_s^{-1}\cdot C m_s^{-1}\sum_j\widetilde A_{ij,\rho,s}^2
  \le C.
  \label{eq:row-bound-proof}
\end{align}
Combining \eqref{eq:trC2-lower} and \eqref{eq:row-bound-proof},
\begin{equation}
\label{eq:row-negligible}
  \frac{\max_i\sum_j C_{ij,\rho,s}^2}{\tr(\bfC_{\rho,s}^2)}
  \le \frac{C}{c_\sigma m_s/2}\to0.
\end{equation}
Equations \eqref{eq:eigen-negligible} and \eqref{eq:row-negligible} verify the de Jong conditions.  Therefore,
\[
  \frac{\bfvarsigma^{\top}\bfC_{\rho,s}\bfvarsigma}{\{2\tr(\bfC_{\rho,s}^2)\}^{1/2}}
  \dto N(0,1)
\]
conditionally on $\calF_s^0$, in probability.  Substitution of \eqref{eq:quad-form-C} and \eqref{eq:var-C} proves the lemma.
\end{proof}

\begin{lemma}[Cross-scan and cross-ridge covariance factorization]\label{lem:cross-factorization}
Under \(H_0\) and Assumptions \ref{ass:elliptical}--\ref{ass:grid}, let
\(\Xi_{\rho,\rho'}^0(s,r)\) be the conditional covariance of
\(\widetilde V_\rho^0(s)-m_s\kappa_\rho^0(s)\) and
\(\widetilde V_{\rho'}^0(r)-m_r\kappa_{\rho'}^0(r)\), divided by
\((m_sm_r)^{1/2}\).  The following explicit factorization holds uniformly
on the natural \(n^{-1}\)-grids of both scan domains in
Convention~\ref{conv:scan}:
\begin{equation}
\label{eq:cross-factorization}
  \frac{\Xi_{\rho,\rho'}^0(s,r)}{
  \sigma_\rho^0(s)\sigma_{\rho'}^0(r)}
  =r_{E,m}(\rho,\rho')K_0(s,r)+O_P(\ell_n m_*^{-1/2}),
\end{equation}
where, for the common pooled size \(m\) of the two statistics,
\[
  r_{E,m}(\rho,\rho')=
  \frac{\mathfrak c_{\rho,\rho',m}}{\{\mathfrak c_{\rho,\rho,m}\mathfrak c_{\rho',\rho',m}\}^{1/2}},
  \qquad
  r_{E,m}(\rho,\rho')\to r_E(c_{\rm pool};\rho,\rho')
  \quad\text{whenever }m/n\to c_{\rm pool},
\]
where
\[
 r_E(c_{\rm pool};\rho,\rho')=
 \frac{\mathfrak c_{\rho,\rho'}(c_{\rm pool})}
 {\{\mathfrak c_{\rho,\rho}(c_{\rm pool})\mathfrak c_{\rho',\rho'}(c_{\rm pool})\}^{1/2}}.
\]
For both global scans, $m=n$ and we abbreviate
$r_E(\rho,\rho')=r_E(1;\rho,\rho')$.  If the construction is applied inside a
recursive window of length $m_0\asymp n$, the finite statement retains
$r_{E,m_0}$; along a subsequence $m_0/n\to c_{\rm pool}$ its limit is
$r_E(c_{\rm pool};\rho,\rho')$.  In all common-pool cases
$r_E(c_{\rm pool};\rho,\rho)=1$, $\abs{r_E(c_{\rm pool};\rho,\rho')}\le1$, and
\begin{equation*}
  K_0(s,r)=\frac{\psi(s,r)^2}{\psi(s)\psi(r)},
  \qquad
  \psi(s,r)=\int_0^1\varphi_s(x)\varphi_r(x)\,dx,
  \qquad
  \psi(s)=\psi(s,s).
\end{equation*}
Here
\begin{equation*}
\varphi_s(x)=
\begin{cases}
0,&0\le x<t_1,\\
-(t_2-t_1)^{-1},&t_1\le x<t_2,\\
(t_3-t_2)^{-1},&t_2\le x<t_3,\\
0,&t_3\le x\le1.
\end{cases}
\end{equation*}
For \(s_t=(0,t,1)\) and \(s_u=(0,u,1)\),
\begin{equation*}
  K_0(s_t,s_u)=
  \frac{\min(t,u)\{1-\max(t,u)\}}
  {\max(t,u)\{1-\min(t,u)\}}.
\end{equation*}
\end{lemma}

\begin{proof}[Proof of Lemma \ref{lem:cross-factorization}]
Extend every weight vector by zero outside its pool.  Conditionally on
\(\calF^0\),
\[
 \widetilde V_\rho^0(s)-m_s\kappa_\rho^0(s)
 =2m_s\sum_{i<j}c_{ij}^{\rho,s}\varsigma_i\varsigma_j,
 \qquad
 c_{ij}^{\rho,s}
 =\beta_i(s)\beta_j(s)\widetilde A_{ij,\rho,s}.
\]
Here
\[
 c_{ij}^{\rho,s}\in L^0(\calF^0),
 \qquad
 (\varsigma_1,\ldots,\varsigma_n)\perp\calF^0,
 \qquad
 \varsigma_i\stackrel{\rm iid}{\sim}{\rm Unif}\{-1,1\}.
\]
Thus, for \(i<j\) and \(k<l\),
\[
 \E(\varsigma_i\varsigma_j\varsigma_k\varsigma_l\mid\calF^0)
 =\ind\{\{i,j\}=\{k,l\}\},
\]
and therefore
\begin{align*}
 &\Cov\{\widetilde V_\rho^0(s),
          \widetilde V_{\rho'}^0(r)\mid\calF^0\}\\
 &\quad=4m_sm_r\sum_{i<j}c_{ij}^{\rho,s}c_{ij}^{\rho',r}\\
 &\quad=2m_sm_r\sum_{i\ne j}
 \beta_i(s)\beta_i(r)\beta_j(s)\beta_j(r)
 \widetilde A_{ij,\rho,s}\widetilde A_{ij,\rho',r}.
\end{align*}
Hence
\[
 \Xi_{\rho,\rho'}^0(s,r)
 =2(m_sm_r)^{1/2}\sum_{i\ne j}
 \beta_i(s)\beta_i(r)\beta_j(s)\beta_j(r)
 \widetilde A_{ij,\rho,s}\widetilde A_{ij,\rho',r}.
\]
For the common-pool scans,
\begin{equation}
\label{eq:common-pool-proof}
 J(s)=J(r)=J_0,
 \qquad
 m_s=m_r=m,
 \qquad
 \widetilde\bfA_{\rho,s}=\widetilde\bfA_{\rho,J_0},
 \qquad
 \widetilde\bfA_{\rho',r}=\widetilde\bfA_{\rho',J_0}.
\end{equation}
Without \eqref{eq:common-pool-proof}, the exact covariance contains
\[
 2m_sm_r\sum_{\{i,j\}\subset J(s)\cap J(r)}
 \beta_i(s)\beta_j(s)\beta_i(r)\beta_j(r)
 \widetilde A_{ij,\rho,J(s)}\widetilde A_{ij,\rho',J(r)},
\]
so the one-resolvent reduction below is used only under the common-pool
condition.

By \eqref{eq:beta-b0-intermediate},
\[
 \bfbeta_s=\mu_{w,p}^{-1}\bfb^{\rm cus}(s)+\bfr_{\beta,s},
 \qquad
 \|\bfr_{\beta,s}\|_2=O_P(\ell_{0,n}^2m^{-1/2}),
 \qquad
 \|\bfr_{\beta,s}\|_\infty=O_P(\ell_{0,n}^2m^{-1}).
\]
For four vectors with sup norm at most \(Cm^{-1/2}\), define
\[
 \mathfrak L_{\rho,\rho'}(\bfu,\bfv,\bfa,\bfb)
 =m\sum_{i\ne j}u_iv_ia_jb_j
 \widetilde A_{ij,\rho}\widetilde A_{ij,\rho'}.
\]
The companion contractions imply, for a representative first-order term,
\begin{align*}
 &|\mathfrak L_{\rho,\rho'}(
 \bfr_{\beta,s},\bfb^{\rm cus}(r),
 \bfb^{\rm cus}(s),\bfb^{\rm cus}(r))|\\
 &\quad\le
 Cm\|\bfr_{\beta,s}\|_2
 \|\bfb^{\rm cus}(r)\|_\infty
 \left[\sum_j\{b_j^{\rm cus}(s)b_j^{\rm cus}(r)\}^2\right]^{1/2}\\
 &\quad=O_P(\ell_{0,n}^2m^{-1/2}).
\end{align*}
Expanding the four arguments gives
\begin{align*}
 &\mathfrak L_{\rho,\rho'}(\bfbeta_s,\bfbeta_r,
 \bfbeta_s,\bfbeta_r)
 -\mu_{w,p}^{-4}\mathfrak L_{\rho,\rho'}
 (\bfb^{\rm cus}(s),\bfb^{\rm cus}(r),
  \bfb^{\rm cus}(s),\bfb^{\rm cus}(r))\\
 &\quad=
 O_P(\ell_nm^{-1/2})+O_P(\ell_{0,n}^4m^{-1})
 =O_P(\ell_nm^{-1/2}).
\end{align*}
Therefore,
\begin{equation}
\label{eq:Xi-b-reduction-proof}
 \Xi_{\rho,\rho'}^0(s,r)
 =2m\mu_{w,p}^{-4}\sum_{i\ne j}
 b_i^{\rm cus}(s)b_i^{\rm cus}(r)
 b_j^{\rm cus}(s)b_j^{\rm cus}(r)
 \widetilde A_{ij,\rho}\widetilde A_{ij,\rho'}
 +O_P(\ell_nm^{-1/2}).
\end{equation}
The factor \(\mu_{w,p}^{-4}\) appears in both marginal variances and
cancels after standardization.

Lemma~\ref{lem:weighted-offdiag} gives
\begin{align*}
 &m\sum_{i\ne j}
 b_i^{\rm cus}(s)b_i^{\rm cus}(r)
 b_j^{\rm cus}(s)b_j^{\rm cus}(r)
 \widetilde A_{ij,\rho}\widetilde A_{ij,\rho'}\\
 &\quad=
 \mathfrak c_{\rho,\rho',m}
 \left\{\sum_i b_i^{\rm cus}(s)b_i^{\rm cus}(r)\right\}^2
 +O_P(\ell_nm^{-1/2}).
\end{align*}
Consequently,
\begin{align*}
 r_{E,m}(\rho,\rho')
 &=\frac{\mathfrak c_{\rho,\rho',m}}
 {\{\mathfrak c_{\rho,\rho,m}
    \mathfrak c_{\rho',\rho',m}\}^{1/2}}
 \longrightarrow r_E(c_{\rm pool};\rho,\rho'),\\
 r_E(c_{\rm pool};\rho,\rho')
 &=\frac{\mathfrak c_{\rho,\rho'}(c_{\rm pool})}
 {\{\mathfrak c_{\rho,\rho}(c_{\rm pool})
    \mathfrak c_{\rho',\rho'}(c_{\rm pool})\}^{1/2}}.
\end{align*}
In particular,
\[
 r_E(c_{\rm pool};\rho,\rho)=1,
 \qquad
 |r_E(c_{\rm pool};\rho,\rho')|\le1.
\]
For the two global scans, \(c_{\rm pool}=1\), and we write
\(r_E(\rho,\rho')\).

Relabel the common pool as \(1,\ldots,m\), set \(x_i=(i-1)/m\), and use the
left-endpoint convention.  Direct substitution gives
\begin{equation*}
 b_i^{\rm cus}(s)
 =m^{-1/2}\psi(s)^{-1/2}\varphi_s(x_i),
 \qquad
 \psi(s)=\int_0^1\varphi_s(x)^2\,dx.
\end{equation*}
Indeed, on the two adjacent segments the right-hand side equals
\[
 -m^{-1/2}
 \left\{\frac{t_3-t_2}{(t_2-t_1)(t_3-t_1)}\right\}^{1/2},
 \qquad
 m^{-1/2}
 \left\{\frac{t_2-t_1}{(t_3-t_2)(t_3-t_1)}\right\}^{1/2},
\]
and it vanishes elsewhere.  Because
\(\varphi_s\varphi_r\) is constant on the same grid cells,
\[
 \sum_{i=1}^m b_i^{\rm cus}(s)b_i^{\rm cus}(r)
 =\frac{m^{-1}\sum_{i=1}^m\varphi_s(x_i)\varphi_r(x_i)}
 {\{\psi(s)\psi(r)\}^{1/2}}
 =\frac{\psi(s,r)}{\{\psi(s)\psi(r)\}^{1/2}},
\]
where
\[
 \psi(s,r)=\int_0^1\varphi_s(x)\varphi_r(x)\,dx.
\]
Substitution into \eqref{eq:Xi-b-reduction-proof} proves
\eqref{eq:cross-factorization}.

For \(s_t=(0,t,1)\),
\[
 \varphi_{s_t}(x)
 =-t^{-1}\ind(0\le x<t)
 +(1-t)^{-1}\ind(t\le x\le1).
\]
If \(t\le u\),
\begin{align*}
 \psi(s_t,s_u)
 &=\int_0^t\frac{dx}{tu}
 -\int_t^u\frac{dx}{u(1-t)}
 +\int_u^1\frac{dx}{(1-t)(1-u)}\\
 &=\frac{1}{u(1-t)}.
\end{align*}
By symmetry,
\[
 \psi(s_t,s_u)
 =\frac{\min(t,u)\{1-\max(t,u)\}}
 {tu(1-t)(1-u)},
 \qquad
 \psi(s_t)=\{t(1-t)\}^{-1}.
\]
The formula for \(K_0\) follows.
\end{proof}

\begin{lemma}[Oracle variance limit]\label{lem:oracle-variance-limit}
Under \(H_0\) and Assumptions \ref{ass:elliptical}--\ref{ass:grid}, for a
common pool of size \(n\),
\[
 \sigma_{\rho,n}^{\circ 2}
 =2\mu_{w,p}^{-4}\mathfrak c_{\rho,\rho,n}
 \longrightarrow
 \sigma_\rho^2
 :=2\zeta_{-1}^{-4}\mathfrak c_{\rho,\rho}\in(0,\infty)
\]
uniformly for \(\rho\in[\rho_0,\rho_1]\).  For either global scan in
Convention~\ref{conv:scan},
\[
 \sup_{s,\rho}|\sigma_\rho^{0,2}(s)-\sigma_\rho^2|=o_P(1),
 \qquad
 \sup_{s,\rho}|\widehat\sigma_\rho^2(s)-\sigma_\rho^2|=o_P(1).
\]
For a local pool sequence \(m/n\to c_{\rm pool}\in(0,\infty)\), the deterministic
limit is instead
\[
 \sigma_\rho^2(c_{\rm pool})=2\zeta_{-1}^{-4}\mathfrak c_{\rho,\rho}(c_{\rm pool}),
\]
where \(\mathfrak c_{\rho,\rho}(c_{\rm pool})\) uses the companion law at aspect ratio
\(\gamma/c_{\rm pool}\).
\end{lemma}

\begin{proof}
By Lemmas~\ref{lem:angular-w-xi} and \ref{lem:weighted-offdiag},
\[
 \mu_{w,p}\longrightarrow\zeta_{-1},
 \qquad
 \sup_{\rho\in[\rho_0,\rho_1]}
 |\mathfrak c_{\rho,\rho,n}-\mathfrak c_{\rho,\rho}(1)|\longrightarrow0,
\]
with
\[
 0<c\le\inf_\rho\mathfrak c_{\rho,\rho}(1)
 \le\sup_\rho\mathfrak c_{\rho,\rho}(1)\le C<\infty.
\]
Therefore the deterministic variance proxy satisfies
\[
 \sup_\rho\left|
 \sigma_{\rho,n}^{\circ2}
 -2\zeta_{-1}^{-4}\mathfrak c_{\rho,\rho}(1)
 \right|\longrightarrow0,
\]
and the limit is uniformly positive.  Equations
\eqref{eq:oracle-sigma-proxy-rate} and \eqref{eq:sigma-feasible-rate} then
give
\[
 \sup_s|\sigma_\rho^{0,2}(s)-\sigma_{\rho,n}^{\circ2}|=o_P(1),
 \qquad
 \sup_s|\widehat\sigma_\rho^2(s)-\sigma_\rho^{0,2}(s)|=o_P(1).
\]
For a local pool with \(m/n\to c_{\rm pool}\),
\eqref{eq:companion-cov-local-aspect} gives
\[
 \mathfrak c_{\rho,\rho,m}
 \longrightarrow\mathfrak c_{\rho,\rho}(c_{\rm pool}),
\]
so the local-pool variance uses the aspect ratio \(p/m\to\gamma/c_{\rm pool}\),
not the full-sample ratio \(p/n\to\gamma\).
\end{proof}

\begin{lemma}[Multivariate quadratic-form CLT]\label{lem:multi-clt}
For $a=1,\ldots,M$, let $s_{a,n}$ be deterministic candidates on the
natural grid of either global scan in Convention~\ref{conv:scan}, and suppose
$s_{a,n}\to s_a$ in the corresponding trimmed scan domain.  Let
$\rho_n^{(1)},\ldots,\rho_n^{(K)}\in[\rho_0,\rho_1]$ be deterministic with
$\rho_n^{(k)}\to\rho^{(k)}$.  Under $H_0$
and Assumptions~\ref{ass:elliptical}--\ref{ass:grid}, the vector
\[
  \left\{
  \frac{\widetilde V_{\rho_n^{(k)}}^0(s_{a,n})-n\kappa_{\rho_n^{(k)}}^0(s_{a,n})}
  {\{n\sigma_{\rho_n^{(k)}}^{0,2}(s_{a,n})\}^{1/2}}
  :1\le a\le M,\ 1\le k\le K
  \right\}
\]
converges conditionally, in probability, to a centered Gaussian vector with
covariance
\begin{equation*}
  \mathcal K_{k\ell}(a,b)
  =r_E(\rho^{(k)},\rho^{(\ell)})K_0(s_a,s_b).
\end{equation*}
The covariance matrix is allowed to be singular, for example when candidate
points or limiting ridge values coincide.
\end{lemma}

\begin{proof}
For arbitrary fixed real coefficients $b_{ak}$, consider the linear combination
\begin{equation*}
  \mathcal L_n=\sum_{a=1}^M\sum_{k=1}^K b_{ak}
  \frac{\widetilde V_{\rho_n^{(k)}}^0(s_{a,n})-m_{s_{a,n}}\kappa_{\rho_n^{(k)}}^0(s_{a,n})}
  {\{m_{s_{a,n}}\sigma_{\rho_n^{(k)}}^{0,2}(s_{a,n})\}^{1/2}}.
\end{equation*}
Using the sign representation in \eqref{eq:V-sign-expansion}, the same variable has the conditional quadratic-form representation
\begin{equation*}
  \mathcal L_n=\bfvarsigma^{\top}\bfC_n\bfvarsigma .
\end{equation*}
For each pair \((a,k)\), extend \(\beta_i(s_{a,n})\) and \(\widetilde A_{ij,\rho_n^{(k)},s_{a,n}}\) by zero whenever the relevant index is outside \(J(s_{a,n})\), and define the zero-diagonal symmetric coefficient matrix \(\bfC^{(a,k)}\) by
\begin{equation*}
  C_{ij}^{(a,k)}
  =\frac{m_{s_{a,n}}\beta_i(s_{a,n})\beta_j(s_{a,n})
       \widetilde A_{ij,\rho_n^{(k)},s_{a,n}}}
       {\{m_{s_{a,n}}\sigma_{\rho_n^{(k)}}^{0,2}(s_{a,n})\}^{1/2}}\ind(i\ne j),
  \qquad
  1\le i,j\le n.
\end{equation*}
Then \(\bfC^{(a,k)}\) is \(\calF^0\)-measurable and
\begin{equation*}
  \frac{\widetilde V_{\rho_n^{(k)}}^0(s_{a,n})-m_{s_{a,n}}\kappa_{\rho_n^{(k)}}^0(s_{a,n})}
       {\{m_{s_{a,n}}\sigma_{\rho_n^{(k)}}^{0,2}(s_{a,n})\}^{1/2}}
  =\bfvarsigma^\top\bfC^{(a,k)}\bfvarsigma.
\end{equation*}
The diagonal of \(\bfC^{(a,k)}\) is zero, and the sign variables enter only through
\begin{equation*}
  A_{ij,\rho_n^{(k)},s_{a,n}}^0=\varsigma_i\varsigma_j\widetilde A_{ij,\rho_n^{(k)},s_{a,n}},
  \qquad
  \E(\varsigma_i\varsigma_j\varsigma_u\varsigma_v\mid\calF^0)
  =\ind\{\{i,j\}=\{u,v\}\}.
\end{equation*}
Therefore
\[
  \bfC_n=\sum_{a=1}^M\sum_{k=1}^K b_{ak}\bfC^{(a,k)}.
\]
Let \(\bfW_{\beta,a,n}=\diag\{\beta_1(s_{a,n}),\ldots,\beta_n(s_{a,n})\}\), and first include the diagonal:
\[
 \overline\bfC^{(a,k)}=
 \frac{m_{s_{a,n}}\bfW_{\beta,a,n}
       \widetilde\bfA_{\rho_n^{(k)},s_{a,n}}\bfW_{\beta,a,n}}
      {\{m_{s_{a,n}}\sigma_{\rho_n^{(k)}}^{0,2}(s_{a,n})\}^{1/2}}.
\]
Since \(0\preceq \widetilde\bfA_{\rho_n^{(k)},s_{a,n}}\preceq \bfI\),
\(\max_i |\beta_i(s_{a,n})|\le C m_{s_{a,n}}^{-1/2}\), and
\(\sigma_{\rho_n^{(k)}}^{0,2}(s_{a,n})\ge c_\sigma\) with probability tending to one,
\[
 \|\overline\bfC^{(a,k)}\|_{\rm op}\le C m_*^{-1/2}.
\]
Moreover,
\[
 \|\diag(\overline\bfC^{(a,k)})\|_{\rm op}
 \le\max_i\frac{m_{s_{a,n}}\beta_i(s_{a,n})^2
       \widetilde A_{ii,\rho_n^{(k)},s_{a,n}}}
      {\{m_{s_{a,n}}\sigma_{\rho_n^{(k)}}^{0,2}(s_{a,n})\}^{1/2}}
 \le C m_*^{-1/2}.
\]
Because
\(\bfC^{(a,k)}=\overline\bfC^{(a,k)}-
\diag(\overline\bfC^{(a,k)})\), the required zero-diagonal matrix satisfies
\begin{equation*}
  \opnorm{\bfC^{(a,k)}}\le C m_*^{-1/2}.
\end{equation*}
The fixed dimension of the Cramer--Wold combination gives
\begin{equation}
\label{eq:Cn-op-multi}
  \opnorm{\bfC_n}
  \le \sum_{a,k}|b_{ak}|\,\opnorm{\bfC^{(a,k)}}
  =O_P(m_*^{-1/2}).
\end{equation}
For the row-influence bound, \eqref{eq:row-bound-proof} gives, uniformly in \((a,k)\),
\begin{equation*}
  \max_i\sum_j \{C_{ij}^{(a,k)}\}^2
  \le C m_*^{-1}.
\end{equation*}
Thus
\begin{equation}
\label{eq:Cn-row-multi}
  \max_i\sum_j C_{n,ij}^2
  \le MK\sum_{a=1}^M\sum_{k=1}^K b_{ak}^2
      \max_i\sum_j \{C_{ij}^{(a,k)}\}^2
  =O_P(m_*^{-1}).
\end{equation}
The conditional variance of $\mathcal L_n$ is
\begin{equation*}
  \Var(\mathcal L_n\mid\calF^0)=2\tr(\bfC_n^2).
\end{equation*}
Lemma~\ref{lem:cross-factorization}, the ridge convergence, and continuity
of $K_0$ on the trimmed scan domains give
\begin{align*}
  2\tr(\bfC_n^2)
  &=\sum_{a,b=1}^M\sum_{k,\ell=1}^K b_{ak}b_{b\ell}
    r_E(\rho^{(k)},\rho^{(\ell)})K_0(s_a,s_b)+o_P(1)\\
  &=\sigma_L^2+o_P(1),
\end{align*}
where
\[
  \sigma_L^2=\sum_{a,b=1}^M\sum_{k,\ell=1}^K b_{ak}b_{b\ell}
  r_E(\rho^{(k)},\rho^{(\ell)})K_0(s_a,s_b).
\]
If $\sigma_L^2>0$, then \eqref{eq:Cn-op-multi} and \eqref{eq:Cn-row-multi} imply
\[
  \frac{\opnorm{\bfC_n}}{\{\tr(\bfC_n^2)\}^{1/2}}\to0,
  \qquad
  \frac{\max_i\sum_j C_{n,ij}^2}{\tr(\bfC_n^2)}\to0.
\]
Moreover,
\begin{equation*}
 \frac{\tr(\bfC_n^4)}{\{\tr(\bfC_n^2)\}^2}
 \le
 \frac{\opnorm{\bfC_n}^2}{\tr(\bfC_n^2)}\to0,
\end{equation*}
so these are precisely the Rademacher quadratic-form conditions used in
Lemma~\ref{lem:clt}.  Hence, conditionally on \(\calF^0\),
\begin{equation*}
  \mathcal L_n\dto N(0,\sigma_L^2)\quad(\sigma_L^2>0).
\end{equation*}
If \(\sigma_L^2=0\), then
\begin{equation*}
  \E(\mathcal L_n^2\mid\calF^0)=2\tr(\bfC_n^2)\pto0,
  \qquad
  \mathcal L_n\pto0.
\end{equation*}
The Cramer--Wold device gives the asserted vector limit.
\end{proof}

For the continuous single-change scan, we use the following increment bound.

\begin{lemma}[Tightness for the common-pool scans]\label{lem:tightness}\label{lem:tightness-mc}
Under \(H_0\) and Assumptions \ref{ass:elliptical}--\ref{ass:grid}, the
oracle and feasible processes are tight on both scan domains in
Convention~\ref{conv:scan}.  Define the oracle process normalized by the
common deterministic proxy as
\[
 \bar Z_\rho^0(s)=
 \frac{\widetilde V_\rho^0(s)-n\kappa_\rho^0(s)}
      {\sqrt n\,\sigma_{\rho,n}^\circ},
\]
and define the candidate-normalized oracle process as
\[
 Z_\rho^0(s)=
 \frac{\widetilde V_\rho^0(s)-n\kappa_\rho^0(s)}
      {\{n\sigma_\rho^{0,2}(s)\}^{1/2}}.
\]
More precisely, there are events \(\mathcal E_n^{\rm proc}\), with
\(\Pbb(\mathcal E_n^{\rm proc})\to1\), such that, conditionally on the
unsigned direction representatives and inverse-distance weights, that is, on \(\calF^0\), \(\bar Z_\rho^0\)
satisfies on \(\mathcal E_n^{\rm proc}\)
\begin{align}
 \E_{\varsigma}|\bar Z_\rho^0(s_t)-\bar Z_\rho^0(s_u)|^4
 &\le C\{|t-u|+n^{-1}\}^2,\label{eq:increment-fourth}\\
 \E_{\varsigma}|\bar Z_\rho^0(s)-\bar Z_\rho^0(r)|^8
 &\le C\{\|s-r\|_1+n^{-1}\}^4,
 \qquad s,r\in\calS_{\rm mc}(\varepsilon).
 \label{eq:mc-increment-eighth}
\end{align}
Furthermore,
\begin{align}
 \sup_{s\in\calS_{\rm sc}(\varepsilon)}
 |Z_\rho(s)-\bar Z_\rho^0(s)|&=o_P(1),
 \label{eq:tightness-feasible-replacement}\\
 \sup_{s\in\calS_{\rm mc}(\varepsilon)}
 |Z_\rho(s)-\bar Z_\rho^0(s)|&=o_P(1).
 \label{eq:mc-tightness-feasible-replacement}
\end{align}
\end{lemma}

\begin{proof}
Work on the natural grid and on the event \(\mathcal E_n^{\rm proc}\) on
which all radial, angular, inverse-distance, and companion-contraction bounds
hold simultaneously.  By construction,
\[
 \Pbb\{(\mathcal E_n^{\rm proc})^c\}=o(1).
\]

Let \(I=\{a+1,\ldots,b\}\) and
\(I'=\{a'+1,\ldots,b'\}\), with
\(|I|,|I'|\ge\varepsilon n\), and define
\[
 d_{\rm end}(I,I')=\frac{|a-a'|+|b-b'|}{n}.
\]
On \(\mathcal E_n^{\rm proc}\),
\[
 \left|\sum_{j=u+1}^{v}(w_j-\mu_{w,p})\right|
 \le C\ell_{0,n}\sqrt{|v-u|+1},
 \qquad 0\le u<v\le n.
\]
Decomposing \(I\) and \(I'\) into their intersection and symmetric
difference gives
\begin{equation}
\label{eq:interval-average-increment-general}
 |e_I-e_{I'}|
 \le Cd_{\rm end}(I,I')
 +C\ell_{0,n}n^{-1/2}
 \{d_{\rm end}(I,I')+n^{-1}\}^{1/2}.
\end{equation}

Let \(\bfb^{\rm cus}(s)\) be the ordinary adjacent-segment CUSUM vector.
On the trimmed domain, its two nonzero coefficients are
\(n^{-1/2}\) times Lipschitz functions of the normalized segment lengths.
Hence
\begin{equation}
\label{eq:bcus-common-increment}
 \|\bfb^{\rm cus}(s)-\bfb^{\rm cus}(r)\|_2^2
 \le C\{d_{\calS}(s,r)+n^{-1}\},
 \qquad
 \max_i\{|b_i^{\rm cus}(s)|+|b_i^{\rm cus}(r)|\}
 \le Cn^{-1/2}.
\end{equation}
Writing \(\bfb_a^{\rm cus}(s)\) for the restriction to \(I_a(s)\),
\[
 \bfbeta_s=e_{1,s}^{-1}\bfb_1^{\rm cus}(s)
           +e_{2,s}^{-1}\bfb_2^{\rm cus}(s).
\]
The uniform lower bound for \(e_{a,s}\), the reciprocal identity
\[
 e^{-1}-{e'}^{-1}=\frac{e'-e}{ee'},
\]
and \eqref{eq:interval-average-increment-general}--
\eqref{eq:bcus-common-increment} imply
\begin{equation}
\label{eq:beta-common-increment}
 \|\bfbeta_s-\bfbeta_r\|_2^2
 \le C\{d_{\calS}(s,r)+n^{-1}\},
 \qquad
 \|\bfbeta_s-\bfbeta_r\|_\infty\le Cn^{-1/2}.
\end{equation}
Here
\[
 d_{\calS}(s_t,s_u)=|t-u|
 \quad\text{and}\quad
 d_{\calS}(s,r)=\|s-r\|_1
\]
for the single-change and adjacent-triple scans, respectively.

The unsigned companion matrix is common to every candidate in either global
scan.  Define the zero-diagonal standardized coefficient matrix
\[
 \bfC_{\rho,s}^{\rm std}
 =\frac{\sqrt n}{\sigma_{\rho,n}^\circ}
 \left[
 \bfW_{\beta,s}\widetilde\bfA_\rho\bfW_{\beta,s}
 -\diag\{\bfW_{\beta,s}\widetilde\bfA_\rho\bfW_{\beta,s}\}
 \right].
\]
Using \(0\preceq\widetilde\bfA_\rho\preceq\bfI\),
\eqref{eq:beta-common-increment}, and
\(\inf_\rho\sigma_{\rho,n}^\circ>c>0\),
\begin{equation}
\label{eq:common-coefficient-increment}
 \tr\{(\bfC_{\rho,s}^{\rm std}-
        \bfC_{\rho,r}^{\rm std})^2\}
 \le C\{d_{\calS}(s,r)+n^{-1}\},
 \qquad
 \|\bfC_{\rho,s}^{\rm std}-
   \bfC_{\rho,r}^{\rm std}\|_{\rm op}
 \le C\{d_{\calS}(s,r)+n^{-1}\}^{1/2}.
\end{equation}
Conditionally on \(\calF^0\),
\[
 \bar Z_\rho^0(s)-\bar Z_\rho^0(r)
 =\bfvarsigma^\top
 (\bfC_{\rho,s}^{\rm std}-\bfC_{\rho,r}^{\rm std})
 \bfvarsigma.
\]
For \(q\in\{2,4,8\}\), degree-two Walsh hypercontractivity and
\eqref{eq:common-coefficient-increment} yield
\[
 \E_{\varsigma}|\bar Z_\rho^0(s)-\bar Z_\rho^0(r)|^q
 \le C_q\{d_{\calS}(s,r)+n^{-1}\}^{q/2}.
\]
This proves \eqref{eq:increment-fourth} and
\eqref{eq:mc-increment-eighth} on \(\mathcal E_n^{\rm proc}\).

Set
\[
 (d_{\rm scan},q_0)
 =\begin{cases}
 (1,4),&\calS_{\rm sc}(\varepsilon),\\
 (3,8),&\calS_{\rm mc}(\varepsilon),
 \end{cases}
 \qquad
 q_0/2-d_{\rm scan}=1.
\]
Construct nested admissible nets \(\mathcal N_{j,n}\) satisfying
\[
 \mathcal N_{j-1,n}\subset\mathcal N_{j,n},
 \qquad
 |\mathcal N_{j,n}|\le C2^{d_{\rm scan}j},
\]
\[
 d_{\calS}\{u,\operatorname{par}(u)\}\le C2^{-j},
 \qquad u\in\mathcal N_{j,n}\setminus\mathcal N_{j-1,n}.
\]
Stop at
\[
 J_n=\min\{j:2^{-j}\le n^{-1}\},
\]
and adjoin all remaining natural-grid points at level \(J_n\).  Then
\[
 |\mathcal N_{J_n,n}|=O(n^{d_{\rm scan}}),
 \qquad
 d_{\calS}\{u,\operatorname{par}(u)\}=O(n^{-1})
\]
for the added final-level points.

For a parent edge at level \(j\),
\[
 \E_{\varsigma}
 |\bar Z_\rho^0(u)-\bar Z_\rho^0(\operatorname{par}u)|^{q_0}
 \le C2^{-q_0j/2}.
\]
Therefore, for \(x_j>0\),
\[
 \Pbb_{\varsigma}\left(
 \max_{u\in\mathcal N_{j,n}}
 |\bar Z_\rho^0(u)-\bar Z_\rho^0(\operatorname{par}u)|>x_j
 \right)
 \le C2^{-j(q_0/2-d_{\rm scan})}x_j^{-q_0}.
\]
Choose
\[
 0<\theta_{\rm ch}<\frac{q_0/2-d_{\rm scan}}{q_0},
 \qquad
 x_j=C_{\rm ch}\delta
 2^{-\theta_{\rm ch}(j-j_0)},
 \qquad
 \sum_{j>j_0}x_j\le\delta.
\]
Then
\begin{align*}
 &\sum_{j>j_0}
 2^{-j(q_0/2-d_{\rm scan})}x_j^{-q_0}
 \longrightarrow0
 \qquad(j_0\to\infty),\\
 &n^{d_{\rm scan}}n^{-q_0/2}x_{J_n}^{-q_0}
 \longrightarrow0.
\end{align*}
Thus every natural-grid point differs from its level-\(j_0\) ancestor by at
most \(\delta\), except on a conditional event whose probability tends to
zero uniformly in \(n\).

For the base level, define
\[
 \mathcal N_{j_0,n}^{\rm pair}
 =\{(u,v)\in\mathcal N_{j_0,n}^2:
 d_{\calS}(u,v)\le C_02^{-j_0}\}.
\]
Packing gives
\[
 |\mathcal N_{j_0,n}^{\rm pair}|
 \le C2^{d_{\rm scan}j_0}.
\]
Hence
\[
 \Pbb_{\varsigma}\left(
 \max_{(u,v)\in\mathcal N_{j_0,n}^{\rm pair}}
 |\bar Z_\rho^0(u)-\bar Z_\rho^0(v)|>\delta
 \right)
 \le C\delta^{-q_0}2^{-j_0(q_0/2-d_{\rm scan})}.
\]
Joining each point to its ancestor, crossing one base-level edge, and
following the second ancestor chain gives
\begin{equation}
\label{eq:direct-grid-chaining}
 \lim_{j_0\to\infty}\limsup_{n\to\infty}
 \Pbb_{\varsigma}\left(
 \sup_{\substack{s,r\text{ on the natural grid}\\
 d_{\calS}(s,r)\le C2^{-j_0}}}
 |\bar Z_\rho^0(s)-\bar Z_\rho^0(r)|>3\delta
 \right)=0.
\end{equation}
Taking expectations and adding
\(\Pbb\{(\mathcal E_n^{\rm proc})^c\}=o(1)\) converts
\eqref{eq:direct-grid-chaining} into the unconditional modulus bound.  With
the finite-dimensional convergence, this proves tightness of
\(\bar Z_\rho^0\) on both scan domains.

Finally,
\[
 \sup_s\left|\frac{\sigma_\rho^0(s)}
 {\sigma_{\rho,n}^\circ}-1\right|=o_P(1),
 \qquad
 \sup_s|\bar Z_\rho^0(s)|=O_P(1),
\]
so
\[
 \sup_s|Z_\rho^0(s)-\bar Z_\rho^0(s)|=o_P(1).
\]
Proposition~\ref{prop:raw-score},
\eqref{eq:kappa-feasible-rate},
\eqref{eq:sigma-feasible-rate}, and \(\ell_n=o(\sqrt n)\) give
\[
 \sup_s|Z_\rho(s)-Z_\rho^0(s)|=o_P(1).
\]
The triangle inequality proves
\eqref{eq:tightness-feasible-replacement} and
\eqref{eq:mc-tightness-feasible-replacement}.
\end{proof}

\begin{lemma}[Uniform small-shift reduction]\label{lem:alternative-expansion}\label{lem:alternative-perturbation}\label{lem:alt-cross-term}
Assume Assumptions \ref{ass:elliptical}--\ref{ass:grid}.  Suppose there is one
change at \(\tau_\star\in[\varepsilon,1-\varepsilon]\), put
\(\vartheta_i^\star=\ind\{i>\lfloor n\tau_\star\rfloor\}\),
\(r_{\Delta,n}=\|\bfDelta_p\|\), and assume
\[
 r_{\Delta,n}\to0,\qquad \ell_n r_{\Delta,n}\to0.
\]
For every scanned segment \(I\), let
\[
 \bar\vartheta_I^\star=|I|^{-1}\sum_{i\in I}\vartheta_i^\star.
\]
The centered-error median \(\widehat\bftheta^{(0)}(I)\) is the object defined
in \eqref{eq:centered-error-spatial-median}.  For every scan candidate, put
\[
 \widehat\bfDelta_s^{(0)}
 =\widehat\bftheta^{(0)}(I_2(s))
  -\widehat\bftheta^{(0)}(I_1(s)).
\]
For \(i\in I_a(s)\), \(a=1,2\), define the centered-error feasible inverse-distance weights and their segment averages by
\[
 \widehat w_{i,a,s}^{(0)}=
 \begin{cases}
  \sqrt p\big/\|\bfvarepsilon_i-\widehat\bftheta^{(0)}(I_a(s))\|,
   &\bfvarepsilon_i\ne\widehat\bftheta^{(0)}(I_a(s)),\\
  0,&\bfvarepsilon_i=\widehat\bftheta^{(0)}(I_a(s)),
 \end{cases}
 \qquad
 \widehat e_{a,s}^{(0)}=\frac1{n_a(s)}\sum_{i\in I_a(s)}\widehat w_{i,a,s}^{(0)}.
\]
Then, uniformly over the trimmed segment collection,
\begin{align}
 \widehat\bftheta(I)
 &=\bftheta^{(1)}+\widehat\bftheta^{(0)}(I)
   +\bar\vartheta_I^\star\bfDelta_p+\bfr_{\Delta,I}^{\rm alt},\notag\\
 \sup_I\|\bfr_{\Delta,I}^{\rm alt}\|
 &=O_P(\ell_nm_*^{-1})+O_P(r_{\Delta,n}m_*^{-1/2})
   +O(r_{\Delta,n}^2p^{-1/2})+O_P(\ell_np^{-1}r_{\Delta,n}).
 \label{eq:alternative-median-reduction}
\end{align}
Consequently,
\[
 \widehat\bfDelta_s
 =\widehat\bfDelta_s^{(0)}
  +\vartheta_{s,n}(\tau_\star)\bfDelta_p+\bfr_{\Delta,s}^{\rm alt},
 \qquad
 \vartheta_{s,n}(\tau_\star)=\bar\vartheta_{I_2(s)}^\star-\bar\vartheta_{I_1(s)}^\star.
\]
On the single-change path,
\[
 \vartheta_{s_t,n}(\tau_\star)=
 \begin{cases}
 (1-\tau_\star)/(1-t),&t<\tau_\star,\\
 \tau_\star/t,&t\ge\tau_\star,
 \end{cases}
 +O(n^{-1})
\]
uniformly on the natural grid.

Let \(\widehat\bfR_s^{(0)}\), \(\widehat\bfQ_{\rho,s}^{(0)}\),
\(\widehat\kappa_\rho^{(0)}(s)\), and
\(\widehat\sigma_\rho^{(0),2}(s)\) denote the feasible quantities computed
from the centered errors.  Write \(\widehat\sigma_\rho^{(0)}(s)\) for the
nonnegative square root of \(\widehat\sigma_\rho^{(0),2}(s)\).  With
\[
 \mathfrak e_{\sigma,n}^{\rm pool}
 =\ell_nr_{\Delta,n}+r_{\Delta,n}^2+\ell_nm_*^{-1/2},
\]
the following bounds hold uniformly:
\begin{align}
 \|\widehat\bfR_s-\widehat\bfR_s^{(0)}\|_{\rm op}
 &=O_P(\ell_nr_{\Delta,n}+r_{\Delta,n}^2+\ell_nm_*^{-1})=o_P(1),
 \label{eq:alt-scatter-op}\\
 |\bfDelta_p^\top
  (\widehat\bfQ_{\rho,s}-\widehat\bfQ_{\rho,s}^{(0)})
  \bfDelta_p|&=o_P(r_{\Delta,n}^2),\notag\\
 \bfDelta_p^\top\widehat\bfQ_{\rho,s}^{(0)}\bfDelta_p
 &=\bfDelta_p^\top\bfD_{\rho,m_s}\bfDelta_p+o_P(r_{\Delta,n}^2),
 \label{eq:alt-det-quadratic}\\
 \sqrt{m_s}|\widehat\kappa_\rho(s)
              -\widehat\kappa_\rho^{(0)}(s)|
 &=O_P(\ell_nr_{\Delta,n}+\ell_nm_*^{-1/2})=o_P(1),\notag\\
 |\widehat\sigma_\rho^2(s)
   -\widehat\sigma_\rho^{(0),2}(s)|
 &=O_P(\mathfrak e_{\sigma,n}^{\rm pool})=o_P(1),\notag\\
 \frac{N_s|(\widehat\bfDelta_s^{(0)})^\top
 (\widehat\bfQ_{\rho,s}-\widehat\bfQ_{\rho,s}^{(0)})
 \widehat\bfDelta_s^{(0)}|}{\sqrt{m_s}}
 &=O_P(\mathfrak e_{\sigma,n}^{\rm pool})=o_P(1).
 \label{eq:alt-null-quadratic}
\end{align}
For the full-pool single-change scan, with
\(\bfB_s=\sum_i\beta_i(s)\bfY_i\),
\begin{equation}
\label{eq:alt-mixed-resolvent-perturbation}
 \sup_{s,\rho}N_s^{1/2}n^{-1/2}
 |\bfDelta_p^\top
 (\widehat\bfQ_{\rho,s}-\widehat\bfQ_{\rho,s}^{(0)})\bfB_s|
 =O_P(\ell_nr_{\Delta,n}^2+r_{\Delta,n}^3+\ell_nr_{\Delta,n}n^{-1/2})=o_P(1).
\end{equation}
Finally,
\begin{align}
 \sup_t\frac{N_{s_t}|\bfDelta_p^\top\widehat\bfQ_{\rho,s_t}^{(0)}
          \widehat\bfDelta_{s_t}^{(0)}|}{\sqrt n}
 &=O_P\{\ell_n(\bfDelta_p^\top
          \bfD_{\rho,n}\bfDelta_p)^{1/2}\}+o_P(1),
 \label{eq:alt-cross-centeredQ}\\
 \sup_t\frac{N_{s_t}|\bfDelta_p^\top\widehat\bfQ_{\rho,s_t}
          \widehat\bfDelta_{s_t}^{(0)}|}{\sqrt n}
 &=O_P\{\ell_n(\bfDelta_p^\top
          \bfD_{\rho,n}\bfDelta_p)^{1/2}\}+o_P(1).
 \label{eq:alt-cross-feasibleQ}
\end{align}
For the full-pool single-change scan, \(m_*=n\) and
\(r_{\Delta,n}=\|\bfDelta_p\|\), so
\(\mathfrak e_{\sigma,n}^{\rm pool}=\mathfrak e_{\sigma,n}\), where the latter
is the global quantity used in Theorem~\ref{thm:localization}.
\end{lemma}

\begin{proof}
Fix a segment \(I\), put \(m=|I|\), and define
\[
 \bfa_i=(\vartheta_i^\star-\bar\vartheta_I^\star)\bfDelta_p,
 \qquad
 m^{-1}\sum_{i\in I}\bfa_i=\mathbf0,
\]
\[
 \widehat\bfd_I
 =\widehat\bftheta(I)-\bftheta^{(1)}
  -\bar\vartheta_I^\star\bfDelta_p.
\]
After removing the average shift, the score is
\[
 \bfS_I(\bfd)
 =m^{-1}\sum_{i\in I}\sqrt p\,
 U(\bfvarepsilon_i+\bfa_i-\bfd).
\]
The spatial-sign Taylor formula gives
\[
 \bfS_I(\bfd)
 =\bar\bfY_I+\bfJ_I^a-\bfJ_I\bfd
  +\bfT_{2,I}(\bfd)+\bfT_{3,I}(\bfd),
\]
where
\[
 \bfJ_I^a=m^{-1}\sum_iw_i\bfP_i\bfa_i,
 \qquad
 \bfJ_I=m^{-1}\sum_iw_i\bfP_i.
\]
On \(\max_iw_i\le W_n=C\ell_{0,n}\), the expansion is valid for
\(\|\bfd\|\le M\ell_{0,n}\), because
\[
 \max_i\|\bfa_i\|\le r_{\Delta,n}=o(1),
 \qquad
 W_n(M\ell_{0,n}+r_{\Delta,n})=o(\sqrt p).
\]
For \(\|\bfd\|=M\ell_{0,n}\),
\begin{align*}
 \bfd^\top\bfS_I(\bfd)
 &\le M\ell_{0,n}\sup_I\|\bar\bfY_I\|
  +M\ell_{0,n}\sup_I e_Ir_{\Delta,n}
  -cM^2\ell_{0,n}^2\\
 &\quad+O_P\left\{
 p^{-1/2}M\ell_{0,n}(M\ell_{0,n}+r_{\Delta,n})^2
 +p^{-1}M\ell_{0,n}(M\ell_{0,n}+r_{\Delta,n})^3
 \right\}.
\end{align*}
For sufficiently large fixed \(M\), the right-hand side is negative with
arbitrarily high probability.  Convexity gives
\[
 \sup_I\|\widehat\bfd_I\|=O_P(\ell_{0,n}).
\]
Moreover,
\[
 \inf_I\min_{i\in I}
 \|\bfvarepsilon_i+\bfa_i-\widehat\bfd_I\|
 \ge\frac{\sqrt p}{W_n}-r_{\Delta,n}
   -\sup_I\|\widehat\bfd_I\|>0
\]
with probability tending to one.  Thus every shifted median is unique,
avoids the observations, and satisfies
\[
 \bfS_I(\widehat\bfd_I)=\mathbf0.
\]

Since \(\sum_i\bfa_i=0\),
\[
 \bfJ_I^a
 =m^{-1}\sum_i(w_i-e_I)\bfa_i
  -m^{-1}\sum_iw_i\bfU_i\bfU_i^\top\bfa_i.
\]
Bernstein's inequality and the Jacobian bound imply
\[
 \sup_I\|\bfJ_I^a\|
 =O_P(r_{\Delta,n}m_*^{-1/2})
  +O_P(\ell_np^{-1}r_{\Delta,n}).
\]
After reduction of repeated signs, the quadratic Taylor average consists of
one deterministic degree-zero term and Walsh chaoses of degrees one and
three, with
\[
 \|\text{degree-zero term}\|
 =O(r_{\Delta,n}^2p^{-1/2}),
 \qquad
 \sum_{\mathfrak m}\|c_{\mathfrak m}\|^2=O((mp)^{-1}).
\]
Hypercontractivity and the polynomial scan union give
\[
 \sup_I\|\bfT_{2,I}\|
 =O_P(\ell_nm_*^{-1})+O(r_{\Delta,n}^2p^{-1/2}),
 \qquad
 \sup_I\|\bfT_{3,I}\|=O_P(\ell_nm_*^{-1}).
\]
Using
\[
 \bfJ_I^{-1}=e_I^{-1}\bfI_p+O_P(\ell_np^{-1}),
 \qquad
 \inf_Ie_I>c>0,
\]
and subtracting the centered-error score equation proves
\eqref{eq:alternative-median-reduction}.

For a pooled segment \(J\), define
\[
 \bfz_{i,J}=\bfvarepsilon_i-\widehat\bftheta^{(0)}(J),
 \qquad
 \bfa_{i,J}=(\vartheta_i^\star-\bar\vartheta_J^\star)\bfDelta_p-\bfd_J,
 \qquad
 \bfd_J=\bfr_{\Delta,J}^{\rm alt}.
\]
With \(\bfY_{i,J}^{(0)}=\sqrt p\,U(\bfz_{i,J})\),
\[
 \widehat\bfY_{i,J}-\bfY_{i,J}^{(0)}
 =\widetilde w_{i,J}\widetilde\bfP_{i,J}\bfa_{i,J}
  +\bfH_{i,J}(\bfa_{i,J})+\bfr_{i,J}.
\]
Separate the scatter perturbation as
\[
 \widehat\bfR_J-\widehat\bfR_J^{(0)}
 =\bfM_{1,J}^{o}+\bfM_{2,J},
\]
where
\[
 \bfM_{1,J}^{o}
 =m^{-1}\sum_{i\in J}
 (\vartheta_i^\star-\bar\vartheta_J^\star)
 \left[
 \bfY_i\{w_i\bfP_i\bfDelta_p\}^\top
 +\{w_i\bfP_i\bfDelta_p\}\bfY_i^\top
 \right].
\]
The centered-column expansion, the median remainder, and matrix Bernstein
give
\begin{align}
 \sup_J\|\bfM_{1,J}^{o}\|_{\rm op}
 &=O_P(\ell_nr_{\Delta,n}),\notag\\
 \sup_J\|\bfM_{2,J}\|_{\rm op}
 &=O_P(r_{\Delta,n}^2+\ell_nm_*^{-1}
      +\ell_nr_{\Delta,n}m_*^{-1/2}).
 \label{eq:alt-M12-op}
\end{align}
This is \eqref{eq:alt-scatter-op}.

Conditionally on \(\calF^0\), write
\[
 \bfY_i=\varsigma_i\widetilde\bfY_i,
 \qquad
 \widetilde\vartheta_{i,J}^\star
 =\vartheta_i^\star-\bar\vartheta_J^\star,
 \qquad
 \bfv_i=w_i\bfP_i\bfDelta_p.
\]
Then
\begin{equation*}
 \bfM_{1,J}^{o}
 =m^{-1}\sum_{i\in J}
 \widetilde\vartheta_{i,J}^\star\varsigma_i
 \{\widetilde\bfY_i\bfv_i^\top
   +\bfv_i\widetilde\bfY_i^\top\}.
\end{equation*}
For
\[
 \bfB_s=\sum_j\beta_j(s)\varsigma_j\widetilde\bfY_j
\]
and an unsigned ridge inverse \(\bfQ_0\),
\begin{align}
 &m^{-1/2}\bfB_s^\top\bfQ_0\bfM_{1,J}^{o}\bfQ_0\bfB_s\notag\\
 &\quad=\frac{1}{m^{3/2}}\sum_{i,j,k}
 \widetilde\vartheta_{i,J}^\star\beta_j(s)\beta_k(s)
 \varsigma_i\varsigma_j\varsigma_k\notag\\
 &\qquad\times\left[
 (\widetilde\bfY_j^\top\bfQ_0\widetilde\bfY_i)
 (\bfv_i^\top\bfQ_0\widetilde\bfY_k)
 +(\widetilde\bfY_j^\top\bfQ_0\bfv_i)
 (\widetilde\bfY_i^\top\bfQ_0\widetilde\bfY_k)
 \right].
 \label{eq:alternative-projected-Walsh}
\end{align}
Repeated indices reduce this expression to Walsh degrees one and three.  The
contractions
\[
 m^{-1}\widetilde\bfY^\top\bfQ_0\widetilde\bfY\preceq\bfI_m,
 \qquad
 \sum_j\beta_j(s)^2
 (\widetilde\bfY_j^\top\bfQ_0\bfv_i)^2
 \le C\|\bfv_i\|^2,
 \qquad
 \max_j\beta_j(s)^2\le Cm^{-1}
\]
and
\[
 m^{-1}\sum_iw_i^r=O_P(1),
 \qquad r\le4,
\]
give the coefficient-energy bound
\[
 \sum_{\mathfrak m}|c_{\mathfrak m}^{\rm ss}|^2
 \le Cr_{\Delta,n}^2.
\]
Similarly,
\begin{align*}
 \bfDelta_p^\top\bfQ_0\bfM_{1,J}^{o}\bfQ_0\bfB_s
 &=\frac1m\sum_{i,k}
 \widetilde\vartheta_{i,J}^\star\beta_k(s)
 \varsigma_i\varsigma_k\\
 &\quad\times\left[
 (\bfDelta_p^\top\bfQ_0\widetilde\bfY_i)
 (\bfv_i^\top\bfQ_0\widetilde\bfY_k)
 +(\bfDelta_p^\top\bfQ_0\bfv_i)
 (\widetilde\bfY_i^\top\bfQ_0\widetilde\bfY_k)
 \right],
\end{align*}
which reduces to Walsh degrees zero and two, with
\[
 \sum_{\mathfrak m}|c_{\mathfrak m}^{\rm ds}|^2
 \le Cr_{\Delta,n}^4.
\]
Hypercontractivity therefore yields
\begin{align*}
 m^{-1/2}|\bfB_s^\top\bfQ_0\bfM_{1,J}^{o}\bfQ_0\bfB_s|
 &=O_P(\ell_nr_{\Delta,n}),\\
 |\bfDelta_p^\top\bfQ_0\bfM_{1,J}^{o}\bfQ_0\bfB_s|
 &=O_P(\ell_nr_{\Delta,n}^2).
\end{align*}

Expanding the perturbed columns before the resolvent keeps the direct
column terms.  The induced first-order companion increments satisfy
\begin{equation*}
 \sum_{\mathfrak m}|c_{\mathfrak m}^{\rm diag}|^2
 \le Cr_{\Delta,n}^2,
 \qquad
 \sum_{\mathfrak m}|c_{\mathfrak m}^{\rm off}|^2
 \le Cr_{\Delta,n}^2.
\end{equation*}
Indeed,
\[
 \sum_j\widetilde A_{ij,\rho}^2\le1,
 \qquad
 \|\bfW_{\beta,s}\|_{\rm op}\le Cm^{-1/2},
 \qquad
 m^{-1}\sum_i\|\bfv_i\|^2=O_P(r_{\Delta,n}^2).
\]

For the second-order perturbation \(\bfM_{2,J}\), the same projected
contractions give
\begin{align}
 m^{-1/2}|\bfB_s^\top\bfQ_0\bfM_{2,J}\bfQ_0\bfB_s|
 &=O_P(r_{\Delta,n}^2+\ell_nm_*^{-1/2}),\notag\\
 |\bfDelta_p^\top\bfQ_0\bfM_{2,J}\bfQ_0\bfB_s|
 &=O_P(r_{\Delta,n}^3+\ell_nr_{\Delta,n}m_*^{-1/2}),\notag\\
 \sqrt m\left|\sum_i\beta_i(s)^2\delta A_{ii}^{(2)}\right|
 &=O_P(\ell_nm_*^{-1/2}),\notag\\
 m\left|\sum_{i\ne j}\beta_i(s)^2\beta_j(s)^2
 A_{ij}\delta A_{ij}^{(2)}\right|
 &=O_P(r_{\Delta,n}^2+\ell_nm_*^{-1/2}).
 \label{eq:alternative-M2-projections}
\end{align}
For example,
\[
 \E\{(\bfDelta_p^\top\bfQ_0\bfB_s)^2\mid\calF^0\}
 \le C r_{\Delta,n}^2,
\]
so a rank-one term \(\bfDelta_p\bfDelta_p^\top\) is projected at order
\(r_{\Delta,n}^2\), rather than bounded by
\(\|\bfB_s\|^2r_{\Delta,n}^2\).  Combining
\eqref{eq:alternative-projected-Walsh}--
\eqref{eq:alternative-M2-projections} gives, uniformly,
\begin{align*}
 m^{-1/2}\bfB_s^\top\bfQ_0(\,\cdot\,)\bfQ_0\bfB_s
 &=O_P(\ell_nr_{\Delta,n}+r_{\Delta,n}^2+\ell_nm_*^{-1/2}),\\
 \bfDelta_p^\top\bfQ_0(\,\cdot\,)\bfQ_0\bfB_s
 &=O_P(\ell_nr_{\Delta,n}^2+r_{\Delta,n}^3
       +\ell_nr_{\Delta,n}m_*^{-1/2}),\\
 \sqrt m\sum_i\beta_i(s)^2\delta A_{ii}
 &=O_P(\ell_nr_{\Delta,n}+\ell_nm_*^{-1/2}),\\
 m\sum_{i\ne j}\beta_i(s)^2\beta_j(s)^2A_{ij}\delta A_{ij}
 &=O_P(\ell_nr_{\Delta,n}+r_{\Delta,n}^2
       +\ell_nm_*^{-1/2}).
\end{align*}
Score--score projections use Lemmas~\ref{lem:score-resolvent} and
\ref{lem:companion-cancellation}; deterministic projections use
Lemma~\ref{lem:deterministic-centered-transfer}.

Let
\[
 \bfM_J=\bfM_{1,J}^{o}+\bfM_{2,J}.
\]
The exact second-order resolvent identity is
\[
 \widehat\bfQ-\widehat\bfQ^{(0)}
 =-\widehat\bfQ^{(0)}\bfM_J\widehat\bfQ^{(0)}
 +\widehat\bfQ^{(0)}\bfM_J\widehat\bfQ^{(0)}
  \bfM_J\widehat\bfQ.
\]
The first term is controlled by the projected bounds above.  For the second,
\begin{align*}
 &|\bfu^\top\widehat\bfQ^{(0)}\bfM_J
 \widehat\bfQ^{(0)}\bfM_J\widehat\bfQ\bfv|\\
 &\quad\le
 C\|\bfM_J\|_{\rm op}
 |\bfu^\top\widehat\bfQ^{(0)}\bfM_J
 \widehat\bfQ^{(0)}\widetilde\bfv|,
\end{align*}
where \(\|\widetilde\bfv\|\le C\|\bfv\|\).  By
\eqref{eq:alt-M12-op} and \(\ell_nr_{\Delta,n}=o(1)\), this is of smaller
order.  Hence
\eqref{eq:alt-det-quadratic}, \eqref{eq:alt-null-quadratic}, and
\eqref{eq:alt-mixed-resolvent-perturbation} follow.

For the inverse-distance averages,
\[
 D\left(\frac{\sqrt p}{\|\bfx\|}\right)[\bfu]
 =-\sqrt p\,\frac{\bfx^\top\bfu}{\|\bfx\|^3}.
\]
Taylor expansion at the centered-error median gives
\[
 \sup_{s,a}|\widehat e_{a,s}-\widehat e_{a,s}^{(0)}|
 =O_P(r_{\Delta,n}p^{-1/2})+O_P(\ell_nm_*^{-1}).
\]
The multiplicative weight identity and the last two projected bounds yield
the stated centering and variance perturbations.

Finally, Lemma~\ref{lem:uniform-linearization} gives
\[
 \widehat\bfDelta_{s_t}^{(0)}
 =N_{s_t}^{-1/2}\bfB_{s_t}+O_P(\ell_nn^{-1}).
\]
After the finite-rank transfer, the leading cross term is
\[
 \mathcal X_t
 =N_{s_t}^{1/2}n^{-1/2}
 \sum_i\beta_i(s_t)\varsigma_i
 \bfDelta_p^\top\bfQ_{\rho,s_t}^0\widetilde\bfY_i.
\]
Conditionally on \(\calF^0\),
\begin{align*}
 \Var(\mathcal X_t\mid\calF^0)
 &\le C\bfDelta_p^\top\bfQ_{\rho,s_t}^0\bfDelta_p,\\
 \E\{(\mathcal X_t-\mathcal X_u)^2\mid\calF^0\}
 &\le C(|t-u|+n^{-1})
 \bfDelta_p^\top\bfQ_{\rho,s_t}^0\bfDelta_p.
\end{align*}
Dyadic chaining and Lemma~\ref{lem:bilinear-de} prove
\eqref{eq:alt-cross-centeredQ}; adding
\eqref{eq:alt-mixed-resolvent-perturbation} gives
\eqref{eq:alt-cross-feasibleQ}.
\end{proof}

\subsection{Proofs of the main results}\label{sec:proofs}

The proofs in this subsection use only the primitive assumptions and the auxiliary lemmas proved above.  All stochastic comparisons are uniform over the trimmed scan set and over the fixed ridge interval unless a fixed point is explicitly specified.
\begin{proof}[Proof of Proposition \ref{prop:raw-score}]
By \eqref{eq:uniform-median-expansion},
\begin{equation}
\label{eq:Delta-exp-proof}
  \widehat\bfDelta_s=N_s^{-1/2}\bfB_s+\bfr_{\Delta,s}^{(0)},
  \qquad
  \sup_s\norm{\bfr_{\Delta,s}^{(0)}}
  =O_P(\ell_{0,n}^4m_*^{-1})=O_P(\ell_nm_*^{-1}).
\end{equation}
Substituting \eqref{eq:Delta-exp-proof} into \eqref{eq:Vraw} yields
\begin{align}
  V^{\rm raw}_{\rho}(s)
  &=N_s\left(N_s^{-1/2}\bfB_s+\bfr_{\Delta,s}^{(0)}\right)^{\top}
  \widehat\bfQ_{\rho,s}
  \left(N_s^{-1/2}\bfB_s+\bfr_{\Delta,s}^{(0)}\right)\notag\\
  &=\bfB_s^{\top}\widehat\bfQ_{\rho,s}\bfB_s
  +2N_s^{1/2}(\bfr_{\Delta,s}^{(0)})^{\top}\widehat\bfQ_{\rho,s}\bfB_s
  +N_s(\bfr_{\Delta,s}^{(0)})^{\top}\widehat\bfQ_{\rho,s}\bfr_{\Delta,s}^{(0)}.
  \label{eq:raw-expansion-proof}
\end{align}
Since $\widehat\bfQ_{\rho,s}$ is a ridge inverse,
\begin{equation*}
  \opnorm{\widehat\bfQ_{\rho,s}}\le \rho_0^{-1}.
\end{equation*}
The interval sizes are trimmed, so $N_s\le m_s/4\le m_*/4\cdot (\sup_s m_s/m_*)$, and in the common-pool single-change and multiple-change scans the ratio $\sup_s m_s/m_*$ is equal to one.  The proof of Lemma~\ref{lem:score-resolvent} gives the intermediate bound
\(\sup_s\|\bfB_s\|=O_P(\ell_{0,n}m_*^{1/2})\).  Together with the
intermediate median-remainder rate in \eqref{eq:Delta-exp-proof},
\begin{align*}
  \sup_{s,\rho}\abs{2N_s^{1/2}(\bfr_{\Delta,s}^{(0)})^{\top}\widehat\bfQ_{\rho,s}\bfB_s}
  &\le 2\rho_0^{-1}\sup_s N_s^{1/2}\sup_s\norm{\bfr_{\Delta,s}^{(0)}}\sup_s\norm{\bfB_s} \\
  &=O_P(\ell_{0,n}^5)=O_P(\ell_n),
\end{align*}
and
\begin{align*}
  \sup_{s,\rho}\abs{N_s(\bfr_{\Delta,s}^{(0)})^{\top}\widehat\bfQ_{\rho,s}\bfr_{\Delta,s}^{(0)}}
  &\le \rho_0^{-1}\sup_sN_s\sup_s\norm{\bfr_{\Delta,s}^{(0)}}^2 \\
  &=O_P(\ell_{0,n}^8m_*^{-1})=O_P(\ell_nm_*^{-1}).
\end{align*}
The leading term satisfies
\begin{equation}
\label{eq:leading-resolvent-proof}
  \bfB_s^{\top}\widehat\bfQ_{\rho,s}\bfB_s
  =\bfB_s^{\top}\bfQ_{\rho,s}^0\bfB_s
  +\bfB_s^{\top}(\widehat\bfQ_{\rho,s}-\bfQ_{\rho,s}^0)\bfB_s.
\end{equation}
The last term in \eqref{eq:leading-resolvent-proof} is $O_P(\ell_n)$ uniformly by Lemma~\ref{lem:score-resolvent}, specifically \eqref{eq:resolvent-quadratic-rate}.  Moreover,
\begin{align}
  \bfB_s^{\top}\bfQ_{\rho,s}^0\bfB_s
  &=\left(\sum_i\beta_i(s)\bfY_i\right)^{\top}
  \bfQ_{\rho,s}^0
  \left(\sum_j\beta_j(s)\bfY_j\right)\notag\\
  &=\sum_{i,j}\beta_i(s)\beta_j(s)\bfY_i^{\top}\bfQ_{\rho,s}^0\bfY_j\notag\\
  &=m_s\sum_{i,j}\beta_i(s)\beta_j(s)A_{ij,\rho,s}^0
  =m_s\bfbeta_s^{\top}\bfA_{\rho,s}^0\bfbeta_s
  =\widetilde V_\rho^0(s).
  \label{eq:companion-ident-proof}
\end{align}
Combining \eqref{eq:raw-expansion-proof}-\eqref{eq:companion-ident-proof},
\[
  \sup_{s,\rho}\abs{V_\rho^{\rm raw}(s)-\widetilde V_\rho^0(s)}
  =O_P(\ell_n)+O_P(\ell_n)+O_P(\ell_nm_*^{-1})=O_P(\ell_n).
\]
Dividing the last display by \(\sqrt{m_s}\) gives
\[
  \sup_{s,\rho}\frac{\abs{V_\rho^{\rm raw}(s)-\widetilde V_\rho^0(s)}}{\sqrt{m_s}}
  =O_P(\ell_nm_*^{-1/2})=o_P(1),
\]
which is the standardized form used in the feasible null law.
\end{proof}

\begin{proof}[Proof of Theorem \ref{thm:pointwise}]
Lemma \ref{lem:clt} gives conditional weak convergence to a standard normal
law.  Since the limit distribution function \(\Phi\) is continuous, Polya's
theorem upgrades this conditional weak convergence to
\begin{equation}
\label{eq:oracle-null-rate}
 \sup_{x\in\mathbb R}
 \left|
 \Pbb\left\{
 \left.
 \frac{\widetilde V_\rho^0(s)-m_s\kappa_\rho^0(s)}
 {\{m_s\sigma_\rho^{0,2}(s)\}^{1/2}}
 \le x
 \,\right|\,\mathcal F_s^0
 \right\}
 -\Phi(x)
 \right|
 \pto0.
\end{equation}
Proposition \ref{prop:raw-score} gives
\begin{equation}
\label{eq:point-proof-raw}
  \frac{V_\rho^{\rm raw}(s)-\widetilde V_\rho^0(s)}{\{m_s\sigma_\rho^{0,2}(s)\}^{1/2}}
  =O_P(\ell_nm_s^{-1/2})=o_P(1).
\end{equation}
The centering rate \eqref{eq:kappa-feasible-rate} gives
\begin{equation*}
  \frac{m_s\{\widehat\kappa_\rho(s)-\kappa_\rho^0(s)\}}{\{m_s\sigma_\rho^{0,2}(s)\}^{1/2}}
  =O_P(\ell_nm_s^{-1/2})=o_P(1).
\end{equation*}
The variance rate \eqref{eq:sigma-feasible-rate} and \eqref{eq:variance-nondegenerate} imply
\begin{equation}
\label{eq:point-proof-var}
  \frac{\widehat\sigma_\rho^2(s)}{\sigma_\rho^{0,2}(s)}
  =1+O_P(\ell_nm_s^{-1/2})=1+o_P(1).
\end{equation}
Combining \eqref{eq:point-proof-raw}-\eqref{eq:point-proof-var},
\begin{align*}
  Z_\rho(s)
  &=\frac{\widetilde V_\rho^0(s)-m_s\kappa_\rho^0(s)}
  {\{m_s\sigma_\rho^{0,2}(s)\}^{1/2}}
  \left\{\frac{\sigma_\rho^{0,2}(s)}{\widehat\sigma_\rho^2(s)}\right\}^{1/2}
  +o_P(1)\\
  &=\frac{\widetilde V_\rho^0(s)-m_s\kappa_\rho^0(s)}
  {\{m_s\sigma_\rho^{0,2}(s)\}^{1/2}}+o_P(1).
\end{align*}
Slutsky's theorem and Lemma \ref{lem:clt} prove
\eqref{eq:pointwise-null}.
\end{proof}

\begin{proof}[Proof of Theorem \ref{thm:sc-process}]
For \(s=(t_1,t_2,t_3)\), let
\[
  \mathfrak p_n(s)
  =\left(\frac{\iota_n(t_1)-1}{n},
          \frac{\iota_n(t_2)-1}{n},
          \frac{\iota_n(t_3)-1}{n}\right).
\]
The definition of \(\iota_n\) gives
\[
  \iota_n\!\left(\frac{\iota_n(t_j)-1}{n}\right)=\iota_n(t_j),
  \qquad j=1,2,3,
\]
and hence
\[
  I_a\{\mathfrak p_n(s)\}=I_a(s),\quad a=1,2,
  \qquad
  \|\mathfrak p_n(s)-s\|_\infty\le n^{-1}.
\]
Therefore all finite-sample quantities indexed by \(s\), including
\(Z_\rho(s)\), are exactly equal to their counterparts indexed by
\(\mathfrak p_n(s)\).  Since every \(s\in\calS_{\rm sc}(\varepsilon)\) is
trimmed, \(\mathfrak p_n(s)\) belongs to the \(\varepsilon/2\)-trimmed
natural grid for all sufficiently large \(n\).

Fix \(s_1,\ldots,s_M\in\calS_{\rm sc}(\varepsilon)\), and define
\[
  Z_\rho^0(s_a)
  =\frac{\widetilde V_\rho^0(s_a)-n\kappa_\rho^0(s_a)}
         {\{n\sigma_\rho^{0,2}(s_a)\}^{1/2}},
  \qquad
  \mathbf Z_{\rho,n}^{0,M}
  =\bigl\{Z_\rho^0(s_1),\ldots,Z_\rho^0(s_M)\bigr\}^{\top}.
\]
Lemmas \ref{lem:multi-clt} and \ref{lem:cross-factorization} yield
\[
  \mathbf Z_{\rho,n}^{0,M}
  \dto N_M(\mathbf0,\mathbf K_M),
  \qquad
  \mathbf K_M=\{K_0(s_a,s_b)\}_{a,b=1}^M.
\]
By Proposition \ref{prop:raw-score} and Lemma
\ref{lem:feasible-companion},
\[
  \left\|
  \bigl\{Z_\rho(s_a)\bigr\}_{a=1}^M-\mathbf Z_{\rho,n}^{0,M}
  \right\|_\infty
  \le \max_{1\le a\le M}|Z_\rho(s_a)-Z_\rho^0(s_a)|
  =o_P(1).
\]
Slutsky's theorem therefore gives the required finite-dimensional limits.

For \(f\in\ell^\infty\{\calS_{\rm sc}(\varepsilon)\}\), put
\[
  \omega_{\rm sc}(f,\delta)
  =\sup_{\substack{s,r\in\calS_{\rm sc}(\varepsilon)\\
                    \|s-r\|_\infty\le\delta}}
    |f(s)-f(r)|.
\]
Lemma \ref{lem:tightness} gives, for every \(\eta>0\),
\[
  \lim_{\delta\downarrow0}\limsup_{n\to\infty}
  \Pbb\{\omega_{\rm sc}(Z_\rho,\delta)>\eta\}=0.
\]
Thus the feasible process is asymptotically tight.  Combining this bound
with the finite-dimensional convergence proves \eqref{eq:sc-process-limit}.
Finally,
\[
  \left|\sup_{s\in\calS_{\rm sc}(\varepsilon)}f(s)
        -\sup_{s\in\calS_{\rm sc}(\varepsilon)}g(s)\right|
  \le \|f-g\|_\infty,
\]
so the supremum functional is Lipschitz on the ambient space.  Applying the
continuous mapping theorem to \eqref{eq:sc-process-limit} proves
\eqref{eq:sc-T-limit}.
\end{proof}

\begin{proof}[Proof of Theorem \ref{thm:mc-process}]
For \(s=(t_1,t_2,t_3)\in\calS_{\rm mc}(\varepsilon)\), use the natural-grid
representative
\[
  \mathfrak p_n(s)
  =\left(\frac{\iota_n(t_1)-1}{n},
          \frac{\iota_n(t_2)-1}{n},
          \frac{\iota_n(t_3)-1}{n}\right).
\]
Exactly as in the single-change case,
\[
  I_a\{\mathfrak p_n(s)\}=I_a(s),\quad a=1,2,
  \qquad
  \|\mathfrak p_n(s)-s\|_\infty\le n^{-1},
\]
so replacing \(s\) by \(\mathfrak p_n(s)\) changes neither the statistic nor
its normalizer and preserves \(\varepsilon/2\)-trimming for all large \(n\).

For fixed \(s_1,\ldots,s_M\in\calS_{\rm mc}(\varepsilon)\), write
\[
  Z_\rho^0(s_a)
  =\frac{\widetilde V_\rho^0(s_a)-n\kappa_\rho^0(s_a)}
         {\{n\sigma_\rho^{0,2}(s_a)\}^{1/2}},
  \qquad
  \mathbf Z_{\rho,n}^{0,M}
  =\bigl\{Z_\rho^0(s_1),\ldots,Z_\rho^0(s_M)\bigr\}^{\top}.
\]
The common full-sample scatter pool and Lemmas \ref{lem:multi-clt} and
\ref{lem:cross-factorization} imply
\[
  \mathbf Z_{\rho,n}^{0,M}
  \dto N_M(\mathbf0,\mathbf K_M),
  \qquad
  \mathbf K_M=\{K_0(s_a,s_b)\}_{a,b=1}^M.
\]
Moreover,
\[
  \sup_{s\in\calS_{\rm mc}(\varepsilon)}
  |Z_\rho(s)-Z_\rho^0(s)|=o_P(1)
\]
by Proposition \ref{prop:raw-score}, Lemma
\ref{lem:feasible-companion}, and Lemma \ref{lem:tightness-mc}.  Hence
\[
  \bigl\{Z_\rho(s_1),\ldots,Z_\rho(s_M)\bigr\}^{\top}
  \dto N_M(\mathbf0,\mathbf K_M).
\]

Define
\[
  \omega_{\rm mc}(f,\delta)
  =\sup_{\substack{s,r\in\calS_{\rm mc}(\varepsilon)\\
                    \|s-r\|_\infty\le\delta}}
    |f(s)-f(r)|.
\]
The eighth-moment increment bound in Lemma \ref{lem:tightness-mc} gives,
for every \(\eta>0\),
\[
  \lim_{\delta\downarrow0}\limsup_{n\to\infty}
  \Pbb\{\omega_{\rm mc}(Z_\rho,\delta)>\eta\}=0.
\]
The preceding finite-dimensional limit and this tightness bound prove
\eqref{eq:mc-process-limit}.  Since
\[
  \left|\sup_{s\in\calS_{\rm mc}(\varepsilon)}f(s)
        -\sup_{s\in\calS_{\rm mc}(\varepsilon)}g(s)\right|
  \le \|f-g\|_\infty,
\]
\eqref{eq:mc-T-limit} follows by the continuous mapping theorem.
\end{proof}

\begin{proof}[Proof of Theorem \ref{thm:sc-joint}]
For \(k=1,\ldots,K\), define the oracle process
\[
  Z_{n,k}^0(s)
  =\frac{\widetilde V_{\rho_n^{(k)}}^0(s)
         -n\kappa_{\rho_n^{(k)}}^0(s)}
        {\{n\sigma_{\rho_n^{(k)}}^{0,2}(s)\}^{1/2}},
  \qquad s\in\calS_{\rm sc}(\varepsilon).
\]
Fix \(s_1,\ldots,s_M\in\calS_{\rm sc}(\varepsilon)\).  Applying Lemma
\ref{lem:multi-clt} to the \(KM\)-vector and then Lemma
\ref{lem:cross-factorization} gives
\[
  \{Z_{n,k}^0(s_a):1\le k\le K,\ 1\le a\le M\}
  \dto
  \{G_k^{\rm sc}(s_a):1\le k\le K,\ 1\le a\le M\},
\]
with
\[
  \Cov\{G_k^{\rm sc}(s_a),G_\ell^{\rm sc}(s_b)\}
  =r_E(\rho^{(k)},\rho^{(\ell)})K_0(s_a,s_b).
\]
The feasible-oracle difference satisfies
\[
  \max_{1\le k\le K}
  \sup_{s\in\calS_{\rm sc}(\varepsilon)}
  |Z_{\rho_n^{(k)}}(s)-Z_{n,k}^0(s)|=o_P(1)
\]
by Proposition \ref{prop:raw-score} and Lemma
\ref{lem:feasible-companion}.  Hence the feasible vector has the same
finite-dimensional limits.

For \(f\in\ell^\infty\{\calS_{\rm sc}(\varepsilon)\}\), define
\[
  \omega_{\rm sc}(f,\delta)
  =\sup_{\substack{s,r\in\calS_{\rm sc}(\varepsilon)\\
                    \|s-r\|_\infty\le\delta}}
    |f(s)-f(r)|.
\]
Then, for every \(\eta>0\),
\[
\begin{split}
  &\Pbb\left\{
    \max_{1\le k\le K}
    \omega_{\rm sc}(Z_{\rho_n^{(k)}},\delta)>\eta\right\}\\
  &\qquad\le
    \sum_{k=1}^K
    \Pbb\left\{
      \omega_{\rm sc}(Z_{\rho_n^{(k)}},\delta)>\eta\right\},
\end{split}
\]
and Lemma \ref{lem:tightness}, uniformly over the compact ridge interval,
implies
\[
  \lim_{\delta\downarrow0}\limsup_{n\to\infty}
  \Pbb\left\{
    \max_{1\le k\le K}
    \omega_{\rm sc}(Z_{\rho_n^{(k)}},\delta)>\eta\right\}=0.
\]
Thus the finite-ridge vector is jointly tight, and the finite-dimensional
limit proves \eqref{eq:sc-joint-process} with covariance
\eqref{eq:sc-joint-cov}.

Set
\[
  S_{n,k}=\sup_{s\in\calS_{\rm sc}(\varepsilon)}
          Z_{\rho_n^{(k)}}(s),
  \qquad
  S_{\infty,k}=\sup_{s\in\calS_{\rm sc}(\varepsilon)}G_k^{\rm sc}(s).
\]
The map from the product process space to \(\R^K\) satisfies
\[
  \max_{1\le k\le K}|\sup_s f_k(s)-\sup_s g_k(s)|
  \le \max_{1\le k\le K}\|f_k-g_k\|_\infty,
\]
so
\[
  (S_{n,1},\ldots,S_{n,K})^\top
  \dto(S_{\infty,1},\ldots,S_{\infty,K})^\top.
\]
Since \(F_{\rm sc}\) is continuous,
\[
  (x_1,\ldots,x_K)\mapsto
  \{1-F_{\rm sc}(x_1),\ldots,1-F_{\rm sc}(x_K)\}
\]
is continuous on \(\R^K\).  A final application of the continuous mapping
theorem proves \eqref{eq:sc-pvector-limit}.
\end{proof}

\begin{proof}[Proof of Theorem \ref{thm:mc-joint}]
For \(k=1,\ldots,K\), put
\[
  Z_{n,k}^0(s)
  =\frac{\widetilde V_{\rho_n^{(k)}}^0(s)
         -n\kappa_{\rho_n^{(k)}}^0(s)}
        {\{n\sigma_{\rho_n^{(k)}}^{0,2}(s)\}^{1/2}},
  \qquad s\in\calS_{\rm mc}(\varepsilon).
\]
For every fixed \(s_1,\ldots,s_M\in\calS_{\rm mc}(\varepsilon)\), Lemmas
\ref{lem:multi-clt} and \ref{lem:cross-factorization} imply
\[
  \{Z_{n,k}^0(s_a):1\le k\le K,\ 1\le a\le M\}
  \dto
  \{G_k^{\rm mc}(s_a):1\le k\le K,\ 1\le a\le M\},
\]
where
\[
  \Cov\{G_k^{\rm mc}(s_a),G_\ell^{\rm mc}(s_b)\}
  =r_E(\rho^{(k)},\rho^{(\ell)})K_0(s_a,s_b).
\]
The common full-sample scatter pool is what makes the temporal factor in this
covariance equal to \(K_0\) for every pair of candidates.

Proposition \ref{prop:raw-score}, Lemma
\ref{lem:feasible-companion}, and Lemma \ref{lem:tightness-mc} give
\[
  \max_{1\le k\le K}
  \sup_{s\in\calS_{\rm mc}(\varepsilon)}
  |Z_{\rho_n^{(k)}}(s)-Z_{n,k}^0(s)|=o_P(1).
\]
For \(f\in\ell^\infty\{\calS_{\rm mc}(\varepsilon)\}\), define
\[
  \omega_{\rm mc}(f,\delta)
  =\sup_{\substack{s,r\in\calS_{\rm mc}(\varepsilon)\\
                    \|s-r\|_\infty\le\delta}}
    |f(s)-f(r)|.
\]
Then, for every \(\eta>0\),
\[
\begin{split}
  &\Pbb\left\{
    \max_{1\le k\le K}
    \omega_{\rm mc}(Z_{\rho_n^{(k)}},\delta)>\eta\right\}\\
  &\qquad\le
    \sum_{k=1}^K
    \Pbb\left\{
      \omega_{\rm mc}(Z_{\rho_n^{(k)}},\delta)>\eta\right\},
\end{split}
\]
so Lemma \ref{lem:tightness-mc}, uniformly over the compact ridge interval,
yields
\[
  \lim_{\delta\downarrow0}\limsup_{n\to\infty}
  \Pbb\left\{
    \max_{1\le k\le K}
    \omega_{\rm mc}(Z_{\rho_n^{(k)}},\delta)>\eta\right\}=0.
\]
The finite-dimensional limits and joint tightness prove
\eqref{eq:mc-joint-process} and \eqref{eq:mc-joint-cov}.

Let
\[
  S_{n,k}=\sup_{s\in\calS_{\rm mc}(\varepsilon)}Z_{\rho_n^{(k)}}(s),
  \qquad
  S_{\infty,k}=\sup_{s\in\calS_{\rm mc}(\varepsilon)}G_k^{\rm mc}(s).
\]
The product supremum map is Lipschitz:
\[
  \max_{1\le k\le K}|\sup_s f_k(s)-\sup_s g_k(s)|
  \le \max_{1\le k\le K}\|f_k-g_k\|_\infty.
\]
Consequently,
\[
  (S_{n,1},\ldots,S_{n,K})^\top
  \dto(S_{\infty,1},\ldots,S_{\infty,K})^\top.
\]
Since \(F_{\rm mc}\) is continuous, applying the continuous map
\[
  (x_1,\ldots,x_K)\mapsto
  \{1-F_{\rm mc}(x_1),\ldots,1-F_{\rm mc}(x_K)\}
\]
proves \eqref{eq:mc-pvector-limit}.
\end{proof}

\begin{proof}[Proof of Theorems \ref{thm:sc-cauchy} and \ref{thm:mc-cauchy}]
Write \({\rm d}\in\{{\rm sc},{\rm mc}\}\), and define
\[
  H(p_1,\ldots,p_K)
  =\sum_{k=1}^K\varpi_k
    \tan\left[\pi\left\{\frac12-p_k\right\}\right],
  \qquad (p_1,\ldots,p_K)\in(0,1)^K.
\]
For each \(k\), continuity of the corresponding Gaussian-supremum
distribution and the probability integral transform give
\[
  P_{\infty,k}^{\rm d}\sim U(0,1),
  \qquad
  \Pbb\{P_{\infty,k}^{\rm d}\in\{0,1\}\}=0.
\]
Hence
\[
  \Pbb\{(P_{\infty,1}^{\rm d},\ldots,P_{\infty,K}^{\rm d})
  \in(0,1)^K\}=1,
\]
and \(H\) is almost surely continuous at the limiting p-value vector.
Theorem \ref{thm:sc-joint} or Theorem \ref{thm:mc-joint}, followed by the
continuous mapping theorem, therefore yields
\[
  T_{\rm CC}^{\rm d}
  =H(P_1^{\rm d},\ldots,P_K^{\rm d})
  \dto
  H(P_{\infty,1}^{\rm d},\ldots,P_{\infty,K}^{\rm d})
  =T_{{\rm CC},\infty}^{\rm d}.
\]
At every continuity point \(c_{\alpha,K}^{\rm d}\) satisfying
\(\Pbb(T_{{\rm CC},\infty}^{\rm d}>c_{\alpha,K}^{\rm d})=\alpha\),
Portmanteau's theorem gives
\[
  \Pbb(T_{\rm CC}^{\rm d}>c_{\alpha,K}^{\rm d})
  \longrightarrow
  \Pbb(T_{{\rm CC},\infty}^{\rm d}>c_{\alpha,K}^{\rm d})
  =\alpha.
\]
Finally, for \(\alpha\in(0,1)\), monotonicity of \(\arctan\) gives the exact
equivalence
\[
\begin{split}
  P_{\rm CC}^{\rm d}\le\alpha
  &\Longleftrightarrow
  \frac12-\frac1\pi\arctan(T_{\rm CC}^{\rm d})\le\alpha\\
  &\Longleftrightarrow
  T_{\rm CC}^{\rm d}\ge\cot(\pi\alpha).
\end{split}
\]
If \(\cot(\pi\alpha)\) is a continuity point of
\(T_{{\rm CC},\infty}^{\rm d}\), then
\[
  \Pbb(P_{\rm CC}^{\rm d}\le\alpha)
  \longrightarrow
  \Pbb\{T_{{\rm CC},\infty}^{\rm d}\ge\cot(\pi\alpha)\}.
\]
Taking \({\rm d}={\rm sc}\) and \({\rm d}={\rm mc}\), respectively, proves
both theorems.
\end{proof}

\begin{corollary}[Finite-grid multiple-change inference]
\label{cor:mc-finite-grid}
Suppose \(H_0\), Assumptions~\ref{ass:elliptical}--\ref{ass:grid}, and
Convention~\ref{conv:scan}(ii) hold.  Recall the fixed finite grid
\[
 \calG_\varepsilon
 =\{j\varepsilon:j\in\mathbb Z\}\cap[0,1],
 \qquad
 \calS_{\rm mc}^*(\varepsilon)
 =\calS_{\rm mc}(\varepsilon)\cap\calG_\varepsilon^3,
\]
and enumerate its distinct candidates by
\(s_1^*,\ldots,s_Q^*\).

For every deterministic sequence
\(\rho_n\to\rho\in[\rho_0,\rho_1]\),
\begin{equation}
\label{eq:mc-star-vector-limit}
 \{Z_{\rho_n}(s_q^*):1\le q\le Q\}
 \dto
 \{G_\rho^{\rm mc}(s_q^*):1\le q\le Q\}.
\end{equation}
Consequently,
\begin{equation}
\label{eq:supp-mc-star-T-limit}
 T_{\rho_n}^{{\rm mc},*}
 :=\max_{s\in\calS_{\rm mc}^*(\varepsilon)}Z_{\rho_n}(s)
 \dto
 T_0^{{\rm mc},*}
 :=\max_{s\in\calS_{\rm mc}^*(\varepsilon)}G_\rho^{\rm mc}(s).
\end{equation}
The distribution of \(T_0^{{\rm mc},*}\) does not depend on \(\rho\);
denote its continuous distribution function by
\begin{equation}
\label{eq:mc-star-cdf}
 F_{\rm mc}^*(x)=\Pbb(T_0^{{\rm mc},*}\le x).
\end{equation}

For the finite regularization-parameter grid
\(\calR_K=\{\rho_n^{(1)},\ldots,\rho_n^{(K)}\}\), jointly,
\begin{equation}
\label{eq:mc-star-joint-vector-limit}
 \{Z_{\rho_n^{(k)}}(s_q^*):
   1\le k\le K,\ 1\le q\le Q\}
 \dto
 \{G_{\rho^{(k)}}^{\rm mc}(s_q^*):
   1\le k\le K,\ 1\le q\le Q\},
\end{equation}
where
\[
 \Cov\{G_{\rho^{(k)}}^{\rm mc}(s_q^*),
       G_{\rho^{(\ell)}}^{\rm mc}(s_{q'}^*)\}
 =r_E(\rho^{(k)},\rho^{(\ell)})K_0(s_q^*,s_{q'}^*).
\]
Define
\begin{align*}
 P_k^{{\rm mc},*}
 &=1-F_{\rm mc}^*(T_{\rho_n^{(k)}}^{{\rm mc},*}),\\
 P_{\infty,k}^{{\rm mc},*}
 &=1-F_{\rm mc}^*\left\{
   \max_{s\in\calS_{\rm mc}^*(\varepsilon)}
   G_{\rho^{(k)}}^{\rm mc}(s)\right\}.
\end{align*}
Then
\begin{equation}
\label{eq:mc-star-pvector-limit}
 (P_1^{{\rm mc},*},\ldots,P_K^{{\rm mc},*})^\top
 \dto
 (P_{\infty,1}^{{\rm mc},*},\ldots,
  P_{\infty,K}^{{\rm mc},*})^\top.
\end{equation}

Finally, for positive weights \(\varpi_k\) summing to one, put
\begin{align*}
 T_{\rm CC}^{{\rm mc},*}
 &=\sum_{k=1}^K\varpi_k
   \tan\left[\pi\left\{\frac12-P_k^{{\rm mc},*}\right\}\right],\\
 P_{\rm CC}^{{\rm mc},*}
 &=\frac12-\frac1\pi\arctan(T_{\rm CC}^{{\rm mc},*}),\\
 T_{{\rm CC},\infty}^{{\rm mc},*}
 &=\sum_{k=1}^K\varpi_k
   \tan\left[\pi\left\{\frac12-P_{\infty,k}^{{\rm mc},*}\right\}\right].
\end{align*}
If \(c_{\alpha,K}^{{\rm mc},*}\) is a continuity point satisfying
\(\Pbb(T_{{\rm CC},\infty}^{{\rm mc},*}>
c_{\alpha,K}^{{\rm mc},*})=\alpha\), then
\begin{equation}
\label{eq:mc-star-cauchy-calibration}
 \Pbb(T_{\rm CC}^{{\rm mc},*}>
 c_{\alpha,K}^{{\rm mc},*})\longrightarrow\alpha.
\end{equation}
Moreover, whenever \(\cot(\pi\alpha)\) is a continuity point of
\(T_{{\rm CC},\infty}^{{\rm mc},*}\),
\begin{equation}
\label{eq:mc-star-analytic-calibration}
 \Pbb(P_{\rm CC}^{{\rm mc},*}\le\alpha)
 \longrightarrow
 \Pbb\{T_{{\rm CC},\infty}^{{\rm mc},*}
       \ge\cot(\pi\alpha)\}.
\end{equation}
\end{corollary}

\begin{proof}
Because \(Q\) and \(K\) are fixed, Lemmas~\ref{lem:multi-clt} and
\ref{lem:cross-factorization}, Proposition~\ref{prop:raw-score}, and
Lemma~\ref{lem:feasible-companion} give
\eqref{eq:mc-star-vector-limit} and
\eqref{eq:mc-star-joint-vector-limit} directly; no process-tightness
argument is needed.  The finite-maximum map is Lipschitz, so the continuous
mapping theorem gives \eqref{eq:supp-mc-star-T-limit}.  Every marginal Gaussian
vector in \eqref{eq:mc-star-vector-limit} has covariance
\(\{K_0(s_q^*,s_{q'}^*)\}_{q,q'=1}^Q\), which is independent of \(\rho\).
Hence \(F_{\rm mc}^*\) is common to all regularization-parameter values.
Furthermore, each \(G_\rho^{\rm mc}(s_q^*)\) is a nondegenerate standard
normal variable, and therefore, for every \(x\in\mathbb R\),
\[
 \Pbb(T_0^{{\rm mc},*}=x)
 \le\sum_{q=1}^Q\Pbb\{G_\rho^{\rm mc}(s_q^*)=x\}=0.
\]
Thus \(F_{\rm mc}^*\) is continuous.  Applying its componentwise
probability transform to the joint maximum vector proves
\eqref{eq:mc-star-pvector-limit}.

Each limiting p-value is consequently uniform on \((0,1)\), so the Cauchy
map is almost surely continuous at the limiting vector.  The continuous
mapping theorem and the same threshold argument used in the proof of
Theorems~\ref{thm:sc-cauchy} and \ref{thm:mc-cauchy} give
\eqref{eq:mc-star-cauchy-calibration} and
\eqref{eq:mc-star-analytic-calibration}.
\end{proof}

\begin{proof}[Proof of Theorem \ref{thm:power}]
Write
\[
 \mathfrak q_{\rho,n}(\bfd)=\bfd^\top\bfD_{\rho,n}\bfd,
 \qquad
 \bfDelta_p=n^{-1/4}\bfd_p
\]
for the local alternative.  Lemmas~\ref{lem:alternative-expansion} and
\ref{lem:alternative-perturbation} give, uniformly over the trimmed
single-change scan and \(\rho\in[\rho_0,\rho_1]\),
\begin{align*}
 \widehat\bfDelta_s
 &=\widehat\bfDelta_s^{(0)}
   +\vartheta_{s,n}(\tau_\star)\bfDelta_p
   +\bfr_{\Delta,s}^{\rm alt},\\
 \sup_s\|\bfr_{\Delta,s}^{\rm alt}\|
 &=O_P(\ell_{0,n}^4n^{-1})+O_P(n^{-3/4})+O_P(n^{-1})
   =o_P(n^{-1/2}),\\
 \sup_{s,\rho}\frac{N_s}{\sqrt n}
 \left| (\widehat\bfDelta_s^{(0)})^\top
 (\widehat\bfQ_{\rho,s}-\widehat\bfQ_{\rho,s}^{(0)})
 \widehat\bfDelta_s^{(0)}\right|
 &=o_P(1),\\
 \sup_{s,\rho}
 \left|\bfDelta_p^\top
 (\widehat\bfQ_{\rho,s}-\widehat\bfQ_{\rho,s}^{(0)})
 \bfDelta_p\right|
 &=o_P(n^{-1/2}),\\
 \bfDelta_p^\top\widehat\bfQ_{\rho,s}^{(0)}\bfDelta_p
 &=n^{-1/2}\{\mathfrak q_{\rho,n}(\bfd_p)+o_P(1)\}.
\end{align*}
Substitution into \eqref{eq:Vraw} yields the exact decomposition
\begin{align*}
 V_\rho^{\rm raw}(s)
 &=N_s(\widehat\bfDelta_s^{(0)})^\top
   \widehat\bfQ_{\rho,s}^{(0)}\widehat\bfDelta_s^{(0)}\\
 &\quad+2N_s\vartheta_{s,n}(\tau_\star)
   \bfDelta_p^\top\widehat\bfQ_{\rho,s}^{(0)}
   \widehat\bfDelta_s^{(0)}\\
 &\quad+N_s\vartheta_{s,n}(\tau_\star)^2
   \bfDelta_p^\top\widehat\bfQ_{\rho,s}^{(0)}\bfDelta_p
   +R_{V,n}(s).
\end{align*}
The uniform score bounds
\[
 \widehat\bfDelta_s^{(0)}
 =N_s^{-1/2}\bfB_s+O_P(\ell_nn^{-1}),
 \qquad
 \sup_s\|\bfB_s\|=O_P(\ell_nn^{1/2}),
\]
together with \eqref{eq:alt-mixed-resolvent-perturbation}, imply
\begin{align*}
 \sup_s\frac{|R_{V,n}(s)|}{\sqrt n}
 &\le C\sup_s\Bigg[
   \sqrt n\|\bfr_{\Delta,s}^{\rm alt}\|^2
   +\|\bfr_{\Delta,s}^{\rm alt}\|\|\bfB_s\|\\
 &\qquad
   +\frac{N_s^{1/2}}{\sqrt n}
    \left|\bfDelta_p^\top
    (\widehat\bfQ_{\rho,s}-\widehat\bfQ_{\rho,s}^{(0)})
    \bfB_s\right|\\
 &\qquad
   +\sqrt n\left|\bfDelta_p^\top
    (\widehat\bfQ_{\rho,s}-\widehat\bfQ_{\rho,s}^{(0)})
    \bfDelta_p\right|\Bigg]+o_P(1)\\
 &=O_P(\ell_{0,n}^8n^{-3/2}+n^{-1})
   +O_P(\ell_{0,n}^5n^{-1/2}+\ell_{0,n}n^{-1/4})\\
 &\quad+O_P(\ell_nn^{-1/2}+n^{-3/4}+\ell_nn^{-3/4})+o_P(1)
 =o_P(1).
\end{align*}

Define the centered-error statistic
\[
 Z_\rho^{(0)}(s)
 =\frac{N_s(\widehat\bfDelta_s^{(0)})^\top
       \widehat\bfQ_{\rho,s}^{(0)}\widehat\bfDelta_s^{(0)}
       -m_s\widehat\kappa_\rho^{(0)}(s)}
      {\{m_s\widehat\sigma_\rho^{(0),2}(s)\}^{1/2}}.
\]
The feasible centering and variance perturbations satisfy
\begin{align*}
 \sup_s\sqrt{m_s}\,
 |\widehat\kappa_\rho(s)-\widehat\kappa_\rho^{(0)}(s)|&=o_P(1),\\
 \sup_s\left|
 \frac{\widehat\sigma_\rho^2(s)}
      {\widehat\sigma_\rho^{(0),2}(s)}-1\right|&=o_P(1),
\end{align*}
and Lemma~\ref{lem:alt-cross-term} gives
\[
 \sup_s\frac{N_s}{\sqrt n}
 \left|\bfDelta_p^\top\widehat\bfQ_{\rho,s}
 \widehat\bfDelta_s^{(0)}\right|
 =O_P(\ell_nn^{-1/4})=o_P(1).
\]
For \(s_t=(0,t,1)\), put
\[
 \Lambda_{\rho,n}
 =\frac{\mathfrak q_{\rho,n}(\bfd_p)}{\sigma_{\rho,n}^{\circ}},
 \qquad
 \mathfrak h_n(t;\tau_\star)
 =\frac{N_{s_t}}{n}\vartheta_{s_t,n}(\tau_\star)^2.
\]
The deterministic-equivalent and floor corrections give
\[
 \Lambda_{\rho,n}\longrightarrow\Lambda_\rho(G_\Delta),
 \qquad
 \sup_{t\in[\varepsilon,1-\varepsilon]}
 |\mathfrak h_n(t;\tau_\star)-\mathfrak h(t;\tau_\star)|
 =O(n^{-1}).
\]
Combining the preceding bounds yields
\begin{equation}
\label{eq:power-process-slutsky}
 \sup_{t\in[\varepsilon,1-\varepsilon]}
 \left|Z_\rho(s_t)-Z_\rho^{(0)}(s_t)
 -\Lambda_{\rho,n}\mathfrak h_n(t;\tau_\star)\right|=o_P(1).
\end{equation}
By Theorem~\ref{thm:sc-process},
\[
 Z_\rho^{(0)}(s_t)\Rightarrow G_\rho^{\rm sc}(s_t)
 \quad\text{in }\ell^\infty[\varepsilon,1-\varepsilon].
\]
Hence \eqref{eq:power-process-slutsky} and functional Slutsky's theorem imply
\[
 Z_\rho(s_t)\Rightarrow
 G_\rho^{\rm sc}(s_t)
 +\Lambda_\rho(G_\Delta)\mathfrak h(t;\tau_\star).
\]
The piecewise-constant interpolation preserves
\eqref{eq:power-process-slutsky} on each natural-grid cell.  Applying the
continuous mapping theorem to \(f\mapsto\sup_t f(t)\) proves the local-power
formula.

For the strong-signal assertion, let
\[
 k_\star=\lfloor n\tau_\star\rfloor,
 \qquad
 s_\star=(0,\tau_\star,1).
\]
Then \(\iota_n(\tau_\star)=k_\star+1\), so the two windows split the sample
at the population boundary and
\[
 \vartheta_{s_\star,n}(\tau_\star)=1,
 \qquad N_{s_\star}\asymp n,
 \qquad m_{s_\star}=n.
\]
Lemmas~\ref{lem:alternative-perturbation} and \ref{lem:ridge-stability} give
\begin{align*}
 \bfDelta_p^\top\widehat\bfQ_{\rho,s_\star}\bfDelta_p
 &=\mathfrak q_{\rho,n}(\bfDelta_p)
   +o_P(\|\bfDelta_p\|^2),\\
 c\|\bfDelta_p\|^2
 &\le\mathfrak q_{\rho,n}(\bfDelta_p)
 \le C\|\bfDelta_p\|^2.
\end{align*}
Therefore
\[
 \frac{N_{s_\star}\bfDelta_p^\top
       \widehat\bfQ_{\rho,s_\star}\bfDelta_p}
      {\sqrt n\,\widehat\sigma_\rho(s_\star)}
 =\{1+o_P(1)\}\frac{N_{s_\star}}{n}
  \frac{\sqrt n\,\mathfrak q_{\rho,n}(\bfDelta_p)}{\sigma_\rho}
 \pto\infty.
\]
The centered quadratic term is \(O_P(1)\), whereas the mixed term satisfies
\begin{align*}
 \frac{N_{s_\star}}
 {\sqrt n}
 \left|\bfDelta_p^\top\widehat\bfQ_{\rho,s_\star}
 \widehat\bfDelta_{s_\star}^{(0)}\right|
 &=O_P\{\ell_n\mathfrak q_{\rho,n}(\bfDelta_p)^{1/2}\}+o_P(1),\\
 \frac{\ell_n\mathfrak q_{\rho,n}(\bfDelta_p)^{1/2}}
      {\sqrt n\,\mathfrak q_{\rho,n}(\bfDelta_p)}
 &\le\frac{C\ell_n}{\sqrt n\|\bfDelta_p\|}
 =\frac{C\ell_n\|\bfDelta_p\|}
       {\sqrt n\|\bfDelta_p\|^2}=o(1).
\end{align*}
Thus, uniformly over \(\rho\in[\rho_0,\rho_1]\),
\[
 T_\rho^{\rm sc}\pto\infty.
\]
Since the ridge grid is fixed,
\[
 \min_{1\le k\le K}T_{\rho_n^{(k)}}^{\rm sc}\pto\infty,
 \qquad
 \max_{1\le k\le K}P_k^{\rm sc}\pto0.
\]
The positivity of all Cauchy weights then gives
\[
 T_{\rm CC}^{\rm sc}
 =\sum_{k=1}^K\varpi_k\cot(\pi P_k^{\rm sc})\pto\infty,
 \qquad
 P_{\rm CC}^{\rm sc}
 =\frac12-\pi^{-1}\arctan(T_{\rm CC}^{\rm sc})\pto0,
\]
which proves consistency.
\end{proof}

\begin{proof}[Proof of Theorem \ref{thm:localization}]
For \(s_t=(0,t,1)\), the common scatter pool is \(\{1,\ldots,n\}\).
Choose
\[
 t_n^\dagger\in\argmin_{t\in\calG_n^{\rm sc}}|t-\tau_\star|.
\]
Then
\[
 |t_n^\dagger-\tau_\star|\le n^{-1},
 \qquad
 \left|\mathfrak h_n(t_n^\dagger;\tau_\star)
       -\mathfrak h(\tau_\star;\tau_\star)\right|=O(n^{-1}).
\]
Since
\[
 \mathfrak s_{n,\rho}
 \le C\sqrt n\|\bfDelta_p\|^2=o(\sqrt n),
\]
the standardized grid error is
\[
 O(\mathfrak s_{n,\rho}/n)=o(1).
\]

Put
\[
 \widehat{\mathfrak q}_{\rho,n}^{\rm alt}(\bfd)
 =\bfd^\top\widehat\bfQ_{\rho,s_{1/2}}\bfd,
 \qquad
 \mathfrak r_{\rho,n}^{\rm alt}
 =\frac{\widehat{\mathfrak q}_{\rho,n}^{\rm alt}(\bfDelta_p)}
       {\mathfrak q_{\rho,n}(\bfDelta_p)},
 \qquad
 \mathfrak r_{\sigma,\rho,n}
 =\frac{\sigma_\rho}{\sigma_{\rho,n}^{\circ}}.
\]
Because the full-pool inverse is common to all \(t\),
Lemma~\ref{lem:alternative-perturbation},
\eqref{eq:oracle-sigma-proxy-rate}, and \eqref{eq:sigma-feasible-rate} imply
\begin{align*}
 \mathfrak r_{\rho,n}^{\rm alt}&=1+o_P(1),
 &\mathfrak r_{\sigma,\rho,n}&\longrightarrow1,\\
 \sup_{t\in[\varepsilon,1-\varepsilon]}
 \left|\frac{\sigma_{\rho,n}^{\circ}}
 {\widehat\sigma_\rho(s_t)}-1\right|
 &=O_P(\mathfrak e_{\sigma,n}).
\end{align*}
Define
\[
 \mathfrak h_n(t;\tau_\star)
 =\frac{N_{s_t}}{n}\vartheta_{s_t,n}(\tau_\star)^2,
 \qquad
 D_{n,\rho}(t)
 =\frac{N_{s_t}\vartheta_{s_t,n}(\tau_\star)^2
       \bfDelta_p^\top\widehat\bfQ_{\rho,s_t}\bfDelta_p}
      {\sqrt n\,\widehat\sigma_\rho(s_t)}.
\]
Then
\[
 \sup_t|\mathfrak h_n(t;\tau_\star)-\mathfrak h(t;\tau_\star)|
 =O(n^{-1})
\]
and
\begin{equation}
\label{eq:localization-deterministic-profile}
 \sup_t\left|D_{n,\rho}(t)
 -\mathfrak s_{n,\rho}
  \mathfrak r_{\sigma,\rho,n}
  \mathfrak r_{\rho,n}^{\rm alt}
  \mathfrak h(t;\tau_\star)\right|
 =O_P(\mathfrak s_{n,\rho}\mathfrak e_{\sigma,n})+o_P(1).
\end{equation}

The stochastic part admits
\[
 Z_\rho(0,t,1)-D_{n,\rho}(t)
 =Z_\rho^{(0)}(0,t,1)+C_\rho(t)+R_{n,\rho}(t),
\]
where
\[
 C_\rho(t)
 =\frac{2N_{s_t}\vartheta_{s_t,n}(\tau_\star)
 \bfDelta_p^\top\widehat\bfQ_{\rho,s_t}
 \widehat\bfDelta_{s_t}^{(0)}}
 {\sqrt n\,\sigma_\rho}.
\]
Lemmas~\ref{lem:alternative-perturbation} and
\ref{lem:alt-cross-term} give
\begin{align*}
 \sup_t|R_{n,\rho}(t)|&=o_P(1),\\
 \sup_t|C_\rho(t)|
 &=O_P\{\ell_n\mathfrak q_{\rho,n}(\bfDelta_p)^{1/2}\}+o_P(1)\\
 &\le O_P(\ell_n\|\bfDelta_p\|)+o_P(1)=o_P(1).
\end{align*}
Theorem~\ref{thm:sc-process} therefore yields
\begin{equation}
\label{eq:localization-stochastic-envelope}
 \sup_t|Z_\rho(0,t,1)-D_{n,\rho}(t)|
 \le\sup_t|Z_\rho^{(0)}(0,t,1)|+o_P(1)
 =O_P(1).
\end{equation}

\begin{equation}
\label{eq:signal-profile-explicit}
 \mathfrak h(t;\tau_\star)
 =\begin{cases}
 t(1-\tau_\star)^2/(1-t),&t<\tau_\star,\\[1mm]
 (1-t)\tau_\star^2/t,&t\ge\tau_\star.
 \end{cases}
\end{equation}
Consequently,
\[
 \mathfrak h(\tau_\star;\tau_\star)-\mathfrak h(t;\tau_\star)
 =\begin{cases}
 (1-\tau_\star)(\tau_\star-t)/(1-t),&t<\tau_\star,\\[1mm]
 \tau_\star(t-\tau_\star)/t,&t>\tau_\star,
 \end{cases}
\]
so
\begin{equation}
\label{eq:signal-profile-separation}
 \mathfrak h(\tau_\star;\tau_\star)-\mathfrak h(t;\tau_\star)
 \ge\varepsilon|t-\tau_\star|,
 \qquad t\in[\varepsilon,1-\varepsilon].
\end{equation}
Using \eqref{eq:localization-deterministic-profile},
\eqref{eq:signal-profile-separation}, and
\(\mathfrak r_{\sigma,\rho,n}\mathfrak r_{\rho,n}^{\rm alt}=1+o_P(1)\),
we obtain, uniformly over \(t\in\calG_n^{\rm sc}\),
\begin{equation}
\label{eq:localization-profile-gap}
 D_{n,\rho}(t_n^\dagger)-D_{n,\rho}(t)
 \ge\frac{\varepsilon}{2}\mathfrak s_{n,\rho}|t-\tau_\star|
 -O_P(\mathfrak s_{n,\rho}\mathfrak e_{\sigma,n})-o_P(1).
\end{equation}

For \(M>0\), define
\[
 \mathcal A_{M,n}
 =\left\{
 \mathfrak s_{n,\rho}|\widehat\tau_\rho-\tau_\star|
 \ge M(1+\mathfrak s_{n,\rho}\mathfrak e_{\sigma,n})
 \right\}.
\]
The argmax property and \eqref{eq:localization-profile-gap} imply on
\(\mathcal A_{M,n}\) that
\begin{align*}
 0
 &\le Z_\rho(0,\widehat\tau_\rho,1)
      -Z_\rho(0,t_n^\dagger,1)\\
 &\le-\frac{\varepsilon M}{2}
      (1+\mathfrak s_{n,\rho}\mathfrak e_{\sigma,n})
   +O_P(\mathfrak s_{n,\rho}\mathfrak e_{\sigma,n})\\
 &\quad+2\sup_t|Z_\rho(0,t,1)-D_{n,\rho}(t)|+o_P(1).
\end{align*}
By \eqref{eq:localization-stochastic-envelope},
\[
 \lim_{M\to\infty}\limsup_{n\to\infty}
 \Pbb(\mathcal A_{M,n})=0.
\]
Hence
\[
 \mathfrak s_{n,\rho}|\widehat\tau_\rho-\tau_\star|
 =O_P(1+\mathfrak s_{n,\rho}\mathfrak e_{\sigma,n}),
\]
and division by \(\mathfrak s_{n,\rho}\) proves
\eqref{eq:localization-consistency}.  If
\(\mathfrak s_{n,\rho}\mathfrak e_{\sigma,n}=O(1)\), then
\[
 |\widehat\tau_\rho-\tau_\star|=O_P(\mathfrak s_{n,\rho}^{-1}).
\]
\end{proof}

\subsection{Auxiliary WBS results and proof of Theorem~\ref{thm:wbs-consistency}}\label{app:wbs-proof}

We first fix the auxiliary constants used only in the proof.  The strict
inequalities in Assumption~\ref{ass:wbs-design} allow us to choose
\(\varepsilon_0\in(\varepsilon,1/2)\) sufficiently close to \(\varepsilon\)
so that
\[
  \delta_{\min}<(1-2\varepsilon_0)\delta_R.
\]
We may then choose
\[
  \delta_{\min}<\delta_W<
  \min\left\{(1-2\varepsilon_0)\delta_R,
             \frac{\delta_0}{2}-\delta_R\right\}
\]
and finally choose \(\delta_B>0\) such that
\[
  \delta_R+\delta_W+\delta_B<\frac{\delta_0}{2}.
\]
For the union bounds below, also put
\begin{equation}
\label{eq:wbs-proof-stochastic-envelope}
  \mathfrak z_n^{\rm WBS}
  =\{\log(C_0Kn^5)\}^{1/2},
\end{equation}
where \(C_0>1\) is fixed.  Since \(K\) is fixed,
\(\mathfrak z_n^{\rm WBS}\asymp\sqrt{\log n}\), so
Assumption~\ref{ass:wbs-signal-tuning} implies
\[
  \mathfrak z_n^{\rm WBS}
  +\overline{\mathfrak s}_n^{\rm WBS}\mathfrak e_n^{\rm WBS}
  =o(\mathfrak t_n^{\rm WBS}).
\]

We next introduce the interval-specific objects used only in the proof.  For
an integer interval \(I=[L,R]\), let \(|I|=R-L+1\).  Relative to the break
\(k_j\), the balance and buffer conditions are
\begin{equation}
\label{eq:wbs-good-balance}
  \min\{k_j-L+1,R-k_j\}\ge\varepsilon_0|I|,
\end{equation}
and
\begin{equation}
\label{eq:wbs-good-buffer}
  L-k_{j-1}\ge\delta_Bn,
  \qquad
  k_{j+1}-R\ge\delta_Bn.
\end{equation}
Let \(\mathfrak G_j\) be the deterministic class of intervals for which
\(k_j\in I\), \(m_{\min}\le |I|\le\delta_Wn\), and both
\eqref{eq:wbs-good-balance} and \eqref{eq:wbs-good-buffer} hold.  Let
\(\mathfrak B_j\) be the class obtained by replacing this length restriction
with
\begin{equation}
\label{eq:wbs-refinement-class-length}
  m_{\min}\le |I|\le3\delta_Rn
\end{equation}
while keeping all its other conditions.  Under
Assumption~\ref{ass:wbs-design}, these classes are nonempty for all
sufficiently large \(n\), and every interval in either class contains only
\(k_j\).

For \((a,b,c)\in\mathcal A(I,\varepsilon)\), define
\[
  n_1=b-a+1,
  \qquad
  n_2=c-b,
  \qquad
  N_{a,b,c}=\frac{n_1n_2}{n_1+n_2},
\]
and
\[
  \bfDelta_I(a,b,c)
  =\frac1{n_2}\sum_{i=b+1}^{c}\bftheta_i
   -\frac1{n_1}\sum_{i=a}^{b}\bftheta_i .
\]
The interval population noncentral component is
\begin{equation}
\label{eq:wbs-population-component}
  D_{\rho,I}(a,b,c)
  =\frac{N_{a,b,c}}{|I|^{1/2}\sigma_{\rho,|I|}^{\circ}}
  \bfDelta_I(a,b,c)^\top\bfD_{\rho,|I|}\bfDelta_I(a,b,c),
  \qquad
  D_I(a,b,c)=\max_{1\le k\le K}
  D_{\rho_n^{(k)},I}(a,b,c).
\end{equation}
Here \(\bfD_{\rho,|I|}\) and
\(\sigma_{\rho,|I|}^{\circ}\) are the deterministic quantities in
\eqref{eq:generic-D-rho-m} and \eqref{eq:finite-sigma-proxy}, evaluated at
the local aspect ratio \(p/|I|\).

For the balanced refinement class, define the interval-specific ridge signal
by
\begin{equation}
\label{eq:wbs-proof-signal}
  \mathfrak s_{n,j}^{\rm WBS}
  =\inf_{I\in\mathfrak B_j}\max_{1\le k\le K}
  \frac{\sqrt{|I|}\,\bfDelta_j^\top
        \bfD_{\rho_n^{(k)},|I|}\bfDelta_j}
       {\sigma_{\rho_n^{(k)},|I|}^{\circ}},
  \qquad
  \mathfrak s_{n,\min}^{\rm WBS}
  =\min_{1\le j\le q}\mathfrak s_{n,j}^{\rm WBS},
  \qquad
  \mathfrak s_{n,\max}^{\rm WBS}
  =\max_{1\le j\le q}\mathfrak s_{n,j}^{\rm WBS}.
\end{equation}
Lemma~\ref{lem:ridge-stability}, \eqref{eq:finite-sigma-proxy}, and the
uniform nondegeneracy of the companion-resolvent variance imply that fixed
constants \(0<c_{\rm WBS}<C_{\rm WBS}<\infty\) exist such that
\begin{equation}
\label{eq:wbs-signal-equivalence}
  c_{\rm WBS}\sqrt n\,\|\bfDelta_j\|^2
  \le\mathfrak s_{n,j}^{\rm WBS}
  \le C_{\rm WBS}\sqrt n\,\|\bfDelta_j\|^2,
  \qquad j=1,\ldots,q.
\end{equation}
Consequently,
\begin{equation}
\label{eq:wbs-extreme-signal-equivalence}
  \mathfrak s_{n,\min}^{\rm WBS}
  \asymp\underline{\mathfrak s}_n^{\rm WBS},
  \qquad
  \mathfrak s_{n,\max}^{\rm WBS}
  \asymp\overline{\mathfrak s}_n^{\rm WBS}.
\end{equation}
Because \(q\) is fixed, differences between averages of the
piecewise-constant locations over any two intervals are bounded by \(qr_n\).
It follows from the same deterministic bounds that the largest population
component over all admissible WBS pools and triples is at most
\(C\mathfrak s_{n,\max}^{\rm WBS}\).

\begin{lemma}[Isolation by uniformly sampled WBS intervals]\label{lem:wbs-random-isolation}
Under Assumption~\ref{ass:wbs-design},
\begin{equation}
\label{eq:wbs-isolation-event}
  \Pbb\left(\bigcap_{j=1}^q\{\mathcal I_{M_n}\cap\mathfrak G_j\ne\varnothing\}\right)\to1 .
\end{equation}
\end{lemma}

\begin{proof}[Proof of Lemma \ref{lem:wbs-random-isolation}]
Because
\[
 \delta_{\min}<\delta_W,
 \qquad
 \varepsilon_0<\frac12,
\]
we may choose \(d_{\rm iso}\) and \(b_{\rm iso}\) such that
\begin{equation}
\label{eq:wbs-isolation-band-constants}
 2(d_{\rm iso}-b_{\rm iso})>\delta_{\min},
 \qquad
 2(d_{\rm iso}+b_{\rm iso})<\delta_W,
 \qquad
 \frac{d_{\rm iso}-b_{\rm iso}}
      {2(d_{\rm iso}+b_{\rm iso})}>\varepsilon_0.
\end{equation}
Indeed, take
\[
 d_{\rm iso}\in(\delta_{\min}/2,\delta_W/2)
\]
and then choose \(b_{\rm iso}>0\) sufficiently small.

For a fixed \(k_j\), define
\begin{align*}
 \mathcal W_{j,n}^{\rm L}
 &=\left\{L:
 k_j-(d_{\rm iso}+b_{\rm iso})n
 \le L\le
 k_j-(d_{\rm iso}-b_{\rm iso})n\right\},\\
 \mathcal W_{j,n}^{\rm R}
 &=\left\{R:
 k_j+(d_{\rm iso}-b_{\rm iso})n
 \le R\le
 k_j+(d_{\rm iso}+b_{\rm iso})n\right\}.
\end{align*}
For every \((L,R)\in
\mathcal W_{j,n}^{\rm L}\times\mathcal W_{j,n}^{\rm R}\), integer rounding
changes the following inequalities by only \(O(1)\):
\begin{align*}
 2(d_{\rm iso}-b_{\rm iso})n-O(1)
 &\le R-L+1
 \le2(d_{\rm iso}+b_{\rm iso})n+O(1),\\
 k_j-L+1&\ge(d_{\rm iso}-b_{\rm iso})n-O(1),\\
 R-k_j&\ge(d_{\rm iso}-b_{\rm iso})n-O(1).
\end{align*}
By \eqref{eq:wbs-isolation-band-constants}, for all sufficiently large
\(n\),
\[
 m_{\min}\le R-L+1\le\delta_Wn,
 \qquad
 \min\{k_j-L+1,R-k_j\}\ge\varepsilon_0(R-L+1).
\]
Moreover,
\[
 d_{\rm iso}+b_{\rm iso}<\delta_W/2,
 \qquad
 \min_{0\le l\le q}(k_{l+1}-k_l)\ge\delta_0n-O(1).
\]
Assumption~\ref{ass:wbs-design} therefore gives
\[
 L-k_{j-1}\ge\delta_Bn,
 \qquad
 k_{j+1}-R\ge\delta_Bn,
\]
so
\[
 \mathcal W_{j,n}^{\rm L}\times\mathcal W_{j,n}^{\rm R}
 \subseteq\mathfrak G_j.
\]

The two bands satisfy
\[
 |\mathcal W_{j,n}^{\rm L}|=2b_{\rm iso}n+O(1),
 \qquad
 |\mathcal W_{j,n}^{\rm R}|=2b_{\rm iso}n+O(1).
\]
Since each WBS interval is generated by a uniformly sampled unordered pair
of distinct endpoints,
\[
 \Pbb\{[\ell_v^{\rm WBS},u_v^{\rm WBS}]\in\mathfrak G_j\}
 \ge
 \frac{|\mathcal W_{j,n}^{\rm L}|\,
       |\mathcal W_{j,n}^{\rm R}|}{\binom n2}
 \ge p_0
\]
for some constant \(p_0>0\), uniformly in \(j\) and all sufficiently large
\(n\).  Independence over \(v\) yields
\[
 \Pbb\left(
 \bigcup_{j=1}^q
 \{\mathcal I_{M_n}\cap\mathfrak G_j=\varnothing\}
 \right)
 \le q(1-p_0)^{M_n}
 \le q\exp(-p_0M_n)\longrightarrow0.
\]
This is \eqref{eq:wbs-isolation-event}.
\end{proof}

\begin{lemma}[Population separation on a balanced isolated WBS interval]\label{lem:wbs-population-separation}
Under Assumptions~\ref{ass:elliptical}--\ref{ass:grid},
\ref{ass:wbs-design}, and \ref{ass:wbs-signal-tuning}, there are constants
\(c,C>0\) such that the following statements hold for all large \(n\).  If
\(I\in\mathfrak G_j\), then the triple \((L,k_j,R)\) belongs to
\(\mathcal A(I,\varepsilon)\) and
\begin{equation}
\label{eq:wbs-good-signal-lemma}
  D_I(L,k_j,R)\ge c \mathfrak s_{n,j}^{\rm WBS}.
\end{equation}
Moreover, if \(I\in\mathfrak B_j\) and \((a,b,c)\in\mathcal A(I,\varepsilon)\) satisfies
\begin{equation*}
  D_I(a,b,c)\ge \sup_{(u,v,w)\in\mathcal A(I,\varepsilon)}D_I(u,v,w)-u_n
\end{equation*}
for some \(u_n=o(\mathfrak s_{n,\min}^{\rm WBS})\), then
\begin{equation}
\label{eq:wbs-pop-location}
  |b-k_j|/n\le C u_n/\mathfrak s_{n,\min}^{\rm WBS}.
\end{equation}
\end{lemma}

\begin{proof}[Proof of Lemma \ref{lem:wbs-population-separation}]
Let \(I=[L,R]\in\mathfrak G_j\) and set
\[
 n_{\rm L}=k_j-L+1,
 \qquad
 n_{\rm R}=R-k_j,
 \qquad
 m=n_{\rm L}+n_{\rm R}=|I|.
\]
By \eqref{eq:wbs-good-balance},
\[
 n_{\rm L}\wedge n_{\rm R}\ge\varepsilon_0m,
 \qquad
 (L,k_j,R)\in\mathcal A(I,\varepsilon),
\]
and the isolated population contrast at \((L,k_j,R)\) equals
\(\bfDelta_j\).  Hence
\begin{align*}
 N_{L,k_j,R}
 &=\frac{n_{\rm L}n_{\rm R}}{m}
 \ge\varepsilon_0^2m,\\
 D_I(L,k_j,R)
 &\ge\varepsilon_0^2
 \max_{1\le k\le K}
 \frac{\sqrt m\,\bfDelta_j^\top
       \bfD_{\rho_n^{(k)},m}\bfDelta_j}
      {\sigma_{\rho_n^{(k)},m}^{\circ}}
 \ge c\mathfrak s_{n,j}^{\rm WBS}.
\end{align*}
This proves \eqref{eq:wbs-good-signal-lemma}.

Now let \(I\in\mathfrak B_j\).  The buffer conditions imply that \(k_j\)
is the only change in \(I\).  For
\((a,b,c)\in\mathcal A(I,\varepsilon)\), write
\[
 \bfDelta_I(a,b,c)
 =d_I(a,b,c;k_j)\bfDelta_j,
 \qquad
 \mathcal P_I^{\rm pop}(a,b,c)
 =N_{a,b,c}d_I(a,b,c;k_j)^2.
\]
At \(b=k_j\), put
\[
 x=k_j-a+1,
 \qquad
 y=c-k_j.
\]
Then
\[
 d_I(a,k_j,c;k_j)=1,
 \qquad
 \mathcal P_I^{\rm pop}(a,k_j,c)=\frac{xy}{x+y},
\]
and
\[
 \partial_x\frac{xy}{x+y}=\frac{y^2}{(x+y)^2}>0,
 \qquad
 \partial_y\frac{xy}{x+y}=\frac{x^2}{(x+y)^2}>0.
\]
Therefore
\[
 \sup_{a,c}\mathcal P_I^{\rm pop}(a,k_j,c)
 =\mathcal P_I^{\rm pop}(L,k_j,R)
 =:\mathcal P_{I,0}^{\rm pop}
 =\frac{n_{\rm L}n_{\rm R}}{n_{\rm L}+n_{\rm R}}.
\]

If \(b=k_j-h<k_j\), let
\[
 x=c-k_j,
 \qquad
 u=b-a+1.
\]
The left window is entirely pre-change and
\[
 \mathcal P_I^{\rm pop}(a,b,c)
 =\frac{u x^2}{(h+x)(u+h+x)}.
\]
Its partial derivatives satisfy
\begin{align*}
 \partial_u\frac{u x^2}{(h+x)(u+h+x)}
 &=\frac{x^2}{(u+h+x)^2}>0,\\
 \partial_x\log\frac{u x^2}{(h+x)(u+h+x)}
 &=\frac2x-\frac1{h+x}-\frac1{u+h+x}>0.
\end{align*}
Thus the maximum over \(a,c\) is attained at \((L,R)\), and
\begin{align*}
 \mathcal P_{I,-}^{\rm pop}(h)
 &:=\sup_{a,c}\mathcal P_I^{\rm pop}(a,k_j-h,c)\\
 &=\frac{(n_{\rm L}-h)n_{\rm R}^2}
 {(n_{\rm R}+h)(n_{\rm L}+n_{\rm R})},\\
 \mathcal P_{I,0}^{\rm pop}-\mathcal P_{I,-}^{\rm pop}(h)
 &=\frac{h n_{\rm R}}{n_{\rm R}+h}
 \ge\varepsilon_0h.
\end{align*}
If \(b=k_j+h>k_j\), the symmetric calculation gives
\begin{align*}
 \mathcal P_{I,+}^{\rm pop}(h)
 &:=\sup_{a,c}\mathcal P_I^{\rm pop}(a,k_j+h,c)\\
 &=\frac{(n_{\rm R}-h)n_{\rm L}^2}
 {(n_{\rm L}+h)(n_{\rm L}+n_{\rm R})},\\
 \mathcal P_{I,0}^{\rm pop}-\mathcal P_{I,+}^{\rm pop}(h)
 &=\frac{h n_{\rm L}}{n_{\rm L}+h}
 \ge\varepsilon_0h.
\end{align*}
Consequently,
\begin{equation}
\label{eq:wbs-single-pop-separation-proof}
 \sup_{(u,v,w)\in\mathcal A(I,\varepsilon)}
 N_{u,v,w}d_I(u,v,w;k_j)^2
 -N_{a,b,c}d_I(a,b,c;k_j)^2
 \ge\varepsilon_0|b-k_j|.
\end{equation}

Define the interval-specific signal factor
\[
 \mathfrak s_{j,I}^{\rm WBS}
 =\max_{1\le k\le K}
 \frac{\sqrt m\,\bfDelta_j^\top
       \bfD_{\rho_n^{(k)},m}\bfDelta_j}
      {\sigma_{\rho_n^{(k)},m}^{\circ}}.
\]
Because the ridge-direction factor is independent of \((a,b,c)\),
\[
 D_I(a,b,c)
 =m^{-1}\mathfrak s_{j,I}^{\rm WBS}
  \mathcal P_I^{\rm pop}(a,b,c).
\]
Moreover,
\[
 \mathfrak s_{j,I}^{\rm WBS}\ge\mathfrak s_{n,j}^{\rm WBS}
 \ge\mathfrak s_{n,\min}^{\rm WBS},
 \qquad
 m\le3\delta_Rn.
\]
Combining these relations with
\eqref{eq:wbs-single-pop-separation-proof} gives
\[
 \sup_{(u,v,w)\in\mathcal A(I,\varepsilon)}D_I(u,v,w)
 -D_I(a,b,c)
 \ge c\mathfrak s_{n,\min}^{\rm WBS}
       \frac{|b-k_j|}{n}.
\]
Hence
\[
 \sup_{(u,v,w)}D_I(u,v,w)-D_I(a,b,c)\le u_n
 \quad\Longrightarrow\quad
 \frac{|b-k_j|}{n}
 \le C\frac{u_n}{\mathfrak s_{n,\min}^{\rm WBS}},
\]
which is \eqref{eq:wbs-pop-location}.
\end{proof}

\begin{lemma}[Refinement-window geometry]\label{lem:wbs-refinement-geometry}
Suppose Assumption~\ref{ass:wbs-design} holds.  Let a recursive
segment \([\ell,r]\) contain \(k_j\) and satisfy, for all large \(n\),
\begin{equation}
\label{eq:wbs-current-margin}
 \min\{k_j-\ell+1,r-k_j\}
 \ge (\delta_R+\delta_W+\delta_B)n.
\end{equation}
If an interval \(I^*\subset[\ell,r]\) contains exactly \(k_j\),
\(|I^*|\le\delta_Wn\), and \(\widetilde k\in I^*\), then the refinement
interval \(B(\widetilde k;\ell,r)\) is not clipped, belongs to
\(\mathfrak B_j\), and contains no change other than \(k_j\).
\end{lemma}

\begin{proof}
Since \(k_j,\widetilde k\in I^*\) and \(|I^*|\le\delta_Wn\),
\[
 |\widetilde k-k_j|\le\delta_Wn.
\]
Using \eqref{eq:wbs-current-margin} and
\(g_n=\lfloor\delta_Rn\rfloor\),
\begin{align*}
 \widetilde k-g_n-\ell
 &\ge k_j-\ell-(\delta_R+\delta_W)n-O(1)
 \ge\delta_Bn-O(1)>0,\\
 r-(\widetilde k+g_n)
 &\ge r-k_j-(\delta_R+\delta_W)n-O(1)
 \ge\delta_Bn-O(1)>0.
\end{align*}
Thus the refinement interval is not clipped and
\[
 B=[\widetilde k-g_n,\widetilde k+g_n],
 \qquad
 |B|=2g_n+1,
 \qquad
 m_{\min}\le|B|\le3\delta_Rn
\]
for all sufficiently large \(n\).  Its two distances from \(k_j\) satisfy
\begin{align*}
 k_j-(\widetilde k-g_n)+1
 &\ge(\delta_R-\delta_W)n-O(1),\\
 \widetilde k+g_n-k_j
 &\ge(\delta_R-\delta_W)n-O(1).
\end{align*}
Since
\[
 \delta_W<(1-2\varepsilon_0)\delta_R,
\]
we have
\[
 (\delta_R-\delta_W)n-O(1)
 \ge\varepsilon_0(2g_n+1)=\varepsilon_0|B|
\]
for all large \(n\).  Furthermore, every endpoint of \(B\) is at distance
at most
\[
 (\delta_R+\delta_W)n+O(1)
\]
from \(k_j\).  The spacing and buffer inequalities give
\[
 \min_{0\le l\le q}(k_{l+1}-k_l)
 \ge\delta_0n-O(1),
 \qquad
 \delta_R+\delta_W+\delta_B<\delta_0/2,
\]
and hence
\[
 \inf B-k_{j-1}\ge\delta_Bn,
 \qquad
 k_{j+1}-\sup B\ge\delta_Bn.
\]
Therefore \(B\) contains no change other than \(k_j\) and satisfies all
defining inequalities of \(\mathfrak B_j\), so \(B\in\mathfrak B_j\).
\end{proof}

\begin{lemma}[Uniform local envelope and fixed-
\(q\) multiple-jump expansion]\label{lem:wbs-envelope}
Let
\begin{equation*}
 \mathcal J_n=\{[u,v]:1\le u\le v\le n,\ v-u+1\ge m_{\min}\}.
\end{equation*}
For \(I\in\mathcal J_n\), let
\(Z_{\rho,I}^{(0)}(a,b,c)\) denote the fully feasible statistic obtained by
applying exactly the construction in Section~\ref{sec:calibration} to the
centered errors \(\{\bfvarepsilon_i:i\in I\}\), with pool \(I\).  For
\((a,b,c)\in\mathcal A(I,\varepsilon)\), define its centered-error median
contrast by
\[
 \widehat\bfDelta_I^{(0)}(a,b,c)
 =\widehat\bftheta^{(0)}([b+1,c])
  -\widehat\bftheta^{(0)}([a,b]).
\]
When the candidate arguments are suppressed, write
\(\widehat\bfR_I\), \(\widehat\bfQ_{\rho,I}\), and
\(\widehat\sigma_{\rho,I}\) for the local feasible SSCM, ridge inverse,
and standard deviation constructed from the observed sample with common
scatter pool \(I\).  Their centered-error counterparts are denoted by
\(\widehat\bfR_I^{(0)}\), \(\widehat\bfQ_{\rho,I}^{(0)}\), and
\(\widehat\sigma_{\rho,I}^{(0)}\).  Also write
\(V_{\rho,I}^{\rm raw,(0)}\), \(\widehat\kappa_{\rho,I}^{(0)}\), and
\(\widehat\sigma_{\rho,I}^{(0),2}\) for, respectively, the raw statistic,
centering term, and variance obtained from this centered-error construction;
\(\widehat\sigma_{\rho,I}^{(0)}\) is the nonnegative square root of the last
quantity.  Use
\(\widetilde V_{\rho,I}^{0}\), \(\kappa_{\rho,I}^{0}\),
and \(\bfQ_{\rho,I}^{0}\) for the corresponding local-pool oracle quantities
obtained by applying the oracle definitions at the beginning of this appendix
with common scatter pool \(I\) and adjacent windows \([a,b]\) and
\([b+1,c]\).  To make the candidate dependence explicit, put
\begin{align*}
 e_{1,I}(a,b,c)&=n_1^{-1}\sum_{i=a}^{b}w_i,
 &e_{2,I}(a,b,c)&=n_2^{-1}\sum_{i=b+1}^{c}w_i,\\
 \beta_{i,I}(a,b,c)
 &=\sqrt{N_{a,b,c}}
 \left\{
 \frac{\ind(b+1\le i\le c)}{n_2e_{2,I}(a,b,c)}
 -\frac{\ind(a\le i\le b)}{n_1e_{1,I}(a,b,c)}
 \right\},& & i\in I,
\end{align*}
where \(\beta_{i,I}(a,b,c)=0\) for indices in the scatter pool outside the two
candidate windows.  Define
\begin{equation*}
 \bfB_{I,a,b,c}=\sum_{i\in I}\beta_{i,I}(a,b,c)\bfY_i,
 \qquad
 \widetilde A_{ij,\rho,I}
 =\varsigma_i\varsigma_j A_{ij,\rho,I}^{0},
\end{equation*}
where \(\bfY_I\) has columns \(\bfY_i\), \(i\in I\), and
\(A_{ij,\rho,I}^{0}\) are the entries of the local oracle companion matrix
\begin{equation*}
 \bfA_{\rho,I}^{0}
 =|I|^{-1}\bfY_I^\top\bfQ_{\rho,I}^{0}\bfY_I.
\end{equation*}
Equivalently,
\begin{equation*}
 \widetilde V_{\rho,I}^{0}(a,b,c)
 =|I|\sum_{i,j\in I}\beta_{i,I}(a,b,c)\beta_{j,I}(a,b,c)
        A_{ij,\rho,I}^{0},
 \qquad
 \kappa_{\rho,I}^{0}(a,b,c)
 =\sum_{i\in I}\beta_{i,I}(a,b,c)^2A_{ii,\rho,I}^{0}.
\end{equation*}
Thus \(\widetilde A_{ij,\rho,I}\) is measurable with respect to
\(\calF_I^0\).  For the structural Rademacher signs restricted to the pool,
write \(\bfvarsigma_I=(\varsigma_i:i\in I)^\top\).  Also set
\begin{equation*}
 \bar Z_{\rho,I}^{0}(a,b,c)
 =\frac{\widetilde V_{\rho,I}^{0}(a,b,c)
       -|I|\kappa_{\rho,I}^{0}(a,b,c)}
       {|I|^{1/2}\sigma_{\rho,|I|}^{\circ}},
\end{equation*}
where the numerator is the centered-error oracle Rademacher quadratic form.
Under \(H_0\), Assumptions~\ref{ass:elliptical}--\ref{ass:grid}, and
Assumption~\ref{ass:wbs-design},
\begin{align}
 \max_{I\in\mathcal J_n}\max_{1\le k\le K}
 \max_{(a,b,c)\in\mathcal A(I,\varepsilon)}
 |\bar Z_{\rho_n^{(k)},I}^{0}(a,b,c)|
 &=O_P(\mathfrak z_n^{\rm WBS}),
 \label{eq:wbs-oracle-envelope}\\
 \max_{I\in\mathcal J_n}\max_{1\le k\le K}
 \max_{(a,b,c)\in\mathcal A(I,\varepsilon)}
 |Z_{\rho_n^{(k)},I}^{(0)}(a,b,c)
  -\bar Z_{\rho_n^{(k)},I}^{0}(a,b,c)|
 &=o_P(1).
 \label{eq:wbs-local-feasible-transfer}
\end{align}
Consequently,
\begin{equation}
\label{eq:wbs-proof-null-envelope}
 \max_{I\in\mathcal J_n}\max_{1\le k\le K}
 \max_{(a,b,c)\in\mathcal A(I,\varepsilon)}
 |Z_{\rho_n^{(k)},I}^{(0)}(a,b,c)|=O_P(\mathfrak z_n^{\rm WBS}).
\end{equation}

Under the fixed-\(q\) multiple-change model and Assumptions
\ref{ass:elliptical}--\ref{ass:grid} and
\ref{ass:wbs-design}--\ref{ass:wbs-signal-tuning}, uniformly over the same
local collection,
\begin{equation}
\label{eq:wbs-uniform-expansion}
 Z_{\rho,I}(a,b,c)
 =Z_{\rho,I}^{(0)}(a,b,c)+D_{\rho,I}(a,b,c)
  +R_{\rho,I}(a,b,c),
\end{equation}
where every maximum over \(\rho\) is taken over the finite grid \(\calR_K\), and
\begin{equation}
\label{eq:wbs-uniform-remainder}
 \max_{I,\rho,a,b,c}|R_{\rho,I}(a,b,c)|
 =O_P\{\mathfrak s_{n,\max}^{\rm WBS}\mathfrak e_n^{\rm WBS}\}
  +O_P(\mathfrak z_n^{\rm WBS}r_n)+o_P(1)
 =O_P\{\mathfrak s_{n,\max}^{\rm WBS}\mathfrak e_n^{\rm WBS}\}+o_P(1).
\end{equation}
\end{lemma}

\begin{proof}
Let
\[
 \mathcal N_n
 =K\sum_{I\in\mathcal J_n}|\mathcal A(I,\varepsilon)|
 \le Kn^5.
\]
Every candidate window contains at least
\(\varepsilon m_{\min}\asymp n\) observations.  The uniform bounds for
inverse distances, spatial medians, companion matrices, diagonal cancellation,
and off-diagonal terms therefore apply to this polynomial collection:
\begin{align}
 &\max_{I,\rho,a,b,c}
 \frac{|V_{\rho,I}^{\rm raw,(0)}-
              \widetilde V_{\rho,I}^{0}|}{|I|^{1/2}}
 =O_P(\ell_nm_{\min}^{-1/2}),\notag\\
 &\max_{I,\rho,a,b,c}|I|^{1/2}
 |\widehat\kappa_{\rho,I}^{(0)}-\kappa_{\rho,I}^{0}|
 =O_P(\ell_nm_{\min}^{-1/2}),\notag\\
 &\max_{I,\rho,a,b,c}
 |\widehat\sigma_{\rho,I}^{(0),2}
       -\sigma_{\rho,|I|}^{\circ2}|
 =O_P(\ell_nm_{\min}^{-1/2}).
 \label{eq:wbs-local-feasible-rates}
\end{align}
The deterministic variance proxies satisfy
\[
 0<c_\sigma
 \le\inf_{I,\rho}\sigma_{\rho,|I|}^{\circ2}
 \le\sup_{I,\rho}\sigma_{\rho,|I|}^{\circ2}
 \le C_\sigma<\infty.
\]
Since \(m_{\min}\asymp n\) and \(\ell_n=o(n^{1/2})\),
\eqref{eq:wbs-local-feasible-rates} implies
\[
 \max_{I,\rho,a,b,c}
 |Z_{\rho,I}^{(0)}(a,b,c)-\bar Z_{\rho,I}^{0}(a,b,c)|
 =o_P(1),
\]
which is \eqref{eq:wbs-local-feasible-transfer}.

Conditionally on \(\calF^0\),
\[
 \bar Z_{\rho,I}^{0}(a,b,c)
 =\bfvarsigma_I^\top
  \bfC_{\rho,I,a,b,c}\bfvarsigma_I,
\]
where the diagonal of \(\bfC_{\rho,I,a,b,c}\) is zero and
\[
 C_{ij,\rho,I,a,b,c}
 =\frac{|I|^{1/2}\beta_{i,I}(a,b,c)
 \beta_{j,I}(a,b,c)\widetilde A_{ij,\rho,I}}
 {\sigma_{\rho,|I|}^{\circ}},
 \qquad i\ne j.
\]
The trimming bounds and companion contraction yield
\[
 \tr(\bfC_{\rho,I,a,b,c}^2)\le C,
 \qquad
 \|\bfC_{\rho,I,a,b,c}\|_{\rm op}
 \le Cm_{\min}^{-1/2}.
\]
Hence the conditional Hanson--Wright inequality gives
\[
 \Pbb_{\varsigma}
 \{ |\bar Z_{\rho,I}^{0}(a,b,c)|>x\}
 \le2\exp\{-c\min(x^2,xm_{\min}^{1/2})\}.
\]
With \(x=A\mathfrak z_n^{\rm WBS}\), where
\(\mathfrak z_n^{\rm WBS}=o(m_{\min}^{1/2})\) by
\eqref{eq:wbs-proof-stochastic-envelope} and \(m_{\min}\asymp n\), and with
\(A\) sufficiently large,
\[
 \Pbb_{\varsigma}\left
 \{\max_{I,\rho,a,b,c}
 |\bar Z_{\rho,I}^{0}(a,b,c)|
 >A\mathfrak z_n^{\rm WBS}\}
 \middle|\calF^0\right)
 \le2Kn^5e^{-cA^2(\mathfrak z_n^{\rm WBS})^2}
 \longrightarrow0.
\]
This proves \eqref{eq:wbs-oracle-envelope}; together with
\eqref{eq:wbs-local-feasible-transfer}, it gives
\eqref{eq:wbs-proof-null-envelope}.

Under the fixed-\(q\) alternative, write
\[
 \bftheta_i=\bftheta^{(1)}+
 \sum_{j=1}^q\ind(i>k_j)\bfDelta_j,
 \qquad
 \bar\vartheta_{j,A}^{\rm ch}
 =|A|^{-1}\sum_{i\in A}\ind(i>k_j),
 \qquad
 \bar\bftheta_A=|A|^{-1}\sum_{i\in A}\bftheta_i.
\]
For every scanned segment \(A\),
\begin{equation}
\label{eq:wbs-offset-decomposition}
 \bfa_{i,A}
 =\bftheta_i-\bar\bftheta_A
 =\sum_{j=1}^q
 \{\ind(i>k_j)-\bar\vartheta_{j,A}^{\rm ch}\}\bfDelta_j,
 \qquad
 |A|^{-1}\sum_{i\in A}\bfa_{i,A}=\mathbf0.
\end{equation}
Consequently,
\[
 \max_{i\in A}\|\bfa_{i,A}\|
 \le\sum_{j=1}^q\|\bfDelta_j\|
 \le qr_n.
\]
Every first- and second-order Taylor component generated by
\eqref{eq:wbs-offset-decomposition} has the form
\[
 \sum_{j=1}^q\mathcal T_{1,j},
 \qquad\text{or}\qquad
 \sum_{j,l=1}^q\mathcal T_{2,jl},
\]
and hence contains at most \(q\) or \(q^2\) terms.  Since \(q\) is fixed,
the bounds in Lemma~\ref{lem:alternative-expansion} apply uniformly with
\(r_n\) in place of \(r_{\Delta,n}\).  Thus
\begin{equation}
\label{eq:wbs-multijump-median-expansion}
 \widehat\bftheta(A)
 =\widehat\bftheta^{(0)}(A)+\bar\bftheta_A+\bfr_A,
 \qquad
 \begin{aligned}
 r_{*,n}:=\max_A\|\bfr_A\|
 &=O_P(\ell_nm_{\min}^{-1})
   +O_P(r_nm_{\min}^{-1/2})\\
 &\quad+O(r_n^2p^{-1/2})
   +O_P(\ell_np^{-1}r_n).
 \end{aligned}
\end{equation}
For a candidate triple,
\[
 \widehat\bfDelta_I(a,b,c)
 =\widehat\bfDelta_I^{(0)}(a,b,c)
  +\bfDelta_I(a,b,c)+\bfr_{\Delta,I}(a,b,c),
 \qquad
 \max_{I,a,b,c}\|\bfr_{\Delta,I}(a,b,c)\|
 \le2r_{*,n}.
\]
The same expansion gives
\begin{align*}
 \max_I\|\widehat\bfR_I-
             \widehat\bfR_I^{(0)}\|_{\rm op}
 &=O_P(\ell_nr_n+r_n^2+\ell_nn^{-1})=o_P(1),\\
 \max_{I,\rho}
 \left|\frac{\widehat\sigma_{\rho,I}}
 {\widehat\sigma_{\rho,I}^{(0)}}-1\right|
 &=O_P(\mathfrak e_n^{\rm WBS}).
\end{align*}

Fix \((I,a,b,c,\rho)\) and abbreviate
\[
 \bfDelta_I=\bfDelta_I(a,b,c),
 \qquad
 \|\bfDelta_I\|\le qr_n.
\]
By Lemmas~\ref{lem:bilinear-de} and
\ref{lem:deterministic-centered-transfer},
\[
 \max_{I,\rho,a,b,c}
 \left|\bfDelta_I^\top
 \{\widehat\bfQ_{\rho,I}^{(0)}-\bfD_{\rho,|I|}\}
 \bfDelta_I\right|
 =O_P(\ell_nn^{-1/2}r_n^2).
\]
Since
\[
 \frac{N_{a,b,c}}{|I|^{1/2}}\le Cn^{1/2},
 \qquad
 \mathfrak s_{n,\max}^{\rm WBS}\asymp n^{1/2}r_n^2,
\]
we obtain
\begin{equation}
\label{eq:wbs-multijump-deterministic-transfer}
 \max_{I,\rho,a,b,c}
 \frac{N_{a,b,c}}{|I|^{1/2}}
 \left|\bfDelta_I^\top
 \{\widehat\bfQ_{\rho,I}^{(0)}-\bfD_{\rho,|I|}\}
 \bfDelta_I\right|
 =O_P(\mathfrak s_{n,\max}^{\rm WBS}\ell_nn^{-1/2})=o_P(1).
\end{equation}
All deterministic local components are bounded by
\[
 \max_{I,\rho,a,b,c}|D_{\rho,I}(a,b,c)|
 \le C\mathfrak s_{n,\max}^{\rm WBS}.
\]
Therefore the scatter, centering, and studentization perturbations attached
to deterministic components are
\[
 O_P(\mathfrak s_{n,\max}^{\rm WBS}
       \mathfrak e_n^{\rm WBS}).
\]

For the stochastic cross term, the centered-error contrast satisfies
\[
 \widehat\bfDelta_I^{(0)}(a,b,c)
 =N_{a,b,c}^{-1/2}\bfB_{I,a,b,c}
  +\bfr_I^{(0)},
 \qquad
 \max_{I,a,b,c}\|\bfr_I^{(0)}\|
 =O_P(\ell_nn^{-1}).
\]
Conditionally on \(\calF^0\),
\(\bfQ_{\rho,I}^{0}\) is measurable and
\(\bfB_{I,a,b,c}\) is a degree-one Rademacher sum.  Hence
\begin{align*}
 \Var_{\varsigma}
 \{\bfDelta_I^\top\bfQ_{\rho,I}^{0}\bfB_{I,a,b,c}\}
 &\le C\|\bfDelta_I\|^2
 \le Cr_n^2,\\
 \Pbb_{\varsigma}
 \left(
 |\bfDelta_I^\top\bfQ_{\rho,I}^{0}\bfB_{I,a,b,c}|>x r_n
 \middle|\calF^0\right)
 &\le2e^{-cx^2}.
\end{align*}
A union bound over at most \(Kn^5\) candidates gives
\begin{equation}
\label{eq:wbs-oracle-cross-bound}
 \max_{I,\rho,a,b,c}
 |\bfDelta_I^\top\bfQ_{\rho,I}^{0}\bfB_{I,a,b,c}|
 =O_P(\mathfrak z_n^{\rm WBS}r_n).
\end{equation}
The centered and alternative mixed-resolvent transfers give, respectively,
\begin{align*}
 \max_{I,\rho,a,b,c}
 |\bfDelta_I^\top
 (\widehat\bfQ_{\rho,I}^{(0)}-\bfQ_{\rho,I}^{0})
 \bfB_{I,a,b,c}|
 &=O_P(\ell_nn^{-1/2}r_n)=o_P(1),\\
 \max_{I,\rho,a,b,c}
 |\bfDelta_I^\top
 (\widehat\bfQ_{\rho,I}-\widehat\bfQ_{\rho,I}^{(0)})
 \bfB_{I,a,b,c}|
 &=O_P(\ell_nr_n^2+r_n^3+\ell_nr_nn^{-1/2})=o_P(1).
\end{align*}
Therefore the fully feasible cross component is
\[
 O_P(\mathfrak z_n^{\rm WBS}r_n)+o_P(1).
\]

The terms containing at least one median remainder satisfy
\begin{align*}
 &\max_{I,\rho,a,b,c}
 |\mathcal R_{\rho,I}^{\rm med}(a,b,c)|\\
 &\quad\le C\{
 \ell_{0,n}n^{1/2}r_{*,n}
 +n^{1/2}r_nr_{*,n}
 +n^{1/2}r_{*,n}^2\}=o_P(1),
\end{align*}
by \eqref{eq:wbs-multijump-median-expansion},
\(p\asymp n\), \(r_n\to0\), and \(\ell_nr_n\to0\).
The studentization of the centered-error component contributes
\[
 O_P(\mathfrak z_n^{\rm WBS}\mathfrak e_n^{\rm WBS}).
\]
If \(\mathfrak s_{n,\max}^{\rm WBS}\ge\mathfrak z_n^{\rm WBS}\), then
\[
 \mathfrak z_n^{\rm WBS}\mathfrak e_n^{\rm WBS}
 \le\mathfrak s_{n,\max}^{\rm WBS}\mathfrak e_n^{\rm WBS}.
\]
If \(\mathfrak s_{n,\max}^{\rm WBS}<\mathfrak z_n^{\rm WBS}\), then
\eqref{eq:wbs-signal-equivalence} gives
\[
 r_n^2\le C\mathfrak z_n^{\rm WBS}n^{-1/2},
\]
and hence
\[
 \mathfrak z_n^{\rm WBS}\mathfrak e_n^{\rm WBS}
 \le C\left\{
 \ell_n(\mathfrak z_n^{\rm WBS})^{3/2}n^{-1/4}
 +(\mathfrak z_n^{\rm WBS})^2n^{-1/2}
 +\ell_n\mathfrak z_n^{\rm WBS}n^{-1/2}
 \right\}=o(1).
\]
Combining these bounds with
\eqref{eq:wbs-multijump-deterministic-transfer} and
\eqref{eq:wbs-oracle-cross-bound} yields
\[
 Z_{\rho,I}(a,b,c)
 =Z_{\rho,I}^{(0)}(a,b,c)+D_{\rho,I}(a,b,c)
  +R_{\rho,I}(a,b,c),
\]
with
\[
 \max_{I,\rho,a,b,c}|R_{\rho,I}(a,b,c)|
 =O_P(\mathfrak s_{n,\max}^{\rm WBS}
       \mathfrak e_n^{\rm WBS})
  +O_P(\mathfrak z_n^{\rm WBS}r_n)+o_P(1).
\]
Finally,
\[
 \mathfrak z_n^{\rm WBS}=O(\ell_n),
 \qquad
 \ell_nr_n\to0
 \quad\Longrightarrow\quad
 \mathfrak z_n^{\rm WBS}r_n=o(1),
\]
which proves \eqref{eq:wbs-uniform-expansion} and
\eqref{eq:wbs-uniform-remainder}.
\end{proof}

\begin{proof}[Proof of Theorem \ref{thm:wbs-consistency}]
Assumption~\ref{ass:wbs-signal-tuning}, together with
\eqref{eq:wbs-extreme-signal-equivalence}, permits a deterministic sequence
\(c_n\uparrow\infty\) such that
\begin{equation}
\label{eq:wbs-slow-envelope}
 c_n\{\mathfrak z_n^{\rm WBS}
 +\mathfrak s_{n,\max}^{\rm WBS}\mathfrak e_n^{\rm WBS}+1\}
 =o(\mathfrak t_n^{\rm WBS}).
\end{equation}
Set
\[
 \mathfrak u_n^{\rm WBS}
 =c_n\{\mathfrak z_n^{\rm WBS}
 +\mathfrak s_{n,\max}^{\rm WBS}\mathfrak e_n^{\rm WBS}+1\}.
\]
Let \(\mathcal E_n^{\rm WBS}\) be the intersection of
\begin{enumerate}[label=(\roman*)]
\item the isolation event in \eqref{eq:wbs-isolation-event};
\item the events on which every uniform \(O_P\)-bound in
Lemma~\ref{lem:wbs-envelope} is bounded by \(c_n\) times its deterministic
rate;
\item the events on which every uniform \(o_P(1)\)-term in that lemma has
absolute value at most one.
\end{enumerate}
Then
\[
 \Pbb(\mathcal E_n^{\rm WBS})\longrightarrow1.
\]
All subsequent inequalities are established on
\(\mathcal E_n^{\rm WBS}\).

First,
\begin{align*}
 \mathfrak s_{n,\min}^{\rm WBS}
 &\le\mathfrak s_{n,\max}^{\rm WBS}
 \le C\sqrt n\,r_n^2=o(\sqrt n),\\
 \frac{\mathfrak t_n^{\rm WBS}}
      {\mathfrak z_n^{\rm WBS}}&\longrightarrow\infty,
 \qquad
 \mathfrak z_n^{\rm WBS}\longrightarrow\infty.
\end{align*}
Hence
\[
 \frac{n\mathfrak t_n^{\rm WBS}}
      {\mathfrak s_{n,\min}^{\rm WBS}}\longrightarrow\infty.
\]
By Assumption~\ref{ass:wbs-signal-tuning} and
\eqref{eq:wbs-extreme-signal-equivalence},
\[
 \frac{n\mathfrak t_n^{\rm WBS}}
      {\mathfrak s_{n,\min}^{\rm WBS}}=o(h_n),
\]
so
\[
 h_n\longrightarrow\infty,
 \qquad
 h_n/n\longrightarrow0.
\]

If a pool \(I\) contains no true change, location equivariance gives the
exact identities
\[
 \widehat\bfDelta_I=\widehat\bfDelta_I^{(0)},
 \qquad
 \widehat\bfR_I=\widehat\bfR_I^{(0)},
 \qquad
 Z_{\rho,I}=Z_{\rho,I}^{(0)}.
\]
Therefore, by \eqref{eq:wbs-proof-null-envelope},
\eqref{eq:wbs-slow-envelope}, and
\(\mathfrak u_n^{\rm WBS}=o(\mathfrak t_n^{\rm WBS})\),
\[
 \max_{\rho,a,b,c}|Z_{\rho,I}(a,b,c)|
 \le c_n\mathfrak z_n^{\rm WBS}
 <\mathfrak t_n^{\rm WBS}
\]
for all large \(n\).  Thus a homogeneous recursive segment stops and cannot
produce a false estimate.

For each \(j\), choose
\[
 I_j^\circ=[L_j^\circ,R_j^\circ]
 \in\mathcal I_{M_n}\cap\mathfrak G_j.
\]
Theorem event \(\mathcal E_n^{\rm WBS}\),
Lemma~\ref{lem:wbs-population-separation}, and
\eqref{eq:wbs-uniform-expansion} give
\begin{align*}
 \mathcal T(I_j^\circ)
 &\ge D_{I_j^\circ}(L_j^\circ,k_j,R_j^\circ)
  -\max_{\rho,a,b,c}|Z_{\rho,I_j^\circ}^{(0)}(a,b,c)|
  -\max_{\rho,a,b,c}|R_{\rho,I_j^\circ}(a,b,c)|\\
 &\ge c\mathfrak s_{n,j}^{\rm WBS}
  -C\mathfrak u_n^{\rm WBS}.
\end{align*}
Since
\[
 \mathfrak u_n^{\rm WBS}=o(\mathfrak t_n^{\rm WBS}),
 \qquad
 \mathfrak t_n^{\rm WBS}=o(\mathfrak s_{n,\min}^{\rm WBS}),
\]
we have
\[
 \mathcal T(I_j^\circ)>\mathfrak t_n^{\rm WBS},
 \qquad j=1,\ldots,q,
\]
for all sufficiently large \(n\).

We now prove the recursion by induction.  For an active segment
\(S=[\ell,r]\), let its unselected true changes form the consecutive block
\[
 \mathcal K(S)=\{k_u,\ldots,k_v\}.
\]
The induction invariant is
\begin{align}
 I_j^\circ&\subseteq S,
 &&j=u,\ldots,v,
 \label{eq:wbs-invariant-good-intervals}\\
 \ell=1\quad&\text{or}\quad
 |\ell-k_{u-1}|\le2h_n+o(h_n),
 &r=n\quad&\text{or}\quad
 |r-k_{v+1}|\le2h_n+o(h_n).
 \label{eq:wbs-invariant-boundaries}
\end{align}
The initial segment \([1,n]\) satisfies both relations.

Consider an active segment with \(\mathcal K(S)\ne\varnothing\).
By \eqref{eq:wbs-invariant-good-intervals}, it contains a significant
interval \(I_j^\circ\); hence the call cannot stop.  Let \(I^*\) be the
shortest significant interval selected by the algorithm.  Since every good
interval satisfies \(|I_j^\circ|\le\delta_Wn\),
\[
 |I^*|\le\delta_Wn.
\]
The spacing condition gives
\[
 \min_l(k_{l+1}-k_l)\ge\delta_0n-O(1),
 \qquad
 \delta_W<\delta_0/2,
\]
so \(I^*\) contains at most one true boundary.  It contains at least one,
because otherwise the homogeneous-pool bound would imply
\(\mathcal T(I^*)<\mathfrak t_n^{\rm WBS}\).  Thus
\[
 I^*\cap\{k_1,\ldots,k_q\}=\{k_j\}
\]
for exactly one unselected \(k_j\).  With
\(\widetilde k=\widehat k_{\calR}(I^*)\),
\[
 |\widetilde k-k_j|\le|I^*|\le\delta_Wn.
\]

From \eqref{eq:wbs-invariant-boundaries} and the spacing of true changes,
\[
 \min\{k_j-\ell+1,r-k_j\}
 \ge\delta_0n-2h_n-o(h_n).
\]
Since \(h_n=o(n)\) and
\(\delta_R+\delta_W+\delta_B<\delta_0/2\),
\[
 \delta_0n-2h_n-o(h_n)
 \ge(\delta_R+\delta_W+\delta_B)n
\]
for all large \(n\).  Therefore
\eqref{eq:wbs-current-margin} holds, and
Lemma~\ref{lem:wbs-refinement-geometry} gives
\[
 B^*=B(\widetilde k;\ell,r)\in\mathfrak B_j,
 \qquad
 B^*\cap\{k_1,\ldots,k_q\}=\{k_j\}.
\]

Let
\[
 (\widehat a,\widehat k,\widehat c,\widehat\rho)
 \in\argmax_{\substack{(a,b,c)\in\mathcal A(B^*,\varepsilon)\\
                       \rho\in\calR_K}}
 Z_{\rho,B^*}(a,b,c).
\]
Choose a population maximizer
\[
 (a^o,k^o,c^o)
 \in\argmax_{(a,b,c)\in\mathcal A(B^*,\varepsilon)}
 D_{B^*}(a,b,c),
 \qquad
 \rho^o\in\argmax_{\rho\in\calR_K}
 D_{\rho,B^*}(a^o,k^o,c^o).
\]
The empirical maximizing property gives
\[
 Z_{\widehat\rho,B^*}(\widehat a,\widehat k,\widehat c)
 \ge Z_{\rho^o,B^*}(a^o,k^o,c^o).
\]
Using \eqref{eq:wbs-uniform-expansion} at these two candidate--ridge pairs,
\begin{align*}
 0
 &\le D_{B^*}(a^o,k^o,c^o)
      -D_{B^*}(\widehat a,\widehat k,\widehat c)\\
 &\le D_{\rho^o,B^*}(a^o,k^o,c^o)
      -D_{\widehat\rho,B^*}(\widehat a,\widehat k,\widehat c)\\
 &\le
 |Z_{\rho^o,B^*}^{(0)}(a^o,k^o,c^o)|
 +|Z_{\widehat\rho,B^*}^{(0)}(\widehat a,\widehat k,\widehat c)|\\
 &\quad+|R_{\rho^o,B^*}(a^o,k^o,c^o)|
 +|R_{\widehat\rho,B^*}(\widehat a,\widehat k,\widehat c)|\\
 &\le C\mathfrak u_n^{\rm WBS}
 =o(\mathfrak t_n^{\rm WBS}).
\end{align*}
The second inequality uses
\[
 D_{B^*}(\widehat a,\widehat k,\widehat c)
 =\max_{\rho\in\calR_K}
 D_{\rho,B^*}(\widehat a,\widehat k,\widehat c)
 \ge D_{\widehat\rho,B^*}(\widehat a,\widehat k,\widehat c).
\]
Since \(\mathfrak u_n^{\rm WBS}=o(\mathfrak t_n^{\rm WBS})\) and
\(\mathfrak t_n^{\rm WBS}=o(\mathfrak s_{n,\min}^{\rm WBS})\),
Lemma~\ref{lem:wbs-population-separation} yields
\begin{equation}
\label{eq:wbs-proof-localization}
 \frac{|\widehat k-k_j|}{n}
 \le C\frac{\mathfrak t_n^{\rm WBS}}
              {\mathfrak s_{n,\min}^{\rm WBS}}.
\end{equation}
By Assumption~\ref{ass:wbs-signal-tuning},
\[
 |\widehat k-k_j|=o(h_n).
\]

The deletion band satisfies, for all large \(n\),
\[
 k_j\in[\widehat k-h_n,\widehat k+h_n].
\]
Moreover,
\[
 2h_n+2|\widehat k-k_j|+1=o(n)
 <\min_l(k_{l+1}-k_l),
\]
so this band contains no other true change.  The new recursive endpoints
obey
\begin{align*}
 |(\widehat k-h_n)-k_j|
 &\le h_n+|\widehat k-k_j|,\\
 |(\widehat k+h_n+1)-k_j|
 &\le h_n+|\widehat k-k_j|+1
 \le2h_n+o(h_n).
\end{align*}
Thus \eqref{eq:wbs-invariant-boundaries} is preserved.

For every remaining change \(k_l\) to the right of \(k_j\),
\eqref{eq:wbs-good-buffer} gives
\[
 \inf I_l^\circ-k_j\ge\delta_Bn.
\]
Hence
\[
 \inf I_l^\circ-(\widehat k+h_n+1)
 \ge\delta_Bn-h_n-|\widehat k-k_j|-1>0.
\]
Similarly, for every remaining change to the left,
\[
 (\widehat k-h_n)-\sup I_l^\circ
 \ge\delta_Bn-h_n-|\widehat k-k_j|>0.
\]
Therefore each remaining good interval is contained in the corresponding
child segment, so \eqref{eq:wbs-invariant-good-intervals} is also preserved.
Each child with at least one remaining change contains a good interval of
length at least \(m_{\min}\) and therefore is not removed by the minimum-length
stopping rule.

Every successful call consequently selects exactly one previously
unselected true change, deletes no other true change, and preserves the
induction invariant for all remaining changes.  After exactly \(q\)
successful calls,
\[
 \mathcal K(S)=\varnothing
\]
in every active child; all such homogeneous children stop by the null-pool
bound.  Thus
\[
 \widehat q=q
\]
on \(\mathcal E_n^{\rm WBS}\) for all sufficiently large \(n\).  Applying
\eqref{eq:wbs-proof-localization} to the \(q\) selected changes gives
\[
 \max_{1\le j\le q}\frac{|\widehat k_j-k_j|}{n}
 \le C\frac{\mathfrak t_n^{\rm WBS}}
              {\mathfrak s_{n,\min}^{\rm WBS}}
 \le C'\frac{\mathfrak t_n^{\rm WBS}}
              {\underline{\mathfrak s}_n^{\rm WBS}}.
\]
Because the right-hand side is \(o(h_n/n)=o(1)\), the estimated and true
changes have the same increasing order for all large \(n\).  Finally,
\(\Pbb(\mathcal E_n^{\rm WBS})\to1\), which proves the theorem.
\end{proof}


\begin{thebibliography}{99}

\bibitem[Baranowski et al.(2019)]{baranowski2019not}
Baranowski, R., Chen, Y., and Fryzlewicz, P. (2019).
\newblock Narrowest-over-threshold detection of multiple change points and change-point-like features.
\newblock \emph{Journal of the Royal Statistical Society: Series B} 81, 649--672.

\bibitem[Boucheron et al.(2013)]{boucheron2013}
Boucheron, S., Lugosi, G., and Massart, P. (2013).
\newblock \emph{Concentration Inequalities: A Nonasymptotic Theory of Independence}.
\newblock Oxford University Press, Oxford.

\bibitem[Chen et al.(2011)]{chen2011rht}
Chen, L. S., Paul, D., Prentice, R. L., and Wang, P. (2011).
\newblock A regularized Hotelling's $T^2$ test for pathway analysis in proteomic studies.
\newblock \emph{Journal of the American Statistical Association} 106, 1345--1360.

\bibitem[Chen et al.(2022)]{chen2022breakpoints}
Chen, L., Wang, W., and Wu, W. B. (2022).
\newblock Inference of breakpoints in high-dimensional time series.
\newblock \emph{Journal of the American Statistical Association} 117, 1951--1963.

\bibitem[Cho and Fryzlewicz(2015)]{cho2015sbs}
Cho, H. and Fryzlewicz, P. (2015).
\newblock Multiple-change-point detection for high dimensional time series via sparsified binary segmentation.
\newblock \emph{Journal of the Royal Statistical Society: Series B} 77, 475--507.

\bibitem[Collins and \'{S}niady(2006)]{collins2006}
Collins, B. and \'{S}niady, P. (2006).
\newblock Integration with respect to the Haar measure on unitary, orthogonal and symplectic group.
\newblock \emph{Communications in Mathematical Physics} 264, 773--795.

\bibitem[El Karoui(2009)]{elkaroui2009}
El Karoui, N. (2009).
\newblock Concentration of measure and spectra of random matrices: Applications to correlation matrices, elliptical distributions and beyond.
\newblock \emph{The Annals of Applied Probability} 19, 2362--2405.

\bibitem[Enikeeva and Harchaoui(2019)]{enikeeva2019}
Enikeeva, F. and Harchaoui, Z. (2019).
\newblock High-dimensional change-point detection under sparse alternatives.
\newblock \emph{The Annals of Statistics} 47, 2051--2079.

\bibitem[Fang et al.(1990)]{fang1990}
Fang, K.-T., Kotz, S., and Ng, K.-W. (1990).
\newblock \emph{Symmetric Multivariate and Related Distributions}.
\newblock Chapman \& Hall, London.

\bibitem[Fryzlewicz(2014)]{fryzlewicz2014}
Fryzlewicz, P. (2014).
\newblock Wild binary segmentation for multiple change-point detection.
\newblock \emph{The Annals of Statistics} 42, 2243--2281.

\bibitem[French(2026)]{frenchdata}
French, K. R. (2026).
\newblock 49 industry portfolios.
\newblock Kenneth R. French Data Library.
\newblock \url{https://mba.tuck.dartmouth.edu/pages/faculty/ken.french/data_library.html}.

\bibitem[James and Matteson(2015)]{james2014ecp}
James, N. A. and Matteson, D. S. (2015).
\newblock ecp: An R package for nonparametric multiple change point analysis of multivariate data.
\newblock \emph{Journal of Statistical Software} 62, 1--25.

\bibitem[Jiang et al.(2023)]{jiang2023robust}
Jiang, F., Wang, R., and Shao, X. (2023).
\newblock Robust inference for change points in high dimension.
\newblock \emph{Journal of Multivariate Analysis} 193, 105114.

\bibitem[Jirak(2015)]{jirak2015}
Jirak, M. (2015).
\newblock Uniform change point tests in high dimension.
\newblock \emph{The Annals of Statistics} 43, 2451--2483.

\bibitem[Li and Xu(2026)]{li2026change}
Li, H. and Xu, H. (2026).
\newblock Adaptable high-dimensional change point detection via ridge regularization.
\newblock arXiv preprint.

\bibitem[Li et al.(2020)]{li2020adaptable}
Li, H., Aue, A., Paul, D., Peng, J., and Wang, P. (2020).
\newblock An adaptable generalization of Hotelling's $T^2$ test in high dimension.
\newblock \emph{The Annals of Statistics} 48, 1815--1847.

\bibitem[Li et al.(2024)]{li2024twoway}
Li, J., Chen, L., Wang, W., and Wu, W. B. (2024).
\newblock $\ell_2$ inference for change points in high-dimensional time series via a Two-Way MOSUM.
\newblock \emph{The Annals of Statistics} 52, 602--627.

\bibitem[Lifshits(1984)]{lifshits1984}
Lifshits, M. A. (1984).
\newblock Absolute continuity of functionals of ``supremum'' type for Gaussian processes.
\newblock \emph{Journal of Soviet Mathematics} 27, 3103--3112.

\bibitem[Liu et al.(2020)]{liu2020unified}
Liu, B., Zhou, C., Zhang, X., and Liu, Y. (2020).
\newblock A unified data-adaptive framework for high dimensional change point detection.
\newblock \emph{Journal of the Royal Statistical Society: Series B} 82, 933--963.

\bibitem[Liu et al.(2022)]{liu2022review}
Liu, B., Zhang, X., and Liu, Y. (2022).
\newblock High dimensional change point inference: Recent developments and extensions.
\newblock \emph{Journal of Multivariate Analysis} 188, 104833.

\bibitem[Liu et al.(2025)]{liu2025sscpd}
Liu, J., Feng, L., Peng, L., and Wang, Z. (2025).
\newblock Spatial-sign based high dimensional change point inference.
\newblock arXiv:2504.19306.

\bibitem[Meckes(2019)]{meckes2019}
Meckes, E. S. (2019).
\newblock \emph{The Random Matrix Theory of the Classical Compact Groups}.
\newblock Cambridge University Press, Cambridge.

\bibitem[Matteson and James(2014)]{matteson2014nonparametric}
Matteson, D. S. and James, N. A. (2014).
\newblock A nonparametric approach for multiple change point analysis of multivariate data.
\newblock \emph{Journal of the American Statistical Association} 109, 334--345.

\bibitem[O'Donnell(2014)]{odonnell2014}
O'Donnell, R. (2014).
\newblock \emph{Analysis of Boolean Functions}.
\newblock Cambridge University Press, Cambridge.

\bibitem[Pinelis(1994)]{pinelis1994}
Pinelis, I. (1994).
\newblock Optimum bounds for the distributions of martingales in Banach spaces.
\newblock \emph{The Annals of Probability} 22, 1679--1706.

\bibitem[Shu et al.(2022)]{shu2022spatialrank}
Shu, L., Chen, Y., Zhang, W., and Wang, X. (2022).
\newblock Spatial rank-based high-dimensional change point detection via random integration.
\newblock \emph{Journal of Multivariate Analysis} 189, 104942.

\bibitem[Tropp(2012)]{tropp2012}
Tropp, J. A. (2012).
\newblock User-friendly tail bounds for sums of random matrices.
\newblock \emph{Foundations of Computational Mathematics} 12, 389--434.

\bibitem[Wang and Feng(2023)]{wangfeng2023}
Wang, G. and Feng, L. (2023).
\newblock Computationally efficient and data-adaptive changepoint inference in high dimension.
\newblock \emph{Journal of the Royal Statistical Society: Series B} 85, 936--958.

\bibitem[Wang and Samworth(2018)]{wangsamworth2018}
Wang, T. and Samworth, R. J. (2018).
\newblock High dimensional change point estimation via sparse projection.
\newblock \emph{Journal of the Royal Statistical Society: Series B} 80, 57--83.

\bibitem[Wang and Shao(2023)]{wangshao2023dating}
Wang, R. and Shao, X. (2023).
\newblock Dating the break in high-dimensional data.
\newblock \emph{Bernoulli} 29, 2879--2901.

\bibitem[Wang et al.(2019)]{wang2019multiple}
Wang, Y., Zou, C., Wang, Z., and Yin, G. (2019).
\newblock Multiple change-points detection in high dimension.
\newblock \emph{Random Matrices: Theory and Applications} 8, 1950014.

\bibitem[Wang et al.(2022)]{wang2022self}
Wang, R., Zhu, C., Volgushev, S., and Shao, X. (2022).
\newblock Inference for change points in high-dimensional data via self-normalization.
\newblock \emph{The Annals of Statistics} 50, 781--806.

\bibitem[Wen et al.(2024)]{wen2024activation}
Wen, M., Wang, G., Zou, C., and Wang, Z. (2024).
\newblock Activation discovery with FDR control: Application to fMRI data.
\newblock \emph{Statistica Sinica} 34, 1625--1647.

\bibitem[Yu and Chen(2021)]{yu2021finite}
Yu, M. and Chen, X. (2021).
\newblock Finite sample change point inference and identification for high-dimensional mean vectors.
\newblock \emph{Journal of the Royal Statistical Society: Series B} 83, 247--270.

\bibitem[Zhang and Lavitas(2018)]{zhang2018self}
Zhang, T. and Lavitas, L. (2018).
\newblock Unsupervised self-normalized change-point testing for time series.
\newblock \emph{Journal of the American Statistical Association} 113, 637--648.

\bibitem[Zhang et al.(2010)]{zhang2010simultaneous}
Zhang, N. R., Siegmund, D. O., Ji, H., and Li, J. Z. (2010).
\newblock Detecting simultaneous changepoints in multiple sequences.
\newblock \emph{Biometrika} 97, 631--645.

\bibitem[Zhang et al.(2022)]{zhang2022adaptive}
Zhang, Y., Wang, R., and Shao, X. (2022).
\newblock Adaptive inference for change points in high-dimensional data.
\newblock \emph{Journal of the American Statistical Association} 117, 1751--1762.

\bibitem[Zhao et al.(2026)]{zhao2026rhtcp}
Zhao, P., Zhou, L., and Feng, L. (2026).
\newblock Cauchy aggregation of ridge-regularized Hotelling tests for high-dimensional change-point detection.
\newblock Manuscript submitted to \emph{Random Matrices: Theory and Applications}.

\end{thebibliography}

\begin{thebibliography}{99}

\bibitem[Bai and Silverstein(2010)]{bai2010spectral}
Bai, Z. D. and Silverstein, J. W. (2010).
\newblock \emph{Spectral Analysis of Large Dimensional Random Matrices}, 2nd ed.
\newblock Springer, New York.

\bibitem[de Jong(1987)]{dejong1987}
de Jong, P. (1987).
\newblock A central limit theorem for generalized quadratic forms.
\newblock \emph{Probability Theory and Related Fields} 75, 261--277.

\bibitem[Hachem et al.(2007)]{hachem2007deterministic}
Hachem, W., Loubaton, P., and Najim, J. (2007).
\newblock Deterministic equivalents for certain functionals of large random matrices.
\newblock \emph{The Annals of Applied Probability} 17, 875--930.

\bibitem[Hachem et al.(2013)]{hachem2013bilinear}
Hachem, W., Loubaton, P., Najim, J., and Vallet, P. (2013).
\newblock On bilinear forms based on the resolvent of large random matrices.
\newblock \emph{Annales de l'Institut Henri Poincare, Probabilites et Statistiques} 49, 36--63.


\bibitem[Li and Xu(2022)]{li2022spatialmedian}
Li, W. and Xu, Y. (2022).
\newblock Asymptotic properties of high-dimensional spatial median in elliptical distributions with application.
\newblock \emph{Journal of Multivariate Analysis} 190, 104975.

\bibitem[Li et al.(2022)]{li2022sscm}
Li, W., Wang, Q., Yao, J., and Zhou, W. (2022).
\newblock On eigenvalues of a high-dimensional spatial-sign covariance matrix.
\newblock \emph{Bernoulli} 28, 606--637.

\bibitem[Magyar and Tyler(2011)]{magyar2011}
Magyar, A. and Tyler, D. E. (2011).
\newblock The asymptotic efficiency of the spatial median for elliptically symmetric distributions.
\newblock \emph{Sankhya B} 73, 165--192.

\bibitem[Oja(2010)]{oja2010}
Oja, H. (2010).
\newblock \emph{Multivariate Nonparametric Methods with R: An Approach Based on Spatial Signs and Ranks}.
\newblock Springer, New York.

\bibitem[Silverstein and Choi(1995)]{silverstein1995}
Silverstein, J. W. and Choi, S. I. (1995).
\newblock Analysis of the limiting spectral distribution of large-dimensional random matrices.
\newblock \emph{Journal of Multivariate Analysis} 54, 295--309.

\end{thebibliography}
\end{document}